\newcommand{\krversion}[1]{}
\newcommand{\longversion}[1]{#1}
\newcommand{\krlongversion}[1]{#1}
\newcommand{\satversion}[1]{}
\newcommand{\squeezing}[1]{}
\newcommand{\nosqueezing}[1]{}
\krversion{
  \documentclass[letterpaper]{article}
  \usepackage{aaai}
  \usepackage{times}
  \usepackage{helvet}
  \usepackage{courier}
  \setlength{\pdfpagewidth}{8.5in}
  \setlength{\pdfpageheight}{11in}
  
  \setcounter{secnumdepth}{2}
  
  \setlength\titlebox{1.7in}
  \usepackage{enumitem}
  \setlist{nosep}
  \renewcommand{\paragraph}[1]{\vspace{1pt}\noindent\textbf{#1}~}

  \usepackage{amsthm}
  \usepackage{array}
  \usepackage{booktabs}
  \usepackage{enumerate}
  \newcommand{\href}[2]{#1}
  
  \usepackage[utopia]{mathdesign}

  \renewenvironment{proof}{\vspace{-2mm}\begin{pf}}{\qed\end{pf}}
  
}
\longversion{
  \documentclass[10pt]{article}
  \usepackage{fullpage}
  \usepackage{hyperref}
  \usepackage{amsthm}
  \usepackage{array}
  \usepackage{booktabs}
  \usepackage{arydshln}
  \usepackage{enumerate}

\newcommand{\tagarray}{%
\mbox{}\refstepcounter{equation}%
$(\theequation)$%
}

  \usepackage[lined,boxed,commentsnumbered,noend]{algorithm2e}
}
\satversion{
  \documentclass{llncs}
  \newcommand{\citep}[1]{\cite{#1}}
  
  \pagestyle{empty}
  \pagenumbering{gobble}
  
  \usepackage[lined,boxruled,commentsnumbered,noend]{algorithm2e}
  \usepackage[small]{caption}
  \usepackage{subcaption}
  \usepackage{hyperref}
  \usepackage{booktabs}
  \usepackage{float}
  
  \usepackage[]{microtype}
  
  \DeclareRobustCommand{\DE}[3]{#2}

  \renewenvironment{proof}{\vspace{-3mm}\begin{pf}}{\qed\end{pf}}

  \squeezing{
  \setlength{\parskip}{0pt plus2pt minus0pt}
  
  }
  
  \renewcommand{\theequation}{\Roman{equation}}
  \newcommand{\inlineeqref}[1]{\refstepcounter{equation}(\theequation)\label{#1}}
  
}

\krlongversion{
  \pdfinfo{
  /Title (The Parameterized Complexity of Reasoning Problems Beyond NP)
  /Author (Ronald de Haan, Stefan Szeider)
  /Keywords (parameterized complexity, answer set programming, Polynomial Hierarchy, SAT, fpt-reductions)
  }
  
  \newtheorem{theorem}{Theorem}
  \newtheorem{lemma}[theorem]{Lemma}
  \newtheorem{corollary}[theorem]{Corollary}
  \newtheorem{proposition}[theorem]{Proposition}

  \DeclareRobustCommand{\DE}[3]{#2}
}

\usepackage{color}

\usepackage[usenames,dvipsnames]{xcolor}

\usepackage{tikz}
\usetikzlibrary{arrows,backgrounds,fit}
\usepackage{amssymb}
\usepackage{amsmath}
\usepackage{mathtools}

\usepackage{microtype}

\krversion{
  \newcommand{\citep}[1]{(\citeauthor{#1}, \citeyear{#1})}
  \newcommand{\citet}[1]{\citeauthor{#1} (\citeyear{#1})}
}

\makeatletter
\DeclareRobustCommand{\rvdots}{%
  \vbox{
    \baselineskip4\p@\lineskiplimit\z@
    \kern-\p@
    \hbox{.}\hbox{.}\hbox{.}
  }}
\makeatother

\newcommand{\SigmaP}[1]{\ensuremath{\Sigma^{\mtext{P}}_{#1}}}
\newcommand{\PiP}[1]{\ensuremath{\Pi^{\mtext{P}}_{#1}}}
\newcommand{\DeltaP}[1]{\ensuremath{\Delta^{\mtext{P}}_{#1}}}
\newcommand{\PTIME}{\ensuremath{\mtext{P}}}
\newcommand{\NP}{\ensuremath{\mtext{NP}}}
\newcommand{\DP}{\ensuremath{\mtext{DP}}}
\newcommand{\XP}{\ensuremath{\mtext{XP}}}
\newcommand{\X}{\ensuremath{\mtext{X}}}
\newcommand{\XNP}{\ensuremath{\mtext{XNP}}}
\newcommand{\XcoNP}{\ensuremath{\mtext{Xco-NP}}}
\newcommand{\PSPACE}{\ensuremath{\mtext{PSPACE}}}
\newcommand{\para}[1]{\ensuremath{\mtext{para-#1}}}
\newcommand{\co}[1]{\ensuremath{\mtext{co-#1}}}
\newcommand{\W}[1]{\ensuremath{\mtext{W[#1]}}}

\newcommand{\EkA}{\ensuremath{\exists^k\forall^{*}}}
\newcommand{\AkE}{\ensuremath{\forall^k\exists^{*}}}

\newcommand{\EkAkW}[1]{\ensuremath{\exists^{k}\forall^k\mtext{-W[#1]}}}
\newcommand{\EkAW}[1]{\ensuremath{\exists^{k}\forall^{*}\mtext{-W[#1]}}}
\newcommand{\EAkW}[1]{\ensuremath{\exists^{*}\forall^k\mtext{-W[#1]}}}
\newcommand{\AEkW}[1]{\ensuremath{\forall^{*}\exists^k\mtext{-W[#1]}}}

\newcommand{\FPT}{\ensuremath{\mtext{FPT}}}
\newcommand{\SAT}{\ensuremath{\mtext{SAT}}}
\newcommand{\UNSAT}{\ensuremath{\mtext{UNSAT}}}
\newcommand{\thSAT}{\ensuremath{\mtext{\sc Sat}}}
\newcommand{\thUNSAT}{\ensuremath{\mtext{\sc Unsat}}}

\newcommand{\CNF}{\ensuremath{\mtext{CNF}}}
\newcommand{\DNF}{\ensuremath{\mtext{DNF}}}
\newcommand{\size}[1]{\ensuremath{|\!|#1|\!|}}
\newcommand{\Card}[1]{\ensuremath{|#1|}}
\newcommand{\Var}[1]{\ensuremath{\mtext{Var}(#1)}}
\newcommand{\Lit}[1]{\ensuremath{\mtext{Lit}(#1)}}
\newcommand{\Sub}[1]{\ensuremath{\mtext{Sub}(#1)}}

\newcommand{\Dom}[1]{\ensuremath{\mtext{Dom}(#1)}}
\newcommand{\Rng}[1]{\ensuremath{\mtext{Rng}(#1)}}
\newcommand{\NF}[1]{\ensuremath{\mtext{NF}(#1)}}
\newcommand{\Atoms}[1]{\ensuremath{\mtext{At}(#1)}}
\newcommand{\Free}[1]{\ensuremath{\mtext{Free}(#1)}}
\newcommand{\Nodes}[1]{\ensuremath{\mtext{Nodes}(#1)}}
\newcommand{\Repr}[1]{\ensuremath{\mtext{Repr}(#1)}}
\newcommand{\arity}[1]{\ensuremath{\mtext{arity}(#1)}}
\newcommand{\WSat}[1]{\ensuremath{p\mtext{-}\textsc{WSat}[#1]}}
\newcommand{\QSat}[1]{\ensuremath{\textsc{QSat}_{#1}}}
\newcommand{\fptclosure}[1]{\ensuremath{[\ #1\ ]^{\mtext{fpt}}}}
\newcommand{\fptred}[0]{\ensuremath{\leq_{\mtext{fpt}}}}
\newcommand{\fptequiv}[0]{\ensuremath{\equiv^{\mtext{fpt}}}}

\newcommand{\Comp}[1]{\ensuremath{\mtext{Comp}(#1)}}
\newcommand{\CompA}[2]{\ensuremath{\mtext{Comp}_{\!#1}(#2)}}
\newcommand{\Cont}[1]{\ensuremath{\mtext{Cont}(#1)}}

\newcommand{\kstar}[0]{\ensuremath{k\mtext{-}*}}
\newcommand{\stark}[0]{\ensuremath{*\mtext{-}k}}

\newcommand{\EkAWSat}{\ensuremath{\exists^k \forall^{*} {\mtext{-{\sc WSat}}}}}
\newcommand{\AkEWSat}{\ensuremath{\forall^k \exists^{*} {\mtext{-{\sc WSat}}}}}
\newcommand{\EkAMC}{\ensuremath{\exists^k \forall^{*} {\mtext{-{\sc MC}}}}}
\newcommand{\EAkWSat}{\ensuremath{\exists^{*} \forall^{k} {\mtext{-{\sc WSat}}}}}
\newcommand{\EkAkWSat}{\ensuremath{\exists^{k} \forall^{k} {\mtext{-{\sc WSat}}}}}
\newcommand{\AEkWSat}{\ensuremath{\forall^{*} \exists^{k} {\mtext{-{\sc WSat}}}}}
\newcommand{\AkEMC}{\ensuremath{\forall^k \exists^{*} {\mtext{-{\sc MC}}}}}

\newcommand{\EleqkAWSat}{\ensuremath{\exists^{\leq k} \forall^{*} {\mtext{-{\sc WSat}}}}}
\newcommand{\EgeqkAWSat}{\ensuremath{\exists^{\geq k} \forall^{*} {\mtext{-{\sc WSat}}}}}

\newcommand{\EnminkAWSat}{\ensuremath{\exists^{n-k} \forall^{*} {\mtext{-{\sc WSat}}}}}

\newcommand{\AleqkEWSat}{\ensuremath{\forall^{\leq k} \exists^{*} {\mtext{-{\sc WSat}}}}}

\newcommand{\SSIC}{\ensuremath{\mtext{{\sc Shortest}-}\allowbreak{}\mtext{{\sc Implicant}-}\allowbreak{}\mtext{{\sc Core}}(\mtext{core size})}}
\newcommand{\SIC}{\ensuremath{\mtext{{\sc Shortest}-}\allowbreak{}\mtext{{\sc Implicant}-}\allowbreak{}\mtext{{\sc Core}}}}
\newcommand{\LSIC}{\ensuremath{\mtext{{\sc Shortest}-}\allowbreak{}\mtext{{\sc Implicant}-}\allowbreak{}\mtext{{\sc Core}}(\mtext{reduction size})}}
\newcommand{\LargeMinDNF}{\ensuremath{\mtext{{\sc DNF}-}\allowbreak{}\mtext{{\sc Minimization}}(\mtext{reduction size})}}
\newcommand{\SmallMinDNF}{\ensuremath{\mtext{{\sc DNF}-}\allowbreak{}\mtext{{\sc Minimization}}(\mtext{core size})}}
\newcommand{\MinDNF}{\ensuremath{\mtext{{\sc DNF}-}\allowbreak{}\mtext{{\sc Minimization}}}}

\newcommand{\EASat}{\ensuremath{\exists \forall {\mtext{-{\sc Sat}}}}}
\newcommand{\EAtwSat}{\ensuremath{{\mtext{{\sc QSat}}}_{2}(\mtext{incid.tw})}}
\newcommand{\EAtwESat}{\ensuremath{{\mtext{{\sc QSat}}}_{2}(\mtext{$\exists$-incid.tw})}}
\newcommand{\EAtwASat}{\ensuremath{{\mtext{{\sc QSat}}}_{2}(\mtext{$\forall$-incid.tw})}}

\newcommand{\FewAQBFSat}{\ensuremath{\mtext{{\sc QBF}}{\mtext{-{\sc Sat}}(\mtext{\# $\forall$-vars})}}}

\newcommand{\SmallCliqueExt}{\ensuremath{\mtext{{\sc Small}-}\allowbreak{}\mtext{{\sc Clique}-}\allowbreak{}\mtext{{\sc Extension}}}}

\newcommand{\ASPcons}{\ensuremath{\mtext{{\sc ASP}-}\allowbreak{}\mtext{\sc consistency}}}
\newcommand{\DisjRules}{\ensuremath{(\#\mtext{disj.rules})}}
\newcommand{\ContRules}{\ensuremath{(\#\mtext{cont.rules})}}

\newcommand{\ContAtoms}{\ensuremath{(\#\mtext{cont.atoms})}}
\newcommand{\SNbd}{\ensuremath{(\mtext{str.norm.bd-size})}}
\newcommand{\AtomOcc}{\ensuremath{(\mtext{max.atom.occ.})}}
\newcommand{\NonDualNormalRules}{\ensuremath{(\#\mtext{non-dual-normal.rules})}}

\newcommand{\aspnot}[0]{\ensuremath{\mathit{not}\mtext{ }}}
\newcommand{\RobustCSPSat}{\ensuremath{\mtext{{\sc Robust}-}\allowbreak{}\mtext{{\sc CSP}-}\allowbreak{}\mtext{{\sc SAT}}}}

\newcommand{\EkATMhalt}[1]{\ensuremath{\mtext{\EkA{}-}\allowbreak{}\mtext{{\sc TM}-}\allowbreak{}\mtext{{\sc halt}}^{#1}}}

\newcommand{\AkEFDform}[2]{\ensuremath \forall^{k}\exists^{*} {\mtext{-{\sc FD}}^{#1}_{#2}}}
\newcommand{\AEkFDform}[2]{\ensuremath \forall^{*}\exists^{k} {\mtext{-{\sc FD}}^{#1}_{#2}}}

\newcommand{\BMCwitness}{\ensuremath{\mtext{{\sc Bounded}-}\allowbreak{}\mtext{{\sc Model}-}\allowbreak{}\mtext{{\sc Checking-}}\allowbreak{}\mtext{{\sc Witness}}}}

\newcommand{\ThreeColExt}{\ensuremath{3\mtext{-{\sc Coloring}-}\allowbreak{}\mtext{{\sc Extension}}\allowbreak{}}}

\newcommand{\Degree}{\ensuremath{(\mtext{degree})}}
\newcommand{\NumLeaves}{\ensuremath{(\# \mtext{leaves})}}
\newcommand{\NumColLeaves}{\ensuremath{(\# \mtext{col.leaves})}}
\newcommand{\NumUncolLeaves}{\ensuremath{(\# \mtext{uncol.leaves})}}

\newcommand{\SB}{\{\,}%
\newcommand{\SM}{\;{:}\;}%
\newcommand{\SE}{\,\}}%
\newcommand{\SBs}{\{}%
\newcommand{\SEs}{\}}%

\newcommand{\AAA}{\mathcal{A}}

\newcommand{\CCC}{\mathcal{C}}

 \newcommand{\TTT}{\mathcal{T}}

\newcommand{\mtext}[1]{\text{\normalfont #1}}

\newcommand{\hy}{\hbox{-}\nobreak\hskip0pt}

\newcommand{\EA}[0]{\ensuremath{\exists\forall}}

\newcommand{\itw}{\mtext{incid.tw}}
\newcommand{\Eitw}{\mtext{$\exists$-incid.tw}}
\newcommand{\Aitw}{\mtext{$\forall$-incid.tw}}

\krversion{
  \newenvironment{myquote}{\list{}{\leftmargin=5pt\rightmargin=-10pt\topsep=3pt}\item[]}{\endlist}
}
\longversion{
  \newenvironment{myquote}{\list{}{\leftmargin=\parindent\rightmargin=0in\topsep=3pt}\item[]}{\endlist}
}
\satversion{
  \newenvironment{myquote}{\list{}{\leftmargin=\parindent\rightmargin=0in\topsep=3pt}\item[]}{\endlist}
}

\krversion{
  \newcommand{\probdef}[1]{
    \begin{myquote}
    \framebox[\linewidth-15pt][l]{\parbox{\linewidth-25pt}{
    #1
   }}\end{myquote}
  }
}
\longversion{
  \newcommand{\probdef}[1]{
    \begin{myquote}
    \framebox[\textwidth-35pt][l]{\parbox{\textwidth-45pt}{
    #1
    }}\end{myquote}\vspace{5pt}
  }
}
\satversion{
  \newcommand{\probdef}[1]{
    \begin{myquote}
    \framebox[\linewidth-15pt][l]{\parbox{\linewidth-25pt}{
    #1
   }}\end{myquote}
  }
}

\satversion{
  \newcommand{\boxedprob}[2]{
    \probdef{#1}
  }
}
\longversion{
  \newcommand{\boxedprob}[2]{
    \begin{myquote}
    {\framebox[\textwidth-35pt][l]{\parbox{\textwidth-45pt}{
    #1
    }}\\[-1.2pt]\framebox[\textwidth-35pt][l]{\parbox{\textwidth-45pt}{
    #2
    }}}\end{myquote}\vspace{5pt}
  }
}

\krversion{
  \title{The Parameterized Complexity of Reasoning Problems Beyond NP}
  \author{Ronald de Haan\thanks{Supported by the European Research Council (ERC), project 239962 (COMPLEX REASON),
    and the Austrian Science Fund (FWF), project P26200 (Parameterized Compilation).} \and Stefan Szeider$^{*}$\\
  Institute of Information Systems\\
  Vienna University of Technology\\
  dehaan@kr.tuwien.ac.at\ \ \ stefan@szeider.net\\}
}
\longversion{
  \title{The Parameterized Complexity of Reasoning Problems Beyond
  NP\footnote{This paper contains results that have
  appeared in shortened and preliminary form in the proceedings of
  SAT 2014~\cite{DeHaanSzeider14b}, the proceedings of
  KR 2014~\cite{DeHaanSzeider14}, and the proceedings
  of SOFSEM 2015~\cite{DeHaanSzeider15}.}}
  \author{Ronald de Haan\thanks{Supported by the European Research Council (ERC), project 239962 (COMPLEX REASON),
    and the Austrian Science Fund (FWF), project P26200 (Parameterized Compilation).}, Stefan Szeider$^{\dagger}$\\[.5em]
    \emph{Algorithms \& Complexity Group}\\
    \emph{Technische Universit\"{a}t Wien}\\[.5em]
    dehaan@ac.tuwien.ac.at\ \ \ stefan@szeider.net\\}
  \date{}
}
\satversion{
  \title{Fixed-Parameter Tractable Reductions to SAT}
  \author{
    Ronald de Haan\thanks{Supported by the European Research Council (ERC), project 239962 (COMPLEX REASON),
    and the Austrian Science Fund (FWF), project P26200 (Parameterized Compilation).} \and
    Stefan Szeider$^{\star}$
  }
  \institute{
    Institute of Information Systems,\\
    Vienna University of Technology,\\
    Vienna, Austria
  }
}

\newcommand{\ProofSSICCompleteness}[0]{
\begin{proof}[sketch]
To show hardness, we give
a many-one fpt-reduction from $\EkAWSat(\DNF)$
to \SSIC{}.
Intuitively, the choice for some~$C' \subseteq C$ with~$\Card{C'} = k$
corresponds directly to the choice of some
assignment~$\alpha : X \rightarrow \SBs 0,1 \SEs$ of weight~$k$.
Both involve a choice between~$\binom{n}{k}$ many candidates,
and in both cases verifying whether the chosen candidate witnesses
that the instance is a yes-instance involves solving a \co{}\NP{}-complete
problem.
Any implicant~$C'$ forces those variables~$x$ that are included in~$C'$
to be set to true (and the other variables are not forced to take any truth value).
However, by Lemma~\ref{lem:exactly-k},
any assignment that sets more than~$k$ variables~$x$ to true
will trivially satisfy~$\psi$. Therefore, the only relevant assignment
is the assignment that sets only those~$x$ to true that are forced
to be true by some~$C'$ of length~$k$,
and hence the choice for such a~$C'$
corresponds exactly to the choice for some
assignment~$\alpha$ of weight~$k$.
To verify whether some~$C'$ of length~$k$ is an implicant of the
formula~$\varphi$ is equivalent to checking
whether the formula~$\bigwedge_{c \in C'}c \wedge \varphi$ is valid,
which in turn is equivalent to checking whether
a formula~$\forall Y. \psi[\alpha]$
is true, for some assignment~$\alpha$.

Let~$(\varphi,k)$ be an instance of~$\EkAWSat(\DNF)$,
with~$\varphi = \exists X. \forall Y. \psi$.
By Lemma~\ref{lem:exactly-k},
we may assume without loss of generality
that for any assignment~$\alpha : X \rightarrow \SBs 0,1 \SEs$ of
weight~$m \neq k$,~$\forall Y. \psi[\alpha]$ is true.
We may also assume without loss of generality that~$\Card{X} > k$;
if this were not the case,~$(\varphi,k)$ would trivially
be a no-instance.
We construct an instance~$(\varphi',C,k)$ of \SSIC{}
by letting~$\Var{\varphi'} = X \cup Y$,~$C =
\bigwedge_{x \in X} x$, and~$\varphi' = \psi$.
Clearly,~$\varphi'$ is a Boolean formula in \DNF{}.
Also, consider the assignment~$\alpha : X \rightarrow \SBs 0,1 \SEs$
where~$\alpha(x) = 1$ for all~$x \in X$.
We know that~$\forall Y. \psi[\alpha]$ is true, since~$\alpha$ has weight more than~$k$.
Therefore~$C$ is an implicant of~$\varphi'$.
We omit a detailed proof of correctness for this reduction.

To show membership in \EkA{},
we give a many-one fpt-reduction from \SSIC{}
to \EkAWSat{}.
This reduction uses exactly the same similarity between
the two problems, i.e., the fact that assignments of weight~$k$
correspond exactly to implicants of length~$k$,
and that verifying whether this choice witnesses that the instance
is a yes-instance in both cases involves checking validity of a
propositional formula.
We describe the reduction, and omit a detailed proof of correctness.
Let~$(\varphi,C,k)$ be an instance of \SSIC{},
where~$C = \SBs c_1,\dotsc,c_n \SEs$.
We construct an instance~$(\varphi',k)$ of \EkAWSat{},
where~$\varphi' = \exists X. \forall Y. \psi$,
by defining~$X = \SBs x_1,\dotsc,x_n \SEs$,~$Y =
\Var{\varphi}$,~$\psi = \psi^{X,Y}_{\mtext{corr}} \rightarrow \varphi$,
and~$\psi^{X,Y}_{\mtext{corr}} =
\bigwedge_{1 \leq i \leq n} (x_i \rightarrow c_i)$.
\end{proof}
}

\newcommand{\ProofSSICHardness}[0]{
\longversion{\begin{proof}}
We give \satversion{a many-one}\longversion{an} fpt-reduction
from~$\EkAWSat(\DNF)$
to \SSIC{}.
Let~$(\varphi,k)$ be an instance of~$\EkAWSat(\DNF)$,
with~$\varphi = \exists X. \forall Y. \psi$.
By Lemma~\ref{lem:exactly-k},
we may assume without loss of generality
that for any assignment~$\alpha : X \rightarrow \SBs 0,1 \SEs$ of
weight~$m \neq k$,~$\forall Y. \psi[\alpha]$ is true.
We may assume without loss of generality that~$\Card{X} > k$;
if this were not the case,~$(\varphi,k)$ would trivially
be a no-instance.

We construct an instance~$(\varphi',C,k)$ of \SSIC{}
by letting~$\Var{\varphi'} = X \cup Y$,~$C =
\bigwedge_{x \in X} x$, and~$\varphi' = \psi$.
Clearly,~$\varphi'$ is a Boolean formula in \DNF{}.
Also, consider the assignment~$\alpha : X \rightarrow \SBs 0,1 \SEs$
where~$\alpha(x) = 1$ for all~$x \in X$.
We know that~$\forall Y. \psi[\alpha]$ is true,
since~$\alpha$ has weight more than~$k$.
Therefore~$C$ is an implicant of~$\varphi'$.

\satversion{A full correctness proof of this reduction
can be found in the appendix. }%
Intuitively, the choice for some~$C' \subseteq C$ with~$\Card{C'} = k$
corresponds directly to the choice of some
assignment~$\alpha : X \rightarrow \SBs 0,1 \SEs$.
For any variable~$x$ included in the implicant~$C'$, the value of~$x$ will be
set to true.
Any variable~$x$ that is not in~$C'$ can be set either
to true or to false.
However, any assignment that sets more than~$k$ variables~$x$ to true
will trivially satisfy~$\psi$. Therefore, the only relevant assignment
is that assignment that sets only those~$x$ to true that are forced to
be true by~$C'$.
\longversion{

\SecondProofSSICHardness}
\end{proof}
}

\newcommand{\SecondProofSSICHardness}[0]{
We show that~$(\varphi,k) \in \EkAWSat(\DNF)$ if and only
if~$(\varphi',C,k) \in \SSIC{}$.

$(\Rightarrow)$
Let~$\alpha : X \rightarrow \SBs 0,1 \SEs$ be an assignment of weight~$k$
such that~$\forall Y. \psi[\alpha]$ is true.
Then construct~$C' = \SB x \in X \SM \alpha(x) = 1 \SE$.
Clearly,~$\Card{C'} = k$, and~$C' \subseteq C$.
We show that~$C'$ is an implicant of~$\varphi'$.
Let~$\gamma : \Var{\varphi'} \rightarrow \SBs 0,1 \SEs$ be an arbitrary assignment
that satisfies~$\bigwedge C'$.
We distinguish two cases:
either (i)~$\gamma$ satisfies exactly~$k$ different
variables~$x \in X$, or (ii) this is not the case.

In case (i), we know that~$\gamma(x) = \alpha(x)$ for all~$x \in X$.
Then, since~$\forall Y. \psi[\alpha]$ is true, we know that~$\psi[\gamma]$ is true,
and therefore~$\gamma$ satisfies~$\varphi'$.
In case (ii), we know that the restriction~$\gamma|_{X}$ of~$\gamma$ to~$X$
(i.e.,~$\gamma|_{X} : X \rightarrow \SBs 0,1 \SEs$ such
that~$\gamma|_{X}(x) = \gamma(x)$ for all~$x \in X$) has weight~$m \neq k$.
Therefore, we know that~$\forall Y. \psi[\gamma|_{X}]$ is true,
and thus~$\gamma$ satisfies~$\psi$.
Thus,~$\gamma$ satisfies~$\varphi'$.
This concludes our proof that~$C'$ is an implicant of~$\varphi'$.

$(\Leftarrow)$
Let~$C' \subseteq C$ be an implicant of~$\varphi'$ such that~$\Card{C'} = k$.
Define~$\alpha : X \rightarrow \SBs 0,1 \SEs$ by
letting~$\alpha(x) = 1$ if and only if~$x \in C'$, for all~$x \in X$.
Clearly,~$\alpha$ has weight~$k$.
We show that~$\forall Y. \psi[\alpha]$ is true.
Let~$\beta : Y \rightarrow \SBs 0,1 \SEs$ be an arbitrary assignment.
Now, clearly the assignment~$\alpha \cup \beta$
satisfies~$C'$. Since~$C'$ is an implicant of~$\varphi'$,
we know that~$\alpha \cup \beta$ satisfies~$\varphi' = \psi$.
This concludes our proof that~$\forall Y. \psi[\alpha]$ is true.
}

\newcommand{\ProofSSICMembership}[0]{
\begin{proof}
We give \satversion{a many-one}\longversion{an} fpt-reduction from \SSIC{}
to \EkAWSat{}.
Let~$(\varphi,C,k)$ be an instance of \SSIC{},
where~$C = \SBs c_1,\dotsc,c_n \SEs$.
We construct an instance~$(\varphi,k)$ of \EkAWSat{},
where~$\varphi = \exists X. \forall Y. \psi$.
We define~$X = \SBs x_1,\dotsc,x_n \SEs$,~$Y =
\Var{\varphi}$,~$\psi = \psi^{X,Y}_{\mtext{corr}} \rightarrow \varphi$,
and~$\psi^{X,Y}_{\mtext{corr}} =
\bigwedge_{1 \leq i \leq n} (x_i \rightarrow c_i)$.
\satversion{A detailed proof that~$(\varphi,C,k) \in \SSIC{}$ if and only
if~$(\varphi',k) \in \EkAWSat{}$ can be found in the appendix. }%
Intuitively, the choice for~$\alpha : X \rightarrow \SBs 0,1 \SEs$ of weight~$k$
corresponds to the choice for~$C' \subseteq C$ with~$\Card{C'} = k$.
\longversion{
\SecondProofSSICMembership}
\end{proof}
}

\newcommand{\SecondProofSSICMembership}[0]{
We show that~$(\varphi,C,k) \in \SSIC{}$
if and only if~$(\varphi',k) \in \EkAWSat{}$.

$(\Rightarrow)$
Assume that~$C' \subseteq C$ with~$\Card{C'} = k$ is an implicant of~$\varphi$.
Define the assignment~$\alpha : X \rightarrow \SBs 0,1 \SEs$ by
letting~$\alpha(x_i) = 1$ if and only if~$c_i \in C'$.
Clearly,~$\alpha$ has weight~$k$.
We show that~$\forall Y. \psi[\alpha]$ evaluates to true.
Let~$\beta : Y \rightarrow \SBs 0,1 \SEs$ be an arbitrary assignment.
We distinguish two cases:
either (i) there is some~$1 \leq i \leq n$ such that~$\alpha(x_i) = 1$
but~$\beta$ does not satisfy~$c_i$,
or (ii) this is not the case.
In case (i),~$\alpha \cup \beta$ does not satisfy~$\psi^{X,Y}_{\mtext{corr}}$,
and therefore~$\psi[\alpha \cup \beta]$ is true.
In case (ii),~$\beta$ satisfies~$C'$. Since~$C'$ is an
implicant of~$\varphi$,~$\beta$ also satisfies~$\varphi$.
Therefore~$\psi[\alpha \cup \beta]$ is true.

$(\Leftarrow)$
Assume that there is some assignment~$\alpha : X \rightarrow \SBs 0,1 \SEs$
of weight~$k$ such that~$\forall Y. \psi[\alpha]$ is true.
Construct~$C' = \SB c_i \SM 1 \leq i \leq n, \alpha(x_i) = 1 \SE$.
Clearly,~$C' \subseteq C$ and~$\Card{C'} = k$.
We show that~$C'$ is an implicant of~$\varphi$.
Let~$\beta : \Var{\varphi} \rightarrow \SBs 0,1 \SEs$ be any assignment
that satisfies~$C'$.
We show that~$\beta$ satisfies~$\varphi$.
Since~$Y = \Var{\varphi}$ and~$\forall Y. \psi[\alpha]$ is true,
we know that~$\psi[\alpha \cup \beta]$ is true.
Clearly,~$\alpha \cup \beta$ satisfies~$\psi^{X,Y}_{\mtext{corr}}$.
Therefore,~$\alpha \cup \beta$ satisfies~$\varphi$.
Since~$\Var{\varphi} \cap \Dom{\alpha} = \emptyset$,
we know that~$\beta$ satisfies~$\varphi$.
}

\newcommand{\ProofLargeMinDNFMembership}[0]{
We give \satversion{a many-one}\longversion{an}
fpt-reduction from \LargeMinDNF{}
to \EkAWSat{}.
Let~$(\varphi,k)$ be an instance of \LargeMinDNF{},
where~$\varphi = t_1 \vee \dotsm \vee t_m$,
and~$t_i = l^{1}_{i} \wedge \dotsm \wedge l^{\ell_i}_{i}$
for~$1 \leq i \leq m$.
We construct an instance~$(\varphi',k)$ of~\EkAWSat{}
as follows.
We let~$X = \SB x^j_i \SM 1 \leq i \leq m, 1 \leq j \leq \ell_{i} \SE$,
$Y = \Var{\varphi}$, and we define~$\varphi' = \exists X. \forall Y. \psi$,
where
$\psi = \bigvee_{1 \leq i \leq m}
[ (l^1_i \vee x^1_i) \wedge \dotsm \wedge
(l^{\ell_i}_{i} \vee x^{\ell_i}_{i}) ] \rightarrow
\bigvee_{1 \leq j \leq m} t_j$.
Intuitively, the variable~$x^j_i$ represents whether or not
the literal occurrence~$l^{j}_{i}$ has been removed from~$\varphi$
to form a suitable DNF formula~$\varphi''$.
\SecondProofLargeMinDNFMembership{}
}

\newcommand{\SecondProofLargeMinDNFMembership}[0]{
We show that~$(\varphi,k) \in \LargeMinDNF{}$ if and only
if~$(\varphi',k) \in \EkAWSat{}$.

$(\Rightarrow)$
Assume that there exists a DNF formula~$\varphi'' =
t'_1 \vee \dotsm \vee t'_m$ with~$n-k$ occurrences of literals
such that~$t'_i \subseteq t_i$ for all~$1 \leq i \leq m$,
and such that~$\varphi \equiv \varphi''$.
Then there are exactly~$k$
literals~$l_{a_1}^{b_1},\dotsc,l_{a_k}^{b_k}$
such that~$l_{a_u}^{b_u}$ does not occur in~$t_{a_u}$
for all~$1 \leq u \leq k$.
Let~$\alpha : X \rightarrow \SBs 0,1 \SEs$ be the assignment
defined by letting~$\alpha(x^j_i) = 1$ if and only
if~$i = a_u$ and~$j = b_u$ for some~$1 \leq u \leq k$.
Clearly,~$\alpha$ has weight~$k$.
We show that~$\forall Y. \psi[\alpha]$ is true.
Let~$\beta : Y \rightarrow \SBs 0,1 \SEs$ be an arbitrary assignment.
We consider two cases: either~$\alpha \cup \beta$ does not satisfy~$\varphi''$,
or~$\alpha \cup \beta$ does satisfy~$\varphi''$.

Consider the first case.
Then we know that for no~$1 \leq i \leq m$,
the assignment~$\alpha \cup \beta$ satisfies
the subformula~$(l^1_i \vee x^1_i) \wedge \dotsm \wedge
(l^{\ell_i}_i \vee x^{\ell_i}_i)$.
Therefore,~$\alpha \cup \beta$ satisfies~$\psi$.
Next, consider the second case.
Then, since~$\varphi \equiv \varphi''$, we know
that~$\alpha \cup \beta$ satisfies~$\varphi$,
and thus satisfies~$t_j$ for some~$1 \leq j \leq m$.
Therefore,~$\alpha \cup \beta$ satisfies~$\psi$.

$(\Leftarrow)$
Assume that there exists an assignment~$\alpha : X \rightarrow \SBs 0,1 \SEs$
of weight~$k$ such that~$\forall Y. \psi[\alpha]$ is true.
Construct the DNF formula~$\varphi''$ by removing the~$k$
many occurrences of literals~$l^j_i$ such that~$\alpha(x^j_i) = 1$.
We show that~$\varphi'' \equiv \varphi$.
Clearly,~$\varphi \models \varphi''$.
We show that~$\varphi'' \models \varphi$.
Let~$\beta : Y \rightarrow \SBs 0,1 \SEs$ be an arbitrary assignment,
and assume that~$\beta$ satisfies~$\varphi''$.
Then there exists some~$1 \leq i \leq m$ such that~$\beta$
satisfies~$t'_i$.
We need to show that there exists some~$1 \leq j \leq m$ such
that~$\beta$ satisfies~$t_j$.
We know that~$\alpha \cup \beta$ satisfies~$\varphi'$.
It is now straightforward to verify
that~$\alpha \cup \beta$ satisfies the
subformula~$(l^1_i \vee x^1_i) \wedge \dotsm \wedge
(l^{\ell_i}_{i} \vee x^{\ell_i}_{i})$.
Therefore,~$\alpha \cup \beta$ must satisfy~$t_j$
for some~$1 \leq j \leq m$.
Since~$t_j$ contains only variables in~$Y$,
we know that~$\beta$ satisfies~$t_j$.
This completes our proof that~$\varphi'' \equiv \varphi$.
}

\newcommand{\CompleteSmallMinDNF}[0]{
\longversion{
The following propositions give us
some first lower and upper bounds on
the complexity of \SmallMinDNF{}.
\begin{proposition}
\label{prop:smallmindnf-hardness}
\SmallMinDNF{} is \para{\co\NP}-hard.
\end{proposition}
\begin{proof}
\ProofSmallMinDNFHardness{}
\end{proof}
}
\longversion{\begin{proposition}
\label{prop:smallmindnf-membership1}
\SmallMinDNF{} is in \EkA{}.
\end{proposition}
\begin{proof}
\ProofSmallMinDNF{}
\end{proof}
}
\satversion{%
\smallskip\noindent The following result, which we give without proof,
gives some first upper and lower bounds on the complexity
of \SmallMinDNF{}.
\begin{proposition}
\label{prop:smallmindnf-results1}
\SmallMinDNF{} is \para{\co\NP}-hard
and is in~\EkA{}.
\end{proposition}
}

\noindent Next, we turn our attention to an fpt-algorithm
that solves \SmallMinDNF{} by using~$f(k)$ many SAT calls,
for some computable function~$f$.
In order to so, we will define the notion of relevant variables,
and establish several lemmas that help us to describe and analyze
the algorithm
(the first of which we state without proof).

Let~$\varphi$ be a DNF formula and let~$x \in \Var{\varphi}$
be a variable occurring in~$\varphi$.
We call~$x$ \emph{relevant in~$\varphi$} if
there exists some
assignment~$\alpha : \Var{\varphi} \backslash \SBs x \SEs
\rightarrow \SBs 0,1 \SEs$
such that~$\varphi[\alpha \cup \SBs x \mapsto 0 \SEs] \neq
\varphi[\alpha \cup \SBs x \mapsto 1 \SEs]$.

\begin{lemma}
\label{lem:relevant}
Let~$\varphi$ be a DNF formula and let~$\varphi'$ be a
DNF formula of minimal size that
it is equivalent to~$\varphi$.
Then for every variable~$x \in \Var{\varphi}$ it holds
that~$x \in \Var{\varphi'}$ if and only if~$x$ is
relevant in~$\varphi$.
\end{lemma}
\longversion{%
\begin{proof}
\ProofRelevantLemma{}
\end{proof}
}

\begin{lemma}
\label{lem:relevant-vars-sat}
Given a DNF formula~$\varphi$ and a positive integer~$m$ (given in unary),
deciding whether there are at least~$m$ variables
that are relevant in~$\varphi$ is in NP.
\end{lemma}
\begin{proof}
We describe a guess-and-check algorithm that decides the problem.
The algorithm first guesses~$m$ distinct variables occurring in~$\varphi$,
and for each guessed variable~$x$ the algorithm guesses an
assignment~$\alpha_x$ to the remaining
variables~$\Var{\varphi} \backslash \SBs x \SEs$.
Then, the algorithm verifies whether the guessed variables are really relevant
by checking that, under~$\alpha_x$, assigning different values to~$x$
changes the outcome of the Boolean function represented by~$\varphi$,
i.e.,~$\varphi[\alpha_x \cup \SBs x \mapsto 0 \SEs] \neq 
\varphi[\alpha_x \cup \SBs x \mapsto 1 \SEs]$.
It is straightforward to construct a SAT instance~$\psi$ that
implements this guess-and-check procedure.
Moreover, from any assignment that satisfies~$\psi$
it is easy to extract the relevant variables.
\end{proof}

\begin{lemma}
\label{lem:combinatorics-dnf-formulas}
Let~$x_1,\dotsc,x_k$ be propositional variables.
There are~$2^{O(k \log k)}$ many different DNF
formulas~$\psi$ over the variables~$x_1,\dotsc,x_k$
that are of size~$k$.
\end{lemma}
\begin{proof}
Each suitable DNF formula~$\psi = t_1 \vee \dotsm \vee t_{\ell}$
can be formed by writing down a
sequence~$\sigma = (l_1,\dotsc,l_k)$
of literals~$l_i$ over~$x_1,\dotsc,x_k$,
and splitting this sequence into terms,
i.e., choosing integers~$1 = d_1 < \dotsm < d_{\ell+1} = k+1$
such that~$t_i = \SBs l_{d_i}, \dotsc, l_{d_{i+1}-1} \SEs$
for each~$1 \leq i \leq \ell$.
To see that there are~$2^{O(k \log k)}$ many
formulas~$\psi$, it suffices to see
that there are~$O(k^k)$ many sequences~$\sigma$,
and~$O(2^k)$ many choices for the integers~$d_i$.
\end{proof}

\begin{figure}[h]
  \begin{algorithm}[H]
    \SetAlgoLined
    \SetKwInOut{Input}{input}\SetKwInOut{Output}{output}
    \SetKwData{RVars}{rvars}
    \SetKw{Return}{return}
    \SetKw{Break}{break}
    \Input{an instance $(\varphi,k)$ of \SmallMinDNF{}}
    \Output{YES iff $(\varphi,k) \in \SmallMinDNF{}$}
    \BlankLine
    \RVars $\leftarrow \emptyset$ \tcp*[r]{variables relevant in~$\varphi$}
    $i \leftarrow 0$; $j \leftarrow k+2$ \tcp*[r]{bounds on \# of rvars}
    \While(\hfill \CommentSty{// logarithmic search for the \# of rvars}){$i+1 < j$}{
      $\ell \leftarrow \lceil (i+j)/2 \rceil$ \;
      query the SAT solver whether there exist at least~$\ell$
      variables \\ \qquad that are relevant in~$\varphi$
      \tcp*[r]{for idea behind encoding, see Lemma~\ref{lem:relevant-vars-sat}}
      \If{the SAT solver returns a model~$M$}{
        $\RVars \leftarrow$ the~$\ell$ many relevant variables encoded by the model~$M$ \;
      }\lElse{
        \Break
      }
    }
    \eIf{$\Card{\RVars} > k$}{
      \Return NO \tcp*[r]{too many rvars for any DNF of size~$\leq k$}
    }{
      \ForEach(\hfill \CommentSty{//~$2^{O(k \log k)}$ many})
      {DNF formula~$\psi$ of size~$k$ over \satversion{var's}\longversion{variables} in \RVars}{  
        construct a formula~$\varphi_\psi$ that is unsatisfiable iff~$\psi \equiv \varphi$\;
        \tcp*[f]{the formulas~$\varphi_\psi$ must be variable disjoint}
      }
      query the SAT solver whether~$\bigwedge_{\psi} \varphi_{\psi}$ is satisfiable \;
      \eIf{the SAT solver returns YES}{
        \Return NO \tcp*[r]{no candidate~$\psi$ is equivalent to~$\varphi$}
      }{
        \Return YES \tcp*[r]{some candidate~$\psi$ is equivalent to~$\varphi$}
      }
    }
    \caption{Solving \SmallMinDNF{} in fpt-time using~$(\lceil \log_2 k \rceil +1)$ many SAT calls.}
    \label{alg:smallmindnf}
  \end{algorithm}
  \vspace{-15pt}
\end{figure}

\longversion{\begin{proposition}\label{prop:smallmindnf-membership2}}%
\satversion{\begin{theorem}\label{thm:smallmindnf-membership2}}%
\SmallMinDNF{} can be solved by an fpt-algorithm
that uses~${(\lceil \log_2 k \rceil +1)}$ many SAT calls,
where the SAT solver returns a model for
satisfiable formulas.
Moreover, the first~$\lceil \log_2 k \rceil$ many calls to the solver
use SAT instances of size~$O(k^2n^2)$,
whereas the last call uses a SAT instance of size~$2^{O(k \log k)} \cdot n$,
where~$n$ is the input size.
\longversion{\end{proposition}}%
\satversion{\end{theorem}}%
\begin{proof}
The algorithm given in pseudo-code in Algorithm~\ref{alg:smallmindnf}
solves the problem \SmallMinDNF{}
in the required time bounds.
To obtain the required running time, we assume that each call to a SAT
solver takes only a single time step.
By Lemma~\ref{lem:relevant}, we know that any minimal equivalent
formula of~$\varphi$ must contain all and only the variables
that are relevant in~$\varphi$.
The algorithm firstly determines
how many variables are relevant in~$\varphi$.
By Lemma~\ref{lem:relevant-vars-sat},
we know that this can be done with a binary search
using~$\lceil \log_2 k \rceil$ SAT calls.
If there are more than~$k$ relevant variables,
the algorithm rejects.
Otherwise, the algorithm will have computed
the set~$\mtext{\sf rvars}$ of
relevant variables.
Next, with a single SAT call, it checks whether there exists
some DNF formula~$\psi$ of size~$k$ over the
variables in~$\mtext{\sf rvars}$.
By Lemma~\ref{lem:combinatorics-dnf-formulas},
we know that there are~$2^{O(k \log k)}$ many different DNF
formulas~$\psi$ of size~$k$ over the variables in~$\mtext{\sf rvars}$.
Verifying whether a particular DNF formula~$\psi$ is equivalent to
the original formula~$\varphi$ can be done by checking whether
the formula~$\varphi_{\psi} =
(\psi \wedge \neg \varphi) \vee (\neg \psi \wedge \varphi)$
is unsatisfiable.
Verifying whether there exists some suitable DNF formula~$\psi$
that is equivalent to~$\varphi$ can be done by making
variable-disjoint copies of all~$\varphi_{\psi}$ and checking
whether the conjunction of these copies is unsatisfiable.
\end{proof}

Note that the algorithm
requires that the SAT solver returns
a model if the query is satisfiable.
Also, the algorithm
can be modified straightforwardly to
return a DNF formula~$\psi$ of size at most~$k$
that is equivalent to an input~$\varphi$ if such
a formula~$\psi$ exists.
It would need to search for this~$\psi$ that is equivalent
to~$\varphi$, for which it would need an
additional~${O(k \log k)}$ many SAT calls
(with instances of size~$2^{O(k \log k)} \cdot \size{\varphi}$).
An interesting topic for further research is to investigate
how many SAT calls are needed for an fpt-algorithm to produce
an equivalent DNF when the SAT solver only returns
whether or not the input is satisfiable,
and does not return a satisfying assignment
in case the input is satisfiable.
For decision problems, the difference between these two interfaces
to SAT solvers is (theoretically) not relevant, when allowing
more than a constant number of calls to the
solver~\cite[Lemma~6.3.4]{Krajicek95}.
For instance, when allowing a logarithmic number of calls,
using witnessed and non-witnessed SAT calls yields the same
computational power~\cite[Corollary~6.3.5]{Krajicek95}.
For function problems, on the other hand,
the difference does seem to be relevant,
in cases where the number of calls is bounded
to logarithmically many in the input size
(cf.~\cite[Theorem~5.4]{GottlobFermueller93})
or bounded by a function of the parameter.

From a practical point of view, the algorithm given in
Algorithm~\ref{alg:smallmindnf} might not be the best approach
to solve the problem.
The (single) instance produced for the last SAT call in the
algorithm is rather large (exponential in~$k$).
However, this instance is equivalent to the conjunction
of~$2^{O(k \log k)}$ many instances of linear size,
and these instances can be solved in parallel.
Such a parallel approach involves more (yet easier) SAT calls,
but might be more efficient in practice.

}

\newcommand{\ProofSmallMinDNFHardness}[0]{
We show that the problem~$P = \SB (x,2) \SM (x,2) \in \SmallMinDNF{} \SE$
is \co\NP{}-hard, by giving a polynomial time reduction
from \thUNSAT{} to~$P$.
Let~$\psi$ be an instance of UNSAT.
Assume without loss of generality that~$\psi$ is in CNF.
We construct an instance~$(\varphi,2)$ of~\SmallMinDNF{}
by letting~$\varphi = \neg \psi \vee (x_1 \wedge x_2 \wedge x_3)$,
for fresh variables~$x_1,x_2,x_3 \not\in \Var{\psi}$.
We show that~$\psi$ is unsatisfiable if and only
there exists a DNF formula~$\varphi'$ of size~2
that is equivalent to~$\varphi$.

$(\Rightarrow)$
Assume that~$\psi$ is unsatisfiable.
Then~$\neg\psi$ is valid, and therefore
so is~$\varphi$.
Let~$y \in \Var{\psi}$ be some variable that occurs
both positively and negatively in~$\psi$.
Then~$\varphi' = y \vee \neg y$ is a DNF formula
of size~$2$ that is equivalent to~$\varphi$.

$(\Leftarrow)$
Assume that~$\psi$ is satisfiable.
Then there exists an assignment~$\alpha : \Var{\psi}
\rightarrow \SBs 0,1 \SEs$ such that~$\psi[\alpha] = 1$.
Then the variables~$x_1,x_2,x_3$
are all relevant in~$\varphi$.
The assignment~$\alpha_1 = \alpha \cup \SBs x_2 \mapsto 1,
x_3 \mapsto 1 \SEs$ witnesses that~$x_1$ is relevant in~$\varphi$,
for instance.
Therefore, by Lemma~\ref{lem:relevant},
any DNF formula equivalent to~$\varphi$ must contain
all variables~$x_1,x_2,x_3$, and thus must be of size at least~3.
}

\newcommand{\ProofRelevantLemma}[0]{
Assume that~$x \in \Var{\varphi}$ is relevant in~$\varphi$.
We show that~$x \in \Var{\varphi'}$.
Let~$X = \Var{\varphi} \cup \Var{\varphi'}$.
Then there exists some assignment~$\alpha : X
\backslash \SBs x \SEs \rightarrow \SBs 0,1 \SEs$
such that~$\varphi[\alpha_1] \neq \varphi[\alpha_2]$,
where~$\alpha_1 = \alpha \cup \SBs x \mapsto 0 \SEs$
and~$\alpha_2 = \alpha \cup \SBs x \mapsto 1 \SEs$.
Assume that~$x \not\in \Var{\varphi'}$.
Then~$\varphi'[\alpha_1] = \varphi'[\alpha_2]$
because~$\varphi_1$ and~$\varphi_2$ coincide on~$\Var{\varphi'}$.
This is a contradiction with the assumption that~$\varphi$
and~$\varphi'$ are equivalent.
Therefore,~$x \in \Var{\varphi'}$.

Conversely, assume that~$x \in \Var{\varphi}$ is not relevant
in~$\varphi$.
We show that~$x \not\in \Var{\varphi'}$.
By definition we know that for each assignment~$\alpha : X
\backslash \SBs x \SEs\rightarrow \SBs 0,1 \SEs$ it holds
that~$\varphi[\alpha_1] = \varphi[\alpha_2]$,
where~$\alpha_1 = \alpha \cup \SBs x \mapsto 0 \SEs$
and~$\alpha_2 = \alpha \cup \SBs x \mapsto 1 \SEs$.
Assume that~$x \in \Var{\varphi'}$.
Then~$\varphi'$ is equivalent to the DNF
formula~$\varphi'[x \mapsto 0]$, which is strictly smaller
than~$\varphi'$.
This contradicts minimality of~$\varphi'$.
Therefore,~$x \not\in \Var{\varphi'}$.
}

\newcommand{\ProofSmallMinDNF}[0]{
We give \satversion{a many-one}\longversion{an}
fpt-reduction to \EkAWSat{}.
Let~$(\varphi,k)$ be an instance of \SmallMinDNF{},
where~$\varphi$ is a DNF formula.
Let~$X = \Var{\varphi}$.
We construct an instance~$(\varphi',k')$ of \EkAWSat{}
as follows.
We introduce the set~$Y = \SBs y_{i,x} \SM 1 \leq i \leq k, x \in X \SE$
of variables.
Intuitively, the variables~$y_{i,x}$ represent a choice of (at most)~$k$
many variables~$x \in X$ to be used in the minimized DNF formula.
Let~$\psi_1,\dotsc,\psi_b$ be an enumeration of all possible DNF
formulas of size~$\leq k$ on the (fresh) variables~$z_1,\dotsc,z_k$.
By straightforward counting, we know that~$b \leq f(k)$,
for some function~$f = 2^{O(k \log k)}$.
Moreover, for each~$1 \leq j \leq b$, let~$\psi_j =
t_{j,1} \vee \dotsm \vee t_{j,w_j}$, and for
each~$1 \leq \ell \leq w_j$, let~$t_{j,\ell} =
l_{j,\ell,1} \wedge \dotsm \wedge l_{j,\ell,v_{j,\ell}}$.
We introduce another set~$U = \SB u_j \SM 1 \leq j \leq b \SE$
of variables.
Intuitively, these variables will be used to select the shape~$\psi_j$
of the minimized DNF formula.
We then perform our construction by
letting~$\varphi' = \exists Y. \exists U. \forall X. \varphi''$,
where~$\varphi'' =
\varphi^{Y}_{\mtext{proper}} \wedge
\varphi^{U}_{\mtext{one}} \wedge
\bigwedge_{1 \leq j \leq b} (u_j \rightarrow
(\varphi \leftrightarrow \bigvee_{1 \leq \ell \leq w_j}
\smash{\varphi^{j,\ell}_{\mtext{sat}}}))$.
Here, the formula~$\varphi^{Y}_{\mtext{proper}}$ ensures
that for each~$1 \leq i \leq k$ there is exactly one~$x_i \in X$
such that~$y_{i,x_i}$ is true, and that the~$x_i$ are all distinct.
It consists of clauses~$\bigvee_{x \in X} y_{i,x}$,
$(\neg y_{i,x} \vee \neg y_{i',x})$ and
$(\neg y_{i,x} \vee \neg y_{i,x'})$,
for each~$1 \leq i < i' \leq k$
and each~$x,x' \in X$ such that~$x \neq x'$.
The formula~$\varphi^{U}_{\mtext{one}}$ ensures
there is exactly one~$1 \leq j \leq b$ such that~$u_j$ is true,
and consists of the clause~$\bigvee_{1 \leq j \leq b} u_j$
and clauses~$(\neg u_j \vee \neg u_{j'})$ for each~$1 \leq j < j' \leq u$.
For each~$1 \leq j \leq b$, the assignment to the variables in~$Y$ represents a
DNF formula~$\chi_j$ that is
obtained by taking the~$\psi_j$
and replacing each~$z_i$ in~$\psi_j$ by the unique~$x_i$ for
which~$y_{i,x_i}$ is true.
Next, the formulas~$\varphi^{j,\ell}_{\mtext{sat}} =
\bigwedge_{1 \leq t \leq v_{j,\ell}} \varphi^{j,\ell,t}_{\mtext{sat}}$ encode
whether the~$\ell$-th term of~$\chi_j$ is satisfied by the assignment
to the variables in~$X$.
For this, we let~$\varphi^{j,\ell,t} = \bigwedge_{x \in X} (y_{m,x} \rightarrow x)$
if~$l_{j,\ell,t} = z_m$ for some~$1 \leq m \leq k$,
and we let~$\varphi^{j,\ell,t} = \bigwedge_{x \in X} (y_{m,x} \rightarrow \neg x)$
if~$l_{j,\ell,t} = \neg z_m$ for some~$1 \leq m \leq k$.
Finally, we let~$k' = k+1$.
We show that~$(\varphi,k) \in \SmallMinDNF{}$ if and only
if~$(\varphi',k') \in \EkAWSat{}$.

$(\Rightarrow)$ 
Assume that there exists a DNF~$\chi$ of size~$\leq k$
such that~$\varphi \equiv \chi$.
By Lemma~\ref{lem:relevant}, we may assume without
loss of generality that~$\Var{\chi} \subseteq \Var{\varphi}$.
Then clearly there exists some DNF formula~$\psi_j$ over
the variables~$z_1,\dotsc,z_k$ and some distinct
variables~$x_1,\dotsc,x_k \in X$ such
that~$\chi = \psi_j[z_1 \mapsto x_1,\dotsc,z_k \mapsto x_k]$.
We construct the assignment~$\alpha : Y \cup U \rightarrow \SBs 0,1 \SEs$
of weight~$k'$ as follows.
We let~$\alpha(y_{i,x}) = 1$ if and only if~$x = x_i$,
and we let~$\alpha(u_{j'}) = 1$ if and only if~$j' = j$.

It is straightforward to verify that~$\alpha$ satisfies the
formulas~$\varphi^{Y}_{\mtext{proper}}$ and~$\varphi^{U}_{\mtext{one}}$.
We show that~$\forall X. (\bigwedge_{1 \leq j \leq b} (u_j \rightarrow
(\varphi \leftrightarrow \bigvee_{1 \leq \ell \leq w_j}
\varphi^{j,\ell}_{\mtext{sat}})))[\alpha]$ is true.
Let~$\beta : X \rightarrow \SBs 0,1 \SEs$ be an arbitrary truth assignment.
Clearly, for each~$1 \leq j' \leq u$ such that~$j' \neq j$,
the implication is satisfied, because~$\alpha(u_{j'}) = 0$.
We show that~$(\varphi \leftrightarrow \bigvee_{1 \leq \ell \leq w_j}
\varphi^{j,\ell}_{\mtext{sat}})[\alpha \cup \beta]$ is true.
If~$\varphi[\beta]$ is true, then since~$\chi \equiv \varphi$,
we know that some term~$t_{j,\ell}[z_1 \mapsto x_1,\dotsc,z_k \mapsto x_k]$
of~$\chi$ is satisfied. It is straightforward to verify
that~$\varphi^{j,\ell}_{\mtext{sat}}$ is satisfied then as well.
Conversely, if~$\varphi[\beta]$ is not true, then an analogous argument
shows that for no~$1 \leq \ell \leq w_j$ the formula~$\varphi^{j,\ell}_{\mtext{sat}}$
is satisfied.
This concludes our proof that~$(\varphi',k') \in \EkAWSat{}$.

$(\Leftarrow)$
Assume that there exists some assignment~$\alpha : Y \cup U \rightarrow
\SBs 0,1 \SEs$ of weight~$k$ such that~$\forall X. \varphi''$ is true.
Clearly, for each~$1 \leq i \leq k$,~$\alpha$ there must be exactly
one~$x_i \in X$ such that~$\alpha$ sets~$y_{i,x_i}$ to true,
and there must be exactly one~$1 \leq j \leq b$ such
that~$\alpha$ sets~$u_j$ to true,
since otherwise the formulas~$\varphi^{Y}_{\mtext{proper}}$
and~$\varphi^{U}_{\mtext{one}}$ would not be satisfied.
Consider the formula~$\chi$ that is obtained from~$\psi_j$
by replacing the variables~$z_1,\dotsc,z_k$ by the
variables~$x_1,\dotsc,x_k$, respectively.
We know that the size of~$\chi$ is at most~$k$.
We show that~$\varphi \equiv \chi$.

Let~$\beta : X \rightarrow \SBs 0,1 \SEs$ be an arbitrary truth assignment.
We show that~$\varphi[\beta] = \chi[\beta]$.
Since~$\alpha(u_j) = 1$, we know that~$\alpha \cup \beta$ must satisfy
the formula~$(\varphi \leftrightarrow \bigvee_{1 \leq \ell \leq w_j}
\varphi^{j,\ell}_{\mtext{sat}})$.
Assume that~$\beta$ satisfies~$\chi$, i.e.,~$\beta$ satisfies some
term~$t_{j,\ell}[z_1 \mapsto x_1,\dotsc,z_k \mapsto x_k]$ of~$\chi$.
It is straightforward to verify that
then~$\varphi^{j,\ell}_{\mtext{sat}}[\alpha \cup \beta]$ is true.
Then also~$\varphi[\alpha \cup \beta] = \varphi[\beta]$ is true.
Conversely, assume that~$\beta$ satisfies~$\varphi$.
Then we know that~$\bigvee_{1 \leq \ell \leq w_j}
\varphi^{j,\ell}_{\mtext{sat}}[\alpha \cup \beta]$ is true,
i.e., there exists some~$1 \leq \ell \leq w_j$ such
that~$\varphi^{j,\ell}_{\mtext{sat}}[\alpha \cup \beta]$ is true.
It is then straightforward to verify that the
term~$t_{j,\ell}[z_1 \mapsto x_1,\dotsc,z_k \mapsto x_k]$
is satisfied by~$\beta$.
This concludes our proof that~$(\varphi,k) \in \SmallMinDNF{}$.
}

\newcommand{\ProofEAtwECompleteness}[0]{
\begin{proof}
\sloppypar
Membership in $\para{\SigmaP{2}}$ is obvious.
To show $\para{\SigmaP{2}}$-hardness, it suffices to show that the
problem is already \SigmaP{2}-hard when the parameter value
is restricted to~$1$~\cite{FlumGrohe03}.
We show this by means of a reduction from \EASat{}.
The idea of this reduction is to introduce for each existentially
quantified variable~$x$
a corresponding universally quantified variable~$z_x$
that is used to represent the truth value assigned to~$x$.
Each of the existentially quantified variables then only directly interacts
with universally quantified variables.

Take an arbitrary instance of \EASat{},
specified by~$\varphi = \exists X. \forall Y. \psi(X,Y)$,
where~$\psi(X,Y)$ is in \DNF{}.
We introduce a new set~$Z = \SB z_x \SM x \in X \SE$ of variables.
It is straightforward to verify that~$\varphi = \exists X. \forall Y. \psi(X,Y)$
is equivalent to the formula $\exists Z. \forall X. \forall Y. \chi$,
where~$\chi = \bigvee_{x \in X}
[ (x \wedge \neg z_x) \vee (\neg x \wedge z_x) ]
\vee \psi(X,Y)$.
\longversion{Also, clearly,~$\mtext{IG}(Z,\chi')$ consists only of isolated
paths of length~$2$,
and thus has treewidth~$1$.
Thus, the unparameterized problem consisting of all yes-instances of \EAtwESat{},
where the input formula is in \DNF{} and the parameter value
is~$1$, is \SigmaP{2}-hard.
This proves that \EAtwESat{} is \para{\SigmaP{2}}-hard.
}%
\satversion{
Clearly, if we now delete all universally quantified variables, the incidence graph
of~$\chi$ consists only of isolated paths of length~$2$,
and therefore the treewidth is~$1$.
This proves that \EAtwESat{} is \para{\SigmaP{2}}-hard.
}%
\end{proof}
}

\newcommand{\ProofEAtwACompleteness}[0]{
\begin{proof}
\satversion{
Hardness for \para{\NP} can be proven by showing that the problem
is already \NP{}-hard when restricted to instances
where the parameter value is~$1$~\cite{FlumGrohe03}.
In order to do this, one can reduce an instance of \thSAT{}
to an instance of \EASat{}
whose matrix is in DNF by using the standard Tseitin
transformation, resulting in tree-like interactions
between the universally quantified variables.
Therefore, the resulting formula has universal incidence treewidth~$1$.
}

\longversion{
\sloppypar We show \para{\NP}-hardness by showing that the problem
is already \NP{}-hard when restricted to instances
where the parameter value is~$1$~\cite{FlumGrohe03}.
We reduce from \thSAT{}, and the idea behind this reduction
is to reduce the instance of \thSAT{} to an instance of \EASat{}
whose matrix is in DNF by using the standard Tseitin
transformation, resulting in tree-like interactions
between the universally quantified variables.

Let~$\varphi$ be a propositional formula whose satisfiability
we want to decide, with~$\Var{\varphi} = X$.
Assume without loss of generality that~$\varphi$ contains only the
connectives~$\neg$ and~$\wedge$.
Equivalently, we want to determine whether the
QBF~$\psi = \exists X. \forall \emptyset. \varphi$ is true.
We can transform~$\psi$ into an equivalent
QBF~$\psi' = \exists X. \forall Y. \chi$ whose matrix is
in DNF by using the standard Tseitin transformation~\cite{Tseitin}
as follows.%
\longversion{\SecondProofEAtwAMembership}
}

We now show \para{\NP}-membership of \EAtwASat{}.
\longversion{Let~$\varphi = \exists X. \forall Y. \psi$ be a quantified Boolean formula
where~$\psi = \delta_1 \vee \dotsm \vee \delta_u$,
and let~$(\mathcal{T},(B_t)_{t \in T})$ be a tree decomposition
of~$\mtext{IG}(Y,\psi)$ of width~$k$. }
\longversion{We may assume without loss of generality
that~$(\mathcal{T},(B_t)_{t \in T})$
is a nice tree decomposition. 
We may also assume without loss of generality
that for each~$t \in T$,~$B_t$ contains some~$y \in Y$. }%
We construct a CNF formula~$\varphi'$ that is satisfiable
if and only if~$\varphi$ is true.
The idea is to construct a formula that encodes the following
guess-and-check algorithm.
Firstly, the algorithm guesses an assignment~$\gamma$
to the existential variables.
\longversion{Then, the algorithm uses a dynamic programming approach
using the tree decomposition to decide whether
the formula instantiated with~$\gamma$ is valid.
This dynamic programming approach is widely used to solve
problems for instances where some graph representing the structure
of the instance has small treewidth (cf.~\cite{Bodlaender88}).
}%
\satversion{Note that the incidence graph of the formula
instantiated with~$\gamma$ has a small treewidth, because
instantiating with~$\gamma$ removes all existentially quantified variables.
Then, the algorithm employs dynamic programming
to exploit the fact that the incidence graph of the remaining formula
has small treewidth
to decide validity of the remaining formula.
This dynamic programming approach is widely used to solve
problems for instances where some graph representing the structure
of the instance has small treewidth (cf.~\cite{Bodlaender88}).

Next, we show how to encode this guess-and-check
algorithm into a formula~$\varphi'$ that is satisfiable if and only if
the algorithm accepts.
In order to do so, we formally define treewidth and tree decompositions
of graphs.
A tree decomposition of a graph~$G = (V,E)$ is a
pair~$(\mathcal{T},(B_t)_{t \in T})$ where~$\mathcal{T} = (T,F)$ is a rooted tree
and~$(B_t)_{t \in T}$ is a family of subsets of~$V$ such that:
(i) for every~$v \in V$, the set~$B^{-1}(v) = \SB t \in T \SM v \in B_t \SE$
induces a nonempty subtree of~$\mathcal{T}$; and
(ii) for every edge~$\SBs v,w \SEs \in E$, there is a~$t \in T$ such
that~$v,w \in B_t$.
In order to simplify the proof, we will consider the following
normal form of tree decompositions.
We call a tree decomposition~$(\mathcal{T},(B_t)_{t \in T})$
\emph{nice} if every node~$t \in T$ is of one of the following
four types:
(\emph{leaf node})~$t$ has no children and~$\Card{B_t} = 1$;
(\emph{introduce node})~$t$ has one child~$t'$
    and~$B_t = B_{t'} \cup \SBs v \SEs$
    for some vertex~$v \not\in B_{t'}$;
(\emph{forget node})~$t$ has one child~$t'$
    and~$B_t = B_{t'} \backslash \SBs v \SEs$
    for some vertex~$v \in B_{t'}$; or
(\emph{join node})~$t$ has two children~$t_1,t_2$
    and~$B_t = B_{t_1} = B_{t_2}$.
Given any graph~$G$ and a tree decomposition of~$G$ of width~$k$,
a nice tree decomposition of~$G$ of width~$k$ can be computed
in polynomial time~\cite{Bodlaender12,Kloks94}.
}%
\satversion{

Let~$\varphi = \exists X. \forall Y. \psi$ be a quantified Boolean formula
where~$\psi = \delta_1 \vee \dotsm \vee \delta_u$,
and let~$(\mathcal{T},(B_t)_{t \in T})$ be a tree decomposition
of width~$k$
of the incidence graph of~$\varphi$ after deletion
of the existentially quantified variables.
We may assume without loss of generality
that~$(\mathcal{T},(B_t)_{t \in T})$
is a nice tree decomposition. }%
\satversion{We may also assume without loss of generality
that for each~$t \in T$,~$B_t$ contains some~$y \in Y$.}

\longversion{Next, we show how to encode this guess-and-check
algorithm into a formula~$\varphi'$ that is satisfiable if and only if
the algorithm accepts.}
We let~$\Var{\varphi'} = X \cup Z$
where~$Z = \SB z_{t,\alpha,i} \SM
t \in T, \alpha : \Var{t} \rightarrow \SBs 0,1 \SEs, 1 \leq i \leq u \SE$.
Intuitively, the variables~$z_{t,\alpha,i}$ represent whether
at least one assignment extending~$\alpha$
(to the variables occurring in nodes~$t'$ below~$t$)
violates the term~$\delta_i$ of~$\psi$.
We then construct~$\varphi'$ as follows
by using the structure of the tree decomposition.
For all~$t \in T$, all~$\alpha : \Var{t} \rightarrow \SBs 0,1 \SEs$,
all~$1 \leq i \leq u$, and each literal~$l \in \delta_i$ such that~$\Var{l} \in X$,
we introduce the clause\satversion{ \inlineeqref{eq:tw1}}:
\satversion{$(\overline{l} \rightarrow z_{t,\alpha,i})$.}
\longversion{\begin{equation}
\label{eq:tw1}
(\overline{l} \rightarrow z_{t,\alpha,i}).
\end{equation}}%
Then, for all~$t \in T$, all~$\alpha : \Var{t} \rightarrow \SBs 0,1 \SEs$,
and all~$1 \leq i \leq u$ such that for some~$l \in \delta_i$
it holds that~$\Var{l} \in Y$ and~$\alpha(l) = 0$,
we introduce the clause\satversion{ \inlineeqref{eq:tw2}}:
\satversion{$(z_{t,\alpha,i})$.}
\longversion{\begin{equation}
\label{eq:tw2}
(z_{t,\alpha,i}).
\end{equation}}%
Next, let~$t \in T$ be any introduction node with child~$t'$, and
let~$\alpha : \Var{t'} \rightarrow \SBs 0,1 \SEs$ be an arbitrary assignment. 
For any assignment~$\alpha' : \Var{t} \rightarrow \SBs 0,1 \SEs$
that extends~$\alpha$, and
for each~$1 \leq i \leq u$, we introduce the
clause\satversion{ \inlineeqref{eq:tw3}}:
\satversion{$(z_{t',\alpha,i} \rightarrow z_{t,\alpha',i})$.}
\longversion{\begin{equation}
\label{eq:tw3}
(z_{t',\alpha,i} \rightarrow z_{t,\alpha',i}).
\end{equation}}%
Then, let~$t \in T$ be any forget node with child~$t'$, and
let~$\alpha : \Var{t} \rightarrow \SBs 0,1 \SEs$ be an arbitrary assignment.
For any assignment~$\alpha' : \Var{t'} \rightarrow \SBs 0,1 \SEs$
that extends~$\alpha$, and
for each~$1 \leq i \leq u$, we introduce the
clause\satversion{ \inlineeqref{eq:tw4}}:
\satversion{$(z_{t',\alpha',i} \rightarrow z_{t,\alpha,i})$.}
\longversion{\begin{equation}
\label{eq:tw4}
(z_{t',\alpha',i} \rightarrow z_{t,\alpha,i}).
\end{equation}}%
Next, let~$t \in T$ be any join node with children~$t_1,t_2$, and
let~$\alpha : \Var{t} \rightarrow \SBs 0,1 \SEs$ be an arbitrary assignment.
For each~$1 \leq i \leq u$, we introduce the
clauses\satversion{ \inlineeqref{eq:tw5}}:
\satversion{$(z_{t_1,\alpha,i} \rightarrow z_{t,\alpha,i})$ and
$(z_{t_2,\alpha,i} \rightarrow z_{t,\alpha,i})$.}
\longversion{\begin{equation}
\label{eq:tw5}
(z_{t_1,\alpha,i} \rightarrow z_{t,\alpha,i}) \mtext{ and }
(z_{t_2,\alpha,i} \rightarrow z_{t,\alpha,i}).
\end{equation}}%
Finally, for the root node~$t_{\mtext{root}} \in T$
and for each~$\alpha : \Var{t_{\mtext{root}}} \rightarrow \SBs 0,1 \SEs$
we introduce the clause\satversion{ \inlineeqref{eq:tw6}}:
\satversion{$\bigvee_{1 \leq i \leq u} \neg z_{t_{\mtext{root}},\alpha,i}$.}
\longversion{\begin{equation}
\label{eq:tw6}
\bigvee\limits_{1 \leq i \leq u} \neg z_{t_{\mtext{root}},\alpha,i}.
\end{equation}}%
It is straightforward to verify that~$\varphi'$
contains~$O(2^k\Card{T})$ many clauses.
\satversion{We claim that this reduction is correct,
i.e., that~$\varphi$ is true if and only if~$\varphi'$
is satisfiable.
We omit a detailed proof of this.}%
\longversion{\SecondProofEAtwACompleteness}
\end{proof}
}

\newcommand{\SecondProofEAtwAMembership}[0]{
Let~$\Sub{\varphi} = \SBs r_1,\dotsc,r_s \SEs$ be the set of subformulas
of~$\varphi$.
We let~$Y = \SBs y_1,\dotsc,y_s \SEs$ contain one propositional variable~$y_i$
for each~$r_i \in \Sub{\varphi}$.
Then, we define~$\psi' = \exists X. \forall Y. \chi$, where
$\chi = r_{\varphi} \vee \bigvee_{1 \leq i \leq u} \chi_i$,
and:
\[ \chi_i = \begin{dcases*}
    (r_i \wedge r_j) \vee (\neg r_i \wedge \neg r_j) &
      if $r_i = \neg r_j$; \\
    (r_i \wedge \neg x) \vee (\neg r_i \wedge x) &
      if $r_i = x \in X$; \\
    (r_i \wedge \neg r_j) \vee
    (r_i \wedge \neg r_k) \vee
    (\neg r_i \wedge r_j \wedge r_k) &
      if $r_i = r_j \wedge r_k$. \\
\end{dcases*} \]
It is straightforward to verify that~$\psi$ and~$\psi'$ are equivalent.
Moreover,~$\mtext{IG}(Y,\chi)$ is a tree, and thus has treewidth~$1$.
Therefore, it is easy to construct a
tree decomposition~$(\TTT,(B_t)_{t \in T})$
of~$\mtext{IG}(Y,\chi)$ of width~$1$ in polynomial time.
Then~$(\TTT,(B_t)_{t \in T})$ and~$\psi'$ together are an instance
of~$(\EAtwASat{})_1$ such that~$\psi'$ is true if and only
if~$\varphi$ is satisfiable.
}

\newcommand{\SecondProofEAtwACompleteness}[0]{
We now show that~$\varphi$ is true if and only if~$\varphi'$
is satisfiable.

$(\Rightarrow)$
Assume that there exists an assignment~$\beta : X \rightarrow \SBs 0,1 \SEs$
such that~$\forall Y. \psi[\beta]$ is true.
We construct an assignment~$\gamma : Z \rightarrow \SBs 0,1 \SEs$
such that~$\beta \cup \gamma$ satisfies~$\varphi'$.
Let~$C$ be the set of clauses~$(z_{t,\alpha,i})$
such that~$\varphi'$ contains a clause~$(\overline{l} \rightarrow z_{t,\alpha,i})$
for which~$\beta(\overline{l}) = 1$.
Since~$C$ together with the clauses~(\ref{eq:tw2}--\ref{eq:tw5})
forms a definite Horn formula~$\chi$, we can compute its
unique subset-minimal model (by unit-propagation).
Let~$\gamma$ be this subset-minimal model of~$\chi$.
We show that~$\gamma$ also satisfies the clauses~(\ref{eq:tw6}).
Let~$\alpha : \Var{t_{\mtext{root}}} \rightarrow \SBs 0,1 \SEs$ be
an arbitrary assignment.
Since~$\forall Y. \psi[\beta]$ is true, we know that there exists an
assignment~$\alpha' : Y \rightarrow \SBs 0,1 \SEs$ extending~$\alpha$
such that~$\psi[\alpha' \cup \beta]$ is true,
i.e., for some~$1 \leq i \leq u$ the term~$\delta_i$ is satisfied
by~$\alpha' \cup \beta$.
This assignment~$\alpha'$ induces a family~$(\alpha_t)_{t \in T}$
of assignments as follows: for each~$t \in T$, the assignment~$\alpha_t$
is the restriction of~$\alpha'$ to the variables~$\Var{t}$.
It is straightforward to verify that~$\gamma(z_{t,\alpha_t,i}) = 0$
for all~$t \in T$.
Therefore, since~$\alpha = \alpha_{t_{\mtext{root}}}$,
the clause~$\bigvee_{1 \leq i \leq u} (\neg z_{t_{\mtext{root}},\alpha,i})$
is satisfied.
Since~$\alpha$ was arbitrary, we know that~$\gamma$ satisfies
all clauses in~(\ref{eq:tw6}).
Thus,~$\beta \cup \gamma$ satisfies~$\varphi'$.

$(\Leftarrow)$
Assume that there is an assignment~$\beta : X \rightarrow \SBs 0,1 \SEs$
and an assignment~$\gamma : Z \rightarrow \SBs 0,1 \SEs$
such that~$\beta \cup \gamma$ satisfies~$\varphi'$.
We show that~$\forall Y. \psi[\beta]$ is true.
Let~$\alpha' : Y \rightarrow \SBs 0,1 \SEs$ be an arbitrary assignment.
We show that for some~$1 \leq i \leq u$,~$\alpha' \cup \beta$
satisfies the term~$\delta_i$, and thus that~$\psi[\alpha' \cup \beta]$
is true.
The assignment~$\alpha'$ induces a family~$(\alpha_t)_{t \in T}$
of assignments as follows: for each~$t \in T$, the assignment~$\alpha_t$
is the restriction of~$\alpha'$ to the variables~$\Var{t}$.
Since~$\gamma$ satisfies the clauses~(\ref{eq:tw6}),
we know that there exists some~$1 \leq i \leq u$ such
that~$\gamma(z_{t_{\mtext{root}},\alpha_{t_{\mtext{root}}},i}) = 0$.
It is straightforward to verify that~$\gamma(z_{t,\alpha_{t},i}) = 0$
for all~$t \in T$ then; otherwise, the clauses~(\ref{eq:tw3}--\ref{eq:tw5})
would force~$\gamma(z_{t_{\mtext{root}},\alpha_{t_{\mtext{root}}},i})$
to be~$1$.
Then, by the clauses~(\ref{eq:tw1}--\ref{eq:tw2}) we know that~$\alpha'$
must satisfy~$\delta_i$.
Since~$\alpha'$ was arbitrary, we know that~$\forall Y. \psi[\beta]$ is true,
and thus that~$\varphi$ is true.
}

\newcommand{\PreliminariesTreewidth}[0]{
A tree decomposition of a graph~$G = (V,E)$ is a
pair~$(\mathcal{T},(B_t)_{t \in T})$ where~$\mathcal{T} = (T,F)$ is a rooted tree
and~$(B_t)_{t \in T}$ is a family of subsets of~$V$ such that:
\satversion{(i) for every~$v \in V$, the set~$B^{-1}(v) = \SB t \in T \SM v \in B_t \SE$
induces a nonempty subtree of~$\mathcal{T}$; and
(ii) for every edge~$\SBs v,w \SEs \in E$, there is a~$t \in T$ such
that~$v,w \in B_t$.
}%
\longversion{
\begin{itemize}
  \item for every~$v \in V$, the set~$B^{-1}(v) = \SB t \in T \SM v \in B_t \SE$
    is nonempty and connected in~$\mathcal{T}$; and
  \item for every edge~$\SBs v,w \SEs \in E$, there is a~$t \in T$ such
    that~$v,w \in B_t$.
\end{itemize}}%
The \emph{width} of the decomposition~$(\mathcal{T},(B_t)_{t \in T})$ is the
number~$\max \SB \Card{B_t} \SM t \in T \SE - 1$.
The \emph{treewidth} of~$G$ is the minimum of the widths
of all tree decompositions of~$G$.
Let~$G$ be a graph and~$k$ a nonnegative integer.
There is an fpt-algorithm that computes
a tree decomposition of~$G$ of width~$k$ if it exists,
and fails otherwise~\cite{Bodlaender96}.
We call a tree decomposition~$(\mathcal{T},(B_t)_{t \in T})$
\emph{nice} if every node~$t \in T$ is of one of the following
four types:
\satversion{%
(\emph{leaf node})~$t$ has no children and~$\Card{B_t} = 1$;
(\emph{introduce node})~$t$ has one child~$t'$
    and~$B_t = B_{t'} \cup \SBs v \SEs$
    for some vertex~$v \not\in B_{t'}$;
(\emph{forget node})~$t$ has one child~$t'$
    and~$B_t = B_{t'} \backslash \SBs v \SEs$
    for some vertex~$v \in B_{t'}$; or
(\emph{join node})~$t$ has two children~$t_1,t_2$
    and~$B_t = B_{t_1} = B_{t_2}$.
}%
\longversion{\begin{itemize}
  \item \emph{leaf node}:~$t$ has no children and~$\Card{B_t} = 1$;
  \item \emph{introduce node}:~$t$ has one child~$t'$
    and~$B_t = B_{t'} \cup \SBs v \SEs$
    for some vertex~$v \not\in B_{t'}$;
  \item \emph{forget node}:~$t$ has one child~$t'$
    and~$B_t = B_{t'} \backslash \SBs v \SEs$
    for some vertex~$v \in B_{t'}$; or
  \item \emph{join node}:~$t$ has two children~$t_1,t_2$
    and~$B_t = B_{t_1} = B_{t_2}$.
\end{itemize}}%
Given any graph~$G$ and a tree decomposition of~$G$ of width~$k$,
a nice tree decomposition of~$G$ of width~$k$ can be computed
in polynomial time~\cite{Kloks94}.
}

\newcommand{\ProofKstarNmink}[0]{
We show that \EkAWSat{} \fptequiv{} \EnminkAWSat{}.
The same reduction works for both directions.
We describe the \satversion{many-one} fpt-reduction
from~$\EkAWSat{}$ to~$\EnminkAWSat{}$.
Let~$(\varphi,k)$ be an instance of~$\EkAWSat{}$,
with~$\varphi = \exists X. \forall Y. \psi$.
Assume without loss of generality that~$\varphi$ contains only
binary conjunctions and negations.
We construct an instance~$(\varphi',k)$ of~$\EnminkAWSat{}$.
Intuitively, we construct~$\varphi'$ by replacing in~$\varphi$
all literals over variables in~$X$ by their complement.

Formally, we let~$\varphi' = \exists X. \forall Y. \tau(\psi)$,
where~$\tau(\neg \psi_1) = \neg \tau(\psi_1)$,~$\tau(\psi_1 \wedge \psi_2)
= \tau(\psi_1) \wedge \tau(\psi_2)$,~$\tau(y) = y$ for
all~$y \in Y$, and~$\tau(x) = \neg x$ for all~$x \in X$.

It is straightforward to verify that for each assignment~$\alpha$ to
the propositional variables~$X$ of weight~$k$
we have that~$\psi[\alpha] = \psi'[\overline{\alpha}]$,
where~$\overline{\alpha}$ is the assignment of weight~$n-k$
obtained by flipping all assignments of~$\alpha$,
i.e.,~$\overline{\alpha}(x) = \neg\alpha(x)$.
Thus,~$(\varphi,k) \in \EkAWSat{}$ if and only
if~$(\varphi',k) \in \EnminkAWSat{}$.
}

\begin{document}

%%% NO PAGE NUMBERS (IN KR VERSION)
\krversion{
  \pagenumbering{gobble}
}

\maketitle
% page numbers
\thispagestyle{plain}
\pagestyle{plain}

%%%
%%% ABSTRACT (KR / ARXIV)
%%%
\krlongversion{
\begin{abstract}
  Today's propositional satisfiability (SAT) solvers are extremely
  powerful and can be used as an efficient back-end for 
  solving NP-complete problems.  However, many fundamental problems in
  knowledge representation and reasoning are located at the second
  level of the Polynomial Hierarchy  or even higher, and hence
  polynomial-time transformations to SAT are not possible, unless the
  hierarchy collapses. Recent research shows that in certain cases one can
  break through these complexity barriers by fixed-parameter tractable (fpt)
  reductions which exploit structural aspects of problem instances in
  terms of problem parameters.

  In this paper we develop a general theoretical framework that
  supports the classification of parameterized problems on whether
  they admit such an fpt-reduction to SAT or not.
  \krversion{
  We instantiate our theory by classifying the complexities of
  several case study problems, with respect to various natural parameters.
  These case studies include the consistency problem for disjunctive
  answer set programming and a robust version of
  constraint satisfaction.
  }\longversion{This framework is based on several new parameterized
  complexity classes.
  As a running example,
  we use the framework to classify the complexity of the consistency
  problem for disjunctive answer set programming, with
  respect to various natural parameters.
  We underpin the robustness of our theory by providing a characterization
  of the new complexity classes
  in terms of weighted QBF satisfiability, alternating Turing
  machines, and first-order model checking.
  In addition, we provide a compendium of parameterized problems
  that are complete for the new complexity classes,
  including problems related to
  Knowledge Representation and Reasoning, Logic,
  and Combinatorics.
  }
\end{abstract}
}

\longversion{\tableofcontents}

\krlongversion{
%%% INTRODUCTION (KR / ARXIV)
\section{Introduction}
\label{sec:introduction}

\krversion{\thispagestyle{empty}}
Over the last two decades, propositional satisfiability (SAT) has
become one of the most successful and widely applied techniques for
the solution of NP-complete problems. Today's SAT-solvers are
extremely efficient and robust, instances with hundreds of thousands
of variables and clauses can be solved routinely.
In fact, due to the
success of SAT, NP-complete problems have lost their scariness, as in
many cases one can efficiently encode NP-complete problems to SAT and
solve them by means of a
SAT-solver~\cite{BiereHeuleMaarenWalsh09,GomesKautzSabharwalSelman08,
MalikZhang09,MarquessilvaLynceMalik09,SakallahMarquessilva11}.
However, many important computational problems, most prominently in
knowledge representation and reasoning, are located
above the first level of the Polynomial Hierarchy (PH) and thus
considered ``harder'' than SAT.  Hence we cannot hope for
polynomial-time reductions from these problems to SAT, as such
transformations would cause the (unexpected) collapse of the PH.

Realistic problem instances are not random and often contain some kind
of ``hidden structure.'' Recent research succeeded to
exploit such hidden structure to break the complexity barriers
between levels of the PH, for problems that arise in
disjunctive answer set programming~\cite{FichteSzeider13} and
abductive reasoning~\cite{PfandlerRuemmeleSzeider13}.  The idea is to
exploit problem structure in terms of a problem \emph{parameter}, and
to develop reductions to SAT that can be computed efficiently as long
as the problem parameter is reasonably small.  The theory of
\emph{parameterized complexity}
\cite{DowneyFellows99,DowneyFellows13,FlumGrohe06,Niedermeier06}
provides exactly the
right type of reduction suitable for this purpose, called
\emph{fixed-parameter tractable reductions}, or \emph{fpt-reductions}
for short.  Now, for a suitable choice of the parameter, one can
aim at developing fpt-reductions from the hard problem under consideration to
SAT.

Such positive results 
go significantly beyond
the state-of-the-art of current research in parameterized complexity.
By shifting the scope from
fixed-parameter tractability to fpt-reducibility (to SAT), parameters
can be less restrictive and hence larger classes of inputs can be
processed efficiently. Therefore, the potential for positive tractability
results is greatly enlarged.  In fact, there are some known reductions
that, in retrospect, can be seen as fpt-reductions to SAT. A prominent
example is Bounded Model Checking~\cite{BiereCimattiClarkeZhu99},
which can be seen as an fpt-reduction from the model checking problem
for linear temporal logic (LTL), which is PSPACE-complete, to SAT,
where the parameter is an upper bound on the size of a counterexample.
Bounded Model Checking is widely used for hardware and software
verification at industrial scale~\cite{Biere09}.
\longversion{For a more detailed discussion of this example
we refer to Section~\ref{sec:app-bmc}.}

\paragraph{New Contributions}
The aim of this paper is to establish a general theoretical framework
that supports the classification of hard problems on whether they
admit an fpt-reduction to SAT or not.  The main contribution is the development of a new hardness theory that can be used to
provide evidence that certain problems do not admit  an
fpt-reduction to SAT, similar to 
NP-hardness which provides evidence against polynomial-time
tractability~\cite{GareyJohnson79} and $\W{1}$\hy hardness
which provides evidence against
fixed-parameter tractability~\cite{DowneyFellows99}.  
% These methods have proven to be extremely successful frameworks,
% with hundreds of natural problems that are complete for the
% complexity classes.

At the center of our theory are two hierarchies of parameterized
complexity classes: the \stark{} hierarchy and the \kstar{} hierarchy.
% (as
%we shall see, the latter hierarchy collapses into two single
%complexity classes). 
We define the complexity classes in terms of weighted variants of
the quantified Boolean satisfiability problem with one
quantifier alternation, which is canonical for the second level of the
PH. For the classes in the \kstar{} hierarchy, the (Hamming) weight of
the assignment to the variables in the first quantifier block is
bounded by the parameter~$k$, the weight of the second quantifier
block is unrestricted (``$*$''). For the classes in the \stark{}
hierarchy it is the other way around, the weight in the second block
restricted by~$k$ and the first block is unrestricted. Both
hierarchies span various degrees of hardness between the
classes~$\para{\NP}$ and~$\para{\co\NP}$ at the bottom %and $\para{\PiP{2}}$
and {$\para{\SigmaP{2}}$} at the top ($\para{K}$ contains all
parameterized problems that, after fpt-time preprocessing, ultimately belong
to complexity class~$\mtext{K}$~\cite{FlumGrohe03}).
Figure~\ref{fig:classes} illustrates the relationship between the various
parameterized complexity classes under consideration.

%A notion of tractability is of little use if there are no
%computational problems to apply it to.
%Conversely, a hardness theory benefits from
%having many natural complete problems.

To illustrate the usefulness of our theory, we consider as a running
example the fundamental problem of \emph{answer set programming} which asks
whether a disjuncive logic program has a stable model.
\longversion{Answer set 
programming~\cite{BrewkaEiterTruszczynski11,Gelfond07,
MarekTruszczynski99}
is a form of declarative
programming capable of expressing many nonmonotonic reasoning problems
that often occur in the domain of knowledge representation. }%
This problem is
$\SigmaP{2}$\hy complete~\cite{EiterGottlob95}, and
exhibits completeness or hardness for various of our
complexity classes; see Table~\ref{table:ex_results} for an
overview.
%As a second case study, we will classify the complexity
%of a robust version of constraint satisfaction \cite{AbramskyGottlobKolaitis13}.
\longversion{Moreover, we give alternative characterizations
of the \kstar{} hierarchy in terms of first-order model checking
and alternating Turing machines. }%
In addition we were able to identify many other natural problems
that populate our new complexity classes.\krversion{ We
  refer to a technical report corresponding to this paper
  which contains full proofs of all results,
  contains a compendium of problems,
  and is available on
  arXiv~(\href{http://arxiv.org/abs/1312.1672}{http://arxiv.org/abs/1312.1672}).} \longversion{In order to give a systematic and concise overview
  of these results,
we provide a \emph{compendium} in the appendix, containing completeness
results for our new complexity classes, for a wide range of problems,
including problems related to
Knowledge Representation and Reasoning, Logic,
and Combinatorics.}

\paragraph{Structure of the Paper}
We begin with reviewing basic notions related to (Parameterized)
Complexity Theory and Answer Set Programming in
Section~\ref{sec:prelims}.
Then, in Section~\ref{sec:params-for-asp}, we introduce our running
example by considering
several parameterized variants of the consistency problem of disjunctive
answer set programming. Moreover, in this section, we argue that
existing notions from
parameterized complexity do not suffice to characterize the
complexity of these parameterized problems.
In Section~\ref{sec:define-hierarchies}, we define the two hierarchies
of parameterized complexity classes
that are central to our new hardness theory
(the \kstar{} and \stark{} hierarchies),
on the basis of weighted variants of quantified Boolean satisfiability.
Next, in Section~\ref{sec:eka}, we show that one of these hierarchies
(the \kstar{} hierarchy) collapses into a single complexity class \EkA{}.
Moreover, we show completeness for this class
for one parameterized variant of our running example.
Then, in Section~\ref{sec:eka-characterizations}, we give additional
characterizations of the class \EkA{} in terms of first-order model
checking and alternating Turing machines.
In Section~\ref{sec:eak}, we provide some normalization results
for two levels of the \stark{} hierarchy.
In addition, we prove completeness of another parameterized variant
of our running example for the highest level of the \stark{} hierarchy.
Finally, we relate the new parameterized complexity classes to existing
parameterized complexity classes in
Section~\ref{sec:relating-eka-and-eak},
before concluding in Section~\ref{sec:conclusion}.

Section~\ref{sec:appendix-a} of the appendix contains a compendium
of completeness
results for our new complexity classes, for a wide range of problems,
including problems related to
Knowledge Representation and Reasoning, Logic,
and Combinatorics.
Section~\ref{sec:appendix-b} of the appendix contains proofs that
were omitted from Section~\ref{sec:atm-char} for the sake of readability,
and proofs for results in the compendium that do not appear in the
main text.

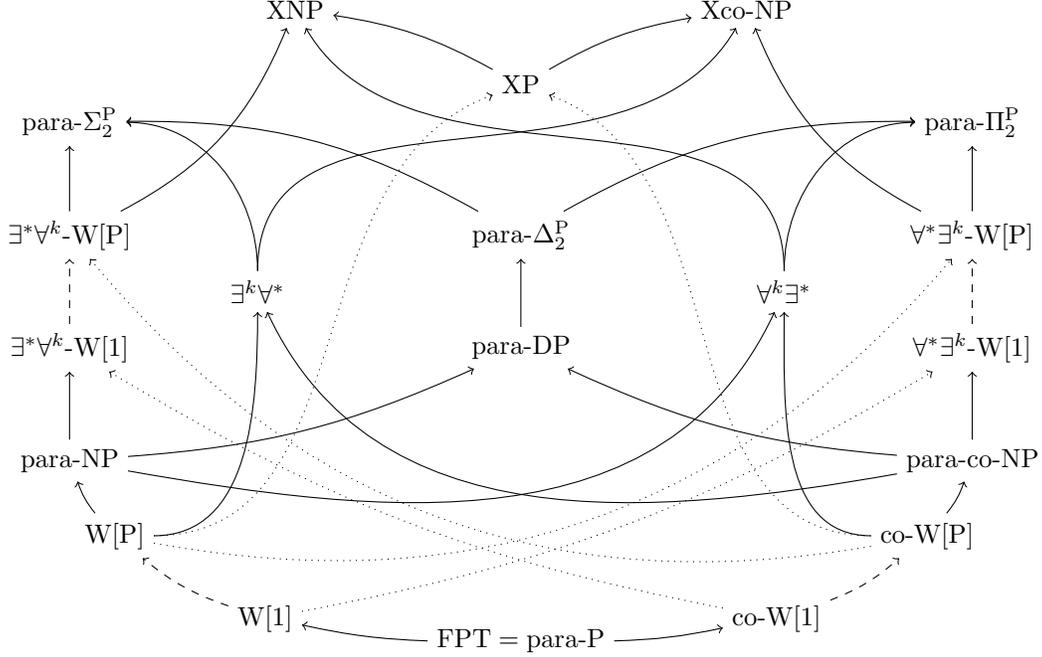
\begin{figure}
\begin{center}
\begin{tikzpicture}[node distance=1.5cm]
  % NODES
  \node[] (sigma2) {$\para{\SigmaP{2}}$};
  \node[right of=sigma2, node distance=12cm] (pi2) {$\para{\PiP{2}}$};
  \node[below of=sigma2, node distance=1.5cm] (eakwp) {$\EAkW{P}$};
  \node[below of=eakwp, node distance=1.5cm] (eakw1) {$\EAkW{1}$};
  \node[below of=eakw1, node distance=1.5cm] (paranp) {$\para{\NP}$};
  \node[below of=paranp, node distance=2.1cm, xshift=2.6cm] (w1) {$\W{1}$};
  \node[below of=paranp, node distance=1cm, xshift=.6cm] (wp) {$\W{P}$};
  \node[below of=pi2, node distance=1.5cm] (aekwp) {$\AEkW{P}$};
  \node[below of=aekwp, node distance=1.5cm] (aekw1) {$\AEkW{1}$};
  \node[below of=aekw1, node distance=1.5cm] (paraconp) {$\para{\co\NP}$};
  \node[below of=paraconp, node distance=1cm, xshift=-.6cm] (cowp) {$\co{}\W{P}$};
  \node[below of=paraconp, node distance=2.1cm, xshift=-2.6cm] (cow1) {$\co{}\W{1}$};
  \node[right of=sigma2, node distance=6cm, yshift=-1.5cm] (delta2) {$\para{\DeltaP{2}}$};
  \node[below of=delta2, node distance=1.5cm] (dp) {$\para{\DP}$};
  \node[left of=dp, node distance=3.5 cm, yshift=.75cm] (eka) {$\EkA$};
  \node[right of=dp, node distance=3.5 cm, yshift=.75cm] (ake) {$\AkE$};
  \node[below of=dp, node distance=3.9cm] (fpt) {$\FPT{} = \para{}\PTIME{}$};
  \node[above of=delta2, node distance=2cm] (xp) {$\XP$};
  \node[left of=xp, node distance=3cm, yshift=1cm] (xnp) {$\XNP$};
  \node[right of=xp, node distance=3cm, yshift=1cm] (xconp) {$\XcoNP$};
  % EDGES
  \path[draw,->] (paranp) -- (eakw1);
  \path[draw,->,dashed] (eakw1) -- (eakwp);
  \path[draw,->] (eakwp) -- (sigma2);
  \path[draw,->] (paraconp) -- (aekw1);
  \path[draw,->,dashed] (aekw1) -- (aekwp);
  \path[draw,->] (aekwp) -- (pi2);
  \path[draw,->] (paranp) edge[bend right=10] (dp);
  \path[draw,->] (paraconp) edge[bend left=10] (dp);
  \path[draw,->] (dp) -- (delta2);
  \path[draw,->] (delta2) edge[bend right=15] (sigma2);
  \path[draw,->] (delta2) edge[bend left=15] (pi2);
  \path[draw,->] (paranp) edge[out=-10, in=-115] (ake);
  \path[draw,->] (paraconp) edge[out=190, in=-65] (eka);
  \path[draw,->] (eka) edge[out=90, in=0] (sigma2);
  \path[draw,->] (ake) edge[out=90, in=180] (pi2);
  \path[draw,->] (fpt) edge[bend left=5] (w1);
  \path[draw,->] (fpt) edge[bend right=5] (cow1);
  \path[draw,->,dashed] (w1) edge[bend left=10] (wp);
  \path[draw,->] (wp) edge[bend left=10] (paranp);
  \path[draw,->,dashed] (cow1) edge[bend right=10] (cowp);
  \path[draw,->] (cowp) edge[bend right=10] (paraconp);
  \path[draw,->] (wp) edge[out=0, in=-90] (eka);
  \path[draw,->] (cowp) edge[out=180, in=-90] (ake);
  \path[draw,->,dotted] (wp) edge[bend right=30] (aekwp);
  \path[draw,->,dotted] (w1) edge[bend right=10] (aekw1);
  \path[draw,->,dotted] (cowp) edge[bend left=30] (eakwp);
  \path[draw,->,dotted] (cow1) edge[bend left=10] (eakw1);
  %%%
  \path[draw,->,dotted] (wp) edge[out=0,in=-160] (xp);
  \path[draw,->,dotted] (cowp) edge[out=180,in=-20] (xp);
  \path[draw,->] (xp) edge[bend right=10] (xnp);
  \path[draw,->] (xp) edge[bend left=10] (xconp);
  \path[draw,->] (eakwp) edge[bend right=20] (xnp);
  \path[draw,->] (aekwp) edge[bend left=20] (xconp);
  \path[draw,->] (eka) edge[out=90, in=-120] (xconp);
  \path[draw,->] (ake) edge[out=90, in=-60] (xnp);
  %%%
\end{tikzpicture}
\end{center}
\caption{The parameterized complexity classes
of the \stark{} and \kstar{} hierarchies
in relation to existing classes.}
\label{fig:classes}
\end{figure}

}

\krlongversion{
%%% PRELIMINARIES
\section{Preliminaries}
\label{sec:prelims}

\subsection{Parameterized Complexity Theory}
\label{sec:pct}
We introduce some core notions from parameterized complexity theory.
For an in-depth treatment we refer to other
sources~\cite{DowneyFellows99,DowneyFellows13,
FlumGrohe06,Niedermeier06}.
A \emph{parameterized problem}~$L$ is a subset of~$\Sigma^{*} \times
\mathbb{N}$ for some finite alphabet~$\Sigma$. For an instance~$(I,k)
\in \Sigma^{*} \times \mathbb{N}$, we call~$I$ the \emph{main part}
and~$k$ the \emph{parameter}.
For each positive integer~$k \geq 1$, we define the $k$-th
slice of~$L$ as the unparameterized problem~$L_k =
\SB x \SM (x,k) \in L \SE$.
The following generalization of
polynomial time computability is commonly regarded as the tractability
notion of parameterized complexity theory.  A parameterized problem
$L$ is \emph{fixed-parameter tractable} if there exists a computable
function~$f$ and a constant~$c$ such that there exists an algorithm
that decides whether~$(I,k) \in L$ in time~$O(f(k)\size{I}^c)$,
where~$\size{I}$ denotes the size of~$I$.  Such an algorithm is called an
\emph{fpt-algorithm}, and this amount of time is called
\emph{fpt-time}. \FPT{} is the class of all fixed-parameter tractable
decision problems.
If the parameter is constant, then fpt-algorithms run in polynomial
time where the order of the polynomial is independent of the
parameter.
This provides a good scalability in the parameter in contrast to
running times of the form~$\size{I}^k$,
which are also polynomial for fixed~$k$, but are already
impractical for, say,~$k>3$.

Parameterized complexity also offers a hardness theory,
similar to the theory of NP-hardness, that allows researchers to give
strong theoretical evidence that some parameterized problems are not
fixed-parameter tractable.
This theory is based on the \emph{Weft hierarchy} of complexity classes
$\FPT{} \subseteq \W{1} \subseteq \W{2} \subseteq \dotsm
\subseteq \W{\SAT} \subseteq \W{P}$,
where all inclusions are believed to be strict.
For a hardness theory, a notion of reduction is needed.
Let~$L \subseteq \Sigma^{*} \times \mathbb{N}$ and
$L' \subseteq (\Sigma')^{*} \times \mathbb{N}$ be two parameterized problems.
%for some finite alphabets~$\Sigma$ and~$\Sigma'$.
An \emph{fpt-reduction} (or \emph{fixed parameter tractable reduction})
from~$L$ to~$L'$ is a mapping
$R : \Sigma^{*} \times \mathbb{N} \rightarrow (\Sigma')^{*} \times \mathbb{N}$
from instances of~$L$ to instances of~$L'$ such that
there exist some computable function~$g : \mathbb{N} \rightarrow \mathbb{N}$
such that for all~$(I,k) \in \Sigma^{*} \times \mathbb{N}$:
(i)~$(I,k)$ is a yes-instance of~$L$ if and only if
$(I',k') = R(I,k)$ is a yes-instance of~$L'$,
(ii)~$k' \leq g(k)$,
and (iii)~$R$ is computable in fpt-time.
We write~$L \fptred{} L'$ if there is an fpt-reduction
from~$L$ to~$L'$.
Similarly, we call reductions that satisfy properties~(i) and~(ii)
but that are computable in time~$O(\size{I}^{f(k)})$,
for some fixed computable function~$f$, \emph{xp-reductions}.

The parameterized complexity
classes~$\W{$t$}$,~$t \geq 1$,~$\W{\SAT}$ and~$\W{P}$
are based on the satisfiability problems
of Boolean circuits and formulas.
We consider \emph{Boolean circuits} with a single output gate.
We call input nodes \emph{variables}.
We distinguish between \emph{small} gates, with fan-in~$\leq 2$,
and \emph{large gates}, with fan-in~$> 2$.
The \emph{depth} of a circuit is the length of a longest path from
any variable to the output gate.
The \emph{weft} of a circuit is the largest number of large gates on any
path from a variable to the output gate.
We let~$\Nodes{C}$ denote the set of all nodes of a circuit~$C$.
We say that a circuit~$C$ is in \emph{negation normal form}
if all negation nodes in~$C$ have variables as inputs.
A \emph{Boolean formula} can be considered as a Boolean circuit
where all gates have fan-out~$\leq 1$.
We adopt the usual notions of truth assignments
and satisfiability of a Boolean circuit.
We say that a truth assignment for a Boolean circuit has \emph{weight~$k$}
if it sets exactly~$k$ of the variables of the circuit to true.
We denote the class of Boolean circuits with depth~$u$ and weft~$t$
by~$\Gamma_{t,u}$.
We denote the class of all Boolean circuits by~$\Gamma$,
and the class of all Boolean formulas by~$\Phi$.
For any class~$\CCC$ of Boolean circuits,
we define the following parameterized problem.
\probdef{
  \WSat{\mathcal{C}}

  \emph{Instance:} A Boolean circuit~$C \in \CCC$, and an integer~$k$.

  \emph{Parameter:} $k$.

  \emph{Question:} Does there exist an assignment of weight~$k$ that satisfies~$C$?
}

We denote closure under fpt-reductions by~$\fptclosure{\cdot}$.
The classes \W{t} are defined by letting
$\W{t} = \fptclosure{\SB \WSat{\Gamma_{t,u}} \SM u \geq 1 \SE}$,
for~$t \geq 1$.
The classes \W{SAT} and \W{P} are defined by letting
$\W{SAT} = \fptclosure{\WSat{\Phi}}$
and $\W{P} = \fptclosure{\WSat{\Gamma}}$.

The following parameterized complexity classes are analogues
to classical complexity classes.
Let~$\mtext{K}$ be a
classical complexity class, e.g., \NP{}.  The parameterized complexity
class~$\para{K}$ is then defined as the class of all parameterized
problems~$L \subseteq \Sigma^{*} \times \mathbb{N}$, for some finite
alphabet~$\Sigma$, for which there exist an alphabet~$\Pi$, a
computable function~$f : \mathbb{N} \rightarrow \Pi^{*}$, and a
problem~$P \subseteq \Sigma^{*} \times \Pi^{*}$ such that~$P \in K$
and for all instances~$(x,k) \in \Sigma^{*} \times \mathbb{N}$ of~$L$
we have that~$(x,k) \in L$ if and only if~$(x,f(k)) \in P$.
Intuitively, the class~$\para{K}$ consists of all problems that are
in~$\mtext{K}$ after a precomputation that only involves the parameter.
The class~$\para{\NP}$ can also be defined via
nondeterministic fpt-algorithms~\cite{FlumGrohe03}.

The following parameterized complexity classes are
different analogues to classical complexity classes.
Let~$\mtext{K}$ be a
classical complexity class, e.g., \PTIME{}.
The parameterized complexity class~$\X\mtext{K}^{\mtext{nu}}$ is defined
as the class of those parameterized problems~$Q$ whose slices~$Q_k$
are in~$\mtext{K}$, i.e., for each positive integer~$k$
the classical problem~$Q_k = \SB x \SM (x,k) \in Q \SE$
is in~$\mtext{K}$~\cite{DowneyFellows99}.
For instance, the class~$\XP^{\mtext{nu}}$ consists
of those parameterized problems
whose slices are decidable in polynomial time.
Note that this definition is non-uniform, that is,
for each positive integer~$k$ there might be a
completely different polynomial-time algorithm that witnesses
that~$P_k$ is polynomial-time solvable.
For the purposes of this paper, we will restrict ourselves to uniform
variants~$\X\mtext{K}$ of these classes~$\X\mtext{K}^{\mtext{nu}}$.
Concretely, we define \XP{} to be the class of parameterized
problems~$P$ for which there exists a computable function~$f$
and an algorithm~$A$ that decides whether~$(x,k) \in P$
in time~$\Card{x}^{f(k)}$~\cite{DowneyFellows99,DowneyFellows13,
FlumGrohe03,FlumGrohe06}.
Similarly, we define \XNP{} to be the class of parameterized
problems that are decidable in nondeterministic time~$\Card{x}^{f(k)}$.
Its dual class we denote by \XcoNP{}.

\subsection{The Polynomial Hierarchy}

There are many natural decision problems that are not
contained in the classical complexity classes \PTIME{} and \NP{}.
The \emph{Polynomial Hierarchy (PH)}
\cite{MeyerStockmeyer72,Papadimitriou94,Stockmeyer76,Wrathall76}
contains a hierarchy
of increasing complexity classes~$\SigmaP{i}$,
for all~$i \geq 0$.
We give a characterization of these classes based on the satisfiability
problem of various classes of quantified Boolean formulas.
A \emph{quantified Boolean formula} is a formula of the form
$Q_1X_1 Q_2X_2 \dotsc Q_m X_m \psi$,
where each~$Q_i$ is either~$\forall$ or~$\exists$,
the~$X_i$ are disjoint sets of propositional variables,
and~$\psi$ is a Boolean formula over the variables in~$\bigcup_{i=1}^{m} X_i$.
The quantifier-free part of such formulas is called the \emph{matrix} of the formula.
Truth of such formulas is defined in the usual way.
%
%We let $\psi[\alpha]$ denote the formula obtained from $\psi$ by instantiation
%variables by their truth values given by a (partial) truth assignment $\alpha$.
%
Let~$\gamma = \SBs x_1 \mapsto d_1, \dotsc, x_n \mapsto d_n \SEs$
be a function that maps some variables
of a formula~$\varphi$ to other variables or to truth values.
We let~$\varphi[\gamma]$ denote the application of
such a substitution~$\gamma$ to the formula~$\varphi$.
We also write~$\varphi[x_1 \mapsto d_1,\dotsc,x_n \mapsto d_n]$
to denote~$\varphi[\gamma]$.
For each~$i \geq 1$ we define the following decision problem.
\probdef{
  $\QSat{i}$

  \emph{Instance:} A quantified Boolean formula
    $\varphi = \exists X_1 \forall X_2 \exists X_3 \dotsc Q_i X_i \psi$,
    where~$Q_i$ is a universal quantifier if~$i$ is even
    and an existential quantifier if~$i$ is odd.

  \emph{Question:} Is~$\varphi$ true?
}

Input formulas to the problem \QSat{i}
are called \SigmaP{i}-formulas.
For each nonnegative integer~$i \leq 0$,
the complexity class~$\SigmaP{i}$ can be characterized as the
closure of the problem~$\QSat{i}$ under polynomial-time
reductions~\cite{Stockmeyer76,Wrathall76}.
The~\SigmaP{i}-hardness of~$\QSat{i}$ holds
already when the matrix of the input formula
is restricted to~$3\CNF$ for odd~$i$,
and restricted to~$3\DNF$ for even~$i$.
Note that the class~$\SigmaP{0}$ coincides with \PTIME{},
and the class~$\SigmaP{1}$ coincides with \NP{}.
For each~$i \geq 1$, the class~$\PiP{i}$
is defined as~$\co{\SigmaP{i}}$.

\longversion{
The classes \SigmaP{i} and \PiP{i} can also be defined by means of
nondeterministic Turing machines with an oracle.
For any complexity class~$C$, we let~$\NP^{C}$ be the set of
decision problems that is decided in polynomial time
by a nondeterministic Turing machine with an oracle
for a problem that is complete for the class~$C$.
Then, the classes \SigmaP{i} and \PiP{i}, for~$i \geq 0$,
can be equivalently defined as follows:
\[ \SigmaP{0} = \PiP{0} = \PTIME{}, \]
and for each~$i \geq 1$:
\[ \SigmaP{i} = \NP^{\SigmaP{i-1}} \qquad\mtext{ and }\qquad
\PiP{i} = \co\NP^{\SigmaP{i-1}}. \]

}

\subsection{Fixed-parameter Tractable Reductions to SAT}

The satisfiability problem for propositional formulas is denoted
informally by SAT.
Formally, we consider the
(non-parameterized) decision
problems~$\thSAT = \SB \varphi \SM
\varphi$ is a satisfiable propositional formula~$\SE$
and~$\thUNSAT = \SB \varphi \SM
\varphi$ is an unsatisfiable propositional formula~$\SE$.
By a slight abuse of notation, we will sometimes also
use~\thSAT{} to refer to the (trivial) parameterized variant
of the problem where the parameter value~$k=1$ is a fixed
constant for all instances,
i.e., to refer to the language~$\SB (\varphi,1) \SM
\varphi$ is a satisfiable propositional formula~$\SE$.
We use a similar convention for the problem~\thUNSAT{}.
In all cases, it is clear from the context whether the parameterized
or the non-parameterized variant is meant.

Many problems can be encoded into a propositional formula
and subsequently solved by a SAT solver.
Every problem in~$\NP \cup \co\NP$ can be solved with one call to a
SAT solver, and every problem in~$\DP = \SB L_1 \cap L_2 \SM
L_1 \in \NP, L_2 \in \co{\NP} \SE$ can be
solved with two calls to a SAT solver.  The Boolean
Hierarchy~\cite{CaiHemachandra86}
contains all problems that can be solved with
a constant number of calls to a SAT solver.  On the other hand,
(assuming that the PH does not collapse)
there are problems in \DeltaP{2}
that cannot be solved efficiently with a constant number of calls to a
SAT solver.

The problem \thSAT{} is \para{}\NP{}-complete,
and the problem \thUNSAT{} is \para{}\co{}\NP{}-complete.
Hence, by showing that a parameterized problem is in
\para{\NP} or \para{\co\NP} (see Section~\ref{sec:pct}) we establish
that the problem admits a (many-to-one) fpt-reduction to \thSAT{}
or \thUNSAT{}, respectively.
%(see Figure~\ref{fig:classes} and Table~\ref{table:ex_results}).
In other words, \para{}\NP{} consists of all parameterized problems
that can be fpt-reduced to \thSAT{},
and \para{}\co{}\NP{} consists of all parameterized problems
that can be fpt-reduced to \thUNSAT{}.
Similarly, we can view \XNP{} as the class of parameterized problems
for which there exists an xp-reduction to \thSAT{}
and \XcoNP{} as the class of parameterized problems for which
there exists an xp-reduction to \thUNSAT{}.

We would like to point out that fpt-reductions to SAT,
with running times of the form~$f(k)n^c$
for some function~$f$ and some constant~$c$, can be
quite efficient for small values of the parameter~$k$.
Therefore, fpt-reductions to SAT offer
possibilities for developing algorithms that use SAT solvers
and that are efficient for small values of the parameter~$k$.
On the other hand, xp-reductions to SAT have a polynomial running
time for fixed values of~$k$, but are already impractical for, say, $k>3$,
even with a running time of~$n^{k}$.
Because of this difference in scalability, we do not consider
membership results in \XNP{} or \XcoNP{} to be tractability results,
whereas we do consider problems in \para{}\NP{}
and \para{}\co{}\NP{} as tractable.

In addition, another tractability notion that we could consider
is the class of parameterized problems
that can be solved by an fpt-algorithm that makes~$f(k)$ many
calls to a SAT solver, for some function~$f$.
This notion opens another possibility to obtain (parameterized)
tractability results for problems beyond NP
(cf.~\cite{EndrissDeHaanSzeider14,EndrissDeHaanSzeider15,DeHaanSzeider14b}).

\subsection{Answer Set Programming}
We will use the logic programming setting of answer set programming
(ASP)
(cf.~\cite{BrewkaEiterTruszczynski11,Gelfond07,MarekTruszczynski99})
as a running example
in the remainder of the paper.
A \emph{disjunctive logic program} (or simply: a \emph{program})~$P$
is a finite set of rules of the form
$r = (a_1 \vee \dotsm \vee a_k \leftarrow$ $b_1,\dotsc,b_m,$
$\aspnot{} c_1, \dotsc, \aspnot{} c_n)$, for~$k,m,n \geq 0$,
where all~$a_i$,~$b_j$ and~$c_l$ are atoms.
A rule is called \emph{disjunctive} if~$k > 1$,
and it is called \emph{normal} if~$k \leq 1$
(note that we only call rules with strictly more than one
disjunct in the head disjunctive).
A rule is called \emph{negation-free} if~$n = 0$.
A program is called normal if all its rules are normal,
and called negation-free if all its rules are negation-free.
A rule is called a \emph{constraint} if~$k = 0$.
A rule is called \emph{dual-normal} if it is either a constraint,
or~$m \leq 1$.
We let~$\Atoms{P}$ denote the set of all atoms occurring in~$P$.
By \emph{literals} we mean atoms~$a$ or their negations~$\aspnot{} a$.
With~$\NF{r}$ we denote the rule~$(a_1 \vee \dotsm \vee a_k \leftarrow
b_1,\dotsc,b_m)$.
The \emph{(GL) reduct} of a program~$P$ with respect to a set~$M$ of atoms,
denoted~$P^M$, is the program obtained from~$P$ by:
(i) removing rules with~$\aspnot{} a$ in the body, for each~$a \in M$, and
(ii) removing literals~$\aspnot{} a$ from all other rules
\cite{GelfondLifschitz91}.
An \emph{answer set}~$A$ of a program~$P$ is a subset-minimal model
of the reduct~$P^A$.
The following decision problem is concerned with the question
of whether a given program has an answer set.
\probdef{
  \ASPcons{}

  \emph{Instance:} A disjunctive logic program~$P$.

  \emph{Question:} Does~$P$ have an answer set?
}

Many implementations of answer set programming already
employ SAT solving techniques, e.g.,
Cmodels \cite{GiunchigliaLierlerMaratea06},
ASSAT \cite{LinZhao04},
and Clasp \cite{GebserKaufmannNeumannSchaub07}.
Work has also been done on translations from ASP
to SAT, both for classes of programs that allow reasoning
within \NP{} or \co\NP{}
\cite{BeneliyahuDechter94,Fages94,
JanhunenNiemelaSeipelSimonsYou06,LinZhao04}
and for classes of programs
for which reasoning is beyond \NP{} and \co\NP{}
\cite{JanhunenNiemelaSeipelSimonsYou06,LeeLifschitz03,
LifschitzRazborov06}.
We hope that our work provides new means for a theoretical study
of these and related approaches to ASP.

\begin{table}[tb]
  \centering
  \begin{tabular}{@{}l@{\quad}l@{}}
    \textbf{Parameter} & \textbf{Complexity} \\
    \toprule
    normality-bd size & \para{\NP}-complete
    \krversion{\\ &}%
    \cite{FichteSzeider13} \\
    \midrule
    $\#$ contingent atoms & \para{\co\NP}-complete (\krversion{Prop}\longversion{Proposition}~\ref{prop:asp-contatoms-completeness}) \\
    \midrule
    $\#$ contingent rules & \EkA{}-complete (\krversion{Thm}\longversion{Theorem}~\ref{thm:asp-contrules-completeness}) \\
    \midrule
    $\#$ disjunctive rules & \EAkW{P}-complete (\krversion{Thm}\longversion{Theorem}~\ref{thm:asp-disjrules}) \\
    \midrule
    max atom occurrence & \para{\SigmaP{2}}-complete (\krversion{Cor}\longversion{Corollary}~\ref{cor:asp-varocc-completeness}) \\
    \bottomrule
  \end{tabular}

  \caption{Complexity results for different parameterizations
    of \ASPcons{}.\krversion{\vspace{-15pt}}}
  \label{table:ex_results}
\end{table}
}

\krlongversion{
%%% PARAMETERIZATIONS FOR ASP (KR / ARXIV)
\section{Parameterizations for Answer Set Programming}
\label{sec:params-for-asp}

\ASPcons{} is \SigmaP{2}-complete in general,
and can therefore
not be reduced to \SAT{} in polynomial time
(assuming that the PH does not collapse).
With the aim of identifying fpt-reductions from
\ASPcons{} to SAT,
we consider several parameterizations.

Fichte and Szeider~\cite{FichteSzeider13} identified
one parameterization
of \ASPcons{} under which the problem
is contained in \para{\NP}.
This parameterization is based on the notion of backdoors to normality
for disjunctive logic programs.
A set~$X$ of atoms is a \emph{normality-backdoor} for a program~$P$
if deleting the atoms~$x \in X$ and their negations
$\aspnot{} x$ from the rules of~$P$ results in a normal program.
\ASPcons{} is contained in \para{\NP},
when parameterized by the size of a smallest normality-backdoor
of the input program.

Two other parameterizations that we consider are related to atoms
that must be part of any answer set of a program~$P$.
We identify a subset \Comp{P} of \emph{compulsory atoms},
that any answer set must include.
Given a program~$P$, we let \Comp{P} be the smallest set such that:
(i) if~$(w \leftarrow \aspnot{} w)$ is a rule of~$P$, then~$w \in \Comp{P}$;
and (ii) if~$(b \leftarrow a_1,\dotsc,a_n)$ is a rule of~$P$, and~$a_1,\dotsc,a_n \in \Comp{P}$,
then~$b \in \Comp{P}$.
We then let the set \Cont{P} of \emph{contingent atoms}
be those atoms that occur in~$P$ but are not in \Comp{P}.
We call a rule \emph{contingent} if it contains contingent atoms
in the head.
(In fact, we could use any polynomial time computable algorithm~$A$ that computes
for every program~$P$ a set~$\CompA{A}{P}$ of atoms that must
be included in any answer set of~$P$.%
\longversion{ In this paper, we restrict ourselves to the algorithm described above
that computes~$\Comp{P}$.})

\sloppypar
The following are candidates for additional parameters that could result in
fpt-reductions to \SAT{}:
(i) the number of disjunctive rules in the program
(i.e., the number of rules with strictly more than one disjunct in the head);
(ii) the number of contingent atoms in the program;
(iii) the number of contingent rules in the program; and
(iv) the number of rules in the program that are not dual-normal.
We will often denote the parameterized problems based on \ASPcons{}
and these parameters
(i) \ASPcons\DisjRules{},
(ii) \ASPcons\ContAtoms{},
(iii) \ASPcons\ContRules{}, and
(v) \ASPcons\NonDualNormalRules{}, respectively.

The question that we would like to answer
is which (if any) of these parameterizations allows an fpt-reduction to \SAT{}.
Tools from classical complexity theory
seem unfit to distinguish these parameters from each other
and from the parameterization by Fichte and
Szeider~\cite{FichteSzeider13}:
if the parameter values are given as part of the input,
the problem remains \SigmaP{2}-complete in all cases;
if we bound the parameter values by a constant,
then in all cases the complexity of the problem decreases
to the first-level of the PH\krversion{ (a proof of this can be found
in the technical report)}\longversion{ (we will prove
this below)}.
However,
some of the parameterizations allow an fpt-reduction
to \SAT{}, whereas others seemingly do not.

Also the existing tools from parameterized complexity theory are unfit
to distinguish between these different parameterizations of
\ASPcons{}.
Practically all existing parameterized complexity classes
that are routinely used to show that an fpt-reduction is unlikely to exist
(such as the classes of the W-hierarchy)
are located below \para{\NP}.
Therefore, these classes do not allow us to
differentiate between problems that are in \para{\NP}
and problems that are not.

However,
using the parameterized complexity classes developed in this paper
we will be able to make the distinction between parameterizations
that allow an fpt-reduction to \SAT{} and parameterizations
that seem not to allow this.
Furthermore, our theory relates the latter ones in such a way that
an fpt-reduction to \SAT{} for any
of them gives us an fpt-reduction to \SAT{} for all of them.
As can be seen in Table~\ref{table:ex_results},
\ASPcons\ContAtoms{} can be fpt-reduced to \SAT{},
whereas we have evidence that this is not possible
for \ASPcons\DisjRules{} and \ASPcons\ContRules{}.

%\krversion{\enlargethispage*{5mm}}
We will use \ASPcons{} together with the various parameterizations
discussed above as a running example,
which allows us to demonstrate the developed
theoretical tools.
We begin with showing a positive result for
\ASPcons\ContAtoms{}.

\begin{proposition}
\label{prop:asp-contatoms-completeness}
\ASPcons\ContAtoms{}
is \para{\co\NP}-complete.
\end{proposition}
\begin{proof}
Hardness for \para{\co\NP} follows from the
reduction of Eiter and Gottlob \cite[Theorem~3]{EiterGottlob95}.
They give a reduction from \QSat{2} to \ASPcons{}.
We can view this reduction as a polynomial-time reduction
from \thUNSAT{} to the slice of \ASPcons\ContAtoms{}
where the parameter value is~$0$.
Namely, considering an instance of \thUNSAT{} as an instance
of \QSat{2} with no existentially quantified variables,
the reduction results in an equivalent instance of \ASPcons{}
that has no contingent atoms.
Therefore, we can conclude that \ASPcons\ContAtoms{}
is \para{}\co{}\NP{}-hard.

We show membership in \para{\co\NP}.
Let~$P$ be a program that contains~$k$ many contingent atoms.
We describe an fpt-reduction to \thSAT{} for the dual problem
whether~$P$ has no answer set.
This is then also an fpt-reduction from \ASPcons\ContAtoms{}
to \thUNSAT{}.
Since each answer set of~$P$ must contain all atoms in~$\Comp{P}$,
there are only~$2^k$ candidate answer sets that we need to consider,
namely~$N \cup \Comp{P}$ for each~$N \subseteq \Cont{P}$.
For each such set~$M_{N} = N \cup \Comp{P}$ it can be checked in
deterministic polynomial time whether~$M_{N}$ is a model of~$P^{M_{N}}$,
and it can be checked by an \NP{}-algorithm whether~$M_{N}$ is not a
minimal model of~$P^{M_{N}}$
(namely, a counterexample consisting of a
model~$M' \subsetneq M_{N}$ of~$P^{M_{N}}$
can be found in nondeterministic polynomial time).
Therefore, by the NP-completeness of \SAT{},
for each~$N \subseteq \Cont{P}$,
there exists a propositional formula~$\varphi_{N}$ that is satisfiable
if and only if~$M_{N}$ is not a minimal model of~$P^{M_{N}}$.
Moreover, we can construct such a formula~$\varphi_{N}$
in polynomial time, for each~$N \subseteq \Cont{P}$.
All together, the statement that for
no~$N \subseteq \Cont{P}$
the set~$N \cup \Comp{P}$ is an answer set holds true
if and only if the disjunction~$\bigvee_{N \subseteq \Cont{P}}\varphi_{N}$
is satisfiable.
\end{proof}

\longversion{
\paragraph{Membership in the first level of the PH for constant parameter values}
We show that bounding the number of disjunctive or contingent
rules to a constant reduces the complexity of \ASPcons{}
to the first level of the polynomial hierarchy.
Firstly, restricting the number of disjunctive rules to a fixed constant
reduces the complexity of the problem to \NP{}.
In order to prove this, we use the following lemma.

\begin{lemma}
\label{lem:asp-disjrules-mu}
Let~$P$ be a negation-free disjunctive logic program,
and let~$M$ be a minimal model of~$P$.
Then there exists a subset~$R \subseteq P$ of
disjunctive rules and a mapping~$\mu : R \rightarrow \Atoms{P}$
such that:
\begin{itemize}
  \item for each~$r \in R$,
    the value~$\mu(r)$ is an atom in the head of~$r$; and
  \item $M = M_{\mu}$, where $M_{\mu}$ is the smallest set such that:
  \begin{itemize}
    \item $\Rng{\mu} \subseteq M_{\mu}$, and
    \item if~$b_1,\dotsc,b_m \in M_{\mu}$, and~$(a \leftarrow b_1,\dotsc,b_m) \in P$,
      then~$a \in M_{\mu}$.
  \end{itemize}
\end{itemize}
\end{lemma}
\begin{proof}
We give an indirect proof.
Assume that~$M$ is a minimal model of~$P$,
but there exist no suitable~$R$ and~$\mu$.
We will derive a contradiction.
Since~$M$ is a model of~$P$,
we know that for each disjunctive rule~$r_i$ either holds
(i) that~$M$ does not satisfy the body, or (ii)
that~$M$ satisfies an atom~$a_i$ in the head.
We construct the set~$R$ and the mapping~$\mu$ as follows.
For each~$r_i$, we let~$r_i \in R$ and~$\mu(r_i) = a_i$ if and only
if~$M$ satisfies an atom~$a_i$ in the head of the disjunctive rule~$r_i$.

Clearly,~$\Rng{\mu} \subseteq M$.
Define~$A_0 = \Rng{\mu}$.
For each~$i \in \mathbb{N}$, we define~$A_{i+1}$ as follows:
\[ A_{i+1} = A_{i} \cup \SB a \SM (a \leftarrow b_1,\dotsc,b_m) \in P,
b_1,\dotsc,b_m \in A_i \SE. \]
We show by induction on~$i$ that~$A_i \subseteq M$ for all~$i \in \mathbb{N}$.
Clearly,~$A_0 \subseteq M$.
Assume that~$A_i \subseteq M$,
and let~$a \in A_{i+1} \backslash A_i$ be an arbitrary atom.
This can only be the case if~$b_1,\dotsc,b_m \in A_i$ and
$(a \leftarrow b_1,\dotsc,b_m) \in P$.
However, since~$M$ is a model of~$P$, also~$a \in M$.
Therefore,~$A_{i+1} \subseteq M$.

Now, define~$M' = \bigcup_{i \in \mathbb{M}}A_i$.
We have that~$M' \subseteq M$.
We show that~$M'$ is a model of~$P$.
Clearly,~$M'$ satisfies all normal rules of~$P$.
Let~$r = (a_1 \vee \dotsm \vee a_n \leftarrow b_1,\dotsc,b_n)$ be a disjunctive rule
of~$P$, and assume that~$M'$ does not satisfy~$r$.
Then it must be the case that~$b_1,\dotsc,b_m \in M'$ and~$a_1,\dotsc,a_n \not\in M'$.
However, since~$M' \subseteq M$, we know that~$M$ satisfies the body of~$r$.
Therefore,~$a_i \in M$, for some~$1 \leq i \leq n$, and
thus~$r \in R$ and~$\mu(r) = a_i$, and so~$a_i \in A_0 \subseteq M'$.
Thus,~$M'$ is a model of~$P$.

If~$M' \subsetneq M$, we have a contradiction with the fact
that~$M$ is a minimal model of~$P$.
Otherwise, if~$M' = M$, then we have a contradiction with the fact
that~$M$ cannot be represented in the way described above by
suitable~$R$ and~$\mu$.
This concludes our proof.
\end{proof}

\begin{proposition}
\label{prop:asp-disjrules-xnp}
Let~$d$ be a fixed positive integer.
The restriction of \ASPcons{} to programs containing at most~$d$
disjunctive rules is in \NP{}.
\end{proposition}
\begin{proof}
We sketch a guess-and-check algorithm~$A$ that solves the problem.
Let~$P$ be a disjunctive logic program with at most~$d$ disjunctive rules.
The algorithm~$A$ guesses a subset~$M \subseteq \Atoms{P}$.
By Lemma~\ref{lem:asp-disjrules-mu} we know that it is possible,
given~$M$ as input,
to verify in polynomial time whether~$M$ is a minimal model of~$P^M$.
This is the case because there are at most~$O((\Card{\Atoms{P}}+1)^d)$ many combinations of a suitable set~$R$ and a suitable mapping~$\mu$,
and verifying for each such~$\mu$ whether~$M = M_{\mu}$ can be done
in polynomial time.
The algorithm~$A$ accepts if and only if~$M$ is a minimal model of~$P^M$,
and thus it accepts if and only if~$P$ has an answer set.
\end{proof}

\noindent Restricting the number of contingent rules to a fixed constant
reduces the complexity of the problem to \co\NP{}.

\begin{lemma}
\label{lem:asp-contrules-mu}
Let~$P$ be a disjunctive logic program,
and let~$M$ be an answer set of~$P$.
Then there exists a subset~$R \subseteq P$
of contingent rules
and a mapping~$\mu : R \rightarrow \Cont{P}$
such that for each~$r \in R$ it holds that~$\mu(r)$ occurs in the head of~$r$,
and~$M = \Comp{P} \cup \Rng{\mu}$.
\end{lemma}
\begin{proof}
We show that~$M = \Rng{\mu} \cup \Comp{P}$ for some subset~$R \subseteq P$
of contingent rules and some
mapping~$\mu : R \rightarrow \Cont{P}$ such that for each~$r \in R$
it holds that~$\mu(r)$ occurs in the head of~$r$.
To show that~$M$ can be represented in this fashion by suitable~$R$ and~$\mu$,
assume the contrary, i.e., that~$M$ is an answer set of~$P$
but for no subset~$R \subseteq P$ of contingent rules there exists
a suitable mapping~$\mu$ such that~$M = \Comp{P} \cup \Rng{\mu}$.
Since~$M$ is an answer set of~$P$, we know that for each contingent rule~$r$ of~$P$
it holds that either (i)~$\NF{r} \not\in P^{M}$,
or (ii)~$\NF{r} \in P^{M}$ and~$M$ satisfies the body of~$\NF{r}$,
or (iii)~$\NF{r} \in P^{M}$ and~$M$ does not satisfy the body of~$\NF{r}$.
Recall that~$\NF{r}$ is the rule~$r$ where all negative literals are removed.
We construct the subset~$R \subseteq P$ of contingent rules and the
mapping~$\mu : R \rightarrow \Cont{P}$ as follows.
For each~$r$, if case (i) or (iii) holds, we let~$r \not\in R$.
If case (ii) holds for rule~$r$, then since~$M$ is a model of~$P^{M}$,
we know that there exists some~$d \in M$ such that~$d$ occurs in the head of~$\NF{r}$.
We again distinguish two cases: either (ii.a) there is some~$d \in \Cont{P} \cap M$
such that~$d$ occurs in the head of~$\NF{r}$,
or (ii.b) this is not the case.
In case (ii.a), we let~$r \in R$ and we let~$\mu(r) = d$.
In case (ii.b), we let~$r \not\in R$.
Clearly,~$\Rng{\mu} \subseteq M$, so~$\Rng{\mu} \cup \Comp{P} \subseteq M$.
Now, to show that~$M \subseteq \Rng{\mu} \cup \Comp{P}$, assume the contrary,
i.e., assume that there exists some~$d \in (M \cap \Cont{P})$
such that~$d \not\in \Rng{\mu}$.
Then~$M \backslash \SBs d \SEs$ is a model of~$P^{M}$, and
therefore~$M$ is not a subset-minimal model of~$P^{M}$.
This is a contradiction with our assumption that~$M$ is an answer set of~$P$.
Therefore,~$M = \Rng{\mu} \cup \Comp{P}$,
which contradicts our assumption that no suitable~$R$ and~$\mu$ exist.
\end{proof}

\begin{proposition}
\label{prop:asp-contrules-xconp}
Let~$d$ be a fixed positive integer.
The restriction of \ASPcons{} to programs containing at most~$d$
contingent rules is in \co\NP{}.
\end{proposition}
\begin{proof}
We sketch a guess-and-check algorithm~$A$ that decides whether programs~$P$
containing at most~$d$ contingent rules have no answer set.
Let~$P$ be an arbitrary program.
Then any answer set can be represented by means of a suitable
subset~$R \subseteq P$ and a suitable mapping~$\mu$,
as described in Lemma~\ref{lem:asp-contrules-mu}.
Let~$M_{R,\mu}$ denote the possible answer set corresponding to~$R$ and~$\mu$.
There are at most~$O((\Card{P}+1)^d)$ many candidate sets~$M_{R,\mu}$.
The algorithm~$A$ guesses a subset~$M'_{R,\mu} \subsetneq M_{R,\mu}$
for each such~$R$ and~$\mu$.
Then, it verifies whether for all~$M'_{R,\mu}$ it holds that~$M'_{R,\mu}$
is a model of~$P^{M_{R,\mu}}$, and
it accepts if and only if this is the case.
Therefore, it accepts if and only if~$P$ has no answer set.
\end{proof}
}

Observe that the algorithms given in the proofs of
Propositions~\ref{prop:asp-disjrules-xnp}
and~\ref{prop:asp-contrules-xconp} are the same for
each positive value~$d$ of the parameter
(only the running times differ for different parameter values).
Therefore, we get the following membership results
in \XNP{} and \XcoNP{}.

\begin{corollary}
\ASPcons\DisjRules{} is in \XNP{}.
\end{corollary}

\begin{corollary}
\ASPcons\ContRules{} is in \XcoNP{}.
\end{corollary}

Moreover, unless~$\NP{} = \co{}\NP{}$, the class \para{}\SigmaP{2}
is neither contained in \XNP{} as a subset,
nor in \XcoNP{} \cite[Proposition~8]{FlumGrohe03}.
Therefore, it is very unlikely that the problems \ASPcons\DisjRules{}
and \ASPcons\ContRules{} are \para{}\SigmaP{2}-complete.
Therefore, if we want to give evidence that these problems
are not contained in \para{}\NP{} or \para{}\co{}\NP{},
we will need a new set of tools.
The hardness theory that we define in the following section
will be this new toolbox.
}

\krlongversion{
%%% THE HIERARCHIES (KR / ARXIV)
\section{The  Hierarchies \stark{} and \kstar{}}
\label{sec:define-hierarchies}

% Maybe we can use this somewhere else:

% An alternative perspective on the theory of fixed-parameter tractability
% is to view fpt-reductions as breaking barriers between
% classical complexity classes.
% For instance, an fpt-reduction breaking the barrier between \NP{} and \PTIME{}
% corresponds to a fixed-parameter tractability result.
% %
% As mentioned above, fpt-reductions can also be used to break barriers beyond \NP{}.
% We are particularly interested in fpt-reductions breaking the barrier between \NP{}
% and \SigmaP{2} (or classes higher in the polynomial hierarchy),
% as these correspond to \para{\NP}-membership results.
% %
% In order to develop a hardness theory that can provide evidence that some problems
% are not in \para{\NP}, we need to identify barriers that separate \para{\NP} from
% harder problems, and that cannot be broken by fpt-reductions.
% %
% We propose various parameterized complexity classes
% that allow us to identify exactly the barriers that we are after.

We are going to define two hierarchies of parameterized complexity classes
that will act as intractability classes in our hardness theory.
All classes will be based on weighted variants
of the satisfiability problem \QSat{2}.
An instance of the  problem \QSat{2} has both an existential quantifier
and a universal quantifier block.
Therefore, there are several ways of restricting the weight of assignments.
Restricting the weight of assignments to the existential quantifier block
will result in the \kstar{} hierarchy, and restricting the weight of
assignments to the universal quantifier block
will result in the \stark{} hierarchy.
The two hierarchies are based on the following two parameterized decision problems.
Let~$\CCC$ be a class of Boolean circuits.
The problem~$\EkAWSat(\CCC)$ provides the foundation for the \kstar{} 
hierarchy.
\probdef{
  $\EkAWSat(\CCC)$
  
  \emph{Instance:} A Boolean circuit~$C \in \CCC$
  over two disjoint sets~$X$ and~$Y$ of variables,
  and an integer~$k$.

  \emph{Parameter:} $k$.

  \emph{Question:} Does there exist a truth assignment~$\alpha$
  to~$X$ with weight~$k$
  such that for all truth assignments~$\beta$ to~$Y$
  the assignment~$\alpha \cup \beta$ satisfies~$C$?
}

\noindent Similarly, the problem~$\EAkWSat(\CCC)$
provides the foundation for the \stark{} hierarchy.
\probdef{
  $\EAkWSat(\CCC)$
  
  \emph{Instance:} A Boolean circuit~$C \in \CCC$
  over two disjoint sets~$X$ and~$Y$ of variables,
  and an integer~$k$.

  \emph{Parameter:} $k$.

  \emph{Question:} Does there exist a truth assignment~$\alpha$
  to~$X$
  such that for all truth assignments~$\beta$ to~$Y$
  with weight~$k$
  the assignment~$\alpha \cup \beta$ satisfies~$C$?
}

For the sake of convenience,
instances to these two problems consisting of a circuit~$C$
over sets~$X$ and~$Y$ of variables and an integer~$k$,
we will denote by~$(\exists X. \forall Y. C, k)$.
\longversion{We now define the following parameterized complexity classes,
that together form the \kstar{} hierarchy:
\begin{eqnarray*}
  \EkAW{t} =& \fptclosure{ \SB \EkAWSat(\Gamma_{t,u}) \SM u \geq 1 \SE }, \\
  \EkAW{SAT} =& \fptclosure{ \EkAWSat(\Phi) }, \mtext{ and} \\
  \EkAW{P} =& \fptclosure{ \EkAWSat(\Gamma) }.
\end{eqnarray*}
Similarly, we define the classes of the \stark{} hierarchy as follows:
\begin{eqnarray*}
  \EAkW{t} =& \fptclosure{ \SB \EAkWSat(\Gamma_{t,u}) \SM u \geq 1 \SE }, \\
  \EAkW{SAT} =& \fptclosure{ \EAkWSat(\Phi) }, \mtext{ and} \\
  \EAkW{P} =& \fptclosure{ \EAkWSat(\Gamma) }.
\end{eqnarray*}
Note that these definitions are entirely analogous to those of the parameterized
complexity classes of the W-hierarchy~\cite{DowneyFellows99}.
}
\krversion{We now define the following parameterized complexity classes,
that together form the \kstar{} hierarchy.
We let~$\EkAW{t} = \fptclosure{ \SB \EkAWSat(\Gamma_{t,u}) \SM u \geq 1 \SE }$,
we let~$\EkAW{SAT} = \fptclosure{ \EkAWSat(\Phi) }$, and
we let~$\EkAW{P} = \fptclosure{ \EkAWSat(\Gamma) }$.

We define the classes of the \stark{} hierarchy similarly.
We let~$\EAkW{t} = \fptclosure{ \SB \EAkWSat(\Gamma_{t,u}) \SM u \geq 1 \SE }$,
we let~$\EAkW{SAT} = \fptclosure{ \EAkWSat(\Phi) }$, and
we let~$\EAkW{P} = \fptclosure{ \EAkWSat(\Gamma) }$.
Note that these definitions are analogous to those of the parameterized
complexity classes of the W-hierarchy~\cite{DowneyFellows99}.
}

\longversion{Dual to the classical complexity class \SigmaP{2} is its co-class
\PiP{2}, whose canonical complete problem is complementary to the
problem \QSat{2}. Similarly, we}\krversion{We} can define dual classes for each of
the parameterized complexity classes in the \kstar{} and \stark{}
hierarchies.  These co-classes are based on problems complementary to
the problems \EkAWSat{} and \EAkWSat{}, i.e., these problems have as
yes-instances exactly the no-instances of \EkAWSat{} and \EAkWSat{},
respectively.  Equivalently, these complementary problems can be
considered as variants of \EkAWSat{} and \EAkWSat{} where the
existential and universal quantifiers are swapped, and are therefore
denoted with \AkEWSat{} and \AEkWSat{}.
We use a similar notation for the dual complexity classes,
e.g., we denote \co{\EAkW{t}} by \AEkW{t}.

\longversion{
\paragraph{More hierarchies}
Similarly to the definition of the complexity classes
of the \kstar{} and \stark{} hierarchies,
the problem \QSat{i} for any~$i$
and any weight restriction on the quantifier blocks
can be used to define hierarchies of parameterized complexity classes.
One example we would like to point out is the hierarchy
of classes \EkAkW{t},
based on a variant of \QSat{2} where both variable blocks
have a restriction on the weight of assignments.
The complexity class~$\EkAkW{SAT}$
has been defined and considered by Gottlob, Scarcello
and Sideri~\cite{GottlobScarcelloSideri02}
under the name $\Sigma_{2}\W{\SAT}$.
}
}

\krlongversion{
%%% THE CLASS \Exists^k \Forall (KR / ARXIV)
\section{The Class \EkA{}}
\label{sec:eka}

In this section, we consider the \kstar{} hierarchy.
It turns out that this hierarchy collapses entirely
into a single parameterized complexity class.
This class we will denote by \EkA{}.
As we will see, the class \EkA{} turns out to be quite robust.
We start this section with showing that
that the \kstar{} hierarchy collapses.
We discuss how this class is related to existing parameterized
complexity classes,
and we show how it can be used to show the intractability
of a variant of the answer set existence problem
whose complexity the existing theory cannot classify properly.

\subsection{Collapse of the \kstar{} Hierarchy}
\label{sec:kstarcollapse}

\begin{theorem}[Collapse of the \kstar{} hierarchy]
\label{thm:kstar-collapse}
$\EkAW{1} = \EkAW{2} = \dotsc = \EkAW{SAT} = \EkAW{P}$.
\end{theorem}
\begin{proof}
Since by definition $\EkAW{1} \subseteq \EkAW{2} \subseteq \dotsc \subseteq \EkAW{P}$,
it suffices to show that $\EkAW{P} \subseteq \EkAW{1}$.
We show this by giving an fpt-reduction from $\EkAWSat(\Gamma)$
to $\EkAWSat(3\DNF)$.
Since $3\DNF \subseteq \Gamma_{1,3}$, this suffices.
We remark that this reduction is based on the standard Tseitin transformation
that transforms arbitrary Boolean formulas into $3\CNF$ by means
of additional variables.

Let~$(\varphi,k)$ be an instance of~$\EkAWSat(\Gamma)$
with~$\varphi = \exists X. \forall Y. C$.
Assume without loss of generality that~$C$ contains
only binary conjunctions and negations.
Let~$o$ denote the output gate of~$C$.
We construct an instance~$(\varphi',k)$ of~$\EkAWSat(3\DNF)$ as follows.
The formula~$\varphi'$ will be over the set of
variables~$X \cup Y \cup Z$,
where~$Z = \SB z_{r} \SM r \in \Nodes{C} \SE$.
For each~$r \in \Nodes{C}$, we define a subformula~$\chi_{r}$.
We distinguish three cases.
If~$r = r_1 \wedge r_2$, then we
let~$\chi_{r} = (z_{r} \wedge \neg z_{r_1}) \vee
(z_{r} \wedge \neg z_{r_2}) \vee
(z_{r_1} \wedge z_{r_2} \wedge \neg z_{r})$.
If~$r = \neg r_1$, then we
let~$\chi_{r} = (z_{r} \wedge z_{r_1}) \vee
(\neg z_{r} \wedge \neg z_{r_1})$.
If~$r = w$, for some~$w \in X \cup Y$,
then we let~$\chi_{r} = (z_{r} \wedge \neg w)
\vee (\neg z_{r} \wedge w)$.
Now we define~$\varphi' = \exists X. \forall Y \cup Z. \psi$,
where~$\psi = \bigvee_{r \in \Nodes{C}} \chi_{r} \vee z_{o}$.
\krversion{It is straightforward to verify that this reduction is correct.
}\longversion{We prove the correctness of this reduction.

$(\Rightarrow)$
Assume that~$(\varphi,k) \in \EkAWSat(\Gamma)$.
This means that there exists an
assignment~$\alpha : X \rightarrow \SBs 0,1 \SEs$
of weight~$k$ such that~$\forall Y. C[\alpha]$ evaluates to true.
We show that~$(\varphi',k) \in \EkAWSat(3\DNF)$,
by showing that~$\forall Y \cup Z. \psi[\alpha]$ evaluates to true.
Let~$\beta : Y \cup Z \rightarrow \SBs 0,1 \SEs$ be an arbitrary assignment
to the variables~$Y \cup Z$,
and let~$\beta'$ be the restriction of~$\beta$ to the variables~$Y$.
We distinguish two cases: either (i) for each~$r \in \Nodes{C}$
it holds that~$\beta(z_{r})$ coincides with the value that gate~$r$ gets
in the circuit~$C$ given assignment~$\alpha \cup \beta'$,
or (ii) this is not the case.
In case (i), by the fact that~$\alpha \cup \beta'$ satisfies~$C$,
we know that~$\beta(z_{o}) = 1$, and therefore~$\alpha \cup \beta$
satisfies~$\psi$.
In case (ii), we know that for some gate~$r \in \Nodes{C}$,
the value of~$\beta(z_{r})$ does not coincide with the value assigned to~$r$
in~$C$ given the assignment~$\alpha \cup \beta'$.
We may assume without loss of generality that for all parent nodes~$r'$ of~$r$
it holds that~$\beta(z_{r'})$ coincides with the value
assigned to~$r'$ by~$\alpha \cup \beta'$.
In this case, there is some term of~$\chi_{r}$ that is satisfied
by~$\alpha \cup \beta$.
From this we can conclude that~$(\varphi',k) \in \EkAWSat(3\DNF)$.

$(\Leftarrow)$
Assume that~$(\varphi',k) \in \EkAWSat(3\DNF)$.
This means that there exists some
assignment~$\alpha : X \rightarrow \SBs 0,1 \SEs$
of weight~$k$ such that~$\forall Y \cup Z. \psi[\alpha]$ evaluates to true.
We now show that~$\forall Y. C[\alpha]$ evaluates to true as well.
Let~$\beta' : Y \rightarrow \SBs 0,1 \SEs$ be an arbitrary assignment to the variables~$Y$.
Define~$\beta'' : Z \rightarrow \SBs 0,1 \SEs$ as follows.
For any~$r \in \Nodes{C}$, we let~$\beta''(r)$ be the value assigned
to the node~$r$ in the circuit~$C$ by the assignment~$\alpha \cup \beta'$.
We then let~$\beta = \beta' \cup \beta''$.
Now, since~$\forall Y \cup Z. \psi[\alpha]$ evaluates to true,
we know that~$\alpha \cup \beta$ satisfies~$\psi$.
By construction of~$\beta$, we know that~$\alpha \cup \beta$ does not satisfy
the term~$\chi_{r}$ for any~$r \in \Nodes{C}$.
Therefore, we know that~$\beta(z_{o}) = 1$.
By construction of~$\beta$, this implies that~$\alpha \cup \beta'$ satisfies~$C$.
Since~$\beta'$ was arbitrary, we can conclude that~$\forall Y. C[\alpha]$
evaluates to true, and therefore that~$(\varphi,k) \in \EkAWSat(\Gamma)$.
}
\end{proof}

\noindent As mentioned above, in order to simplify notation, we will use \EkA{}
to denote the class $\EkAW{1} = \dotsc = \EkAW{P}$.
Also, we will denote $\EkAWSat(\Gamma)$ by \EkAWSat{}.
\krversion{We make some observations about the relation of \EkA{}
to existing parameterized complexity classes.
It is straightforward to see that~$\EkA{} \subseteq \para{\SigmaP{2}}$.
In polynomial time, any formula~$\exists X. \forall Y. \psi$ can be transformed into
a \SigmaP{2}-formula that is true if and only if for some assignment~$\alpha$ of weight~$k$
to the variables~$X$ the formula~$\forall Y. \psi[\alpha]$ is true.
Trivially,~$\para{\co{\NP}} \subseteq \EkA{}$.
To summarize, we obtain the following inclusions:
\krversion{$\para{\co\NP} \subseteq \EkA \subseteq \para{}\SigmaP{2}$, and
$\para{\NP} \subseteq \AkE \subseteq \para{}\PiP{2}$.}
\longversion{\[ \para{\co\NP} \subseteq \EkA \subseteq \para{}\SigmaP{2}, \quad
\para{\NP} \subseteq \AkE \subseteq \para{}\PiP{2}. \]}
This immediately leads to the following result.
\begin{proposition}
If~$\EkA \subseteq \para{\NP}$, then~$\NP = \co\NP$.
\end{proposition}
\noindent A natural question to ask is whether~$\para{\NP} \subseteq \EkA$.
The following result indicates that this is unlikely.
\begin{proposition}
\label{prop:1}
If~$\para{\NP} \subseteq \EkA$,
then~$\NP = \co\NP$.
\end{proposition}
\krversion{\begin{proof}[Proof (sketch).]
Let~$\thSAT{}$ be the language of satisfiable propositional
formulas, and~$\thUNSAT{}$ the language of unsatisfiable
propositional formulas.
The parameterized problem~$P = \SB (\varphi,1) \SM \varphi \in \thSAT \SE$
is in \para{\NP}.
Since the parameter value is constant for all instances of~$P$,
an fpt-reduction from~$P$ to \EkAWSat{} can be transformed
into an polynomial time reduction from \thSAT{}
to \thUNSAT{}.
\end{proof}}
\longversion{\begin{proof}
Assume that~$\para{\NP} \subseteq \EkA$.
Note that the language \thSAT{} of satisfiable propositional formulas
is \NP{}-complete, and the language \thUNSAT{} of unsatisfiable propositional
formulas is \co\NP{}-complete.
The parameterized problem~$P = \SB (\varphi,1) \SM \varphi \in \thSAT \SE$
is in \para{\NP}.
Then also~$P \in \EkA$.
This means that there is an fpt-reduction~$R$ from~$P$ to \EkAWSat{}.
We construct a polynomial-time reduction~$S$ from \thSAT{} to \thUNSAT{}.
Let~$\varphi$ be an instance of \thSAT{}.
The reduction~$R$ maps~$(\varphi,1)$ to an instance~$(\varphi',f(1))$ of \EkAWSat{},
where~$f$ is some computable function and~$\varphi' = \exists X. \forall Y. \psi$,
such that~$(\varphi',f(1)) \in \EkAWSat{}$ if and only if~$\varphi \in \thSAT$.
Note that~$f(1)$ is a constant, since~$f$ is a fixed function.
To emphasize this, we let~$c = f(1)$, and we will use~$c$ to denote~$f(1)$.
By definition, we know that~$(\varphi',c) \in \EkAWSat{}$ if and only if
for some truth assignment~$\alpha$ to the variables~$X$ of weight~$c$,
the formula~$\forall Y. \psi[\alpha]$ is true.
Let~$\mtext{ta}(X,c)$ denote the set of all truth assignments to~$X$
of weight~$c$.
We then get that~$\varphi \in \thSAT$ if and only if
the formula~$\forall Y. \chi$ is true,
\krversion{where~$\chi = \bigvee_{\alpha \in \mtext{ta}(X,c)} 
\psi[\alpha]$.}\longversion{where~$\chi$ is defined as follows:
\[ \chi = \underset{\alpha \in \mtext{ta}(X,c)}{\bigvee} \psi[\alpha] \]}
It is straightforward to verify that the mapping
$\varphi \mapsto \neg \chi$ is a polynomial-time reduction
from \thSAT{} to \thUNSAT{}.
\end{proof}}

\noindent This implies that~$\EkA$ is very likely to be a strict subset
of~$\para{\SigmaP{2}}$.
\begin{corollary}
If~$\EkA = \para{\SigmaP{2}}$, then~$\NP = \co\NP$.
\end{corollary}

\noindent The following result shows
another way in which the class \EkA{} relates
to the existing complexity class \co\NP{}.
Let~$P$ be a parameterized decision problem, and let~$c \geq 1$ be an integer.
Recall that the \emph{$c$-th slice of~$P$}, denoted~$P_{c}$, is the (unparameterized)
decision problem~$\SB x \SM (x,c) \in P \SE$.

\begin{proposition}
Let~$P$ be a parameterized problem complete for \EkA{},
and let~$c \geq 1$ be an integer.
Then~$P_{c}$ is in \co\NP{}.
Moreover,
there exists some integer~$d \geq 1$
such that~$P_{1} \cup \dotsm \cup P_{d}$ is \co\NP-complete.
\end{proposition}
\krversion{A proof of this statement can be found in the
technical report.}\longversion{\begin{proof}
We show \co\NP-membership of~$P_{c}$, by
constructing a polynomial-time reduction~$S$ from~$P_c$ to \thUNSAT{}.
Since~$P \in \EkA{}$, we know that there exists an fpt-reduction~$R$ from~$P$
to~$\EkAWSat(\Phi)$.
Therefore, there exist computable functions~$f$ and~$g$ and a polynomial~$p$
such that for all instances~$(x,k)$ of~$P$,~$R(x,k) = (x',k')$
is computable in time~$f(k) \cdot p(\Card{x})$
and~$k' \leq g(k)$.
We describe the reduction~$S$.
Let~$x$ be an arbitrary instance of~$P_c$.
We know~$R$ maps~$(x,c)$ to~$(\varphi,k')$,
for some~$k' \leq g(c)$,
where~$\varphi = \exists X. \forall Y. \psi$.
Note that~$k'$ is bounded by a constant~$g(c) = d$.
Let~$\mtext{ta}(X,k')$ denote the set of all truth assignment to~$X$
of weight~$k'$.
We then get that~$\varphi$ is equivalent to the
formula~$\forall Y. \chi$,
where~$\chi = \bigvee_{\alpha \in \mtext{ta}(X,k')} \psi[\alpha]$.
Also, the size of~$\chi$ is polynomial in the size of~$\psi$.
We now let~$S(x) = \neg\chi$.
It is straightforward to verify that~$S$ is a correct polynomial-time
reduction from~$P_{c}$ to \thUNSAT{}.

We show that there exists a function~$f$
such that for any positive integer~$s \geq 1$,
there is a polynomial-time reduction from \thUNSAT{}
to~$P_{1} \cup \dotsm \cup P_{f(s)}$.
Then, in particular,~$P_{1} \cup \dotsm \cup P_{f(1)}$
is \co\NP{}-complete.
Let~$s \geq 1$ be an arbitrary integer. We construct the reduction~$S$.
There is a trivial polynomial-time reduction~$S$ from \thUNSAT{}
to~$(\EkAWSat)_{s}$, that maps a Boolean formula~$\varphi$
over a set of variables~$Y$ to the instance~$(\exists \SBs x_1,\dotsc,x_s \SEs.
\forall Y. \neg\chi,s)$ of \EkAWSat{}.
Since~$P$ is \EkA{}-complete,
there exists an fpt-reduction~$R$ from~$\EkAWSat$
to~$P$.
From this, we know that there exists a nondecreasing and unbounded
function~$f$ such that
for each instance~$(x,k)$ of~$\EkAWSat$ it holds that~$k' \leq f(k)$,
where~$R(x,k) = (x',k')$.
This reduction~$R$ is a polynomial-time reduction from~$(\EkAWSat)_{s}$
to~$P_{1} \cup \dotsm \cup P_{f(s)}$.
Composing the polynomial-time reductions~$S$ and~$R$,
we obtain a reduction from \thUNSAT{} to~$P_{1} \cup \dotsm \cup P_{f(s)}$.
\end{proof}}}

\label{sec:kstarasp1}
\subsection{Answer Set Programming and Completeness for the \kstar{} Hierarchy}
Now that we defined this new intractability class \EkA{}
and that we have some basic results about it in place,
we are able to prove the intractability
of a variant of our running example problem.
In fact, we show that one variant of our running example
is complete for the class \EkA{}.

\begin{theorem}
\label{thm:asp-contrules-completeness}
\ASPcons\ContRules{} is \EkA{}-complete.
\end{theorem}

\noindent This completeness result follows from
the following two propositions.

\begin{proposition}
\label{prop:asp-contrules-hardness}
\ASPcons\ContRules{}
is \EkA-hard.
\end{proposition}
\begin{proof}
We give an fpt-reduction from~$\EkAWSat(3\DNF)$.
This reduction is a parameterized version of a reduction of Eiter and
Gottlob~\cite[Theorem~3]{EiterGottlob95}.
Let~$(\varphi,k)$ be an instance of~$\EkAWSat(3\DNF)$,
where~$\varphi = \exists X. \forall Y. \psi$,~$X = \SBs x_1,\dotsc,x_n
\SEs$,~$Y = \SBs y_1,\dotsc,y_m \SEs$,~$\psi = \delta_1 \vee
\dotsm \vee \delta_u$,
and~$\delta_{\ell} = l^{\ell}_1 \wedge l^{\ell}_2 \wedge l^{\ell}_3$
for each~$1 \leq \ell \leq u$.
We construct a disjunctive program~$P$.
We consider the sets~$X$ and~$Y$ of variables as atoms.
In addition,
we introduce fresh atoms~$v_1,\dotsc,v_n$,~$z_1,\dotsc,z_m$,~$w$,
and~$x^j_i$ for all~$1 \leq j \leq k$,~$1 \leq i \leq n$.
We let~$P$ consist of the rules described
%in Table~\ref{table:construction1}.
as follows:\\
%%% SHORT VERSION
\krversion{\begin{small}
\begin{align}
  & \label{eka-rule1} x^{j}_{1} \vee \dotsm \vee x^{j}_{n} \leftarrow & \mtext{for } 1 \leq j \leq k; \\[-4pt]
  & \leftarrow x^{j}_{i}, x^{j'}_{i} & \mtext{for } 1 \leq i \leq n,
    \label{eka-rule2} 1 \leq j < j' \leq k; \\[-1pt]
  & \label{eka-rule3} y_{i} \vee z_{i} \leftarrow, \quad w \leftarrow y_{i},z_{i} & \mtext{for } 1 \leq i \leq m; \\[-1pt]
  &\label{eka-rule4} y_{i} \leftarrow w, \quad z_{i} \leftarrow w & \mtext{for } 1 \leq i \leq m; \\[-1pt]
  &\label{eka-rule7} x_{i} \leftarrow w, \quad v_{i} \leftarrow w & \mtext{for } 1 \leq i \leq n; \\[-3pt]
  &x_{i} \leftarrow x^{j}_{i} & \mtext{for } 1 \leq i \leq n,
    \label{eka-rule8} 1 \leq j \leq k; \\[-3pt]
  &\label{eka-rule10} v_{i} \leftarrow \aspnot{} x^{1}_{i}, \dotsc, \aspnot{} x^{k}_{i} \hspace{-20pt} & \mtext{for } 1 \leq i \leq n; \\[-3pt]
  &\label{eka-rule11} w \leftarrow \sigma(l^{\ell}_1), \sigma(l^{\ell}_2), \sigma(l^{\ell}_3) \hspace{-20pt} &
    \mtext{for } 1 \leq \ell \leq u; \\
  &\label{eka-rule12} w \leftarrow \aspnot{} w.
\end{align}
\end{small}}%
%%% LONG VERSION
%\begin{table}[ht!]
%\setlength\dashlinedash{0.5pt}
%\setlength\dashlinegap{4pt}
\renewcommand{\arraystretch}{1.8}
%\caption{Construction of the disjunctive program~$P$ in the
%proof of Proposition~\ref{prop:asp-contrules-hardness}.}
%\label{table:construction1}
\parbox{\textwidth}{
\centering
\begin{small}
%\vspace{5pt}\framebox{
\vspace{10pt}
\begin{tabular}{r @{\ } l r r}
$x^{j}_{1} \vee \dotsm \vee x^{j}_{n}$ & $\leftarrow $ & for~$1 \leq j \leq k$; & \tagarray\label{eka-rule1} \\
%\hdashline
 & $\leftarrow x^{j}_{i}, x^{j'}_{i}$ & for~$1 \leq i \leq n,$ \\[-8pt]
 & & and~$1 \leq j < j' \leq k$; & \tagarray\label{eka-rule2} \\
%\hdashline
$y_{i} \vee z_{i}$ & $\leftarrow $ & for~$1 \leq i \leq m$; & \tagarray\label{eka-rule3} \\
%\hdashline
$y_{i}$ & $\leftarrow w$ & for~$1 \leq i \leq m$; & \tagarray\label{eka-rule4} \\
%\hdashline
$z_{i}$ & $\leftarrow w$ & for~$1 \leq i \leq m$; & \tagarray\label{eka-rule5} \\
%\hdashline
$w$ & $\leftarrow z_{i}, z_{i}$ & for~$1 \leq i \leq m$; & \tagarray\label{eka-rule6} \\
%\hdashline
$x_{i}$ & $\leftarrow w$ & for~$1 \leq i \leq n$; & \tagarray\label{eka-rule7} \\
%\hdashline
$x_{i}$ & $\leftarrow x^{j}_{i}$ & for~$1 \leq i \leq n$ & \\[-8pt]
& & and~$1 \leq j \leq k$; & \tagarray\label{eka-rule8} \\
%\hdashline
$v_{i}$ & $\leftarrow w$ & for~$1 \leq i \leq n$ & \tagarray\label{eka-rule9} \\
%\hdashline
$v_{i}$ & $\leftarrow \aspnot{} x^{1}_{i}, \dotsc, \aspnot{} x^{k}_{i}$ & for~$1 \leq i \leq n$; & \tagarray\label{eka-rule10} \\
%\hdashline
$w$ & $\leftarrow \sigma(l^{\ell}_1), \sigma(l^{\ell}_2), \sigma(l^{\ell}_3)$ & for~$1 \leq \ell \leq u$ & \tagarray\label{eka-rule11} \\
%\hdashline
$w$ & $\leftarrow \aspnot{} w$. & & \tagarray\label{eka-rule12}
\end{tabular}
%}
\vspace{10pt}
\end{small}%
}
%\end{table}
%%%

\noindent Here we let~$\sigma(x_i) = x_i$ and~$\sigma(\neg x_i) = v_i$
for each~$1 \leq i \leq n$; and
we let~$\sigma(y_j) = y_j$ and~$\sigma(\neg y_j) = z_j$
for each~$1 \leq j \leq m$.
Intuitively,~$v_i$ corresponds to~$\neg x_i$,
and~$z_j$ corresponds to~$\neg y_j$.
The main difference with the reduction of Eiter and
Gottlob~\cite{EiterGottlob95}
is that we use the rules
in~(\ref{eka-rule1}),~(\ref{eka-rule2}),~(\ref{eka-rule8})~and~(\ref{eka-rule10})
to let the variables~$x_i$ and~$v_i$
represent an assignment of weight~$k$ to the variables in~$X$.
\krversion{The rules in~(\ref{eka-rule7})}\longversion{The
rules in~(\ref{eka-rule7})~and~(\ref{eka-rule9})}
ensure that the atoms~$v_i$ and~$x_i$ are compulsory.
It is straightforward to verify that~$\Comp{P} = \SBs w \SEs \cup
\SB x_i,v_i \SM 1 \leq i \leq n \SE \cup
\SB y_i,z_i \SM 1 \leq i \leq m \SE$.
Notice that~$P$ has exactly~$k$ contingent rules, namely the rules in~(\ref{eka-rule1}).
\krversion{A full proof that~$(\varphi,k) \in \EkAWSat{}$ if and only
if~$P$ has an answer set can be found in the technical report.
}\longversion{
We show that~$(\varphi,k) \in \EkAWSat(3\DNF)$ if and only
if~$P$ has an answer set.

$(\Rightarrow)$ Assume that there exists an assignment~$\alpha : X
\rightarrow \SBs 0,1 \SEs$ of weight~$k$ such that~$\forall
Y. \psi[\alpha]$ is true.  Let~$\SBs x_{i_1},\dotsc,x_{i_k} \SEs$
denote the set~$\SB x_i \SM 1 \leq i \leq n, \alpha(x_i) = 1 \SE$.  We
construct an answer set~$M$ of~$P$.  We let~$M = \SB
x^{\ell}_{i_{\ell}} \SM 1 \leq \ell \leq k \SE \cup \Comp{P}$.  The
reduct~$P^{M}$ consists of
Rules~(\ref{eka-rule1})--(\ref{eka-rule9})~and~(\ref{eka-rule11}),
together with rules~$(v_i \leftarrow)$ for all~$1 \leq i \leq n$ such
that~$\alpha(x_i) = 0$.  We show that~$M$ is a minimal model
of~$P^{M}$.  We proceed indirectly and assume to the contrary that there
exists a model~$M' \subsetneq M$ of~$P^{M}$.  By
Rule~(\ref{eka-rule8}) and the rules~$(v_i \leftarrow)$, we know that
for all~$1 \leq i \leq n$ it holds that~$x_i \in M'$
if~$\alpha(x_i) = 1$, and~$v_i \in M'$ if~$\alpha(x_i) = 0$.
If~$x^{\ell}_{i_\ell} \not\in M'$ for any~$1 \leq \ell \leq k$,
then~$M'$ is not a model of~$P^M$.
Therefore,~$\SBs x^1_{i_1},\dotsc,x^k_{i_k} \SEs \subseteq M'$.
Also, it holds that~$w \not\in M'$.
To show this, assume the contrary, i.e., assume that~$w \in M'$.
Then, by Rules~(\ref{eka-rule4}),~(\ref{eka-rule5}),~(\ref{eka-rule7})~and~(\ref{eka-rule9}),
it follows that~$M' = M$, which contradicts our assumption
that~$M' \subsetneq M$.
By Rules~(\ref{eka-rule3})~and~(\ref{eka-rule6}), we know
that~$\Card{\SBs y_i,z_i \SEs \cap M'} = 1$ for each~$1 \leq i \leq m$.  We
define the assignment~$\beta : Y \rightarrow \SBs 0,1 \SEs$ by
letting~$\beta(y_i) = 1$ if and only if~$y_i \in M'$, for all~$1 \leq i \leq
m$.  Since~$\forall Y. \psi[\alpha]$ is true, we know that~$\alpha
\cup \beta$ satisfies some term~$\delta_{\ell}$ of~$\psi$.  It is
straightforward to verify
that~$\sigma(l^{\ell}_1),\sigma(l^{\ell}_2),\sigma(l^{\ell}_3) \in M'$.
Therefore, by Rule~(\ref{eka-rule11}),~$w \in M'$, which is a
contradiction with our assumption that~$w \not\in M'$.  From this we
can conclude that~$M$ is a minimal model of~$P^{M}$, and thus that~$M$
is an answer set of~$P$.

$(\Leftarrow)$ Assume that~$M$ is an answer set of~$P$.
Clearly,~$\Comp{P} \subseteq M$.  Also, since
Rules~(\ref{eka-rule1})~and~(\ref{eka-rule2}) are
rules of~$P^{M}$,~$M$ must contain
atoms~$x^{1}_{i_1},\dotsc,x^{k}_{i_k}$ for
some~$i_1,\dotsc,i_k$.  We know that~$P^{M}$ contains
\nobreak{Rules~(\ref{eka-rule1})--(\ref{eka-rule9})} and~(\ref{eka-rule11}), as
well as the rules~$(v_i \leftarrow)$ for all~$1 \leq i \leq n$ such that
for no~$1 \leq j \leq k$ it holds that~$x^j_i \in M'$.  We define
the assignment~$\alpha : X \rightarrow \SBs 0,1 \SEs$ by
letting~$\alpha(x_i) = 1$ if and only if~$i \in \SBs i_1,\dotsc,i_k \SEs$.
The assignment~$\alpha$ has weight~$k$.  We show that~$\forall
Y. \psi[\alpha]$ is true.  Let~$\beta : Y \rightarrow \SBs 0,1 \SEs$
be an arbitrary assignment.  Construct the set~$M'
\subsetneq M$ by letting~$M' = (M \cap \SB x^{j}_{i} \SM 1
\leq i \leq n, 1 \leq j \leq k \SE) \cup \SB x_i \SM 1 \leq i \leq n,
\alpha(x_i) = 1 \SE \cup \SB v_i \SM 1 \leq i \leq n, \alpha(v_i) = 0
\SE \cup \SB y_i \SM 1 \leq i \leq m, \beta(y_i) = 1 \SE \cup \SB z_i
\SM 1 \leq i \leq m, \beta(y_i) = 0 \SE$.  Since~$M$ is a minimal
model of~$P^{M}$ and~$M' \subsetneq M$, we know that~$M'$ cannot be a
model of~$P^{M}$.  Clearly,~$M'$ satisfies
Rules~(\ref{eka-rule1})--(\ref{eka-rule9}), and all rules of~$P^{M}$ of
the form~$(v_i \leftarrow)$.  Thus there must be some instantiation of
Rule~(\ref{eka-rule11}) that~$M'$ does not satisfy.  This means that
there exists some~$1 \leq \ell \leq u$ such
that~$\sigma(l^{\ell}_1),\sigma(l^{\ell}_2),\sigma(l^{\ell}_3) \in M'$.  By
construction of~$M'$, this means that~$\alpha \cup \beta$
satisfies~$\delta_{\ell}$, and thus satisfies~$\psi$.  Since~$\beta$ was chosen
arbitrarily, we can conclude that~$\forall Y. \psi[\alpha]$ is true,
and therefore~$(\varphi,k) \in \EkAWSat(3\DNF)$.  }
\end{proof}

\begin{proposition}
\label{prop:asp-contrules-membership}
\ASPcons\ContRules{} is in \EkA{}.
\end{proposition}
\begin{proof}
We show membership in \EkA{}
by reducing \ASPcons\-\ContRules{} to \EkAWSat{}.
Let~$P$ be a program,
where~$r_1,\dotsc,r_k$ are the contingent rules of~$P$
and where~$\Atoms{P} = \SBs d_1,\dotsc,d_n \SEs$.
We construct a quantified Boolean
formula~$\varphi = \exists X. \forall Y \cup Z \cup W. \psi$
such that~$(\varphi,k) \in \EkAWSat{}$ if and only if~$P$ has an answer set.

In order to do so, we firstly construct a Boolean
formula~$\psi_{P}(z_1,\dotsc,z_n,z'_1,\dotsc,z'_n)$
(or, for short:~$\psi_{P}$)
over variables~$z_1,\dotsc,z_n,z'_1,\dotsc,z'_n$ such that for
any~$M \subseteq \Atoms{P}$
and any~$M' \subseteq \Atoms{P}$ holds
that~$M$ is a model of~$P^{M'}$ if and only if~$\psi_{P}[\alpha_{M} \cup \alpha_{M'}]$
evaluates to true,
where~$\alpha_{M} : \SBs z_1,\dotsc,z_n \SEs \rightarrow \SBs 0,1 \SEs$
is defined by letting~$\alpha_{M}(z_i) = 1$ if and only if~$d_i \in M$,
and~$\alpha_{M'} : \SBs z'_1,\dotsc,z'_n \SEs \rightarrow \SBs 0,1 \SEs$
is defined by letting~$\alpha_{M'}(z'_i) = 1$ if and only if~$d_i \in M'$,
for all~$1 \leq i \leq n$.
\krversion{
We define $\psi_{P} = \bigwedge_{r \in P}
( \psi^{1}_{r} \vee \psi^{2}_{r} )$,
where $\psi^{1}_{r} = ( z'_{i^3_1} \vee \dotsm \vee z'_{i^3_{c}} )$
and $\psi^{2}_{r} = ( ( z_{i^1_1} \vee \dotsm \vee z_{i^1_{a}} )
\leftarrow ( z_{i^2_1} \wedge \dotsm \wedge z_{i^2_{b}} ) )$
for $r = (d_{i^1_1} \vee \dotsm \vee d_{i^1_{a}} \leftarrow d_{i^2_1},\dotsc,d_{i^2_{b}}, \aspnot{} d_{i^3_1},\dotsc,\aspnot{} d_{i^3_{c}})$.
}\longversion{We define:
\[ \begin{array}{r l}
  \psi_{P} =& \underset{r \in P}{\bigwedge} \left ( \psi^{1}_{r} \vee \psi^{2}_{r} \right ); \\
  \psi^{1}_{r} =& \left ( z'_{i^3_1} \vee \dotsm \vee z'_{i^3_{c}} \right ); \mtext{ and} \\
  \psi^{2}_{r} =& \left ( \left ( z_{i^1_1} \vee \dotsm \vee z_{i^1_{a}} \right ) \leftarrow 
    \left ( z_{i^2_1} \wedge \dotsm \wedge z_{i^2_{b}} \right ) \right ), \\
  & \mtext{where } r = (d_{i^1_1} \vee \dotsm \vee d_{i^1_{a}} \leftarrow d_{i^2_1},\dotsc,d_{i^2_{b}}, \krversion{\\ &}
\aspnot{} d_{i^3_1},\dotsc,\aspnot{} d_{i^3_{c}}).
\end{array} \]}
It is easy to verify that~$\psi_{P}$ satisfies the required property.

We now introduce the set~$X$ of existentially quantified variables of~$\varphi$.
For each contingent rule~$r_i$ of~$P$ we
let~$a^i_1,\dotsc,a^i_{\ell_i}$ denote the atoms that occur in the head of~$r_i$.
For each~$r_i$, we introduce variables~$x^i_0,x^i_1,\dotsc,x^i_{\ell_i}$,
i.e.,~$X = \SB x^i_j \SM 1 \leq i \leq k, 0 \leq j \leq \ell_i \SE$.
Furthermore, for each atom~$d_i$, we add universally quantified
variables~$y_i$,~$z_i$ and~$w_i$, i.e.,~$Y =
\SB y_i \SM 1 \leq i \leq n \SE$,~$Z =
\SB z_i \SM 1 \leq i \leq n \SE$,
and~$W = \SB w_i \SM 1 \leq i \leq n \SE$.

We then construct~$\psi$ as follows:
\[ \begin{array}{r l}
  \psi =& \psi_{X} \wedge \left ( \psi^{1}_{Y} \vee
    \psi_{W} \vee \psi_{\mtext{min}} \right ) \wedge ( \psi^{1}_{Y} \vee \psi^{2}_{Y} ) ; \\
  \psi_{X} =& \underset{1 \leq i \leq k}{\bigwedge} \left (
    \underset{0 \leq j \leq \ell_i}{\bigvee} x^i_j
    \wedge \underset{0 \leq j < j' \leq \ell_i}{\bigwedge} (\neg x^i_j \vee \neg x^i_{j'}) \right ); \\
  \psi^{1}_{Y} =&
    \longversion{\underset{1 \leq i \leq k}{\bigvee}\ \underset{1 \leq j \leq \ell_i}{\bigvee}}
    \krversion{\underset{1 \leq j \leq \ell_i}{\bigvee\limits_{1 \leq i \leq k}}}
    \psi^{i,j}_{y} \vee
    \underset{d_i \in \Cont{P}}{\bigvee} \psi^{d_i}_{y} \vee
    \underset{d_i \in \Comp{P}}{\bigvee} \neg y_i; \\
  \longversion{\psi^{d_m}_{y} =& \begin{dcases*}
    (y_m \wedge \neg x^{i_1}_{j_1} \wedge \dotsm \wedge \neg x^{i_u}_{j_u})
    & if $\SB x^{i}_{j} \SM 1 \leq i \leq k, 1 \leq j \leq \ell_i, a^i_j = d_m \SE =
 \SBs a^{i_1}_{j_1},\dotsc,a^{i_u}_{j_u} \SEs$, \\
    \bot &
    if $\SB x^{i}_{j} \SM 1 \leq i \leq k, 1 \leq j \leq \ell_i, a^i_j = d_m \SE = \emptyset;$ \\
    \end{dcases*} \\}
  \krversion{\psi^{d_m}_{y} =&
    (y_m \wedge \neg x^{i_1}_{j_1} \wedge \dotsm \wedge \neg x^{i_u}_{j_u})
\krversion{\\ &}\longversion{\hfill \qquad}
    \mtext{ if } \SB x^{i}_{j} \SM 1 \leq i \leq k, 1 \leq j \leq \ell_i, a^i_j = d_m \SE =
    \krversion{\\ &} \SBs a^{i_1}_{j_1},\dotsc,a^{i_u}_{j_u} \SEs, \mtext{ and} \\
  \psi^{d_m}_{y} =& 
    \bot
    \krversion{\\ &}\longversion{\hfill}
    \mtext{if } \SB x^{i}_{j} \SM 1 \leq i \leq k, 1 \leq j \leq \ell_i, a^i_j = d_m \SE = \emptyset; \\}
  \psi^{i,j}_{y} =& (x^{i}_{j} \wedge \neg y_m)
    \hfill \mtext{ where $a^i_j = d_m$}; \\
  \psi^{2}_{Y} =& \psi_{P}(y_1,\dotsc,y_n,y_1,\dotsc,y_n); \\
  \psi_{W} =& \underset{1 \leq i \leq n}{\bigvee}
%    (w_i \wedge y_i \wedge z_i) \vee
%    (w_i \wedge \neg y_i \wedge \neg z_i) \krversion{\ } \vee \krversion{\\ &}
%    (\neg w_i \wedge \neg y_i \wedge z_i) \vee
%    (\neg w_i \wedge y_i \wedge \neg z_i); \\
      (w_i \leftrightarrow (y_i \leftrightarrow z_i)); \\
%
%\krversion{\end{array}\]\[\begin{array}{r l}}
%
  \psi_{\mtext{min}} =&
    \psi^{1}_{\mtext{min}} \vee
    \psi^{2}_{\mtext{min}} \vee
    \psi^{3}_{\mtext{min}}; \\
  \psi^{1}_{\mtext{min}} =&
    \underset{1 \leq i \leq n}{\bigvee} \left ( z_i \wedge \neg y_i \right ); \\
  \psi^{2}_{\mtext{min}} =&
    (\neg w_1 \wedge \dotsm \wedge \neg w_m); \mtext{ and} \\
  \psi^{3}_{\mtext{min}} =&
    \neg \psi_{P}(z_1,\dotsc,z_n,y_1,\dotsc,y_n).
\end{array} \]
The idea behind this construction is the following.
The variables in~$X$ represent guessing
at most one atom in the head of each contingent rule to be true.
\krversion{Such a guess represents a possible answer
set~$M \subseteq \Atoms{P}$ (a proof of this can be found
in the technical report).
}\longversion{By Lemma~\ref{lem:asp-contrules-mu},
such a guess represents a possible answer set~$M \subseteq \Atoms{P}$.}
The formula~$\psi_{X}$ ensures that for each~$1 \leq i \leq k$,
exactly one \krversion{$\smash{x^i_j}$}\longversion{$x^i_j$} is set to true.
The formula~$\psi^{1}_{Y}$ filters out every assignment in which
the variables~$Y$ are not set corresponding to~$M$.
The formula~$\psi^{2}_{Y}$ filters out every assignment
corresponding to a candidate~$M \subseteq \Atoms{P}$ such
that~$M \not\models P$.
The formula~$\psi_{W}$ filters out every assignment
such that~$w_i$ is not set to the value~$(y_i \mtext{ \sc xor } z_i)$.
The formula~$\psi^{1}_{\mtext{min}}$ filters out every assignment
where the variables~$Z$ correspond to a set~$M'$
such that~$M' \not\subseteq M$.
The formula~$\psi^{2}_{\mtext{min}}$ filters out every assignment
where the variables~$Z$ correspond to the set~$M$,
by referring to the variables~$w_i$.
The formula~$\psi^{3}_{\mtext{min}}$, finally, ensures that in every
remaining assignment, the variables~$Z$ do not correspond to a
set~$M' \subseteq M$ such that~$M' \models P$.
\krversion{A full proof that~$P$ has an answer set
if and only if~$(\varphi,k) \in \EkAWSat{}$ can be found
in the technical report.
}\longversion{We now formally prove that~$P$ has an answer set
if and only if~$(\varphi,k) \in \EkAWSat{}$.

$(\Rightarrow)$
Assume that~$P$ has an answer set~$M$.
By Lemma~\ref{lem:asp-contrules-mu}, we know that there exist
some subset~$R \subseteq \SBs r_1,\dotsc,r_k \SEs$ and
a mapping~$\mu : R \rightarrow \Cont{P}$
such that for each~$r \in R$,~$\mu(r)$ occurs in the head
of~$r$, and such that~$M = \Rng{\mu} \cup \Comp{P}$.
We construct the mapping~$\alpha_{\mu} : X \rightarrow \SBs 0,1 \SEs$
by letting~$\alpha_{\mu}(x^i_j) = 1$ if and only
if~$r_i \in R$ and~$\mu(r_i) = a^i_j$,
and letting~$\alpha_{\mu}(x^i_0) = 1$ if and only if~$r_i \not\in R$.
Clearly,~$\alpha_{\mu}$ has weight~$k$ and satisfies~$\psi_{X}$.
We show that~$\forall Y. \forall Z. \forall W. \psi$ evaluates to true.
Let~$\beta : Y \cup Z \cup W \rightarrow \SBs 0,1 \SEs$ be an arbitrary assignment.
We let~$M_{Y} = \SB d_i \SM 1 \leq i \leq m, \beta(y_i) = 1 \SE$,
and~$M_{Z} = \SB d_i \SM 1 \leq i \leq m, \beta(z_i) = 1 \SE$.
We distinguish a number of cases:
\begin{enumerate}[(i)]
  \item either~$M_{Y} \neq M$,
  \item or the previous case does not apply and
    for some~$1 \leq i \leq m$,~$\beta(w_i) \neq (\beta(y_i) \mtext{ \sc xor } \beta(z_i))$,
  \item or all previous cases do not apply and~$M_{Z} \not\subseteq M_{Y}$,
  \item or all previous cases do not apply and~$M_{Z} = M_{Y}$,
  \item or all previous cases do not apply and~$M_{Z} \not\models P^{M}$,
  \item or all previous cases do not apply and~$M_{Z} \models P^{M}$.
\end{enumerate}
The following is now straightforward to verify.
In case (i),~$\alpha \cup \beta$ satisfies~$\psi^{1}_{Y}$.
Thus,~$\alpha \cup \beta$ satisfies~$\psi$.
In all further cases, we know that~$\alpha \cup \beta$
satisfies~$\psi^{2}_{Y}$, since~$M_{Y} = M$,
and~$M \models P^{M}$.
In case (ii),~$\alpha \cup \beta$ satisfies~$\psi_{W}$.
In case (iii),~$\alpha \cup \beta$ satisfies~$\psi^{1}_{\mtext{min}}$.
In case (iv),~$\alpha \cup \beta$ satisfies~$\psi^{2}_{\mtext{min}}$.
In case (v),~$\alpha \cup \beta$ satisfies~$\psi^{3}_{\mtext{min}}$.
In case (vi), we get a direct contradiction from the facts
that~$M_{Z} \subsetneq M$, that~$M_{Z} \models P^{M}$, and
that~$M$ is a subset-minimal model of~$P^{M}$.
We can thus conclude, that in any case~$\alpha \cup \beta$ satisfies~$\psi$.
Therefore,~$\forall Y \cup Z \cup W. \psi$ evaluates to true,
and thus~$(\varphi,k) \in \EkAWSat{}$.

$(\Leftarrow)$
Assume that~$(\varphi,k) \in \EkAWSat{}$.
This means that there exists an assignment~$\beta : X \rightarrow \SBs 0,1 \SEs$
of weight~$k$ such that~$\forall Y \cup Z \cup W. \psi[\alpha]$ evaluates to true.
We construct~$M = \Comp{P} \cup \SB d_m \SM 1 \leq i \leq k,
1 \leq j \leq \ell_i, \alpha(x^i_j) = 1, a^i_j = d_m \SE$.
We show that~$M$ is an answer set of~$P$.
Construct an assignment~$\beta_1 : Y \cup Z \cup W \rightarrow \SBs 0,1 \SEs$
as follows.
We let~$\beta_1(y_i) = 1$ if and only if~$d_i \in M$.
For all~$y \in Z \cup W$, the assignment~$\beta_1(y)$ is arbitrary.
We know that~$\alpha \cup \beta_1$ satisfies~$\psi$,
and thus that~$\alpha \cup \beta_1$ satisfies~$(\psi^{1}_{Y} \vee \psi^{2}_{Y})$.
It is straightforward to verify that~$\alpha \cup \beta_1$
does not satisfy~$\psi^{1}_{Y}$.
Therefore~$\alpha \cup \beta_1$ satisfies~$\psi^{2}_{Y}$.
From this, we can conclude that~$M \models P^{M}$.

Now we show that~$M$ is a subset-minimal model of~$P^{M}$.
Let~$M' \subseteq \Atoms{P}$ be an arbitrary set such that~$M' \subsetneq M$.
We show that~$M' \not\models P^{M}$.
We construct an assignment~$\beta_2 : Y \cup Z \cup W
\rightarrow \SBs 0,1 \SEs$
as follows.
For all~$y_i \in Y$, we let~$\beta_2(y_i) = 1$ if and only if~$d_i \in M$.
For all~$z_i \in Z$, let~$\beta_2(z_i) = 1$ if and only if~$d_i \in M'$.
For all~$w_i \in W$, let~$\beta_2(w_i) =
\beta_2(y_i) \mtext{ \sc xor } \beta_2(z_i)$.
It is straightforward to verify that~$\alpha \cup \beta_2$
satisfies $\neg \psi^{1}_{Y}$,~$\neg \psi_{W}$,~$\neg \psi^{1}_{\mtext{min}}$, 
and~$\neg \psi^{2}_{\mtext{min}}$.
Therefore, since~$\alpha \cup \beta_2$ satisfies~$\psi$,
we know that~$\alpha \cup \beta_2$ satisfies~$\psi^{3}_{\mtext{min}}$.
Therefore, we know that~$M' \not\models P^{M}$.
This concludes our proof that~$M$ is a subset-minimal model of~$P^{M}$,
and thus we can conclude that~$M$ is an answer set of~$P$.}
\end{proof}
}

\krlongversion{
%%%
%%% ADDITIONAL CHARACTERIZATIONS OF K-*
%%%
\section{Additional Characterizations of \kstar{}}
\label{sec:eka-characterizations}

Finally, we provide a number of
different equivalent characterizations of \EkA{}.
In particular, we characterize \EkA{} in terms
of model checking of first-order logic formulas
and in terms of alternating Turing machines.

%%%
%%% MC CHARACTERIZATION (KR / ARXIV)
%%%
\longversion{
\subsection{First-order Model Checking Characterization}
\label{sec:mc}

We begin with giving an equivalent characterization of the class \EkA{}
in terms of model checking of first-order logic formulas.%
\footnote{We would like to thank Hubie Chen for suggesting to use
first-order model checking to obtain
an alternative characterization of the class \EkA{}.}
The perspective of model checking
is also used in parameterized complexity theory
to define the A-hierarchy,
which is central to another hardness theory~\cite{FlumGrohe06}.

We introduce a few notions that we need
for defining the model checking perspective on \EkA{}.
A \emph{(relational) vocabulary}~$\tau$ is a finite set of relation symbols.
Each relation symbol~$R$ has an \emph{arity}~$\arity{R} \geq 1$.
A \emph{structure}~$\AAA$ of vocabulary~$\tau$, or \emph{$\tau$-structure}
(or simply \emph{structure}), consists of a set~$A$ called the \emph{domain}
and an interpretation~$R^{\AAA} \subseteq A^{\arity{R}}$ for each
relation symbol~$R \in \tau$.
We use the usual definition of truth of a first-order logic sentence~$\varphi$ over
the vocubulary~$\tau$ in a~$\tau$-structure~$\AAA$.
We let~$\AAA \models \varphi$ denote that the sentence~$\varphi$
is true in structure~$\AAA$.
If~$\varphi$ is a first-order formula with free variables~$\Free{\varphi}$,
and~$\mu : \Free{\varphi} \rightarrow A$ is an assignment,
we use the notation~$\AAA,\mu \models \varphi$ to denote
that~$\varphi$ is true in structure~$\AAA$ under the assignment~$\mu$.
Now consider the following parameterized problem.
\probdef{
  $\EkAMC{}$
  
  \emph{Instance:} A first-order logic sentence
  $\varphi = \exists x_1,\dotsc,x_k. \forall y_1,\dotsc,y_n. \psi$
  over a vocabulary~$\tau$,
  where~$\psi$ is quantifier-free,
  and a finite~$\tau$-structure~$\AAA$.

  \emph{Parameter:} $k$.

  \emph{Question:} Does~$\AAA \models \varphi$?
}

\noindent We show that this problem is complete for the class \EkA{}.

\begin{theorem}
\label{thm:modelchecking-completeness}
\EkAMC{} is \EkA{}-complete.
\end{theorem}

\noindent This completeness result follows from the following two
propositions.

\begin{proposition}
\label{prop:modelchecking-membership}
\EkAMC{} is in \EkA{}.
\end{proposition}
\begin{proof}
We show \EkA{}-membership of \EkAMC{} by giving an fpt-reduction
to \EkAWSat{}.
Let~$(\varphi,\AAA)$ be an instance of \EkAMC{},
where~$\varphi = \exists x_1,\dotsc,x_k. \forall y_1,\dotsc,y_n. \psi$
is a first-order logic sentence over vocabulary~$\tau$,
and~$\AAA$ is a~$\tau$-structure with domain~$A$.
We assume without loss of generality that~$\psi$ contains only
connectives~$\wedge$ and~$\neg$.

We construct an instance~$(\varphi',k)$ of \EkAWSat{},
where~$\varphi$ is of the form~$\exists X'. \forall Y'. \psi'$.
We define:
\[ \begin{array}{r l}
  X' = &
    \SB x'_{i,a} \SM 1 \leq i \leq k, a \in A \SE, \mtext{ and} \\
  Y' = &
    \SB y'_{j,a} \SM 1 \leq j \leq n, a \in A \SE.
\end{array} \]
In order to define~$\psi'$,
we will use the following auxiliary function~$\mu$
on subformulas of~$\psi$:
\[
  \mu(\chi) = \begin{dcases*}
    \mu(\chi_1) \wedge \mu(\chi_2) & if $\chi = \chi_1 \wedge \chi_2$, \\
    \neg \mu(\chi_1) & if $\chi = \neg \chi_1$, \\
    \underset{1 \leq i \leq u}{\bigvee}
      \left (
        \psi_{z_1,a^i_1} \wedge \dotsm \wedge \psi_{z_m,a^i_m}
      \right )
      & if $\chi = R(z_1,\dotsc,z_m)$ and $R^{\mathcal{A}} =
      \SBs (a^1_1,\dotsc,a^1_m),\dotsc,(a^u_1,\dotsc,a^u_m) \SEs$, \\
  \end{dcases*}
\]
where for each~$z \in X \cup Y$ and each~$a \in A$ we define:
\[
  \psi_{z,a} = \begin{dcases*}
    x'_{i,a} & if $z = x_i$, \\
    y'_{j,a} & if $z = y_j$. \\
  \end{dcases*}
\]
Now, we define~$\psi'$ as follows:
\[ \begin{array}{r l}
  \psi' = &
    \psi'_{\mtext{unique-}X'} \wedge
    \left (
      \psi'_{\mtext{unique-}Y'} \rightarrow \mu(\psi)
    \right ), \mtext{ where} \\[5pt]
  \psi'_{\mtext{unique-}X'} = &
    \underset{1 \leq i \leq k}{\bigwedge}
    \left (
      \underset{a \in A}{\bigvee} x'_{i,a} \wedge
      \underset{a \neq a'}{
      \underset{a,a' \in A}{\bigwedge}} (\neg x'_{i,a} \vee \neg x'_{i,a'})
    \right ), \mtext{ and} \\
  \psi'_{\mtext{unique-}Y'} = &
    \underset{1 \leq j \leq n}{\bigwedge}
    \left (
      \underset{a \in A}{\bigvee} y'_{j,a} \wedge
      \underset{a \neq a'}{
      \underset{a,a' \in A}{\bigwedge}} (\neg y'_{j,a} \vee \neg y'_{j,a'})
    \right ). \\
\end{array} \]
We show that~$(\AAA,\varphi) \in \EkAMC{}$ if and only if
$(\varphi',k) \in \EkAWSat{}$.

$(\Rightarrow)$
Assume that there exists an assignment
$\alpha : \SBs x_1,\dotsc,x_k \SEs \rightarrow A$
such that~$\AAA,\alpha \models \forall y_1,\dotsc,y_n. \psi$.
We define the assignment~$\alpha' : X' \rightarrow \SBs 0,1 \SEs$
where~$\alpha'(x'_{i,a}) = 1$ if and only if~$\alpha(x_i) = a$.
Clearly,~$\alpha'$ has weight~$k$.
Also, note that~$\alpha'$ satisfies~$\psi'_{\mtext{unique-}X'}$.
Now, let~$\beta' : Y' \rightarrow \SBs 0,1 \SEs$ be an arbitrary assignment.
We show that~$\alpha' \cup \beta'$ satisfies~$\psi'$.
We distinguish two cases: either (i) for each~$1 \leq j \leq n$,
there is a unique~$a_{j} \in A$ such that~$\beta'(y'_{j,a_{j}}) = 1$,
or (ii) this is not the case.
In case (i),~$\alpha' \cup \beta'$ satisfies~$\psi'_{\mtext{unique-}Y'}$,
so we have to show that~$\alpha' \cup \beta'$ satisfies~$\mu(\psi)$.
Define the assignment~$\beta : \SBs y_1,\dotsc,y_n \SEs \rightarrow A$
by letting~$\beta(y_j) = a_j$.
We know that~$\AAA,\alpha \cup \beta \models \psi$.
It is now straightforward to show by induction on the structure of~$\psi$
that for each subformula~$\chi$ of~$\psi$ holds that
that~$\alpha' \cup \beta'$ satisfies~$\mu(\chi)$
if and only if~$\AAA,\alpha \cup \beta \models \chi$.
We then know in particular that~$\alpha' \cup \beta'$ satisfies~$\mu(\psi)$.
In case (ii), we know that~$\alpha' \cup \beta'$ does not satisfy
$\psi'_{\mtext{unique-}Y'}$, and therefore~$\alpha' \cup \beta'$ satisfies~$\psi'$.
This concludes our proof that~$(\varphi',k) \in \EkAWSat{}$.

$(\Leftarrow)$
Assume that there exists an assignment~$\alpha : X' \rightarrow \SBs 0,1 \SEs$
of weight~$k$ such that~$\forall Y'. \psi'[\alpha]$ is true.
Since~$\psi'_{\mtext{unique-}X'}$ contains only variables in~$X'$,
we know that~$\alpha$ satisfies~$\psi'_{\mtext{unique-}X'}$.
From this, we can conclude that for each~$1 \leq i \leq k$,
there is some unique~$a_i \in A$ such that~$\alpha(x'_{i,a_i}) = 1$.
Now, define the assignment~$\alpha' : \SBs x_1,\dotsc,x_k \SEs \rightarrow A$
by letting~$\alpha'(x_i) = a_i$.

We show that~$\AAA,\alpha' \models \forall y_1,\dotsc,y_n. \psi$.
Let~$\beta' : \SBs y_1,\dotsc,y_n \SEs \rightarrow A$ be an arbitrary assignment.
We define~$\beta : Y' \rightarrow \SBs 0,1 \SEs$ by letting
$\beta(y'_{i,a}) = 1$ if and only if~$\beta(y_i) = a$.
It is straightforward to verify that~$\beta$ satisfies~$\psi'_{\mtext{unique-}Y'}$.
We know that~$\alpha \cup \beta$ satisfies~$\psi'$,
so therefore~$\alpha \cup \beta$ satisfies~$\mu(\psi)$.
It is now straightforward to show by induction on the structure of~$\psi$
that for each subformula~$\chi$ of~$\psi$ holds that
that~$\alpha \cup \beta$ satisfies~$\mu(\chi)$
if and only if~$\AAA,\alpha' \cup \beta' \models \chi$.
We then know in particular that~$\AAA,\alpha' \cup \beta' \models \psi$.
This concludes our proof that~$(\AAA,\varphi) \in \EkAMC{}$.
\end{proof}

\begin{proposition}
\label{prop:modelchecking-hardness}
\EkAMC{} is \EkA{}-hard.
\end{proposition}
\begin{proof}
We show \EkA{}-hardness by giving an fpt-reduction
from~$\EkAWSat(\DNF)$.
Let~$(\varphi,k)$ specify an instance of~$\EkAWSat$,
where~$\varphi = \exists X. \forall Y. \psi$,~$X =
\SBs x_1,\dotsc,x_n \SEs$,~$Y =
\SBs y_1,\dotsc,y_m \SEs$,~$\psi =
\delta_1 \vee \dotsm \vee \delta_u$,
and for each~$1 \leq \ell \leq u$,~$\delta_{\ell} = 
l^{\ell}_1 \vee l^{\ell}_2 \vee l^{\ell}_3$.
We construct an instance~$(\AAA,\varphi')$ of \EkAMC{}.
In order to do so, we first fix the following vocabulary~$\tau$
(which does not depend on the instance~$(\varphi,k)$):
it contains unary relation symbols~$D$,~$X$ and~$Y$,
and binary relation symbols~$C_1$,~$C_2$,~$C_3$ and~$O$.
We construct the domain~$A$ of~$\AAA$ as follows:
\[ A = X \cup Y \cup \SB \delta_{\ell} \SM 1 \leq \ell \leq u \SE \cup \SBs \star \SEs. \]
Then, we define:
\[ \begin{array}{r l}
  D^{\AAA} =& \SB \delta_{\ell} \SM 1 \leq \ell \leq u \SE; \\
  X^{\AAA} =& X; \\
  Y^{\AAA} =& Y \cup \SBs \star \SEs ; \\
  C^{\AAA}_{d} =& \SB (\delta_{\ell},z) \SM 1 \leq \ell \leq u,
    z \in X \cup Y, l^{\ell}_{d} \in \SBs x, \neg x \SEs \SE
    \qquad \mtext{ for $1 \leq d \leq 3$; and} \\
  O^{\AAA} =& \SB (\delta_{\ell},\delta_{\ell'}) \SM 1 \leq \ell < \ell' \leq u \SE. \\
\end{array} \]
Intuitively, the relations~$D$,~$X$ and~$Y$ serve to distinguish the
various subsets of the domain~$A$.
The relations~$C_d$, for~$1 \leq d \leq 3$, encode (part of) the structure
of the matrix~$\psi$ of the formula~$\varphi$.
The relation~$O$ encodes a linear ordering on the terms~$\delta_\ell$.

We now define the formula~$\varphi'$ as follows:
\[ \varphi' = \exists u_1,\dotsc,u_k. \forall v_1,\dotsc,v_m. \forall w_1,\dotsc,w_u. \chi, \]
where we define~$\chi$ to be of the following form:
\[ \chi = \chi^{U}_{\mtext{proper}} \wedge (( \chi^{V}_{\mtext{proper}}
  \wedge \chi^{W}_{\mtext{exact}}) \rightarrow \chi_{\mtext{sat}}). \]
We will define the subformulas of~$\chi$ below.
Intuitively, the assignment of the variables~$u_i$ will correspond to an
assignment~$\alpha : X \rightarrow \SBs 0,1 \SEs$ of weight~$k$
that sets a variable~$x \in X$
to true if and only if some~$u_i$ is assigned to~$x$.
Similarly, any assignment of the variables~$v_i$ will correspond to an
assignment~$\beta : Y \rightarrow \SB 0,1 \SEs$ that sets~$y \in Y$ to
true if and only if
some~$v_i$ is assigned to~$y$.
The variables~$w_{\ell}$ will function to refer to the elements~$\delta_{\ell} \in A$.

The formula~$\chi^{U}_{\mtext{proper}}$ ensures that
the variables~$u_1,\dotsc,u_k$ select exactly~$k$ different elements from~$X$.
We define:
\[ \chi^{U}_{\mtext{proper}} = \bigwedge\limits_{1 \leq i \leq k} X(u_i) \wedge
\bigwedge\limits_{1 \leq i < i' \leq k} (u_i \neq u_{i'}). \]
For the sake of clarity, we use the formula~$\chi^{V}_{\mtext{proper}}$ to check
whether each variable~$v_i$ is assigned to a value in~$Y \cup \SBs \star \SEs$.
We define:
\[ \chi^{V}_{\mtext{proper}} = \bigwedge\limits_{1 \leq i \leq m} Y(v_i). \]
Next, the formula~$\chi^{W}_{\mtext{exact}}$ encodes whether the
variables~$w_1,\dotsc,w_u$ get assigned exactly to the
elements~$\delta_{1},\dotsc,\delta_{u}$ (and also in that order,
i.e.,~$w_{\ell}$ gets assigned~$\delta_{\ell}$ for each~$1 \leq \ell \leq u$).
We let:
\[ \chi^{W}_{\mtext{exact}} = \bigwedge\limits_{1 \leq \ell \leq u} D(w_{\ell}) \wedge
\bigwedge\limits_{1 \leq \ell < \ell' \leq u} O(w_{\ell},w_{\ell'}). \]
Finally, we can turn to the formula~$\chi_{\mtext{sat}}$, which represents
whether the assignments~$\alpha$ and~$\beta$ represented by the assignment
to the variables~$u_i$ and~$v_j$ satisfies~$\psi$.
We define:
\[ \chi_{\mtext{sat}} = \bigvee\limits_{1 \leq \ell \leq u} \chi^{\ell}_{\mtext{sat}}, \]
where we let:
\[ \chi^{\ell}_{\mtext{sat}} = \chi^{\ell,1}_{\mtext{sat}} \wedge
\chi^{\ell,2}_{\mtext{sat}} \wedge \chi^{\ell,3}_{\mtext{sat}}, \]
and for each~$1 \leq d \leq 3$ we let:
\[ \chi^{\ell,d}_{\mtext{sat}} = \begin{dcases*}
  \bigvee\limits_{1 \leq j \leq k} C_{d}(w_{\ell},u_j) & if $l^{\ell}_{d} = x \in X$, \\
  \bigwedge\limits_{1 \leq j \leq k} \neg C_{d}(w_{\ell},u_j) & if $l^{\ell}_{d} = \neg x$ for some $x \in X$, \\
  \bigvee\limits_{1 \leq j \leq m} C_{d}(w_{\ell},v_j) & if $l^{\ell}_{d} = y \in Y$, \\
  \bigwedge\limits_{1 \leq j \leq m} \neg C_{d}(w_{\ell},v_j) & if $l^{\ell}_{d} = \neg y$ for some $y \in X$. \\
\end{dcases*} \]
Intuitively, for each~$1 \leq \ell \leq u$ and each~$1 \leq d \leq 3$,
the formula~$\chi^{\ell,d}_{\mtext{sat}}$ will be satsified by the assignments
to the variables~$u_i$ and~$v_i$ if and only if the corresponding
assignments~$\alpha$ to~$X$ and~$\beta$ to~$Y$ satisfy~$l^{\ell}_{d}$.

It is straightforward to verify that the instance~$(\AAA,\varphi')$
can be constructed in polynomial time.
We show that~$(\varphi,k) \in \EkAWSat(3\DNF)$ if and only
if~$(\AAA,\varphi') \in \EkAMC$.

$(\Rightarrow)$
Assume that there exists an assignment~$\alpha : X \rightarrow \SBs 0,1 \SEs$
of weight~$k$ such that~$\forall Y. \psi[\alpha]$ is true.
We show that~$\AAA \models \varphi'$.
Let~$\SB x \in X \SM \alpha(x) = 1 \SE = \SBs x_{i_1},\dotsc,x_{i_k} \SEs$.
We define the assignment~$\mu : \SBs u_1,\dotsc,u_k \SEs \rightarrow A$
by letting~$\mu(x_j) = x_{i_j}$ for all~$1 \leq j \leq k$.
It is straightforward to verify that~$\AAA,\mu \models \psi^{U}_{\mtext{proper}}$.
Now, let~$\nu : \SBs y_1,\dotsc,y_m,w_1,\dotsc,w_u \SEs \rightarrow A$
be an arbitrary assignment.
We need to show that~$\AAA, \mu \cup \nu \models \chi$,
and we thus need to show that
$\AAA, \mu \cup \nu \models ( \chi^{V}_{\mtext{proper}}
\wedge \chi^{W}_{\mtext{exact}}) \rightarrow \chi_{\mtext{sat}}$.
We distinguish several cases:
either (i)~$\nu(y_i) \not\in Y \cup \SBs \star \SEs$ for some~$1 \leq i \leq m$,
or (ii) the above is not the case and~$\nu(w_\ell) \neq \delta_{\ell}$ for
some~$1 \leq \ell \leq u$,
or (iii) neither of the above is the case.
In case (i), it is straightforward to verify that~$\AAA, \mu \cup \nu \models
\neg \chi^{V}_{\mtext{proper}}$.
In case (ii), it is straightforward to verify that~$\AAA, \mu \cup \nu \models
\neg \chi^{W}_{\mtext{exact}}$.
Consider case (iii).
We construct the assignment~$\beta : Y \rightarrow \SBs 0,1 \SEs$
by letting~$\beta(y) = 1$ if and only if~$\nu(v_i) = y$ for some~$1 \leq i \leq m$.
We know that~$\alpha \cup \beta$ satisfies~$\psi$, and thus in particular
that~$\alpha \cup \beta$ satisfies some term~$\delta_{\ell}$.
It is now straightforward to verify that
$\AAA, \mu \cup \nu \models \chi^{\ell}_{\mtext{sat}}$,
and thus that~$\AAA, \mu \cup \nu \models \chi_{\mtext{sat}}$.
This concludes our proof that~$\AAA \models \varphi'$.

$(\Leftarrow)$
Assume that~$\AAA \models \varphi'$.
We show that~$(\varphi,k) \in \EkAWSat(3\DNF)$.
We know that there exists an assignment
$\mu : \SBs u_1,\dotsc,u_k \SEs \rightarrow A$
such that~$\AAA,\mu \models \forall u_1,\dotsc,u_m. \forall w_1,\dotsc,w_u. \chi$.
Since~$\AAA,\mu \models \chi^{U}_{\mtext{proper}}$,
we know that~$\mu$ assigns the variables~$u_i$ to~$k$ many
different values~$x \in X$.
Define~$\alpha : X \rightarrow \SBs 0,1 \SEs$ by letting~$\alpha(x) = 1$
if and only if~$\mu(u_i) = x$ for some~$1 \leq i \leq k$.
Clearly,~$\alpha$ has weight~$k$.
Now, let~$\beta : Y \rightarrow \SBs 0,1 \SEs$ be an arbitrary assignment.
Construct the assignment~$\nu : \SBs v_1,\dotsc,v_m,w_1,\dotsc,w_u \SEs$
as follows.
For each~$1 \leq i \leq m$, we let~$\nu(v_i) = y_i$ if~$\beta(y_i) = 1$,
and we let~$\nu(v_i) = \star$ otherwise.
Also, for each~$1 \leq \ell \leq u$, we let~$\nu(w_\ell) = \delta_{\ell}$.
It is straightforward to verify that
$\AAA,\mu \cup \nu \models
\chi^{V}_{\mtext{proper}} \wedge \chi^{W}_{\mtext{exact}}$.
Therefore, we know that~$\AAA,\mu \cup \nu \models \chi_{\mtext{sat}}$,
and thus that for some~$1 \leq \ell \leq u$ it holds that
$\AAA,\mu \cup \nu \models \chi^{\ell}_{\mtext{sat}}$.
It is now straightforward to verify that~$\alpha \cup \beta$
satisfies~$\delta_{\ell}$.
Since~$\beta$ was arbitrary, this concludes our proof
that~$(\varphi,k) \in \EkAWSat$.
\end{proof}

\noindent The problem \EkAMC{} takes the relational vocabulary~$\tau$
over which the structure~$\AAA$ and the first-order logic sentence~$\varphi$
are defined as part of the input.
However, the proof of Proposition~\ref{prop:modelchecking-hardness} shows
that the problem \EkAMC{} is \EkA{}-hard already when the vocabulary~$\tau$
is fixed and contains only unary and binary relation symbols.
}
}

\krlongversion{
%%%
%%% ATM CHARACTERIZATION
%%%
\longversion{

\subsection{Alternating Turing Machine Characterization}
\label{sec:atm-char}
Next, we give another equivalent characterization of the class \EkA{},
by means of alternating Turing machines
(cf.~\cite{Papadimitriou94}).
We use the same notation as Flum and 
Grohe~\cite[Appendix~A.1]{FlumGrohe06}.

%%% DEFINITIONS
\paragraph{Definitions}

Let~$m \geq 1$ be a positive integer.
An \emph{alternating Turing machine (ATM)} with~$m$ tapes is a
$6$-tuple~$\mathbb{M} = (S_{\exists},S_{\forall},\Sigma,\Delta,s_0,F)$, where:
\begin{itemize}
  \item $S_{\exists}$ and~$S_{\forall}$ are disjoint sets;
  \item $S = S_{\exists} \cup S_{\forall}$ is the finite set of \emph{states};
  \item $\Sigma$ is the alphabet;
  \item $s_0 \in S$ is the \emph{initial state};
  \item $F \subseteq S$ is the set of \emph{accepting states};
  \item $\Delta \subseteq S \times (\Sigma \cup \SBs \$,\Box \SEs)^m \times S \times
    (\Sigma \cup \SBs \$ \SEs)^m \times \SBs \mathbf{L},\mathbf{R},\mathbf{S} \SEs^{m}$
    is the \emph{transition relation}. The elements of~$\Delta$ are the \emph{transitions}.    
  \item $\$,\Box \not\in \Sigma$ are special symbols.
    ``$\$$'' marks the left end of any tape. It cannot be overwritten and only allows
    $\mathbf{R}$-transitions.\footnote{To formally achieve that ``$\$$''
    marks the left end of the tapes, whenever
    $(s,(a_1,\dotsc,a_m),s',(a'_1,\dotsc,a'_m),(d_1,\dotsc,d_m)) \in \Delta$,
    then for all~$1 \leq i \leq m$ we have that~$a_i = \$$
    if and only if~$a'_i = \$$ and
    that~$a_i = \$$ implies~$d_i = \mathbf{R}$.}
    ``$\Box$'' is the \emph{blank symbol}.
\end{itemize}

Intuitively, the tapes of our machine are bounded to the left and unbounded to the right.
The leftmost cell, the~$0$-th cell, of each tape carries a ``$\$$'',
and initially, all other tape cells carry the blank symbol.
The input is written on the first tape, starting with the first cell,
the cell immediately to the right of the ``$\$$''.

A \emph{configuration} is a tuple~$C = (s,x_1,p_1,\dotsc,x_m,p_m)$,
where~$s \in S$,~$x_i \in \Sigma^{*}$, and~$0 \leq p_i \leq \Card{x_i}+1$
for each~$1 \leq i \leq k$.
Intuitively,~$\$x_i\Box\Box\dotsc$ is the sequence of symbols in the
cells of tape~$i$,
and the head of tape~$i$ scans the~$p_i$-th cell.
The \emph{initial configuration} for an input~$x \in \Sigma^{*}$
is~$C_0(x) = (s_0,x,1,\epsilon,1,\dotsc,\epsilon,1)$,
where~$\epsilon$ denotes the empty word.

A \emph{computation step} of~$\mathbb{M}$ is a pair~$(C,C')$ of configurations
such that the transformation from~$C$ to~$C'$ obeys the transition relation.
We omit the formal details.
We write~$C \rightarrow C'$ to denote that~$(C,C')$ is a computation step
of~$\mathbb{M}$.
If~$C \rightarrow C'$, we call~$C'$ a \emph{successor configuration} of~$C$.
A \emph{halting configuration} is a configuration that has no successor configuration.
A halting configuration is \emph{accepting} if its state is in~$F$.
A step~$C \rightarrow C'$ is \emph{nondeterministic} if there is
a configuration~$C'' \neq C'$ such that~$C \rightarrow C''$,
and is \emph{existential} if~$C$ is an existential configuration.
A state~$s \in S$ is called \emph{deterministic} if for
any~$a_1,\dotsc,a_m \in \Sigma \cup \SBs \$,\Box \SEs$,
there is at most one~$(s,(a_1,\dotsc,a_m),s',(a'_1,\dotsc,a'_m),(d_1,\dotsc,d_m)) \in \Delta$.
Similarly, we call a non-halting configuration \emph{deterministic} if its state is deterministic,
and \emph{nondeterministic} otherwise.

A configuration is called \emph{existential} if it is not a halting configuration and its state
is in~$S_{\exists}$, and \emph{universal} if it is not a halting configuration and its state
is in~$S_{\forall}$. Intuitively, in an existential configuration, there must be one
possible run that leads to acceptance, whereas in a universal configuration,
all runs must lead to acceptance. Formally, a \emph{run} of an
ATM~$\mathbb{M}$
is a directed tree where each node is labeled with a configuration
of~$\mathbb{M}$
such that:
\begin{itemize}
  \item The root is labeled with an initial configuration.
  \item If a vertex is labeled with an existential configuration~$C$,
    then the vertex has precisely one child that is labeled with a successor
    configuration of~$C$.
  \item If a vertex is labeled with a universal configuration~$C$,
    then for every successor configuration~$C'$ of~$C$ the vertex has a child
    that is labeled with~$C'$.
\end{itemize}
We often identify nodes of the tree with the configurations with which they are labeled.
The run is \emph{finite} if the tree is finite, and \emph{infinite} otherwise.
The \emph{length} of the run is the height of the tree.
The run is \emph{accepting} if it is finite and every leaf is labeled
with an accepting configuration.
If the root of a run~$\rho$ is labeled with~$C_0(x)$,
then~$\rho$ is a run \emph{with input~$x$}.
Any path from the root of a run~$\rho$ to a leaf of~$\rho$
is called a \emph{computation path}.

The \emph{language (or problem) accepted by~$\mathbb{M}$}
is the set~$Q_{\mathbb{M}}$ of all~$x \in \Sigma^{*}$
such that there is an accepting run of~$\mathbb{M}$
with initial configuration~$C_0(x)$.
\emph{$\mathbb{M}$ runs in time~$t : \mathbb{N} \rightarrow \mathbb{N}$}
if for every~$x \in \Sigma^{*}$ the length of every run
of~$\mathbb{M}$ with input~$x$ is at most~$t(\Card{x})$.

A \emph{step}~$C \rightarrow C'$ is an \emph{alternation}
if either~$C$ is existential and~$C'$ is universal, or vice versa.
A run~$\rho$ of~$\mathbb{M}$ is \emph{$\ell$-alternating}, for
an~$\ell \in \mathbb{N}$,
if on every path in the tree associated with~$\rho$,
there are less than~$\ell$ alternations between existential and universal configurations.
The machine~$\mathbb{M}$ is \emph{$\ell$-alternating}
if every run of~$\mathbb{M}$ is $\ell$-alternating.

We now consider two particular types of ATMs.
An \emph{\EA{}-Turing machine} (or simply {\EA{}-machine})
is a~$2$-alternating ATM~$(S_{\exists},S_{\forall},\Sigma,\Delta,s_0,F)$,
where~$s_0 \in S_{\exists}$.
Let~$\ell,t \geq 1$ be positive integers.
We say that an \EA{}-machine \emph{$\mathbb{M}$ halts (on the empty string)
with existential cost~$\ell$ and universal cost~$t$}
if:
\begin{itemize}
  \item there is an accepting run of~$\mathbb{M}$ with input~$\epsilon$,
  \item each computation path of~$\mathbb{M}$ contains at
    most~$\ell$ existential configurations
    and at most~$t$ universal configurations.
\end{itemize}

Let~$P$ be a parameterized problem.
An \emph{\EkA{}-machine for~$P$} is a \EA{}-machine~$\mathbb{M}$
such that there exists a computable function~$f$
and a polynomial~$p$ such that
\begin{itemize}
  \item $\mathbb{M}$ decides~$P$ in time~$f(k) \cdot p(\Card{x})$; and
  \item and for all instances~$(x,k)$ of~$P$
    and each computation path~$R$ of~$\mathbb{M}$ with input~$(x,k)$,
    at most~$f(k) \cdot \log \Card{x}$ of the existential configurations
    of~$R$ are nondeterministic.
\end{itemize}
%Whenever $\mathbb{M}$ is an \EkA{}-machine for $P$ such that $P$
%is the language accepted by $\mathbb{M}$, we say that
%\emph{$\mathbb{M}$ decides $P$}.
We say that a parameterized problem~$P$ \emph{is decided by some
$\EkA$-machine} if there exists a \EkA{}-machine for~$P$.

Let~$m \in \mathbb{N}$ be a positive integer.
We consider the following parameterized decision problems.
\probdef{
  \EkATMhalt{*}.
  
  \emph{Instance:} An~$\exists\forall$-machine~$\mathbb{M}$
    with~$m$ tapes,
    and positive integers~$k,t \geq 1$.

  \emph{Parameter:} $k$.

  \emph{Question:} Does~$\mathbb{M}$ halt
    on the empty string
    with existential cost~$k$ and universal cost~$t$?
}
\probdef{
  \EkATMhalt{m}.
  
  \emph{Instance:} An $\exists\forall$-machine~$\mathbb{M}$
    with~$m$ tapes,
    and positive integers~$k,t \geq 1$.

  \emph{Parameter:} $k$.

  \emph{Question:} Does~$\mathbb{M}$ halt
    on the empty string
    with existential cost~$k$ and universal cost~$t$?
}
\noindent Note that for \EkATMhalt{m}, the number~$m$ of tapes
of the \EA{}-machines in the input is a fixed constant,
whereas for \EkATMhalt{*}, the number of tapes is given as part of the input.

The parameterized complexity class \EkA{}
is characterized by alternating Turing machines in the following way.
These results can be seen as an analogue to the Cook-Levin Theorem
for the complexity class \EkA{}.
\begin{theorem}
\label{thm:atm-char1}
The problem \EkATMhalt{*} is \EkA{}-complete,
and so is the problem \EkATMhalt{m} for each~$m \in \mathbb{N}$.
\end{theorem}
\begin{theorem}
\label{thm:atm-char2}
\EkA{} is exactly the class of parameterized decision problems~$P$
that are decided by some \EkA{}-machine.
\end{theorem}
\begin{proof}[Proof of Theorems~\ref{thm:atm-char1}~and~\ref{thm:atm-char2}.]
In order to show these results, concretely,
we will use the following statements.
We show how the results follow from these statements.
Detailed proofs of the statements can be found
in the appendix, in Section~\ref{sec:app-atm-char}.
\begin{enumerate}[(i)]
  \item \EkATMhalt{*} \fptred{} \EkAMC{} (Proposition~\ref{prop:tm-char-1}).
  \item For any parameterized problem~$P$ that is decided by some \EkA{}-machine
    with~$m$ tapes, it holds that~$P \fptred{} \EkATMhalt{m+1}$
    (Proposition~\ref{prop:tm-char-2}).
  \item There is an \EkA{}-machine with a single tape that decides \EleqkAWSat{}
    (Proposition~\ref{prop:tm-char-3}).
  \item Let~$A$ and~$B$ be parameterized problem.
    If~$B$ is decided by some \EkA{}-machine with~$m$ tapes,
    and if~$A \fptred{} B$,
    then~$A$ is decided by some \EkA{}-machine with~$m$ tapes
    (Proposition~\ref{prop:tm-char-4}).
\end{enumerate}
In addition to these statements, we will need one result known
from the literature (Proposition~\ref{prop:multitape-TM-simulation}
and Corollary~\ref{cor:multitape-TM-simulation}).
To see that these statements imply the desired results,
observe the following.

Together, (ii) and (iii) imply that~$\EleqkAWSat{} \fptred{} \EkATMhalt{2}$.
Clearly, for all~$m \geq 2$, $\EkATMhalt{2} \fptred{} \EkATMhalt{m}$.
This gives us \EkA{}-hardness of \EkATMhalt{m}, for all~$m \geq 2$.
\EkA{}-hardness of \EkATMhalt{1} follows from
Corollary~\ref{cor:multitape-TM-simulation}, which implies that there is
an fpt-reduction from \EkATMhalt{2} to \EkATMhalt{1}.
This also implies that \EkATMhalt{*} is \EkA{}-hard.
Then, by (i), and since \EkAMC{} is in \EkA{}
by Theorem~\ref{thm:modelchecking-completeness},
we obtain \EkA{}-completeness of \EkATMhalt{*} 
and \EkATMhalt{m}, for each~$m \geq 1$.

By (i) and (ii), and by transitivity of fpt-reductions,
we have that any parameterized problem~$P$ that is decided by an \EkA{}-machine
is fpt-reducible to \EkAWSat{}, and thus is in \EkA{}.
Conversely, let~$P$ be any parameterized problem in \EkA{}.
Then, by \EkA{}-hardness of \EleqkAWSat{}, we know that~$P \fptred{} \EleqkAWSat{}$.
By (iii) and (iv), we know that~$P$ is decided by some \EkA{}-machine with a single tape.
From this we conclude that
\EkA{} is exactly the class of parameterized problems~$P$
decided by some \EkA{}-machine.
\end{proof}
}
}

\krlongversion{
%%% THE \stark{} HIERARCHY (KR / ARXIV)
\section{The \stark{} Hierarchy}
\label{sec:eak}

We now turn our attention to the \stark{} hierarchy.
Unlike in the \kstar{} hierarchy, in the canonical quantified
Boolean satisfiability problems
of the \stark{} hierarchy, we cannot add auxiliary variables
to the second quantifier block whose truth assignment
is not restricted.
Therefore, because of similarity to the W-hierarchy,
we believe that the classes of  the \stark{} hierarchy are distinct.
%We will mainly focus on the first level of the \stark{} hierarchy.
\krversion{We begin with proving some basic properties.}%
The main results in this section are normalization
results for \EAkW{1} and \EAkW{P}.
In addition, we show a completeness result for the \stark{} hierarchy
for another variant of our running example.

\krversion{
\begin{proposition}
If $\para{\co\NP} \subseteq \EAkW{P}$,
then $\NP = \co\NP$.
\end{proposition}
\begin{proof}[Proof (sketch)]
With an argument similar to the one in the proof of Proposition~\ref{prop:1},
a polynomial-time reduction from $\thUNSAT{}$ to $\thSAT{}$
can be constructed.
An additional technical observation needed for this case is that
\thSAT{} is in \NP{} also when the input is a Boolean circuit.
\end{proof}
}

%%% NORMALIZATION RESULTS FOR THE *-K HIERARCHY
\subsection{Normalization Results for the \stark{} Hierarchy}

\krversion{Next, w}\longversion{W}e
show that the problem \EAkWSat{} is already \EAkW{1}-hard
when the input circuits are restricted to formulas in $c\mtext{-}\DNF$,
for any constant $c \geq 2$.
In order to make our life easier, we switch our perspective
to the co-problem \AEkWSat{}
when stating and proving the following results.
Note that the proofs of the following results
make heavy use of the original normalization
proof for the class \W{1} by Downey and
Fellows~\cite{DowneyFellows95,DowneyFellows99,DowneyFellows13}.
Therefore, we provide only proof sketches.

\begin{lemma}
\label{lem:aeknorm}
For any $u \geq 1$,
$\AEkWSat(\Gamma_{1,u}) \fptred{} \AEkWSat(s\mtext{-}\CNF)$,
where $s = 2^u+1$.
\end{lemma}
\begin{proof}[Proof (sketch).]
The reduction is completely analogous to the
reduction used in the proof of 
Downey and Fellows~\cite[Lemma~2.1]{DowneyFellows95},
where the presence of universally quantified variables
is handled in four steps.
In Steps~1~and~2, in which only the form of the circuit is modified,
no changes are needed.
In Step~3, universally quantified variables can be handled
exactly as existentially quantified variables.
Step~4 can be performed with only a slight modification,
the only difference being that universally quantified variables
appearing in the input circuit
will also appear in the resulting clauses that verify
whether a given product-of-sums or sum-of-products is satisfied.
It is straightforward to verify that this reduction with the mentioned modifications
works for our purposes.
\end{proof}

\begin{theorem}
\label{thm:stark-2cnf-hardness}
$\AEkWSat(2\CNF)$ is \AEkW{1}-complete.
\end{theorem}
\begin{proof}[Proof (sketch).]
Clearly $\AEkWSat(2\CNF)$ is in $\AEkW{1}$, since
a $2\CNF{}$ formula can be considered as a
constant-depth circuit of weft 1.
To show that $\AEkWSat(2\CNF)$ is $\AEkW{1}$-hard,
we give an fpt-reduction from $\AEkWSat(\Gamma_{1,u})$ to $\AEkWSat(2\CNF)$,
for arbitrary $u \geq 1$.
By Lemma~\ref{lem:aeknorm}, we know that we can reduce $\AEkWSat(\Gamma_{1,u})$
to $\AEkWSat(s\mtext{-}\CNF)$, for $s = 2^u+1$.
We continue the reduction in multiple steps.
In each step, we let $C$ denote the circuit resulting from the previous step,
and we let $Y$ denote the universally quantified
and $X$ the existentially quantified variables of $C$,
and we let $k$ denote the parameter value.
We only briefly describe the last two steps,
since these are completely analogous to constructions in the work of
Downey and Fellows~\cite{DowneyFellows99}.

\paragraph{Step 1: contracting the universally quantified variables.}
This step transforms~$C$ into a \CNF{} formula~$C'$
such that each clause contains at most one variable in~$Y$
such that~$(C,k)$ is a yes-instance if and only if~$(C',k)$ is a yes-instance.
We introduce new universally quantified variables~$Y'$
containing a variable~$y'_{A}$ for each set~$A$ of literals
over~$Y$ of size at least~$1$ and at most~$s$.
Now, it is straightforward to construct a set~$D$ of
polynomially many ternary clauses
over~$Y$ and~$Y'$ such that the following property holds.
An assignment~$\alpha$ to~$Y \cup Y'$ satisfies~$D$
if and only if for each subset~$A = \SBs l_1,\dotsc,l_b \SEs$ of literals over $Y$
it holds that~$\alpha(l_1) = \alpha(l_2) = \dotsc = \alpha(l_b) = 1$
if and only if~$\alpha(y'_{A}) = 1$.
Note that we do not directly add the set~$D$ of clauses to the formula~$C'$.

We introduce~$k-1$ many new existentially quantified
variables~$x^{\star}_1,\dotsc,x^{\star}_{k-1}$.
We add binary clauses to~$C'$ that enforce that
the variables~$x^{\star}_1,\dotsc,x^{\star}_{k-1}$
all get the same truth assignment.
Also, we add binary clauses to~$C'$ that enforce that each~$x \in X$
is set to false if~$x^{\star}_1$ is set to true.

We introduce~$\Card{D}$ many existentially quantified variables,
including a variable~$x''_d$ for each clause~$d \in D$.
For each~$d \in D$, we add the following clauses to~$C'$.
Let~$d = (l_1,l_2,l_3)$, where each~$l_i$ is a literal over~$Y \cup Y'$.
We add the clauses~$(\neg x''_d \vee \neg l_1)$,~$(\neg x''_d \vee \neg l_2)$ and~$(\neg x''_d \vee \neg l_3)$,
enforcing that the clause~$d$ cannot be satisfied
if~$x''_d$ is set to true.

We then modify the clauses of~$C$ as follows.
Let~$c = (l^x_1,\dotsc,l^x_{s_1},l^y_1,\dotsc,l^y_{s_2})$ be a clause of~$C$,
where~$l^x_1,\dotsc,l^x_{s_1}$ are literals over~$X$,
and~$l^y_1,\dotsc,l^y_{s_2}$ are literals over~$Y$.
We replace~$c$ by the clause~$(l^x_1,\dotsc,l^x_{s_1}, x^{\star}_1, y'_{B})$,
where~$B = \SBs l^y_1,\dotsc,l^y_{s_2} \SEs$.
Clauses~$c$ of~$C$ that contain no literals over the variables~$Y$
remain unchanged.

The idea of this reduction is the following.
If~$x^{\star}_1$ is set to true, then exactly one of the
variables~$x''_d$ must be set to true,
which can only result in an satisfying assignment
if the clause~$d \in D$ is not satisfied.
Therefore, if an assignment~$\alpha$ to the variables~$Y \cup Y'$
does not satisfy~$D$, there is a satisfying assignment of weight~$k$ that sets
both~$x^{\star}_1$ and~$x''_d$ to true, for some~$d \in D$
that is not satisfied by~$\alpha$.
Otherwise, we know that the value~$\alpha$ assigns to variables~$y'_{A}$
corresponds to the value~$\alpha$ assigns to~$\bigwedge_{a \in A} a$,
for~$A \subseteq \Lit{Y}$.
Then any satisfying assignments of weight~$k$ for~$C$
is also a satisfying assignments of weight~$k$ for~$C'$.

\paragraph{Step 2: making~$C$ antimonotone in~$X$.}
This step transforms~$C$ into a circuit~$C'$ that has only negative
occurrences of existentially quantified variables,
and transforms~$k$ into~$k'$ depending only on~$k$,
such that~$(C,k)$ is a yes-instance if and only if~$(C',k')$
is a yes-instance.
The reduction is completely analogous to the reduction in the proof of
Downey and Fellows~\cite[Theorem~10.6]{DowneyFellows99}.

\paragraph{Step 3: contracting the existentially quantified variables.}
This step transforms~$C$ into a circuit~$C'$ in \CNF{}
that contains only clauses with two variables in~$X$ and no variables in~$Y$
and clauses with one variable in~$X$ and one variable in~$Y$,
and transforms~$k$ into~$k'$ depending only on~$k$,
such that~$(C,k)$ is a yes-instance if and only if~$(C',k')$
is a yes-instance.
The reduction is completely analogous to the reduction in the proof of
Downey and Fellows~\cite[Theorem~10.7]{DowneyFellows99}.
\end{proof}

\begin{corollary}
\label{cor:stark-2cnf-hardness}
For any fixed integer~$r \geq 2$, the problem
$\EAkWSat(r\mtext{-}\DNF)$ is \EAkW{1}-complete.
\end{corollary}

Next, we consider a normalization result for \EAkW{P}.
In order to do so, we will need some definitions.
Let~$C$ be a quantified Boolean circuit over
two disjoint sets~$X$ and~$Y$ of variables
that is in negation normal form.
We say that~$C$ is \emph{monotone in the variables~$Y$} 
if the only negation nodes that occur in the circuit~$C'$
have variables in~$X$ as inputs,
i.e., the variables in~$Y$ can appear only positively in the circuit.
Then, the following restriction of \EAkWSat{} is already \EAkW{P}-hard.

\begin{proposition}
\label{prop:eakwp-monotone}
The problem~\EAkWSat{} is \EAkW{P}-hard,
even when restricted to quantified circuits that
are in negation normal form and that
are monotone in the universal variables.
\end{proposition}
\begin{proof}
We give an fpt-reduction from the problem \EAkWSat{}
to the problem \EAkWSat{} restricted to circuits that are
monotone in the universal variables.
Let~$(C,k)$ be an instance of \EAkWSat{},
where~$C$ is a quantified Boolean circuit over the set~$X$
of existential variables and the set~$Y$ of universal
variables, where~$X = \SBs x_1,\dotsc,x_n \SEs$
and where~$Y = \SBs y_1,\dotsc,y_m \SEs$.
We construct an equivalent instance~$(C',k)$ of \EAkWSat{}
where~$C'$ is a quantified Boolean circuit over the set~$X$
of existential variables and the set~$Y'$ of universal variables,
and where the circuit~$C'$ is monotone in~$Y'$.
We may assume without loss of generality that~$C$ is in
negation normal form.
If this is not the case, we can simply transform~$C$ into an equivalent
circuit that has this property using the De Morgan rule.
The form of the circuit~$C$ is depicted in Figure~\ref{fig:original-circuit}.

This construction bears some resemblance to the construction used
in a proof by Flum and Grohe~\cite[Theorem~3.14]{FlumGrohe06}.
The plan is to replace the variables in~$Y$
by~$k$ copies of them, grouped in sets~$Y^1,\dotsc,Y^k$ of
new variables.
Each assignment of weight~$k$ to the new
variables that sets a copy of a different variable to true
in each set~$Y^i$ corresponds exactly to an assignment of weight~$k$
to the original variables in~$Y$.
Moreover, we will ensure that each assignment of weight~$k$ to the
new variables that does not set a copy of a different variable to true
in each set~$Y^i$ satisfies the newly constructed circuit.
Using these new variables we can then construct internal
nodes~$y_j$ and~$y'_j$ that, for each assignment to the new input nodes~$Y'$, evaluate to the truth value assigned to~$y_j$ and~$\neg y_j$, respectively, by
the corresponding truth assignment to the original input nodes~$Y$.

We will describe this construction in more detail.
The construction is also depicted in Figure~\ref{fig:monotone-circuit}.
We let~$Y' = \SB y^i_j \SM 1 \leq i \leq k, 1 \leq j \leq m \SE$.
We introduce a number of new internal nodes.
For each~$1 \leq j \leq m$, we introduce an internal node~$y_j$,
that is the disjunction
of the input nodes~$y^i_j$, for~$1 \leq i \leq k$.
That is, the internal node~$y_j$ is true if and only
if~$y^i_j$ is true for some~$1 \leq i \leq k$.
Intuitively, this node~$y_j$ corresponds to the input node~$y_j$
in the original circuit~$C$.
Moreover, we introduce an internal node~$y'_{j,i}$
for each~$1 \leq j \leq m$ and each~$1 \leq i \leq k$,
that is the disjunction of~$y^{i}_{j'}$, for each~$1 \leq j' \leq m$
such that~$j \neq j'$.
That is, the node~$y'_{j,i}$ is true if and only
if~$y^i_{j'}$ is true for some~$j'$ that is different from~$j$.
Then, we introduce the node~$y'_j$, for each~$1 \leq j \leq m$,
that is the conjunction of the nodes~$y'_{j,i}$ for~$1 \leq i \leq k$.
That is, the node~$y'_j$ is true if and only
if for each~$1 \leq i \leq k$ there is some~$j' \neq j$
for which the input node~$y^i_j$ is true.
Intuitively, this node~$y'_j$ corresponds to the negated input
node~$\neg y_j$ in the original circuit~$C$.
Also, for each~$1 \leq i \leq k$ and each~$1 \leq j < j' \leq m$,
we add an internal node~$z^{j,j'}_{i}$ that is the conjunction
of the input nodes~$y^i_j$ and~$y^i_{j'}$.
Then, for each~$1 \leq i \leq k$ we add the internal node~$z_i$
that is the conjunction of all nodes~$z_i^{j,j'}$, for~$1 \leq j < j' \leq m$.
Intuitively,~$z_i$ is true if and only if at least two input nodes
in the set~$Y_i$ are set to true.
In addition, we add a subcircuit~$B$ that acts on the nodes~$y'_1,\dotsc,y'_m$,
and that is satisfied if and only if at least~$m-k+1$ of the nodes~$y'_j$
are set to true.
It is straightforward to construct such a circuit~$B$ in polynomial time.
Then, we add the subcircuit~$C$ with input nodes~$x_1,\dotsc,x_n$,
negated input nodes~$\neg x_1,\dotsc,\neg x_n$,
where the input nodes~$y_1,\dotsc,y_m$ are identified with
the internal nodes~$y_1,\dotsc,y_m$ in the newly constructed circuit~$C'$,
and where the negated input nodes~$\neg y_1,\dotsc,\neg y_m$ are
identified with the internal nodes~$y'_1,\dotsc,y'_m$
in the newly constructed circuit~$C'$.
Finally, we let the output node be the disjunction of the
nodes~$z_1,\dotsc,z_k$ and the
output nodes of the subcircuits~$C$ and~$B$.
Since~$C$ is a circuit in negation normal form,
the circuit~$C'$ is monotone in~$Y'$.
We claim that for each assignment~$\alpha : X \rightarrow \SBs 0,1 \SEs$
it holds that the circuit~$C[\alpha]$ is satisfied by all assignments of weight~$k$
if and only if~$C'[\alpha]$ is satisfied by all assignments of weight~$k$.

\begin{figure}[ht!]
\begin{center}
\begin{tikzpicture}[scale=0.5]
  \node[] at (0,0) {$x_1$};
  \node[] at (1,0) {$\dotsc$};
  \node[] at (2,0) {$x_n$};
  \node[] at (3,0) {$\neg x_1$};
  \node[] at (4,0) {$\dotsc$};
  \node[] at (5,0) {$\neg x_n$};
  \node[] at (7,0) {$y_1$};
  \node[] at (8,0) {$\dotsc$};
  \node[] at (9,0) {$y_m$};
  \node[] at (10,0) {$\neg y_1$};
  \node[] at (11,0) {$\dotsc$};
  \node[] at (12,0) {$\neg y_m$};
  \draw[] (-.5,-.5) -- (12.8,-.5) -- (6,-5) -- cycle;
  \node[] at (6,-2.5) {$C$};
\end{tikzpicture}
\end{center}
\caption{Original quantified Boolean circuit~$C$
over the set~$X = \SBs x_1,\dotsc,x_n \SEs$ of existential variables
and the set~$Y = \SBs y_1,\dotsc,y_m \SEs$ of universal variables.}
\label{fig:original-circuit}
\end{figure}
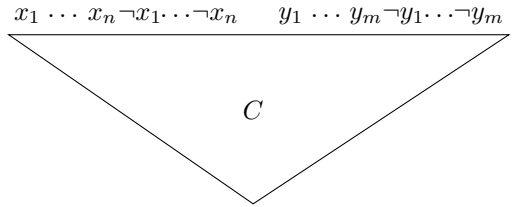

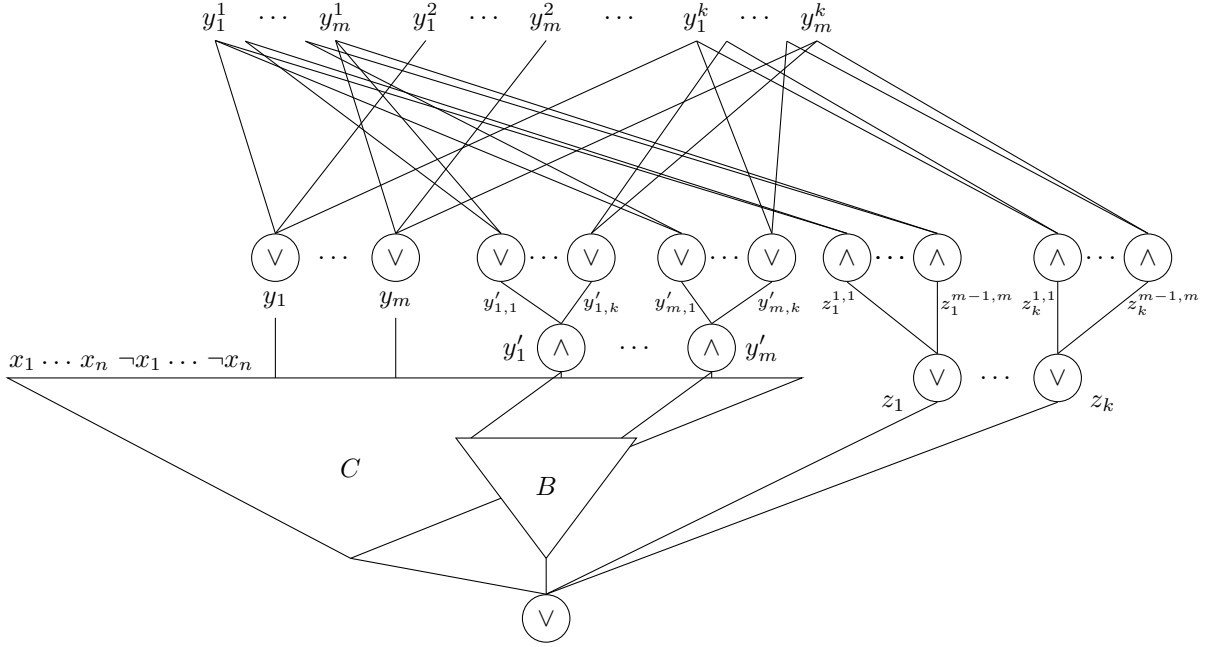
\begin{figure}[ht!]
\begin{center}
\begin{tikzpicture}[scale=0.8]
  \node[] (x1) at (2.8,-3.75) {$x_1$};
  \node[] at (3.4,-3.75) {$\dotsc$};
  \node[] (xn) at (4,-3.75) {$x_n$};
  \node[] (nx1) at (4.75,-3.75) {$\neg x_1$};
  \node[] at (5.5,-3.75) {$\dotsc$};
  \node[] (nxn) at (6.25,-3.75) {$\neg x_n$};
  \node[] (y11) at (6,2) {$y^1_1$};
  \node[] at (7,2) {$\dotsc$};
  \node[] (y1m) at (8,2) {$y^1_m$};
  \node[] (y21) at (9.5,2) {$y^2_1$};
  \node[] at (10.5,2) {$\dotsc$};
  \node[] (y2m) at (11.5,2) {$y^2_m$};
  \node[] at (12.75,2) {$\dotsc$};
  \node[] (yk1) at (14,2) {$y^k_1$};
  \node[] at (15,2) {$\dotsc$};
  \node[] (ykm) at (16,2) {$y^k_m$};
  \node[label=270:{$y_1$},draw,circle] (y1) at (7,-2) {$\vee$};
  \node[] at (8,-2) {$\dotsc$};
  \draw[-] (y11.south) -- (y1.north);
  \draw[-] (y21.south) -- (y1.north);
  \draw[-] (yk1.south) -- (y1.north);
  \node[label=-90:{$y_m$},draw,circle] (ym) at (9,-2) {$\vee$};
  \draw[-] (y1m.south) -- (ym.north);
  \draw[-] (y2m.south) -- (ym.north);
  \draw[-] (ykm.south) -- (ym.north);
  \node[label=180:{$y'_1$},draw,circle] (ny1) at (11.75,-3.5) {$\wedge$};
  \node[label={[xshift=-0pt]-90:{\scriptsize $y'_{1,1}$}},draw,circle] (ny11) at (10.75,-2) {$\vee$};
  \draw[-] (6.5,1.6) -- (ny11.north);
  \draw[-] (y1m.south) -- (ny11.north);
  \node[] at (11.5,-2) {$\dotsc$};
  \node[label={[xshift=4pt]-90:{\scriptsize $y'_{1,k}$}},draw,circle] (ny1k) at (12.25,-2) {$\vee$};
  \draw[-] (14.5,1.6) -- (ny1k.north);
  \draw[-] (ykm.south) -- (ny1k.north);
  \draw[-] (ny11.south) -- (ny1.north);
  \draw[-] (ny1k.south) -- (ny1.north);
  \node[] at (13,-3.5) {$\dotsc$};
  \node[label=0:{$y'_m$},draw,circle] (nym) at (14.25,-3.5) {$\wedge$};
  \node[label={[xshift=-2pt]-90:{\scriptsize $y'_{m,1}$}},draw,circle] (nym1) at (13.75,-2) {$\vee$};
  \draw[-] (y11.south) -- (nym1.north);
  \draw[-] (7.5,1.6) -- (nym1.north);
  \node[] at (14.5,-2) {$\dotsc$};
  \node[label={[xshift=3pt]-90:{\scriptsize $y'_{m,k}$}},draw,circle] (nymk) at (15.25,-2) {$\vee$};
  \draw[-] (yk1.south) -- (nymk.north);
  \draw[-] (15.5,1.6) -- (nymk.north);
  \draw[-] (nym1.south) -- (nym.north);
  \draw[-] (nymk.south) -- (nym.north);
  \draw[-] (7,-3) -- (7,-4);
  \draw[-] (9,-3) -- (9,-4);
  \draw[-] (ny1.south) -- (11.75,-4);
  \draw[-] (nym.south) -- (14.25,-4);
  \node[draw,circle,label={200:{$z_1$}}] (z1) at (18,-4) {$\vee$};
  \node[] at (19,-4) {$\dotsc$};
  \node[draw,circle,label={-20:{$z_k$}}] (zk) at (20,-4) {$\vee$};
  \node[draw,circle,label={[xshift=-3pt]-90:{\scriptsize $z_1^{1,1}$}}] (z111) at (16.5,-2) {$\wedge$};
  \node[] at (17.25,-2) {$\dotsc$};
  \node[draw,circle,label={[xshift=15pt]-90:{\scriptsize $z_1^{m-1,m}$}}] (z1m-1m) at (18,-2) {$\wedge$};
  \node[] at (17.25,-2) {$\dotsc$};
  \node[draw,circle,label={[xshift=-7pt]-90:{\scriptsize $z_k^{1,1}$}}] (zk11) at (20,-2) {$\wedge$};
  \node[] at (20.75,-2) {$\dotsc$};
  \node[draw,circle,label={[xshift=6pt]-90:{\scriptsize $z_k^{m-1,m}$}}] (zkm-1m) at (21.5,-2) {$\wedge$};
  \draw[-] (z111.south) -- (z1.north);
  \draw[-] (z1m-1m.south) -- (z1.north);
  \draw[-] (zk11.south) -- (zk.north);
  \draw[-] (zkm-1m.south) -- (zk.north);
  \draw[-] (y11.south) -- (z111.north);
  \draw[-] (6.5,1.6) -- (z111.north);
  \draw[-] (7.5,1.6) -- (z1m-1m.north);
  \draw[-] (y1m.south) -- (z1m-1m.north);
  \draw[-] (yk1.south) -- (zk11.north);
  \draw[-] (14.5,1.6) -- (zk11.north);
  \draw[-] (15.5,1.6) -- (zkm-1m.north);
  \draw[-] (ykm.south) -- (zkm-1m.north);
  % OLD CIRCUIT C
  \draw[] (2.55,-4) -- (15.75,-4) -- (8.25,-7) -- cycle;
  \node[] at (8.25,-5.5) {$C$};
  % COUNTING CIRCUIT
  \draw[fill=white] (10,-5) -- (13,-5) -- (11.5,-7) -- cycle;
  \draw[-] (ny1.south) -- (10.25,-5);
  \draw[-] (nym.south) -- (12.75,-5);
  \node[] at (11.5,-5.8) {$B$};
  % FINAL DISJUNCTION
  \node[draw,circle] (output) at (11.5,-8) {$\vee$};
  \draw[-] (11.5,-7) -- (output.north);
  \draw[-] (8.25,-7) -- (output.north);
  \draw[-] (z1.south) -- (output.north);
  \draw[-] (zk.south) -- (output.north);
\end{tikzpicture}
\end{center}
\caption{Equivalent circuit~$C'$
over the set~$X = \SBs x_1,\dotsc,x_n \SEs$
of existential variables and the
set~$Y' = \SBs y^1_1,\dotsc,y^k_m \SEs$ of
universal variables,
that is monotone in~$Y'$.}
\label{fig:monotone-circuit}
\end{figure}

$(\Rightarrow)$
Let~$\alpha : X \rightarrow \SBs 0,1 \SEs$ be an arbitrary truth assignment.
Assume that~$C[\alpha]$ is satisfied by all truth assignments~$\beta : Y
\rightarrow \SBs 0,1 \SEs$ of weight~$k$.
We show that~$C'[\alpha]$ is satisfied by all truth assignments~$\beta' : Y'
\rightarrow \SBs 0,1 \SEs$ of weight~$k$.
Let~$\beta' : Y' \rightarrow \SBs 0,1 \SEs$ be an arbitrary truth assignment
of weight~$k$.
We distinguish several cases:
either (i)~for some~$1 \leq i \leq k$
there are some~$1 \leq j < j' \leq m$ such that $\beta'(y^i_j) = \beta'(y^i_{j'}) = 1$,
or (ii)~for each~$1 \leq i \leq k$ there is exactly one~$\ell_i$ such
that~$\beta'(y^i_{\ell_i}) = 1$ and for some~$1 \leq i < i' \leq k$
it holds that~$\ell_i = \ell_{i'}$,
or (iii)~for each~$1 \leq i \leq k$ there is exactly one~$\ell_i$ such
that~$\beta'(y^i_{\ell_i}) = 1$ and for each~$1 \leq i < i' \leq k$
it holds that~$\ell_i \neq \ell_{i'}$.
In case~(i), we know that the assignment~$\beta'$ sets the node~$z_i^{j,j'}$
to true.
Therefore,~$\beta'$ sets the node~$z_i$ to true,
and thus satisfies the circuit~$C'[\alpha]$.
In case~(ii), we know that~$\beta'$ sets~$y'_j$ to true
for at least~$m-k+1$ many different values of~$j$.
Therefore,~$\beta'$ satisfies the subcircuit~$B$,
and thus satisfies~$C'[\alpha]$.
Finally, in case~(iii), we know that~$\beta'$ sets exactly~$k$ different
internal nodes~$y_j$ to true, and for each~$1 \leq j \leq m$
sets the internal node~$y'_j$ to true if and only if it sets~$y_j$ to false.
Then, since~$C[\alpha]$ is satisfied by all truth assignments of
weight~$k$, we know that~$\beta'$ satisfies the subcircuit~$C$,
and thus satisfies~$C'[\alpha]$.
Since~$\beta'$ was arbitrary, we can conclude
that~$C'[\alpha]$ is satisfied by all truth assignments~$\beta' : Y'
\rightarrow \SBs 0,1 \SEs$ of weight~$k$.

$(\Leftarrow)$
Let~$\alpha : X \rightarrow \SBs 0,1 \SEs$ be an arbitrary truth assignment.
Assume that~$C'[\alpha]$ is satisfied by all truth assignments~$\beta' : Y'
\rightarrow \SBs 0,1 \SEs$ of weight~$k$.
We show that~$C[\alpha]$ is satisfied by all truth assignments~$\beta : Y
\rightarrow \SBs 0,1 \SEs$ of weight~$k$.
Let~$\beta : Y \rightarrow \SBs 0,1 \SEs$ be an arbitrary truth assignment
of weight~$k$.
We now define the truth assignment~$\beta' : Y' \rightarrow \SBs 0,1 \SEs$
as follows.
Let~$\SBs y_{\ell_1},\dotsc,y_{\ell_k} \SEs = \SB y_j \SM 1 \leq j \leq m,
\beta(y_j) = 1 \SE$.
For each~$1 \leq i \leq k$ and each~$1 \leq j \leq m$
we let~$\beta'(y^i_j) = 1$ if and only if~$j = \ell_i$.
Clearly,~$\beta'$ has weight~$k$.
Moreover, the assignment~$\beta'$ sets the nodes~$z_1,\dotsc,z_k$
to false.
Furthermore, it is the case that~$\beta'$ sets the internal node~$y_j$
in~$C'$ to true for exactly those~$1 \leq j \leq m$ for which~$\beta(y_j) = 1$,
and it sets the internal node~$y'_j$ in~$C'$ to true for exactly
those~$1 \leq j \leq m$ for which~$\beta(y_j) = 0$.
Thus,~$\beta'$ sets (the output node of) the subcircuit~$B$ to false.
Therefore, since~$\beta'$ satisfies the circuit~$C'[\alpha]$,
we can conclude that~$\beta'$ satisfies the subcircuit~$C$,
and thus that~$\beta$ satisfies~$C[\alpha]$.
Since~$\beta$ was arbitrary, we can conclude
that~$C[\alpha]$ is satisfied by all truth assignments~$\beta : Y
\rightarrow \SBs 0,1 \SEs$ of weight~$k$.
\end{proof}

\subsection{Answer Set Programming and Completeness for the \stark{} Hierarchy}

\krversion{%
We now turn to another variant of our running example problem.}
\longversion{%
We now turn to two other variants of our running example problem,
that turn out to be complete for \EAkW{P}.
We begin with showing \EAkW{P}-completeness
for \ASPcons\DisjRules{}.}

%%% KR VERSION %%%
\krversion{
\begin{theorem}
\label{thm:asp-disjrules-hardness}
\ASPcons\DisjRules{}
is \EAkW{1}-hard.
\end{theorem}
\begin{proof}
We give an fpt-reduction from~$\EAkWSat(3\DNF)$,
which we know to be \EAkW{1}-hard from Corollary~\ref{cor:stark-2cnf-hardness}.
This is, like the reduction in the proof of
Proposition~\ref{prop:asp-contrules-hardness} above, a parameterized
version of a reduction of Eiter and Gottlob~\cite[Theorem~3]{EiterGottlob95}.
Let~$(\varphi,k)$ be an instance of~$\EAkWSat(3\DNF)$,
where~$\varphi = \exists X. \forall Y. \psi$,
$X = \SBs x_1,\dotsc,x_n \SEs$,
$Y = \SBs y_1,\dotsc,y_m \SEs$,
$\psi = \delta_1 \vee \dotsm \vee \delta_u$,
and~$\delta_{\ell} = l^{\ell}_1 \wedge l^{\ell}_2 \wedge l^{\ell}_3$ for each~$1 \leq \ell \leq u$.
By step 2 in the proof of Theorem~\ref{thm:stark-2cnf-hardness},
we may assume without loss of generality that all universally quantified
variables~$y_1,\dotsc,y_m$ occur only positively in~$\psi$.
We construct a disjunctive program~$P$.
We consider the variables~$X$ and~$Y$ as atoms.
In addition,
we introduce fresh atoms~$v_1,\dotsc,v_n$,~$w$, and~$y^j_i$
for all~$1 \leq j \leq k$,~$1 \leq i \leq m$.
We let~$P$ consist of the following rules:
%%% SHORT VERSION
\krversion{\begin{small}
\begin{align}
  &\label{eak-rule1} x_i \leftarrow  \aspnot{} v_i & \mtext{for } 1 \leq i \leq n; \\[-.5pt]
  &\label{eak-rule2} v_i \leftarrow  \aspnot{} x_i & \mtext{for } 1 \leq i \leq n; \\[-3pt]
  &\label{eak-rule3} y^{j}_{1} \vee \dotsm \vee y^{j}_{m} \leftarrow & \mtext{for } 1 \leq j \leq k; \\[-3pt]
  &y_{i} \leftarrow  y^{j}_{i} & \mtext{for } 1 \leq i \leq m,
    \label{eak-rule4} 1 \leq j \leq k; \\[-3pt]
  &y^{j}_{i} \leftarrow  w & \mtext{for } 1 \leq i \leq m,
    \label{eak-rule5} 1 \leq j \leq k; \\[-4pt]
  &w \leftarrow  y^{j}_{i}, y^{j'}_{i} & \mtext{for } 1 \leq i \leq m,
  \label{eak-rule6} 1 \leq j < j' \leq k; \\[-3pt]
  &\label{eak-rule7} w \leftarrow  \sigma(l^{\ell}_1), \sigma(l^{\ell}_2), \sigma(l^{\ell}_3) \hspace{-20pt} &
    \mtext{for } 1 \leq \ell \leq u; \\[0pt]
  &\label{eak-rule8} w \leftarrow \aspnot{} w.
\end{align}
\end{small}}%
%%% LONG VERSION
\longversion{\begin{small}
\begin{align}
  \label{eak-rule1} x_i \leftarrow &\ \aspnot{} v_i & \mtext{for } 1 \leq i \leq n; \\
  \label{eak-rule2} v_i \leftarrow &\ \aspnot{} x_i & \mtext{for } 1 \leq i \leq n; \\
  \label{eak-rule3} y^{j}_{1} \vee \dotsm \vee y^{j}_{m} \leftarrow & & \mtext{for } 1 \leq j \leq k; \\
  y_{i} \leftarrow &\ y^{j}_{i} & \mtext{for } 1 \leq i \leq m, \krversion{\nonumber \\ && }
    \label{eak-rule4} 1 \leq j \leq k; \\
  y^{j}_{i} \leftarrow &\ w & \mtext{for } 1 \leq i \leq m, \krversion{\nonumber \\ && }
    \label{eak-rule5} 1 \leq j \leq k; \\
  w \leftarrow &\ y^{j}_{i}, y^{j'}_{i} & \mtext{for } 1 \leq i \leq m, \krversion{\nonumber \\ && }
  \label{eak-rule6} 1 \leq j < j' \leq k; \\
  \label{eak-rule7} w \leftarrow &\ \sigma(l^{\ell}_1), \sigma(l^{\ell}_2), \sigma(l^{\ell}_3) &
    \mtext{for } 1 \leq \ell \leq u; \\
  \label{eak-rule8} w \leftarrow &\ \aspnot{} w.
\end{align}
\end{small}}%
%%%
Here we let~$\sigma(\neg x_i) = v_i$ for each~$1 \leq i \leq n$;
we let~$\sigma(x_i) = x_i$ for each~$1 \leq i \leq n$; and
we let~$\sigma(y_i) = y_i$ for each~$1 \leq i \leq m$.
Intuitively,~$v_i$ corresponds to~$\neg x_i$.
The main difference with the reduction of Eiter and
Gottlob~\cite{EiterGottlob95}
is that we use the rules in~(\ref{eak-rule3})--(\ref{eak-rule6})
to let the variables~$y_i$
represent an assignment of weight~$k$ to the variables in~$Y$.
Note that~$P$ has~$k$ disjunctive rules,
namely the rules~(\ref{eak-rule3}).
\krversion{A full proof that~$(\varphi,k) \in \EAkWSat{}$ if and only if
$P$ has an answer set can be found in the technical report.
}\longversion{We show that~$(\varphi,k) \in \EAkWSat{}$ if and only if
$P$ has an answer set.

$(\Rightarrow)$
Assume there exists an assignment~$\alpha : X \rightarrow \SBs 0,1 \SEs$
such that for each assignment~$\beta : Y \rightarrow \SBs 0,1 \SEs$ of weight~$k$
it holds that~$\psi[\alpha \cup \beta]$ evaluates to true.
We show that~$M = \SB x_i \SM \alpha(x_i) = 1, 1 \leq i \leq n \SE \cup
\SB y^{j}_{i}, y_{i} \SM 1 \leq i \leq m, 1 \leq j \leq k \SE \cup \SBs w \SEs$
is an answer set of~$P$.
We have that~$P^{M}$ consists of Rules~(\ref{eak-rule3})--(\ref{eak-rule7}) and
the rules~$(x_i \leftarrow)$ for all~$1 \leq i \leq n$ such that~$\alpha(x_i) = 1$,
and the rules~$(v_i \leftarrow)$ for all~$1 \leq i \leq n$ such that~$\alpha(x_i) = 0$.
Clearly,~$M$ is a model of~$P^{M}$.
We show that~$M$ is a minimal model of~$P^{M}$.
Assume~$M' \subsetneq M$ is a minimal model of~$P^{M}$.
If~$M'$ does not coincide with~$M$ on the atoms~$x_i$ and~$v_i$,
then~$M'$ is not a model of~$P^{M}$.
If~$w \in M'$, then by Rules~(\ref{eak-rule4})~and~(\ref{eak-rule5}),~$M' = M$.
Therefore,~$w \not\in M'$.
By Rule~(\ref{eak-rule3}), we have that~$y^{1}_{i_1},y^{2}_{i_2},\dotsc,y^{k}_{i_k} \in M'$,
for some~$1 \leq i_1,\dotsc,i_k \leq m$.
By Rule~(\ref{eak-rule6}), we know that~$i_1,\dotsc,i_k$ are all different,
since otherwise it would have to hold that~$w \in M'$.
By Rule~(\ref{eak-rule4}), it holds that~$y_{i_1},\dotsc,y_{i_k} \in M'$.
By minimality of~$M'$, we know that~$\SB 1 \leq i \leq m \SM i \not\in \SBs i_1,\dotsc,i_k \SEs \SE \cap M' = \emptyset$.
Define the assignment~$\gamma : X \cup Y \rightarrow \SBs 0,1 \SEs$ by
letting~$\gamma(x_i) = 1$ if and only if~$x_i \in M'$
and~$\gamma(y_i) = 1$ if and only if~$y_i \in M'$.
Clearly,~$\gamma$ coincides with~$\alpha$ on~$X$,
and~$\gamma$ assigns exactly~$k$ variables~$y_j$ to true.
Now, since~$(\varphi,k) \in \EAkWSat{}$, we know that~$\gamma$ satisfies~$\psi$.
Therefore, there must be some~$1 \leq \ell \leq u$ such that~$\delta_{\ell}[\gamma]$
evaluates to true.
Now consider the instantiation of Rule~(\ref{eak-rule7}) for~$\ell$.
We know that~$M'$ satisfies the body of this rule.
Therefore~$w \in M'$, which is a contradiction.
From this we can conclude that no model~$M' \subsetneq M$ of~$P^{M}$ exists,
and thus~$M$ is an answer set of~$P$.

$(\Leftarrow)$
Assume~$P$ has an answer set~$M$.
We know that~$w \in M$, since otherwise~$(w \leftarrow)$ would be a rule
of~$P^{M}$, and then~$M$ would not be a model of~$P^{M}$.
Then, by Rules~(\ref{eak-rule4})~and~(\ref{eak-rule5}), also~$y_i,y^{j}_{i} \in M$
for all~$1 \leq i \leq m$ and~$1 \leq j \leq k$.
We show that for each~$1 \leq i \leq n$ it holds that~$\Card{M \cap \SBs x_i,v_i \SEs} = 1$.
Assume that for some~$1 \leq i \leq n$,~$M \cap \SBs x_i,v_i \SEs = \emptyset$.
Then~$(x_i \leftarrow)$ and~$(v_i \leftarrow)$ would be rules of~$P^{M}$,
and then~$M$ would not be a model of~$P^{M}$, which is a contradiction.
Assume instead that~$\SBs x_i,v_i \SEs \subseteq M$.
Then~$P^{M}$ would contain no rules with~$x_i$ and~$v_i$ in the head,
and hence~$M$ would not be a minimal model of~$P^{M}$, which is a contradiction.

We now construct an assigment~$\alpha : X \rightarrow \SBs 0,1 \SEs$
such that for all assignments~$\beta : Y \rightarrow \SBs 0,1 \SEs$ of
weight~$k$ it holds that~$\psi[\alpha \cup \beta]$ evaluates to true.
Define~$\alpha$ by letting~$\alpha(x_i) = 1$ if and only if~$x_i \in
M$.  Now let~$\beta$ be an arbitrary truth assignment to~$Y$ of weight
$k$.  We show that~$\psi[\alpha \cup \beta]$ evaluates to true.  We
proceed indirectly, and assume to the contrary that~$\alpha \cup
\beta$ does not satisfy any term~$\delta_{\ell}$ of~$\psi$.  We
construct a model~$M' \subsetneq M$ of~$P^{M}$.  We
let~$y_{i_1},\dotsc,y_{i_k}$ denote the~$k$ variables~$y_i$ such that
$\beta(y_i) = 1$.  We let~$M' = (M \cap \SB x_i,v_i \SM 1 \leq i \leq
n \SE) \cup \SB y^{\ell}_{i_{\ell}},y_{i_{\ell}} \SM 1 \leq \ell \leq
k \SE$.  For all rules of~$P^{M}$ other than Rule~(\ref{eak-rule7}),
it is clear that~$M'$ satisfies them.  We show that~$M'$ also
satisfies Rule~(\ref{eak-rule7}).  Assume that for some~$1 \leq \ell
\leq u$,~$M'$ satisfies the body of the instantiation of
Rule~(\ref{eak-rule7}) for~$\ell$, but not the head.  By construction
of~$M'$, this would imply that~$\alpha \cup \beta$ satisfies the
term~$\delta_{\ell}$ of~$\psi$, which contradicts our assumption.
Therefore, we can conclude that~$M'$ satisfies all instantiations of
Rule~(\ref{eak-rule7}).  It then holds that~$M' \subsetneq M$ is a
model of~$P^{M}$, which is a contradiction with the fact that~$M$ is
an answer set of~$P$.
From this we can conclude that~$\alpha \cup \beta$ satisfies
some term~$\delta_{\ell}$ of~$\psi$.
Since~$\beta$ was arbitrary, we know this holds for all truth
assignments~$\beta$ to~$Y$ of weight~$k$.
Therefore,~$(\varphi,k) \in \EAkWSat{}$.
}
\end{proof}
}
%%%

%%% ARXIV VERSION
\longversion{
\begin{theorem}
\label{thm:asp-disjrules}
\ASPcons\DisjRules{} is \EAkW{P}-complete.
\end{theorem}
\begin{proof}
In order to show hardness, we give an fpt-reduction from~$\EAkWSat$.
Let~$(C,k)$ be an instance of~$\EAkWSat$,
where~$C$ is a quantified Boolean circuit over existential variables~$X$
and universal variables~$Y$, where~$X =
\SBs x_1,\dotsc,x_n \SEs$, and where~$Y = \SBs y_1,\dotsc,y_m \SEs$.
By Proposition~\ref{prop:eakwp-monotone},
we may assume without loss of generality that~$C$ is in negation
normal form and that it is monotone in~$Y$,
i.e., that the variables~$y_1,\dotsc,y_m$ occur only positively in~$C$.
We construct a disjunctive program~$P$ as follows.
We consider the variables in~$X$ and~$Y$ as atoms.
In addition,
we introduce fresh atoms~$v_1,\dotsc,v_n$,~$w$, and~$y^j_i$
for all~$1 \leq j \leq k$,~$1 \leq i \leq m$.
Also, for each internal node~$g$ of~$C$, we introduce
a fresh atom~$z_g$.
We let~$P$ consist of the rules described
%in Table~\ref{table:construction2}.
as follows:\\

%\begin{table}[ht!]
%\setlength\dashlinedash{0.5pt}
%\setlength\dashlinegap{4pt}
\renewcommand{\arraystretch}{1.8}
%\caption{Construction of the disjunctive program~$P$ in the
%proof of Theorem~\ref{thm:asp-disjrules}.}
%\label{table:construction2}
\noindent \parbox{\textwidth}{
\centering
\begin{small}
\vspace{10pt}
%\vspace{5pt}\framebox{
\begin{tabular}{r @{\ } l r r}
$x_i$ & $\leftarrow \aspnot{} v_i$ & for~$1 \leq i \leq n$; & \tagarray\label{eak-rule1} \\
%\hdashline
$v_i$ & $\leftarrow \aspnot{} x_i$ & for~$1 \leq i \leq n$; & \tagarray\label{eak-rule2} \\
%\hdashline
$y^{j}_{1} \vee \dotsm \vee y^{j}_{m}$ & $\leftarrow $ & for~$1 \leq j \leq k$; & \tagarray\label{eak-rule3} \\
%\hdashline
$y_{i}$ & $\leftarrow y^{j}_{i}$ & for~$1 \leq i \leq m$; & \tagarray\label{eak-rule4} \\
%\hdashline
$y^{j}_{i}$ & $\leftarrow w$ & for~$1 \leq i \leq m$ & \\[-8pt]
&& and~$1 \leq j \leq k$; & \tagarray\label{eak-rule5} \\
%\hdashline
$w$ & $\leftarrow y^{j}_{i}, y^{j'}_{i}$ & for~$1 \leq i \leq m,$ & \\[-8pt]
&& and~$1 \leq j < j' \leq k$; & \tagarray\label{eak-rule6} \\
%\hdashline
$z_g$ & $\leftarrow w$ & for each internal node~$g$ of~$C$; & \tagarray\label{eak-rule9} \\
%\hdashline
$z_g$ & $\leftarrow \sigma(g_1), \dotsc, \sigma(g_u)$ & for each conjunction node~$g$ of~$C$ & \\[-8pt]
&& with inputs~$g_1,\dotsc,g_u$; & \tagarray\label{eak-rule10} \\
%\hdashline
$z_g$ & $\leftarrow \sigma(g_i)$ & for each disjunction node~$g$ of~$C$ & \\[-8pt]
&& with inputs~$g_1,\dotsc,g_u$ \\
& & and each~$1 \leq i \leq u$; & \tagarray\label{eak-rule11} \\
%\hdashline
$w$ & $\leftarrow z_o$ & where~$o$ is the output node of~$C$; & \tagarray\label{eak-rule7} \\
%\hdashline
$w$ & $\leftarrow \aspnot{} w$. & & \tagarray\label{eak-rule8} \\
\end{tabular}
%}
\vspace{10pt}
\end{small}%
}
%\end{table}
%%%

\noindent Here, we define the following mapping~$\sigma$ from nodes of~$C$
to variables in~$V$.
For each non-negated input node~$x_i \in X$, we let~$\sigma(x_i) = x_i$.
For each negated input node~$\neg x_i$, for~$x_i \in X$,
we let~$\sigma(\neg x_i) = v_i$.
For each input node~$y_j \in Y$, we let~$\sigma(y) = y_j$.
For each internal node~$g$, we let~$\sigma(g) = z_g$.
Intuitively,~$v_i$ corresponds to~$\neg x_i$.
One of the main differences with the reduction of Eiter
and Gottlob~\cite{EiterGottlob95}
is that we use the rules in~(\ref{eak-rule3})--(\ref{eak-rule6})
to let the variables~$y_i$
represent an assignment of weight~$k$ to the variables in~$Y$.
Another main difference is that we encode an arbitrary Boolean circuit,
rather than just a DNF formula,
in the rules~(\ref{eak-rule9})--(\ref{eak-rule7}).
Note that~$P$ has~$k$ disjunctive rules,
namely the rules~(\ref{eak-rule3}).
We show that~$(C,k) \in \EAkWSat{}$ if and only if
$P$ has an answer set.

$(\Rightarrow)$
Assume there exists an assignment~$\alpha : X \rightarrow \SBs 0,1 \SEs$
such that for each assignment~$\beta : Y \rightarrow \SBs 0,1 \SEs$ of weight~$k$
it holds that~$\alpha \cup \beta$ satisfies~$C$.
We show that~$M = \SB x_i \SM \alpha(x_i) = 1, 1 \leq i \leq n \SE \cup
\SB y^{j}_{i}, y_{i} \SM 1 \leq i \leq m, 1 \leq j \leq k \SE \cup
\SB z_g \SM g \mtext{ an internal node of } C \SE \cup \SBs w \SEs$
is an answer set of~$P$.
We have that~$P^{M}$ consists of Rules~(\ref{eak-rule3})--(\ref{eak-rule7}),
the rules~$(x_i \leftarrow)$ for all~$1 \leq i \leq n$ such that~$\alpha(x_i) = 1$,
and the rules~$(v_i \leftarrow)$ for all~$1 \leq i \leq n$ such that~$\alpha(x_i) = 0$.
Clearly,~$M$ is a model of~$P^{M}$.
We show that~$M$ is a minimal model of~$P^{M}$.
Assume~$M' \subsetneq M$ is a minimal model of~$P^{M}$.
If~$M'$ does not coincide with~$M$ on the atoms~$x_i$ and~$v_i$,
then~$M'$ is not a model of~$P^{M}$.
If~$w \in M'$, then by Rules~(\ref{eak-rule4}),~(\ref{eak-rule5})
and~(\ref{eak-rule9}),~$M' = M$.
Therefore,~$w \not\in M'$.
By Rule~(\ref{eak-rule3}), we have
that~$y^{1}_{i_1},y^{2}_{i_2},\dotsc,y^{k}_{i_k} \in M'$,
for some~$1 \leq i_1,\dotsc,i_k \leq m$.
By Rule~(\ref{eak-rule6}), we know that~$i_1,\dotsc,i_k$ are all different,
since otherwise it would have to hold that~$w \in M'$.
By Rule~(\ref{eak-rule4}), it holds that~$y_{i_1},\dotsc,y_{i_k} \in M'$.
By minimality of~$M'$, we know that~$\SB 1 \leq i \leq m \SM i \not\in \SBs i_1,\dotsc,i_k \SEs \SE \cap M' = \emptyset$.
Define the assignment~$\gamma : X \cup Y \rightarrow \SBs 0,1 \SEs$ by
letting~$\gamma(x_i) = 1$ if and only if~$x_i \in M'$
and~$\gamma(y_i) = 1$ if and only if~$y_i \in M'$.
Clearly,~$\gamma$ coincides with~$\alpha$ on~$X$,
and~$\gamma$ assigns exactly~$k$ variables~$y_j$ to true.
Now, since~$(C,k) \in \EAkWSat{}$, we know that~$\gamma$ satisfies~$C$.
Using the Rules~(\ref{eak-rule10}) and~(\ref{eak-rule11}),
we can show by an inductive argument that for each internal node~$g$
of~$C$ that is set to true by~$\gamma$ it must hold
that~$z_g \in M'$.
Since~$\gamma$ satisfies~$C$, it sets the output node~$o$
of~$C$ to true, and thus by Rule~(\ref{eak-rule7}),
we know that~$w \in M'$.
This is a contradiction.
From this we can conclude that no model~$M' \subsetneq M$ of~$P^{M}$ exists,
and thus~$M$ is an answer set of~$P$.

$(\Leftarrow)$
Assume~$P$ has an answer set~$M$.
We know that~$w \in M$, since otherwise~$(w \leftarrow)$ would be a rule
of~$P^{M}$, and then~$M$ would not be a model of~$P^{M}$.
Then, by Rules~(\ref{eak-rule4})~and~(\ref{eak-rule5}), also~$y_i,y^{j}_{i} \in M$
for all~$1 \leq i \leq m$ and~$1 \leq j \leq k$.
We show that for each~$1 \leq i \leq n$ it holds
that~$\Card{M \cap \SBs x_i,v_i \SEs} = 1$.
Assume that for some~$1 \leq i \leq n$,~$M \cap \SBs x_i,v_i \SEs = \emptyset$.
Then~$(x_i \leftarrow)$ and~$(v_i \leftarrow)$ would be rules of~$P^{M}$,
and then~$M$ would not be a model of~$P^{M}$, which is a contradiction.
Assume instead that~$\SBs x_i,v_i \SEs \subseteq M$.
Then~$P^{M}$ would contain no rules with~$x_i$ and~$v_i$ in the head,
and hence~$M$ would not be a minimal model of~$P^{M}$, which is a contradiction.

We now construct an assigment~$\alpha : X \rightarrow \SBs 0,1 \SEs$
such that for all assignments~$\beta : Y \rightarrow \SBs 0,1 \SEs$ of
weight~$k$ it holds that~$\psi[\alpha \cup \beta]$ evaluates to true.
Define~$\alpha$ by letting~$\alpha(x_i) = 1$ if and only if~$x_i \in M$.
Now let~$\beta$ be an arbitrary truth assignment to~$Y$ of weight~$k$.
We show that~$\alpha \cup \beta$ satisfies~$C$.
We proceed indirectly, and assume to the contrary that~$\alpha \cup
\beta$ does not satisfy~$C$.
We construct a model~$M' \subsetneq M$ of~$P^{M}$.
We let~$y_{i_1},\dotsc,y_{i_k}$ denote the~$k$ variables~$y_i$ such that
$\beta(y_i) = 1$.  We let~$M'$ consist
of~$(M \cap \SB x_i,v_i \SM 1 \leq i \leq n \SE)$,
of~$\SB y^{\ell}_{i_{\ell}},y_{i_{\ell}} \SM 1 \leq \ell \leq k \SE$,
and of~$\SB z_g \SM g \mtext{ an internal node of } C
\mtext{ set to true by } \alpha \cup \beta \SE$.
For all rules of~$P^{M}$ other than Rules~(\ref{eak-rule10})--(\ref{eak-rule7}),
it is clear that~$M'$ satisfies them.
It holds that~$M'$ also
satisfies Rules~(\ref{eak-rule10})--(\ref{eak-rule11}),
since~$M'$ contains~$z_g$ for all internal nodes~$g$ of~$C$
that are satisfied by~$\alpha \cup \beta$.
Then, since~$\alpha \cup \beta$ does not satisfy the output node~$o$
of~$C$, we know that~$z_o \not\in M'$,
and thus~$M'$ satisfies Rule~(\ref{eak-rule7}).
It then holds that~$M' \subsetneq M$ is a
model of~$P^{M}$, which is a contradiction with the fact that~$M$ is
an answer set of~$P$.
From this we can conclude that~$\alpha \cup \beta$ satisfies~$C$.
Since~$\beta$ was arbitrary, we know this holds for all truth
assignments~$\beta$ to~$Y$ of weight~$k$.
Therefore,~$(C,k) \in \EAkWSat{}$.

% MEMBERSHIP
To show membership in \EAkW{P}, we give
an fpt-reduction to \EAkWSat{}.
Let~$P$ be an instance of~\ASPcons\DisjRules{},
where~$P$ is a disjunctive logic program that
contains~$k$ disjunctive rules
and that contains atoms~$\SBs a_1,\dotsc,a_n \SEs$.
Let~$r_1,\dotsc,r_k$ be the disjunctive rules of~$P$,
where the head of rule~$r_i$ is~$a^1_i \vee \dotsm \vee a^{\ell_i}_i$,
for each~$1 \leq i \leq k$.
Moreover, let~$\ell = \sum_{1 \leq i \leq k} \ell_i + k$.
We sketch an algorithm that takes as input two bitstrings:
one string~$\overline{x} = x_1\dotsc x_n$, of length~$n$,
and one string~$\overline{z} = z_1\dotsc z_\ell$ of length~$\ell$
containing exactly~$k$ $1$'s.
Moreover, we consider the string~$\overline{z}$
as the concatenation of the strings~$\overline{z}_1,\dotsc,\overline{z}_k$,
where~$\overline{z}_i = z_{i,0}\dotsc z_{i,\ell_i}$.
Firstly, the algorithm checks~(0)~whether each~$z_i$ contains exactly one~$1$.
If this checks fails, the algorithm accepts the input.
Otherwise, the algorithm constructs the set~$M = \SB a_i \SM x_i = 1 \SE$,
and it checks~(1)~whether~$M$ is a model of the reduct~$P^M$.
Moreover, the algorithm constructs the subset~$R$ of disjunctive rules
that is defined as follows.
For each~$1 \leq i \leq k$, it holds that~$r_i \in R$ if and only if~$z_{i,0} = 0$.
In addition, it constructs the mapping~$\mu : R \rightarrow \Atoms{P}$,
by letting~$\mu(r_i) = a^{j}_{i}$ for the unique~$1 \leq j \leq \ell_i$
for which~$z_{i,j} = 1$.
The algorithm constructs the set~$M_{\mu}$, as defined in
Lemma~\ref{lem:asp-disjrules-mu}.
Then, the algorithm checks~(2)~whether~$M_{\mu}$ is not a strict subset of~$M$.
The algorithm accepts the input if and only if both
checks~(1) and~(2) pass.
We can choose this algorithm so that it runs in polynomial time.

Now, by Lemma~\ref{lem:asp-disjrules-mu},
it is straightforward to verify that~$P \in \ASPcons\DisjRules{}$ if and only
if there exists some strings~$\overline{x}$
such that for all strings~$\overline{z}$ containing exactly~$k$ $1$'s
the algorithm accepts.
Since the algorithm runs in polynomial time, we can construct
in polynomial time a quantified
Boolean circuit~$C$ over the set~$X$ of existential variables
and the set~$Z$ of universal variables,
with the property that
the algorithm accepts for some strings~$\overline{x}$
and for all suitable strings~$\overline{z}$
if and only if
there is a truth assignment~$\alpha : X \rightarrow \SBs 0,1 \SEs$
such that for all truth assignments~$\beta : Z \rightarrow \SBs 0,1 \SEs$ of
weight~$k$ the assignment~$\alpha \cup \beta$ satisfies~$C$.
In other words,~$(C,k) \in \EAkWSat{}$
if and only if~$P \in \ASPcons\DisjRules{}$.
\end{proof}
}
%%%

This hardness result holds even for the case where
each atom occurs only a constant number of times in the input program.
In order to show this, we consider the following polynomial-time transformation
on disjunctive logic programs.
Let $P$ be an arbitrary program, and let $x$ be an atom of $P$ that occurs
$\ell \geq 3$ many times.
We introduce new atoms $x_j$ for all $1 \leq j \leq \ell$.
We replace each occurrence of $x$ in $P$ by a unique atom $x_j$.
Then, to $P$, we add the rules $(x_{j+1} \leftarrow x_{j})$, for all $1 \leq j < \ell$,
and the rule $(x_{1} \leftarrow x_{\ell})$.
We call the resulting program $P'$.
It is straightforward to verify that $P$ has an answer set if and only if $P'$
has an answer set.

\begin{proposition}
Let~$\ell \geq 3$.
\ASPcons\DisjRules{} is \EAkW{P}-hard,
even when each atom occurs at most~$\ell$ times.
\end{proposition}
\begin{proof}
The transformation described above can be repeatedly applied to
ensure that each atom occurs at most 3 times.
Since this transformation does not introduce new disjunctive rules,
the result follows by Theorem~\ref{thm:asp-disjrules}.
\end{proof}

\noindent Repeated application of the described transformation
also gives us the following result.

\begin{corollary}
\label{cor:asp-varocc-completeness}
\ASPcons\AtomOcc{} is \para{\SigmaP{2}}-complete.
\end{corollary}
}

\longversion{
\noindent Finally, we show \EAkW{P}-completeness
for \ASPcons\NonDualNormalRules{}.

\begin{theorem}
\label{thm:asp-nondualnormalrules}
\ASPcons\NonDualNormalRules{} is \EAkW{P}-complete.
\end{theorem}

\begin{proof}
To show hardness,
we give an fpt-reduction from~$\EAkWSat$.
Let~$(C,k)$ be an instance of~$\EAkWSat$,
where~$C = \exists X. \forall Y. \CCC$,~$X =
\SBs x_1,\dotsc,x_n \SEs$, and~$Y = \SBs y_1,\dotsc,y_m \SEs$.
By Proposition~\ref{prop:eakwp-monotone},
we may assume without loss of generality that~$C$ is monotone
in~$Y$,
i.e., that the variables~$y_1,\dotsc,y_m$ occur only positively in~$\CCC$.
We construct a disjunctive program~$P$ as follows.
We consider the variables~$X$ and~$Y$ as atoms.
In addition,
we introduce fresh atoms~$v_1,\dotsc,v_n$,~$w$, and~$y^j_i$
for all~$1 \leq j \leq k$,~$1 \leq i \leq m$.
Also, for each internal node~$g$ of~$\CCC$, we introduce
a fresh atom~$z_g$.
We let~$P$ consist of the following rules:
\begin{small}
\begin{align}
  \label{eak2-rule1} x_i \leftarrow &\ \aspnot{} v_i & \mtext{for } 1 \leq i \leq n; \\
  \label{eak2-rule2} v_i \leftarrow &\ \aspnot{} x_i & \mtext{for } 1 \leq i \leq n; \\
  \label{eak2-rule3} w \leftarrow &\ y^{j}_{1}, \dotsc, y^{j}_{m} & \mtext{for } 1 \leq j \leq k; \\
  y^{j}_{i} \leftarrow &\ y_{i} & \mtext{for } 1 \leq i \leq m, \label{eak2-rule4} 1 \leq j \leq k; \\
  y^{j}_{i} \leftarrow &\ w & \mtext{for } 1 \leq i \leq m, \label{eak2-rule5} 1 \leq j \leq k; \\
  \label{eak2-rule9} z_g \leftarrow &\ w & \mtext{for each internal node~$g$ of~$\CCC$}; \\
  \label{eak2-rule13} z \leftarrow &\ w & \\
  \label{eak2-rule14} y_i \leftarrow &\ w & \mtext{for } 1 \leq i \leq m \\
  y^{j}_{i} \vee y^{j'}_{i} \leftarrow &\ z & \mtext{for } 1 \leq i \leq m, \label{eak2-rule6} 1 \leq j < j' \leq k; \\
  \label{eak2-rule10} \sigma(g_1) \vee \dotsm \vee \sigma(g_u) \leftarrow &\ z_g & \hspace{-30pt}\mtext{for each conjunction node~$g$ of~$\CCC$ with inputs~$g_1,\dotsc,g_u$}; \\
  \sigma(g_i) \leftarrow &\ z_g & \hspace{-30pt}\mtext{for each disjunction node~$g$ of~$\CCC$ with inputs~$g_1,\dotsc,g_u$,} \nonumber \\
  \label{eak2-rule11} && \mtext{and each~$1 \leq i \leq u$}; \\
  \label{eak2-rule7} z_o \leftarrow &\ z & \mtext{where~$o$ is the output node of~$\CCC$}; \\
  \label{eak2-rule12} w \vee z \leftarrow & & \\
  \label{eak2-rule8} w \leftarrow &\ \aspnot{} w.
\end{align}
\end{small}%
%%%
Here, we define the following mapping~$\sigma$ from nodes of~$\CCC$
to variables in~$V$.
For each non-negated input node~$x_i \in X$, we let~$\sigma(x_i) = v_i$.
For each negated input node~$\neg x_i$, for~$x_i \in X$,
we let~$\sigma(\neg x_i) = x_i$.
For each input node~$y_j \in Y$, we let~$\sigma(y) = y_j$.
For each internal node~$g$, we let~$\sigma(g) = z_g$.
Intuitively,~$v_i$ corresponds to~$\neg x_i$.
Note, however, that~$\sigma$ negates literals over
variables in~$X$.
Note that~$P$ has~$k$ rules that are not dual-Horn,
namely the rules~(\ref{eak2-rule3}).
We show that~$(\varphi,k) \in \EAkWSat{}$ if and only if
$P$ has an answer set.

$(\Rightarrow)$
Assume there exists an assignment~$\alpha : X \rightarrow \SBs 0,1 \SEs$
such that for each assignment~$\beta : Y \rightarrow \SBs 0,1 \SEs$ of weight~$k$
it holds that~$\alpha \cup \beta$ satisfies~$\CCC$.
We show that~$M = \SB x_i \SM \alpha(x_i) = 1, 1 \leq i \leq n \SE \cup
\SB y^{j}_{i}, y_{i} \SM 1 \leq i \leq m, 1 \leq j \leq k \SE \cup
\SB z_g \SM g \mtext{ an internal node of } \CCC \SE \cup \SBs w,z \SEs$
is an answer set of~$P$.
We have that~$P^{M}$ consists of Rules~(\ref{eak2-rule3})--(\ref{eak2-rule12}),
the rules~$(x_i \leftarrow)$ for all~$1 \leq i \leq n$ such that~$\alpha(x_i) = 1$,
and the rules~$(v_i \leftarrow)$ for all~$1 \leq i \leq n$ such that~$\alpha(x_i) = 0$.
Clearly,~$M$ is a model of~$P^{M}$.
We show that~$M$ is a minimal model of~$P^{M}$.
Assume~$M' \subsetneq M$ is a minimal model of~$P^{M}$.
If~$M'$ does not coincide with~$M$ on the atoms~$x_i$ and~$v_i$,
then~$M'$ is not a model of~$P^{M}$.
If~$w \in M'$, then by Rules~(\ref{eak2-rule5})--(\ref{eak2-rule14}),~$M' = M$.
Therefore,~$w \not\in M'$.
By Rule~(\ref{eak2-rule3}), we have
that~$y^{1}_{i_1},y^{2}_{i_2},\dotsc,y^{k}_{i_k} \not\in M'$,
for some~$1 \leq i_1,\dotsc,i_k \leq m$.
By Rule~(\ref{eak2-rule6}) and~(\ref{eak2-rule12}),
we know that~$i_1,\dotsc,i_k$ are all different,
since otherwise it would have to hold that~$w \in M'$.
By Rule~(\ref{eak2-rule4}), it holds that~$y_{i_1},\dotsc,y_{i_k} \not\in M'$.
Define the assignment~$\gamma : X \cup Y \rightarrow \SBs 0,1 \SEs$ by
letting~$\gamma(x_i) = 1$ if and only if~$x_i \in M'$
and~$\gamma(y_i) = 1$ if and only if~$y_i \in M'$.
Clearly,~$\gamma$ coincides with~$\alpha$ on~$X$,
and~$\gamma$ assigns exactly~$k$ variables~$y_j$ to true.
Now, since~$(C,k) \in \EAkWSat{}$, we know that~$\gamma$ satisfies~$\CCC$.
Using the Rules~(\ref{eak2-rule10}) and~(\ref{eak2-rule11}),
we can show by an inductive argument that for each internal node~$g$
of~$\CCC$ that is forced to true by~$\gamma$ it must hold
that~$z_g \not\in M'$.
Since~$\gamma$ satisfies~$\CCC$, it forces the output node~$o$
of~$\CCC$ to be true, and thus by Rule~(\ref{eak2-rule7}),
we know that~$z \not\in M'$.
Then, by Rule~(\ref{eak2-rule12}), we know that~$w \in M'$,
which is a contradiction.
From this we can conclude that no model~$M' \subsetneq M$ of~$P^{M}$ exists,
and thus~$M$ is an answer set of~$P$.

$(\Leftarrow)$
Assume~$P$ has an answer set~$M$.
We know that~$w \in M$, since otherwise~$(w \leftarrow)$ would be a rule
of~$P^{M}$, and then~$M$ would not be a model of~$P^{M}$.
Then, by Rules~(\ref{eak2-rule5})--(\ref{eak2-rule14}), also~$y_i,y^{j}_{i} \in M$
for all~$1 \leq i \leq m$ and~$1 \leq j \leq k$,~$z_g \in M$
for all internal nodes~$g$ of~$\CCC$, and~$z \in M$.
We show that for each~$1 \leq i \leq n$ it holds
that~$\Card{M \cap \SBs x_i,v_i \SEs} = 1$.
Assume that for some~$1 \leq i \leq n$,~$M \cap \SBs x_i,v_i \SEs = \emptyset$.
Then~$(x_i \leftarrow)$ and~$(v_i \leftarrow)$ would be rules of~$P^{M}$,
and then~$M$ would not be a model of~$P^{M}$, which is a contradiction.
Assume instead that~$\SBs x_i,v_i \SEs \subseteq M$.
Then~$P^{M}$ would contain no rules with~$x_i$ and~$v_i$ in the head,
and hence~$M$ would not be a minimal model of~$P^{M}$, which is a contradiction.

We now construct an assigment~$\alpha : X \rightarrow \SBs 0,1 \SEs$
such that for all assignments~$\beta : Y \rightarrow \SBs 0,1 \SEs$ of
weight~$k$ it holds that~$\psi[\alpha \cup \beta]$ evaluates to true.
Define~$\alpha$ by letting~$\alpha(x_i) = 1$ if and only if~$x_i \in M$.
Now let~$\beta$ be an arbitrary truth assignment to~$Y$ of weight~$k$.
We show that~$\alpha \cup \beta$ satisfies~$\CCC$.
We proceed indirectly, and assume to the contrary that~$\alpha \cup
\beta$ does not satisfy~$\CCC$.
We construct a model~$M' \subsetneq M$ of~$P^{M}$.
We let~$y_{i_1},\dotsc,y_{i_k}$ denote the~$k$ variables~$y_i$ such that
$\beta(y_i) = 1$.  We let~$M'$ consist
of~$(M \cap \SB x_i,v_i \SM 1 \leq i \leq n \SE)$,
of~$\SB y^j_i, y_i \SM 1 \leq i \leq m, 1 \leq j \leq k \SE \backslash
\SB y^{\ell}_{i_{\ell}},y_{i_{\ell}} \SM 1 \leq \ell \leq k \SE$,
of~$\SB z_g \SM g \mtext{ an internal node of } \CCC
\mtext{ not satisfied by } \alpha \cup \beta \SE$,
and of~$\SBs z \SEs$.
For all rules of~$P^{M}$ other than Rules~(\ref{eak2-rule10})--(\ref{eak2-rule7}),
it is clear that~$M'$ satisfies them.
It holds that~$M'$ also
satisfies Rules~(\ref{eak2-rule10})--(\ref{eak2-rule11}),
since~$M'$ does not contain~$z_g$ for all internal nodes~$g$ of~$\CCC$
that are satisfied by~$\alpha \cup \beta$.
Then, since~$\alpha \cup \beta$ does not satisfy the output node~$o$
of~$\CCC$, we know that~$z_o \in M'$,
and thus~$M'$ satisfies Rule~(\ref{eak2-rule7}).
It then holds that~$M' \subsetneq M$ is a
model of~$P^{M}$, which is a contradiction with the fact that~$M$ is
an answer set of~$P$.
From this we can conclude that~$\alpha \cup \beta$ satisfies~$\CCC$.
Since~$\beta$ was arbitrary, we know this holds for all truth
assignments~$\beta$ to~$Y$ of weight~$k$.
Therefore,~$(C,k) \in \EAkWSat{}$.

Membership in \EAkW{P} can be shown similarly to
\EAkW{P}-membership for \ASPcons\DisjRules{},
as in the proof of Theorem~\ref{thm:asp-disjrules}.
We omit a detailed membership proof here.
\end{proof}
}

\longversion{
%%% RELATING K-* AND *-K TO EXISTING COMPLEXITY CLASSES
\section{Relating \kstar{} and \stark{} to Existing Complexity Classes}
\label{sec:relating-eka-and-eak}

In this section, we relate the classes of the \kstar{} and \stark{} hierarchies
to existing parameterized complexity classes.
In particular, we give evidence that these classes differ from
the parameterized complexity classes \para{}\NP{}, \para{}\co{}\NP{},
\para{}\SigmaP{2} and \para{}\PiP{2}.
We will do this by providing separation results that are based on
complexity-theoretic assumptions related to
well-known classical complexity classes
(that is, the classes of the first two levels of the Polynomial Hierarchy).

%% SITUATING K-*
\subsection{Situating \kstar{}}
We begin with making some observations about the relation of \EkA{}
to existing parameterized complexity classes.
It is straightforward to see that~$\EkA{} \subseteq \para{\SigmaP{2}}$.
In polynomial time, any formula~$\exists X. \forall Y. \psi$ can be transformed into
a \SigmaP{2}-formula that is true if and only if for some assignment~$\alpha$ of weight~$k$
to the variables~$X$ the formula~$\forall Y. \psi[\alpha]$ is true.
Also, trivially,~$\para{\co{\NP}} \subseteq \EkA{}$.
To summarize, we obtain the following inclusions:
\[ \para{\co\NP} \subseteq \EkA \subseteq \para{}\SigmaP{2}
\quad \mtext{ and } \quad
\para{\NP} \subseteq \AkE \subseteq \para{}\PiP{2}. \]
This immediately leads to the following result.
\begin{proposition}
If~$\EkA \subseteq \para{\NP}$, then~$\NP = \co\NP$.
\end{proposition}

It is also not so difficult to see that~$\EkA{} \subseteq \XcoNP{}$.
A straightforward brute-force algorithm to solve \EkAWSat{}
that tries out all~$\binom{n}{k} = O(n^k)$ assignments of weight~$k$
to the existentially quantified variables
(and that uses nondeterminism to handle the assignment
to the universally quantified variables)
witnesses this.

A natural question to ask is whether~$\para{\NP} \subseteq \EkA$.
Since~$\para{}\NP{} \subseteq \XcoNP{}$ implies~$\NP = \co{}\NP{}$
\cite[Proposition~8]{FlumGrohe03}, this is unlikely.
We give a direct proof of this consequence.
\begin{proposition}
\label{prop:1}
If~$\para{\NP} \subseteq \EkA$,
then~$\NP = \co\NP$.
\end{proposition}
\begin{proof}
Assume that~$\para{\NP} \subseteq \EkA$.
Note that the language \thSAT{} of satisfiable propositional formulas
is \NP{}-complete, and the language \thUNSAT{} of unsatisfiable propositional
formulas is \co\NP{}-complete.
The parameterized problem~$P = \SB (\varphi,1) \SM \varphi \in \thSAT \SE$
is in \para{\NP}.
Then also~$P \in \EkA$.
This means that there is an fpt-reduction~$R$ from~$P$ to \EkAWSat{}.
We construct a polynomial-time reduction~$S$ from \thSAT{} to \thUNSAT{}.
Let~$\varphi$ be an instance of \thSAT{}.
The reduction~$R$ maps~$(\varphi,1)$ to an instance~$(\varphi',f(1))$ of \EkAWSat{},
where~$f$ is some computable function and~$\varphi' = \exists X. \forall Y. \psi$,
such that~$(\varphi',f(1)) \in \EkAWSat{}$ if and only if~$\varphi \in \thSAT$.
Note that~$f(1)$ is a constant, since~$f$ is a fixed function.
To emphasize this, we let~$c = f(1)$, and we will use~$c$ to denote~$f(1)$.
By definition, we know that~$(\varphi',c) \in \EkAWSat{}$ if and only if
for some truth assignment~$\alpha$ to the variables~$X$ of weight~$c$,
the formula~$\forall Y. \psi[\alpha]$ is true.
Let~$\mtext{ta}(X,c)$ denote the set of all truth assignments to~$X$
of weight~$c$.
We then get that~$\varphi \in \thSAT$ if and only if
the formula~$\forall Y. \chi$ is true,
where~$\chi$ is defined as follows:
\[ \chi = \underset{\alpha \in \mtext{ta}(X,c)}{\bigvee} \psi[\alpha] \]
It is straightforward to verify that the mapping
$\varphi \mapsto \neg \chi$ is a polynomial-time reduction
from \thSAT{} to \thUNSAT{}.
\end{proof}

This implies that~$\EkA$ is very likely to be a strict subset
of~$\para{\SigmaP{2}}$.
\begin{corollary}
If~$\EkA = \para{\SigmaP{2}}$, then~$\NP = \co\NP$.
\end{corollary}

The following result shows
another way in which the class \EkA{} relates
to the existing complexity class \co\NP{}.
Let~$P$ be a parameterized decision problem, and let~$c \geq 1$ be an integer.
Recall that the \emph{$c$-th slice of~$P$}, denoted~$P_{c}$, is the (unparameterized)
decision problem~$\SB x \SM (x,c) \in P \SE$.

\begin{proposition}
\label{prop:eka-conp}
Let~$P$ be a parameterized problem complete for \EkA{},
and let~$c \geq 1$ be an integer.
Then~$P_{c}$ is in \co\NP{}.
Moreover,
there exists some integer~$d \geq 1$
such that~$P_{1} \cup \dotsm \cup P_{d}$ is \co\NP-complete.
\end{proposition}
\begin{proof}
We show \co\NP-membership of~$P_{c}$, by
constructing a polynomial-time reduction~$S$ from~$P_c$ to \thUNSAT{}.
Since~$P \in \EkA{}$, we know that there exists an fpt-reduction~$R$ from~$P$
to~$\EkAWSat(\Phi)$.
Therefore, there exist computable functions~$f$ and~$g$ and a polynomial~$p$
such that for all instances~$(x,k)$ of~$P$,~$R(x,k) = (x',k')$
is computable in time~$f(k) \cdot p(\Card{x})$
and~$k' \leq g(k)$.
We describe the reduction~$S$.
Let~$x$ be an arbitrary instance of~$P_c$.
We know~$R$ maps~$(x,c)$ to~$(\varphi,k')$,
for some~$k' \leq g(c)$,
where~$\varphi = \exists X. \forall Y. \psi$.
Note that~$k'$ is bounded by a constant~$g(c) = d$.
Let~$\mtext{ta}(X,k')$ denote the set of all truth assignment to~$X$
of weight~$k'$.
Then, for each~$\alpha \in \mtext{ta}(X,k')$,
we let~$Y^{\alpha}$ be the set containing of a copy~$y^{\alpha}$ of
each variable~$y \in Y$,
and we let~$Y' = \bigcup_{\alpha \in \mtext{ta}(X,k')} Y^{\alpha}$.
We then get that~$\varphi$ is equivalent to the
formula~$\forall Y'. \chi$,
where~$\chi = \bigvee_{\alpha \in \mtext{ta}(X,k')} (\psi[\alpha])^{\alpha}$.
Here,~$(\psi[\alpha])^{\alpha}$ denotes the formula~$\psi[\alpha]$
where each variable~$y \in Y$ is replaced by its copy~$y^{\alpha}$.
Also, the size of~$\chi$ is polynomial in the size of~$\psi$.
We then let~$S(x) = \neg\chi$.
It is straightforward to verify that~$S$ is a correct polynomial-time
reduction from~$P_{c}$ to \thUNSAT{}.

We show that there exists a function~$f$
such that for any positive integer~$s \geq 1$,
there is a polynomial-time reduction from \thUNSAT{}
to~$P_{1} \cup \dotsm \cup P_{f(s)}$.
Then, in particular,~$P_{1} \cup \dotsm \cup P_{f(1)}$
is \co\NP{}-complete.
Let~$s \geq 1$ be an arbitrary integer. We construct the reduction~$S$.
There is a trivial polynomial-time reduction~$S$ from \thUNSAT{}
to~$(\EkAWSat)_{s}$, that maps a Boolean formula~$\varphi$
over a set of variables~$Y$ to the instance~$(\exists \SBs x_1,\dotsc,x_s \SEs.
\forall Y. \neg\chi,s)$ of \EkAWSat{}.
Since~$P$ is \EkA{}-complete,
there exists an fpt-reduction~$R$ from~$\EkAWSat$
to~$P$.
From this, we know that there exists a nondecreasing and unbounded
function~$f$ such that
for each instance~$(x,k)$ of~$\EkAWSat$ it holds that~$k' \leq f(k)$,
where~$R(x,k) = (x',k')$.
This reduction~$R$ is a polynomial-time reduction from~$(\EkAWSat)_{s}$
to~$P_{1} \cup \dotsm \cup P_{f(s)}$.
Composing the polynomial-time reductions~$S$ and~$R$,
we obtain a reduction from \thUNSAT{} to~$P_{1} \cup \dotsm \cup P_{f(s)}$.
\end{proof}

%% SITUATION *-K HIERARCHY
\subsection{Situating the \stark{} Hierarchy}

Next, we continue with relating the classes of the \stark{} hierarchy
to existing parameterized complexity classes.
Similarly to the case of \kstar{},
we obtain the following inclusions:
\[ \para{\co\NP} \subseteq \AEkW{1} \subseteq \dotsm \subseteq
\AEkW{P} \subseteq \para{}\PiP{2} \]
and
\[ \para{\NP} \subseteq \EAkW{1} \subseteq \dotsm \subseteq
\EAkW{P} \subseteq \para{}\SigmaP{2}. \]
This immediately leads to the following result.
\begin{proposition}
If~$\EAkW{1} \subseteq \para{}\co{}\NP{}$, then~$\NP = \co\NP$.
\end{proposition}

It is also not so difficult to see that~$\EAkW{P} \subseteq \XNP{}$.
A straightforward brute-force algorithm to solve \EAkWSat{}
that tries out all~$\binom{n}{k} = O(n^k)$ assignments of weight~$k$
to the universally quantified variables
(and that uses nondeterminism to handle the assignment
to the existentially quantified variables)
witnesses this.

A natural question to ask is whether~\para{\co\NP} is contained
in any of the classes \EAkW{t}.
Since $\para{}\co{}\NP{} \subseteq \XNP{}$ implies~$\NP = \co{}\NP{}$
\cite[Proposition~8]{FlumGrohe03}, this is unlikely.
We show how to prove this consequence directly.

\begin{proposition}
If $\para{\co\NP} \subseteq \EAkW{P}$,
then $\NP = \co\NP$.
\end{proposition}
\begin{proof}[Proof (sketch)]
With an argument similar to the one in the proof of Proposition~\ref{prop:1},
a polynomial-time reduction from $\thUNSAT{}$ to $\thSAT{}$
can be constructed.
An additional technical observation needed for this case is that
\thSAT{} is in \NP{} also when the input is a Boolean circuit.
\end{proof}

Then in particular for \EAkW{1},
we get the following result.

\begin{corollary}
If $\para{\co\NP} \subseteq \EAkW{1}$,
then $\NP = \co\NP$.
\end{corollary}

Moreover, this gives us the following separation.
\begin{corollary}
If~$\EAkW{P} = \para{}\SigmaP{2}$, then~$\NP = \co{}\NP{}$.
\end{corollary}

Similarly to the class \EkA{},
the classes \AEkW{t} relate to the existing complexity
class \co\NP{} in the following way.

\begin{proposition}
\label{prop:aek-conp}
Let~$P$ be a parameterized problem that is contained in \AEkW{P}
and that is hard for \AEkW{1},
and let~$c \geq 1$ be an integer.
Then~$P_{c}$ is in \co\NP{}.
Moreover,
there exists some integer~$d \geq 1$
such that~$P_{1} \cup \dotsm \cup P_{d}$ is \co\NP{}-complete.
\end{proposition}
\begin{proof}
The proof of this proposition is similar to the proof
of Proposition~\ref{prop:eka-conp}.
We show \co\NP-membership of~$P_{c}$, by
constructing a polynomial-time reduction~$S$ from~$P_c$ to \thUNSAT{}.
Since~$P \in \AEkW{P}$, we know that there exists an fpt-reduction~$R$
from~$P$ to~$\AEkWSat(\Gamma)$.
Therefore, there exist computable functions~$f$ and~$g$ and a polynomial~$p$
such that for all instances~$(x,k)$ of~$P$,~$R(x,k) = (x',k')$
is computable in time~$f(k) \cdot p(\Card{x})$
and~$k' \leq g(k)$.
We describe the reduction~$S$.
Let~$x$ be an arbitrary instance of~$P_c$.
We know~$R$ maps~$(x,c)$ to~$(\varphi,k')$,
for some~$k' \leq g(c)$,
where~$\varphi = \exists X. \forall Y. C$.
Note that~$k'$ is bounded by a constant~$g(c) = d$.
Let~$\mtext{ta}(X,k')$ denote the set of all truth assignment to~$X$
of weight~$k'$.
We then get that~$\varphi$ is equivalent to the
quantified circuit~$\forall Y. C'$,
where~$C' \equiv \bigvee_{\alpha \in \mtext{ta}(X,k')} C[\alpha]$.
Also, the size of~$C'$ is polynomial in the size of~$C$.
It is straightforward to construct a propositional formula~$\psi$
that is valid if and only if~$C'$ is valid.
We then let~$S(x) = \neg\psi$.
Then,~$S$ is a correct polynomial-time
reduction from~$P_{c}$ to \thUNSAT{}.

We show that there exists a function~$f$
such that for any positive integer~$s \geq 1$,
there is a polynomial-time reduction from \thUNSAT{}
to~$P_{1} \cup \dotsm \cup P_{f(s)}$.
Then, in particular,~$P_{1} \cup \dotsm \cup P_{f(1)}$
is \co\NP{}-complete.
Let~$s \geq 1$ be an arbitrary integer. We construct the reduction~$S$.
There is a trivial polynomial-time reduction~$S$ from \thUNSAT{}
to~$(\AEkWSat(3\DNF))_{s}$, that maps a Boolean formula~$\chi$
in~$3\CNF$ over a set of variables~$Y$ to the
instance~$(\forall Y. \exists \SBs x_1,\dotsc,x_s \SEs. \neg\chi,s)$
of $\AEkWSat(3\DNF)$.
Since~$P$ is \AEkW{1}-hard,
there exists an fpt-reduction~$R$ from~$\AEkWSat(3\DNF)$
to~$P$.
From this, we know that there exists a nondecreasing and unbounded
function~$f$ such that
for each instance~$(x,k)$ of~$\AEkWSat(3\DNF)$ it holds that~$k' \leq f(k)$,
where~$R(x,k) = (x',k')$.
This reduction~$R$ is a polynomial-time reduction
from~$(\AEkWSat(3\DNF))_{s}$
to~$P_{1} \cup \dotsm \cup P_{f(s)}$.
Composing the polynomial-time reductions~$S$ and~$R$,
we obtain a reduction from \thUNSAT{} to~$P_{1} \cup \dotsm \cup P_{f(s)}$.
\end{proof}

%% NONTRIVIAL OBSTACLES TO SHOWING EQUALITY
\subsection{Nontrivial Obstacles to Showing Equality}
Above, we have seen that \EkA{} is different from \para{}\NP{},
under some common complexity-theoretic assumptions.
However, it might in principle be the case that \EkA{}
coincides with \para{}\co{}\NP{}.
In this section, we obtain a result that shows that this has
highly nontrivial consequences.
Moreover, this result can be used to show 
that a similarly nontrivial obstacle would need to be overcome
to show that \EAkW{t} coincides with \para{}\NP{}, for~$t \geq 2$.
Concretely, we will show that showing any of the above equalities
would involve establishing that there exists a subexponential-time
reduction from \QSat{i} to either \SAT{} or \UNSAT{}.
In order to show the existence of such a reduction
one would need to a novel, profound understanding
of the problems \QSat{i}, \SAT{} and \UNSAT{}.

The nonexistence of such a subexponential-time reduction
is not commonly used as a complexity-theoretic assumption.
However, using this conjecture that no such subexponential-time
reductions exist as an assumption, one could consider the
results in this section as additional evidence that
the class \EkA{} is different from \para{}\co{}\NP{}
and that the class \EAkW{2} is different from \para{}\NP{}.

To obtain the result, we will use the following auxiliary
problem $\EkAkWSat(\CCC)$, which is a variant of \QSat{2}
where the assignment to both quantifier blocks is weighted.
Here~$\CCC$ is an arbitrary class of Boolean circuits.
\probdef{
  $\EkAkWSat(\CCC)$
  
  \emph{Instance:} A Boolean circuit~$C \in \CCC$
  over two disjoint sets~$X$ and~$Y$ of variables,
  and an integer~$k$.

  \emph{Parameter:} $k$.

  \emph{Question:} Does there exist a truth assignment~$\alpha$
  to~$X$ of weight~$k$
  such that for all truth assignments~$\beta$ to~$Y$
  with weight~$k$
  the assignment~$\alpha \cup \beta$ satisfies~$C$?
}

Moreover, we consider the following effective notion of ``little oh,''
as used by Flum and Grohe~\cite[Definition~3.22]{FlumGrohe06}.
Let~$f,g$ be computable functions with the positive integers
as domain and range. We say that~$f$ is~$o(g)$ if there
is a computable function~$h$ such that for all~$\ell \geq 1$
and~$n \geq h(\ell)$, we have that~$f(n) \leq g(n)/\ell$.
Equivalently,~$f$ is~$o(g)$ if and only if there exists a
positive integer~$n_0$ and a nondecreasing and unbounded
computable function~$\iota$ such that for all~$n \geq n_0$
we have that~$f(n) \leq g(n)/\iota(n)$
\cite[Lemma~3.23]{FlumGrohe06}.

Now, we prove the following technical result,
whose proof is inspired by results relating the parameterized
complexity classes of the Weft hierarchy
to intractability notions from the field of subexponential-time complexity
\cite{ChenChorFellowsHuangJuedesKanjXia05,ChenHuangKanjXia06,ChenKanj12}.
\begin{theorem}
\label{thm:separation-tech-result}
If there exists an~$f(k) n^{o(k)} m^{O(1)}$ time reduction
from $\EkAkWSat(\Gamma_{2,4})$ to
\UNSAT{},
where~$k$ denotes the parameter value,~$n$
denotes the number of variables
and~$m$ denotes the instance size,
then there exists a subexponential-time reduction
from \QSat{2} to \UNSAT{},
i.e., a reduction that runs in time~$2^{o(n)}m^{O(1)}$,
where~$n$ denotes the number of variables
and~$m$ denotes the instance size.
\end{theorem}
\begin{proof}
Assume that there exists a reduction~$R$
from $\EkAkWSat(\Gamma_{2,4})$ to
\UNSAT{}
that, for large enough values of~$n$,
runs in time~$f(k) n^{k/\lambda(k)} m^{O(1)}$,
for some computable function~$f$
and some nondecreasing and unbounded computable function~$\lambda$.

We now construct a reduction from
\QSat{2} to \UNSAT{}
that, for sufficiently large values of~$n$,
runs in time~$2^{o(n)} m^{O(1)}$,
where~$n$ is the number of variables,
and~$m$ is the instance size.
Let~$(\varphi, k)$ be an instance
of \QSat{2},
where~$\varphi = \exists X_1. \forall X_2. \psi$
and where~$\psi$ is in~$3\DNF$.
We may assume without loss of generality
that~$\Card{X_1} = \Card{X_2} = n$;
we can add dummy variables to~$X_1$ and~$X_2$
to ensure this.
We denote the size of~$\psi$ by~$m$.
Let~$X = X_1 \cup X_2$.
We may assume without loss of generality that~$f(k) \geq 2^k$
and that~$f$ is nondecreasing and unbounded.
Define~$f^{-1}$ by letting~$f^{-1}(h) = \max \SB q \SM f(q) \leq h \SE$.
Since the function~$f$ is nondecreasing and unbounded,
the function~$f^{-1}$ is also nondecreasing and unbounded,
and satisfies~$f(f^{-1}(h)) \leq h$.
Since~$f(k) \geq 2^k$, we know that~$f^{-1}(h) \leq \log h$.
We then choose the integers~$r,k'$ as follows.
We let~$r = \lfloor n / f^{-1}(n) \rfloor$,
and let~$k' = \lceil n/r \rceil$.

We firstly construct an instance~$(\varphi',k')$ of $\EkAkWSat(\Gamma_{2,4})$
that is a yes-instance if and only if~$\varphi$ is a yes-instance
of \QSat{2}.
We will describe~$\varphi'$ as a quantified Boolean formula
whose matrix corresponds to a circuit of depth~$4$ and weft~$2$.
In order to do so, for each~$1 \leq i \leq 2$,
we split~$X_i$ into~$k'$ disjoint sets~$X_{i,1},\dotsc,X_{i,k'}$.
We do this in such a way that each set~$X_{i,j}$
has at most~$n/k'$ elements, i.e.,~$\Card{X_{i,j}} \leq r$
for all~$1 \leq i \leq 2$ and all~$1 \leq j \leq k'$.
Now, for each truth assignment~$\alpha : X_{i,j} \rightarrow \SBs 0,1 \SEs$
we introduce a new variable~$y^{\alpha}_{i,j}$.
Formally, we define a set of variables~$Y_{i,j}$
for each~$1 \leq i \leq 2$ and each~$1 \leq j \leq k'$:
\[ Y_{i,j} = \SB y^{\alpha}_{i,j} \SM (\alpha : X_{i,j} \rightarrow \SBs 0,1 \SEs) \SE. \]
We have that~$\Card{Y_{i,j}} \leq 2^r$, for each~$1 \leq i \leq 2$
and each~$1 \leq j \leq k'$.
We let~$Y_i = \bigcup_{1 \leq j \leq k'} Y_{i,j}$,
and we let~$Y = Y_1 \cup Y_2$.

We continue the construction of the formula~$\varphi'$.
For each~$1 \leq i \leq 2$, we define the formula~$\psi_{Y_i}$
as follows:
\[ \psi_{Y_i} =
  \underset{1 \leq j \leq k'}{\bigwedge}\ 
  \underset{\alpha \neq \alpha'}{
  \underset{\alpha,\alpha' : X_{i,j} \rightarrow \SBs 0,1 \SEs}{\bigwedge}}
  \left ( \neg y^{\alpha}_{i,j} \vee \neg y^{\alpha'}_{i,j} \right ).
\]
Then we define the auxiliary function~$\sigma : X \rightarrow 2^Y$,
that maps variables in~$X$ to sets of variables in~$Y$.
For each~$x \in X_{i,j}$, we let:
\[
  \sigma(x) =
    \SB y^{\alpha}_{i,j} \SM (\alpha : X_{i,j} \rightarrow \SBs 0,1 \SEs), \alpha(x) = 1 \SE.
\]
Intuitively,~$\sigma(x)$ corresponds to those variables~$y^{\alpha}_{i,j}$
where~$\alpha$ is an assignment that satisfies~$x$.

Now, we construct a formula~$\psi''$,
by transforming the formula~$\psi$ in the following way.
We replace each occurrence of a variable~$x \in X_{i,j}$
in~$\psi$
by the formula~$\chi_{x}$, that is defined as follows:
\[
  \chi_{x} =
    \underset{y^{\alpha}_{i,j} \in \sigma(x)}{\bigvee}
      y^{\alpha}_{i,j}.
\]

We can now define the quantified Boolean formula~$\varphi'$.
We let~$\varphi' = \exists Y_1. \forall Y_2. \psi'$,
where~$\psi'$ is defined as follows:
\[ \psi' = \psi_{Y_1} \wedge (\psi_{Y_2} \rightarrow \psi''). \]
The formula~$\psi'$ can be seen as a circuit of depth~$4$ and weft~$2$.
In the remainder, we will refer to~$\psi'$ as a circuit.

We verify that~$\varphi \in \QSat{2}$
if and only if~$(\varphi',k') \in \EkAkWSat(\Gamma_{2,4})$.

$(\Rightarrow)$
Assume that~$\varphi \in \QSat{2}$, i.e., that there exists
a truth assignment~$\beta_1 : X_1 \rightarrow \SBs 0,1 \SEs$
such that for all truth assignments~$\beta_2 : X_2 \rightarrow \SBs 0,1 \SEs$
it holds that~$\psi[\beta_1 \cup \beta_2]$ is true.
We show that~$(\varphi',k') \in \EkAkWSat{}$.
We define the truth assignment~$\gamma_1 : Y_1 \rightarrow \SBs 0,1 \SEs$
by letting~$\gamma_1(y^{\alpha}_{1,j}) = 1$ if and only
if~$\beta_1$ coincides with~$\alpha$ on the variables~$X_{1,j}$,
for each~$1 \leq j \leq k'$ and each~$\alpha : X_{1,j} \rightarrow \SBs 0,1 \SEs$.
Clearly,~$\gamma_1$ has weight~$k$.
Moreover,~$\gamma_1$ satisfies~$\psi_{Y_1}$.
We show that for each truth assignment~$\gamma_2 : Y_2
\rightarrow \SBs 0,1 \SEs$ of weight~$k$ it holds
that~$\psi'[\gamma_1 \cup \gamma_2]$ is true.
Let~$\gamma_2$ be an arbitrary truth assignment of weight~$k$.
We distinguish two cases: either~(i)~$\gamma_2$ does not
satisfy~$\psi_{Y_2}$, or~(ii)~$\gamma_2$ does satisfy~$\psi_{Y_2}$.
In case~(i), clearly,~$\psi'[\gamma_1 \cup \gamma_2]$ is true.
In case~(ii), we know that for each~$1 \leq j \leq k'$,
there is exactly one~$\alpha_j : X_{2,j} \rightarrow \SBs 0,1 \SEs$
such that~$\gamma_2(y^{\alpha}_{2,j}) = 1$.
Now let~$\beta_2 : X_2 \rightarrow \SBs 0,1 \SEs$ be the assignment
that coincides with~$\alpha_j$ on the variables~$Y_{2,j}$,
for each~$1 \leq j \leq k'$.
We know that~$\psi[\beta_1 \cup \beta_2]$ is true.
Then, by definition of~$\psi''$, it follows that~$\psi''[\gamma_1 \cup \gamma_2]$
is true as well.
Since~$\gamma_2$ was arbitrary, we can conclude
that~$(\varphi',k') \in \EkAkWSat{}$.

$(\Leftarrow)$
Conversely, assume that~$(\varphi',k') \in \EkAkWSat{}$,
i.e., that there exists a truth assignment~$\gamma_1 : Y_1 \rightarrow
\SBs 0,1 \SEs$ of weight~$k'$ such that for all truth
assignments~$\gamma_2 : Y_2 \rightarrow \SBs 0,1 \SEs$ of weight~$k'$
it holds that~$\psi'[\gamma_1 \cup \gamma_2]$ is true.
We show that~$\varphi \in \QSat{2}$.
Since~$\psi_{Y_1}$ contains only variables in~$Y_1$,
we know that~$\gamma_1$ satisfies~$\psi_{Y_1}$,
i.e., that for each~$1 \leq j \leq k'$ there is a unique~$\alpha_j : X_{1,j}
\rightarrow \SBs 0,1 \SEs$ such that~$\gamma_1(y^{\alpha_j}_{1,j}) = 1$.
We define the truth assignment~$\beta_1 : X_1 \rightarrow \SBs 0,1 \SEs$
to be the unique truth assignment that coincides with~$\alpha_j$
for each~$1 \leq j \leq k'$.
We show that for all truth assignments~$\beta_2 : X_2 \rightarrow \SBs 0,1 \SEs$
it holds that~$\psi[\beta_1 \cup \beta_2]$ is true.
Let~$\beta_2$ be an arbitrary truth assignment.
We construct the truth assignment~$\gamma_2 : Y_2 \rightarrow \SBs 0,1 \SEs$
by letting~$\gamma_2(y^{\alpha}_{2,j}) = 1$ if and only
if~$\beta_2$ coincides with~$\alpha$ on the variables in~$Y_{2,j}$,
for each~$1 \leq j \leq k'$ and each~$\alpha : X_{2,j} \rightarrow \SBs 0,1 \SEs$.
Clearly,~$\gamma_2$ has weight~$k$.
Moreover,~$\gamma_2$ satisfies~$\psi_{Y_2}$.
Therefore, since we know that~$\psi'[\gamma_1 \cup \gamma_2]$ is true,
we know that~$\psi''[\gamma_1 \cup \gamma_2]$ is true.
Then, by definition of~$\psi''$,
it follows that~$\psi[\beta_1 \cup \beta_2]$
is true as well.
Since~$\beta_2$ was arbitrary, we can conclude
that~$\varphi \in \QSat{2}$.

We observe some properties of the quantified Boolean
formula~$\varphi' = \exists Y_1. \forall Y_2. \psi'$.
Each~$Y_i$, for~$1 \leq i \leq 2$, contains at most~$n' = k' 2^r$ variables.
Furthermore, the circuit~$\psi'$ has
size~$m' \leq O(k' 2^{2r} + 2^{r}m) \leq O(k'2^{2r}m)$.
Finally, it is straightforward to verify that the circuit~$\psi'$ can be constructed
in time~$O((m')^2)$.

Since~$(\varphi',k')$ is an instance of $\EkAkWSat(\Gamma_{2,4})$,
we can apply the reduction~$R$
to obtain an equivalent instance~$(\varphi'',k'')$ of \UNSAT{}.
This reduction runs in time~$f(k') (n')^{k'/\lambda(k')} (m')^{O(1)}$.
By first constructing~$(\varphi',k')$ from~$\varphi$,
and then constructing~$\varphi''$ from~$(\varphi',k')$,
we get a reduction~$R'$ from \QSat{2}
to \UNSAT{},
that runs in time~$f(k') (n')^{k'/\lambda(k')} (m')^{O(1)} + O((m')^2)$.
We analyze the running time of this reduction~$R'$
in terms of the values~$n$ and~$m$.

Since~$k' = \lceil n/r \rceil \leq f^{-1}(n) \leq \log n$,
we have~$f(k') \leq f(f^{-1}(n)) \leq n$.
Moreover,
\[ k' = \lceil n/r \rceil \geq n/r \geq n/(n/f^{-1}(n)) = f^{-1}(n). \]
Therefore if we set~$\lambda'(n) = \lambda(f^{-1}(n))$,
then~$\lambda(k') \geq \lambda'(n)$.
Since both~$\lambda$ and~$f^{-1}$ are nondecreasing and unbounded,
$\lambda'(n)$ is a nondecreasing and unbounded function of~$n$.
%Note that~$k' = f^{-1}(n) \leq \log n$.
We have,
\[ \begin{array}{r l}
  (n')^{k' / \lambda(k')} = &
    (k' 2^r)^{k' / \lambda(k')} \leq
    (k')^{k'} 2^{k' r / \lambda(k')} \leq
    (k')^{k'} 2^{k' n / (\lambda(k') f^{-1}(n))} \leq
    (k')^{k'} 2^{n / \lambda(k')} \\
  \leq &
    (k')^{k'} 2^{n / \lambda'(n)} =
    2^{o(n)}.
\end{array} \]
Finally, consider the factor~$m'$.
Since~$f^{-1}$ is nondecreasing and unbounded,
\[
  m' \leq
  O(k'2^{2r}m) =
  O(2\log n 2^{2n / f^{-1}(n)}m) =
  2^{o(n)} m.
\]
Therefore, both terms~$(m')^{O(1)}$ and~$O((m')^2)$
in the running time of~$R'$
are bounded by~$2^{o(n)}m^{O(1)}$.
Combining all these, we conclude that the running
time~$f(k') (n')^{k'/\lambda(k')} (m')^{O(1)} + O((m')^2)$
of~$R'$ is bounded by~$2^{o(n)}m^{O(1)}$.
Therefore,~$R'$ is a subexponential-time reduction
from \QSat{2} to \UNSAT{}.
This completes our proof.
\end{proof}

In addition, we get the following technical result by a
completely analogous argument.
\begin{corollary}
\label{cor:separation-tech-result}
If there exists an~$f(k) n^{o(k)} m^{O(1)}$ time reduction
from $\EkAkWSat(\Gamma_{2,4})$ to \SAT{},
where~$k$ denotes the parameter value,~$n$
denotes the number of variables
and~$m$ denotes the instance size,
then there exists a subexponential-time reduction
from \QSat{2} to \SAT{},
i.e., a reduction that runs in time~$2^{o(n)}m^{O(1)}$,
where~$n$ denotes the number of variables
and~$m$ denotes the instance size.
\end{corollary}

In order to use this technical result to show that
it is nontrivial
to prove that~$\EkA{} = \para{}\co{}\NP{}$,
or to prove that~$\EAkW{2} = \para{}\NP{}$,
we show that~$\EkAkWSat(\Gamma_{2,4}) \in \EkA{}$
and~$\EkAkWSat(\Gamma_{2,4}) \in \EAkW{2}$.

\begin{proposition}
\label{prop:ekak2-in-eka}
$\EkAkWSat(\Gamma_{2,4})$ is in \EkA{}.
\end{proposition}
\begin{proof}[Proof (sketch)]
We show that $\EkAkWSat(\Gamma_{2,4}) \in \EkA{}$
by giving an fpt-reduction from $\EkAkWSat(\Gamma_{2,4})$
to \EkAWSat{}.
Let~$(\varphi,k)$ be an instance of $\EkAkWSat(\Gamma_{2,4})$,
where~$\varphi = \exists X. \forall Y. C$
and where~$C \in \Gamma_{2,4}$.
It is straightforward to construct a circuit~$B$ with input nodes~$Y$
that is satisfied if and only if exactly~$k$ variables in~$Y$
are set to true.
We then construct the instance~$(\varphi',k)$ of \EkAWSat{}
by letting~$\varphi' = \exists X. \forall Y. C'$,
where~$C'$ is a circuit that is equivalent to the
expression~$B \rightarrow C$.
This is an fpt-reduction from $\EkAkWSat(\Gamma_{2,4})$
to \EkAWSat{}.
\end{proof}

\begin{proposition}
\label{prop:ekak2-in-eak2}
$\EkAkWSat(\Gamma_{2,4})$ is in \EAkW{2}.
\end{proposition}
\begin{proof}[Proof (sketch)]
We show that $\EkAkWSat(\Gamma_{2,4}) \in \EAkW{2}$
by giving an fpt-reduction from $\EkAkWSat(\Gamma_{2,4})$
to $\EAkWSat(\Gamma_{2,4})$.
Let~$(\varphi,k)$ be an instance of $\EkAkWSat(\Gamma_{2,4})$,
where~$\varphi = \exists X. \forall Y. C$
and where~$C \in \Gamma_{2,4}$.
It is straightforward to construct a circuit~$B$ with input nodes~$X$
that is satisfied if and only if exactly~$k$ variables in~$X$
are set to true.
Then, by means of the standard Tseitin transformation~\cite{Tseitin},
we can transform~$B$ into a 3\CNF{} formula over the sets~$X$
and~$X'$ of variables, where~$X'$ is a set of fresh variables,
such that for each truth assignment~$\alpha : X \rightarrow \SBs 0,1 \SEs$,
the formula~$B'[\alpha]$ is satisfiable if and only if~$B[\alpha]$ is true.
Moreover,~$B'$ can be constructed in time polynomial in the size of~$B$.
We then construct the instance~$(\varphi',k)$ of \EkAWSat{}
by letting~$\varphi' = \exists X \cup X'. \forall Y. C'$,
where~$C'$ is a circuit that is equivalent to the
expression~$B' \wedge C$.
Also, we can choose~$C'$ in such a way that~$C' \in \Gamma_{2,4}$.
Then this is an fpt-reduction from $\EkAkWSat(\Gamma_{2,4})$
to \EkAWSat{}.
\end{proof}

Since the classes \EkA{} and \EAkW{2} are closed under fpt-reductions,
the results of Theorem~\ref{thm:separation-tech-result},
Corollary~\ref{cor:separation-tech-result},
and Propositions~\ref{prop:ekak2-in-eka}~and~\ref{prop:ekak2-in-eak2}
give us the following results.

\begin{corollary}
If $\EkA{} \subseteq \para{}\co{}\NP{}$,
then there exists a subexponential-time reduction from
\QSat{2} to \UNSAT{}.
\end{corollary}

\begin{corollary}
If $\EAkW{2} \subseteq \para{}\NP{}$,
then there exists a subexponential-time reduction from
\QSat{2} to \SAT{}.
\end{corollary}

We conjecture that when proving the
inclusion~$\EAkW{1} \subseteq \para{}\NP{}$
one faces similar obstacles
in terms of notions from classical complexity theory.
However, it is not clear whether the proof techniques used for the case
of \EAkW{2} can be used to show this, or whether entirely different
proof techniques would be required.

%\ronald{to-do; do we want to say that $\A{2} \subseteq \EkAkW{2}$?
%maybe just as observation without further proof (code 5)}
}

\krlongversion{
%%% CONCLUSION (KR / ARXIV)
\section{Conclusion}
\label{sec:conclusion}

We developed a general theoretical framework that
supports the classification of parameterized problems on whether
they admit an fpt-reduction to SAT or not.
Our theory is based on two new hierarchies of complexity classes,
the \kstar{} and \stark{} hierarchies.
We illustrated the use of this theoretical toolbox by
means of a case study, in which we studied
the complexity of the consistency problem for disjunctive
answer set programming with respect to various
natural parameters.
\longversion{%
%We showed that the newly introduced
%parameterized complexity classes are populated
%by many natural problems from various domains.
%We presented a systematic and concise overview of results
%witnessing this,
%in the form of a compendium in the appendix, containing
%a wide range of problems.
The newly introduced parameterized complexity classes are
populated by many natural problems from various domains.
A systematic and concise overview
of (completeness) results for a wide range of problems
can be found as a compendium in the appendix.
Additionally, we provided alternative characterizations of the
complexity classes in the \kstar{} hierarchy based on first-order model checking
and alternating Turing machines.
Hence we have quite a clear understanding of the \kstar{} hierarchy.
The \stark{} hierarchy seems to be more complicated and we
have focused only on two of its levels.
A more comprehensive study of the \stark{} hierarchy including
alternative characterizations of its classes
using first-order model checking and machine models
are left for future research.
}%
\krversion{There are many more problems
that occur in knowledge representation and reasoning
where our theory can be used.
In the technical report corresponding to this paper, we use our theory to
analyze various additional problems, including
the problem of minimizing DNF formulas and
the problem of minimizing implicant cores.
Additionally, we illustrate the robustness of our theory
by showing a number of complete problems for the newly
introduced classes from various domains,
as well as providing alternative characterizations of the
complexity classes based on first-order model checking
and alternating Turing machines.
There are many more problems that can be analyzed
within our framework.
}

We focused our attention on the range between the first
and the second level of the PH, since many natural problems
lie there~\cite{SchaeferUmans02}.
In general, any fpt-reduction from a problem whose complexity
is higher in the PH to a lower level in the PH would be interesting.
With this more general aim in mind, it would be helpful
to have tools to gather evidence that an fpt-reduction across
some complexity border in the PH is not possible.
We hope that this paper provides a starting point for further developments.

}

\satversion{
%%% ABSTRACT (SAT)
\begin{abstract}
  Today's SAT solvers have an enormous importance and impact
  in many practical settings.
  They are used as efficient back-end to solve many NP-complete problems.
  However, many computational problems are located at the second level
  of the Polynomial Hierarchy or even higher, and hence
  polynomial-time transformations to SAT are not possible, unless the
  hierarchy collapses.
  In certain cases one can
  break through these complexity barriers by fixed-parameter tractable (fpt)
  reductions which exploit structural aspects of problem instances in
  terms of problem parameters.
  Recent research established a general theoretical framework that
  supports the classification of parameterized problems on whether
  they admit such an fpt-reduction to SAT or not.
  We use this framework to analyze some problems that are related
  to Boolean satisfiability.
  We consider several natural parameterizations of these problems,
  and we identify for which of these an fpt-reduction to SAT is possible.
  The problems that we look at are related to
  minimizing an implicant of a DNF formula,
  minimizing a DNF formula,
  and satisfiability of quantified Boolean formulas.
\end{abstract}
}

\satversion{
%%% INTRODUCTION (SAT)
\section{Introduction}

Modern SAT solvers have an enormous importance and impact
in many practical settings that require solutions
to NP-complete problems.
In fact, due to the
success of SAT, NP-complete problems have lost their scariness, as in
many cases one can efficiently encode NP-complete problems to SAT and
solve them by means of a SAT solver
\cite{BiereHeuleMaarenWalsh09,GomesKautzSabharwalSelman08,MalikZhang09,SakallahMarquessilva11}.
However, many important computational problems are located
above the first level of the Polynomial Hierarchy (PH) and thus
considered significantly ``harder'' than SAT.  Hence we cannot hope for
polynomial-time reductions from these problems to SAT, as such
transformations would cause the (unexpected) collapse of the PH.

Realistic problem instances are not random and often contain some kind
of structure. Recent research succeeded to
exploit such structure to break the complexity barriers
between levels of the PH \cite{FichteSzeider13,PfandlerRuemmeleSzeider13}.
The idea is to
exploit problem structure in terms of a problem \emph{parameter}, and
to develop reductions to SAT that can be computed efficiently as long
as the problem parameter is reasonably small.  The theory of
\emph{parameterized complexity}
\cite{DowneyFellows99,FlumGrohe06,Niedermeier06,DowneyFellows13}
provides exactly the
right type of reduction suitable for this purpose, called
\emph{fixed-parameter tractable reductions}, or \emph{fpt-reductions}
for short.  Now, for a suitable choice of the parameter, one can
aim at developing fpt-reductions from the hard problem under consideration to
SAT.

Such positive results 
go significantly beyond
the state-of-the-art of current research in parameterized complexity.
By shifting the scope from
fixed-parameter tractability to fpt-reducibility (to SAT), parameters
can be less restrictive and hence larger classes of inputs can be
processed efficiently. Therefore, the potential for positive tractability
results is greatly enlarged.  In fact, there are some known reductions
that, in retrospect, can be seen as fpt-reductions to SAT. A prominent
example is Bounded Model Checking \cite{Biere09,BiereCimattiClarkeZhu99},
which can be seen as an fpt-reduction from the model checking problem
for linear temporal logic (LTL), which is PSPACE-complete, to SAT,
where the parameter is an upper bound on the size of a counterexample.
%Bounded Model Checking is widely used for hardware and software
%verification at industrial scale \cite{Biere09}.

Recently, extending the work of Flum and Grohe \cite{FlumGrohe03},
we initiated the development of a general theoretical framework
to support the classification of hard problems on whether they
admit an fpt-reduction to SAT or not \cite{DeHaanSzeider14}.
This framework provides a hardness theory that can be used to
provide evidence that certain problems do not admit  an
fpt-reduction to SAT, similar to 
NP-hardness which provides evidence against polynomial-time tractability
\cite{GareyJohnson79} and $\W{1}$\hy hardness which provides  evidence against
fixed-parameter tractability \cite{DowneyFellows99}.
For an overview of the
parameterized complexity classes in this framework
and the relation between them,
see Figure~\ref{fig:classes}.
%In particular,
%the following parameterized complexity classes
%can be used to give evidence
%for and against the existence of reductions to SAT
%(under complexity theoretic assumptions):
%\begin{itemize}
%  \item \para{\NP} (the class of problems allowing
%    many-one fpt-reduction to SAT);
%  \item \EkA{} (the class ``\kstar{}'' of problems allowing
%    a polynomial-time reduction to SAT
%    for constant parameter values, but not allowing
%    an fpt-reduction to SAT); and
%  \item \para{\SigmaP{2}} (the class of problems
%    that do not allow a polynomial-time reduction to SAT
%    for constant parameter values, and so also
%    do not allow an fpt-reduction to SAT).
%\end{itemize}

\paragraph{New Contributions}
We use this new framework to analyze
problems related to Boolean satisfiability.
We focus on problems that are located at the second level of the
PH, i.e., problems complete for \SigmaP{2}. % or \PiP{2}.
This initiates a structured investigation of fpt-reducibility to SAT
of problems related to Boolean satisfiability that are ``beyond NP.''
Concretely, we look at the following problems,
consider several parameterizations of these problems,
and identify for which of these parameterized problems an fpt-reduction to SAT is possible
and for which this is not possible:
\begin{itemize}
  \item minimizing an implicant of a formula in disjunctive normal form (DNF)
    (parameterizations: the size of the minimized implicant,
    and the difference in size between the original and the
    minimized implicant);
  \item minimizing a DNF formula
    (parameterizations: the size of the minimized formula,
    and the difference in size between the original and the
    minimized formula); and
  \item the satisfiability problem of quantified Boolean formulas (QBFs)
    (parameterizations: the treewidth of the incidence graph
    of the formula restricted to several subsets of variables).
\end{itemize}
In particular, we show that minimizing an implicant of a DNF formula
does not become significantly easier when the minimized
implicant is small
(Proposition~\ref{prop:ssic-completeness}), nor when the difference in size
between the original and the minimized implicant is small
(Proposition~\ref{prop:lsic}).
The problem of reducing a DNF formula in size
(while preserving logical equivalence)
also does not become significantly easier when the difference in size is small
(Proposition~\ref{prop:largemindnf-completeness}).
However, the problem of reducing a DNF formula to an equivalent
DNF formula that is small can be done with a small number of
SAT calls (Theorem~\ref{thm:smallmindnf-membership2}).
Moreover, we show that deciding satisfiability of a quantified Boolean
formula with one quantifier alternation can be reduced to a single SAT instance
if the variables in the second quantifier block interact with each other
in a tree-like fashion (Theorem~\ref{thm:qsat2-twA-fpt}), whereas a similar restriction
on the variables in the first quantifier block does not make the problem
any easier (Proposition~\ref{prop:qsat2-twE-fpt}).
\squeezing{See Table~\ref{table:results} for an overview of the main complexity results
pertaining to fpt-reducibility to SAT.}

\paragraph{Related work}
Many of the decision problems analyzed in this paper have been studied before
in a classical complexity setting
\cite{GoldsmithHagenMundhenk08,Stockmeyer76,Umans00}.
The logic minimization problems that we consider in this paper
have been studied since the 1950s (cf.~\cite{Umans00}).
The problem of minimizing an implicant of a DNF formula
plays a central role in the analysis of logic minimization
problems \cite{Umans00}.
Variants of the minimization problems that we consider,
where a subset-minimal solution is sought,
are often solved by calling SAT solvers as subroutines.
One example of such work is related to identifying minimal unsatisfiable subsets
(MUSes) of a CNF formula
\cite{BelovLynceMarquesSilva12,GregoireMazurePiette08,MarquesSilvaJanotaBelov13}.
Recent work on MUS extraction indicates that reducing the number of
SAT calls made in these algorithms is beneficial
for the practical performance of these algorithms
\cite{MarquesSilvaJanotaBelov13}.
Decision procedures using SAT solvers as a subroutine
have also been used to solve problems that lie at the second level
of the PH, e.g., problems related to abstract argumentation
\cite{DvorakJarvisaloWallnerWoltran14}.
}

\satversion{
%%% PRELIMINARIES (SAT)
\section{Preliminaries}

\paragraph{Propositional Logic and Quantified Boolean Formulas}

A \emph{literal} is a propositional variable~$x$ or a negated variable~$\neg x$.
The \emph{complement}~$\overline{x}$ of a positive literal~$x$ is~$\neg x$,
and the complement~$\overline{\neg x}$ of a negative literal~$\neg x$ is~$x$.
For literals~$l \in \SBs x, \neg x \SEs$, we let~$\Var{l} = x$ denote
the variable occurring in~$l$.
A \emph{clause} is a finite set of literals,
not containing a complementary pair~$x$,~$\neg x$,
and is interpreted as the disjunction of these literals.
%We let~$\bot$ denote the empty clause.
A \emph{term} is a finite set of literals,
not containing a complementary pair~$x$,~$\neg x$,
and is interpreted as the conjunction of these literals.
We let~$\top$ denote the empty clause.
%For any term~$t$, we let~$\Pos{t}$ denote the set of
%variables~$x$ occurring positively in~$t$, and we
%let~$\Neg{t}$ denote the set of variables~$x$ occurring
%negatively in~$t$, i.e., those variables~$x$ such
%that~$\neg x \in t$.
A formula in \emph{conjunctive normal form (CNF)}
is a finite set of clauses, interpreted as the conjunction
of these clauses.
A formula in \emph{disjunctive normal form (DNF)}
is a finite set of terms, interpreted as the disjunction
of these terms.
%We define the \emph{size}~$\size{\varphi}$ of a
%CNF formula~$\varphi$ to be~$\sum_{c \in \varphi} \Card{c}$;
%the number of clauses of~$\varphi$ is denoted by~$\Card{\varphi}$.
We say that a DNF formula~$\varphi$ is a \emph{term-wise subformula}
of another DNF formula~$\varphi'$ if for all terms~$t \in \varphi$
there exists a term~$t' \in \varphi'$ such that~$t \subseteq t'$.
We define the \emph{size}~$\size{\varphi}$ of a
DNF formula~$\varphi$ to be~$\sum_{t \in \varphi} \Card{t}$;
the number of terms of~$\varphi$ is denoted by~$\Card{\varphi}$.
For a DNF formula~$\varphi$, the set~$\Var{\varphi}$ denotes
the set of all variables~$x$ such that some term of~$\varphi$
contains~$x$ or~$\neg x$.
We use the standard notion of \emph{(truth)
assignments}~$\alpha : \Var{\varphi} \rightarrow \SBs 0,1 \SEs$
for Boolean formulas and \emph{truth} of a formula
under such an assignment.
We denote the problem of deciding whether a propositional
formula~$\varphi$ is satisfiable by \thSAT{},
and the problem of deciding whether~$\varphi$ is
not satisfiable by \thUNSAT{}.
Let~$\varphi$ be a DNF formula.
We say that a set~$C$ of literals is an \emph{implicant of~$\varphi$}
if all assignments that satisfy~$\bigwedge_{l \in C} l$
also satisfy $\varphi$.
%An implicant~$C$ of~$\varphi$ is a \emph{prime implicant}
%if there is no subset~$C' \subsetneq C$ of~$C$
%that is an implicant of~$\varphi$.
Let~$\gamma = \SBs x_1 \mapsto d_1, \dotsc, x_n \mapsto d_n \SEs$
be a function that maps some variables
of a formula~$\varphi$ to other variables or to truth values.
We let~$\varphi[\gamma]$ denote the application of
such a substitution~$\gamma$ to the formula~$\varphi$.
We also write~$\varphi[x_1 \mapsto d_1,\dotsc,x_n \mapsto d_n]$
to denote~$\varphi[\gamma]$.

\sloppypar
A \emph{(prenex) quantified Boolean formula (QBF)} is a formula of the form
$Q_1X_1 Q_2X_2 \dotsc Q_m X_m \psi$,
where each~$Q_i$ is either~$\forall$ or~$\exists$,
the~$X_i$ are disjoint sets of propositional variables,
and~$\psi$ is a Boolean formula over the variables in~$\bigcup_{i=1}^{m} X_i$.
We call $\psi$ the \emph{matrix} of the formula.
Truth of such formulas is defined in the usual way.
We say that a QBF is \emph{in QDNF} if the matrix is in DNF.
For the remainder of this paper, we will restrict our attention
to QDNF formulas.
Consider the following decision problem.
\probdef{
  $\EASat{}$
  
  \emph{Instance:} A QDNF~$\varphi = \exists X. \forall Y. \psi$,
    where~$\psi$ is quantifier-free.

  \emph{Question:} Is $\varphi$ satisfiable?
}
The complexity class consisting of all problems that are
polynomial-time reducible to~\EASat{} is denoted by~\SigmaP{2},
and its co-class is denoted by~\PiP{2}.
These classes form the second level of the PH
\cite{Papadimitriou94}.

\paragraph{Parameterized Complexity}
We introduce some core notions from parameterized complexity theory.
For an in-depth treatment we refer to other sources
\cite{DowneyFellows99,FlumGrohe06,Niedermeier06,DowneyFellows13}.
A \emph{parameterized problem}~$L$ is a subset of~$\Sigma^{*} \times
\mathbb{N}$ for some finite alphabet~$\Sigma$.  For an instance~$(I,k)
\in \Sigma^{*} \times \mathbb{N}$, we call~$I$ the \emph{main part}
and~$k$ the \emph{parameter}.  The following generalization of
polynomial time computability is commonly regarded as the tractability
notion of parameterized complexity theory.  A parameterized problem
$L$ is \emph{fixed-parameter tractable} if there exists a computable
function~$f$ and a constant~$c$ such that there exists an algorithm
that decides whether~$(I,k) \in L$ in time~$O(f(k)\size{I}^c)$,
where~$\size{I}$ denotes the size of~$I$.  Such an algorithm is called an
\emph{fpt-algorithm}, and this amount of time is called
\emph{fpt-time}. \FPT{} is the class of all fixed-parameter tractable
decision problems.
If the parameter is constant, then fpt-algorithms run in polynomial
time where the order of the polynomial is independent of the
parameter.
This provides a good scalability in the parameter in contrast to
running times of the form $\size{I}^k$,
which are also polynomial for fixed~$k$, but are already
impractical for, say, $k>3$.
By XP we denote the class of all problems~$L$ for which it can be decided
whether~$(I,k) \in L$ in time~$O(\size{I}^{f(k)})$,
for some fixed computable function $f$.

Parameterized complexity also generalizes the notion of polynomial-time
reductions.
Let~$L \subseteq \Sigma^{*} \times \mathbb{N}$ and
$L' \subseteq (\Sigma')^{*} \times \mathbb{N}$ be two parameterized problems.
%for some finite alphabets~$\Sigma$ and~$\Sigma'$.
A \emph{(many-one) fpt-reduction} from~$L$ to~$L'$ is a mapping
$R : \Sigma^{*} \times \mathbb{N} \rightarrow (\Sigma')^{*} \times \mathbb{N}$
from instances of~$L$ to instances of~$L'$ such that
there exist some computable function~$g : \mathbb{N} \rightarrow \mathbb{N}$
such that for all~$(I,k) \in \Sigma^{*} \times \mathbb{N}$:
(i)~$(I,k)$ is a yes-instance of~$L$ if and only if
$(I',k') = R(I,k)$ is a yes-instance of~$L'$,
(ii)~$k' \leq g(k)$,
and (iii)~$R$ is computable in fpt-time.
Similarly, we call reductions that satisfy properties~(i) and~(ii)
but that are computable in time~$O(\size{I}^{f(k)})$,
for some fixed computable function~$f$, \emph{xp-reductions}.

%Parameterized complexity theory also offers complexity classes for
%problems that lie higher in the PH.
Let~$C$ be a
classical complexity class, e.g., \NP{}.  The parameterized complexity
class $\para{C}$ is then defined as the class of all parameterized
problems~$L \subseteq \Sigma^{*} \times \mathbb{N}$, for some finite
alphabet~$\Sigma$, for which there exist an alphabet~$\Pi$, a
computable function~$f : \mathbb{N} \rightarrow \Pi^{*}$, and a
problem~$P \subseteq \Sigma^{*} \times \Pi^{*}$ such that~$P \in C$
and for all instances~$(x,k) \in \Sigma^{*} \times \mathbb{N}$ of~$L$
we have that~$(x,k) \in L$ if and only if~$(x,f(k)) \in P$.
Intuitively, the class $\para{C}$ consists of all problems that are in
$C$ after a precomputation that only involves the parameter
\cite{FlumGrohe03}.

%%%
%%% FPT-REDUCTIONS TO SAT
%%%
\section{Fpt-Reductions to SAT}
Problems in~$\NP$ and~$\co\NP$ can
be encoded into SAT in such a way that
the time required to produce the encoding
and consequently also the size of the resulting SAT instance are
polynomial in the input (the encoding is a polynomial-time many-one
reduction). Typically, the SAT encodings of problems proposed for
practical use are of this kind (cf.~\cite{Prestwich09}).
For problems that
are ``beyond NP,'' say for problems on the second level of the PH, such
polynomial SAT encodings do not exist, unless the PH
collapses. However, for such problems, there still could exist SAT
encodings which can be produced in fpt-time in terms of some parameter
associated with the problem. In fact, such fpt-time SAT encodings have
been obtained for various problems on the second level of the PH
\cite{FichteSzeider13,DeHaanSzeider14,PfandlerRuemmeleSzeider13}. The
classes \para{\NP} and \para{\co\NP} contain exactly those
parameterized problems that admit such a many-one fpt-reduction to
\thSAT{} and \thUNSAT{}, respectively. Thus, with fpt-time encodings,
one can go significantly beyond what is possible by conventional
polynomial-time SAT encodings.

Consider the following example.
The problem of deciding satisfiability of a QBF does not
allow a polynomial-time SAT encoding.
However, if the number of universal variables is small,
one can use known methods in QBF solving
to get an fpt-time encoding into SAT.
%In other words, the following parameterized problem
%is in \para{}\NP{}.
\probdef{
  \FewAQBFSat{}
  
  \emph{Instance:} A QBF~$\varphi$. % = \exists X_1. \forall X_2 \dotsc Q_{j} X_{j}. \psi$.

  \emph{Parameter:} The number of universally quantified variables of~$\varphi$.

  \emph{Question:} Is~$\varphi$ true?
}
The idea behind this encoding is to repeatedly use
universal quantifier expansion~\cite{AyariBasin02,Biere04}.
Eliminating~$k$ many universally quantified variables
in this manner leads to an existentially quantified formula
that is at most a factor of~$2^k$ larger than the original formula.

Fpt-time encodings to SAT also have their limits. Clearly,
\para{\SigmaP{2}}-hard and \para{\PiP{2}}-hard parameterized
problems do not admit fpt-time encodings to SAT, even when the
parameter is a constant, unless the PH collapses. There are problems
that apparently do not admit fpt-time encodings to SAT, but are neither
\para{\SigmaP{2}}-hard nor \para{\PiP{2}}-hard.  In recent work
\cite{DeHaanSzeider14} we have introduced several complexity
classes for such intermediate problems, including the following.
% Recent work in parameterized complexity theory has
% resulted in complexity classes that can be used to provide
% evidence for the non-existence
% of fpt-reductions to \thSAT{}
% also for problems that do allow an xp-reduction to \thSAT{}
% or \thUNSAT{}
% \cite{DeHaanSzeider14}.
The parameterized complexity class~\EkA{} consists of all
parameterized problems that can be many-one fpt-reduced to the
following variant of quantified Boolean satisfiability that is based
on truth assignments of restricted weight.
\probdef{ \EkAWSat{}

  \emph{Instance:} A quantified Boolean
  formula~$\varphi = \exists X. \forall Y. \psi$,
  and an integer~$k$.

  \emph{Parameter:} $k$.

  \emph{Question:} Does there exist a truth assignment~$\alpha$
  to~$X$ with weight~$k$
  such that for all truth assignments~$\beta$ to~$Y$
  the assignment~$\alpha \cup \beta$ satisfies~$\psi$?
}
For each problem in \EkA{} there exists an xp-reduction to \thUNSAT{}.
However, there is evidence that problems that are hard for \EkA{}
do not allow a many-one fpt-reduction to \thSAT{} \cite{DeHaanSzeider14}.
Many natural parameterized problems from various domains are
complete for the class \EkA{}, and for none of them an fpt-reduction
to \thSAT{} or \thUNSAT{} has been found.
If there exists an fpt-reduction to \thSAT{} for any \EkA{}-complete problem
then this is the case for all \EkA{}-complete problems.
The dual complexity class of \EkA{} is denoted by \AkE{},
and has similar (yet dual) properties.
Note that the notion of reducibility underlying hardness for
all parameterized complexity classes mentioned above
is that of many-one fpt-reductions.
For a more detailed discussion on the complexity classes
\EkA{} and \AkE{},
we refer to previous work \cite{DeHaanSzeider14}.

One can also enhance the power of polynomial-time SAT encodings by
considering polynomial-time algorithms that can query a SAT solver
multiple times.
Such an approach has been shown to be quite effective
in practice (see, e.g.,~\cite{BelovLynceMarquesSilva12,DvorakJarvisaloWallnerWoltran14,MarquesSilvaJanotaBelov13}) and extends the scope of
SAT solvers to problems in the class~\DeltaP{2}, but not to problems that are
\SigmaP{2}-hard or \PiP{2}-hard. Also here, switching from
polynomial-time to fpt-time provides a significant increase in
power. The class \para{\DeltaP{2}} contains all parameterized problems
that can be solved by an fpt-algorithm that can query a SAT solver
multiple times (i.e., by an fpt-time Turing reduction to SAT).

%%%
\begin{figure}[H]
\begin{center}
\vspace{-20pt}
\begin{tikzpicture}[xscale=1.6, yscale=0.65]
  % NODES
  \node[] (sigma2) at (0,0) {$\para{\SigmaP{2}}$};
  \node[] (pi2) at (3,0) {$\para{\PiP{2}}$};
  \node[] (paranp) at (0,-4.4) {$\para{\NP}$};
  \node[] (paraconp) at (3,-4.4) {$\para{\co\NP}$};
  \node[] (delta2) at (1.5,-1) {$\para{\DeltaP{2}}$};
%  \node[] (fptnpk) at (1.5,-2.2)
%    {\parbox{3.8cm}{\center \scriptsize [solvability in fpt-time\\using~$f(k)$ many SAT calls]}};
  \node[] (fptnpk) at (1.5,-2.2)
    {\center $\mtext{FPT}^{\mtext{NP}[f(k)]}$};
  \node[] (dp) at (1.5,-3.4) {$\para{\DP}$};
  \node[] (eka) at (-0.8,-1.8) {$\EkA$};
  \node[] (ake) at (3.8,-1.8) {$\AkE$};
%  \node[] (wdots) at (0,-5.2) {$\rvdots$};
  \node[] (w1) at (0,-5.6) {$\W{1}$};
%  \node[] (cowdots) at (3,-5.2) {$\rvdots$};
  \node[] (cow1) at (3,-5.6) {$\co{\W{1}}$};
  \node[] (fpt) at (1.5,-6.2) {$\FPT = \para{\PTIME}$};
  % EDGES
  \path[draw,->] (paranp) edge[bend left=60] (sigma2);
  \path[draw,->] (paraconp) edge[bend right=60] (pi2);
  \path[draw,->] (paranp) edge[bend right=25] (dp);
  \path[draw,->] (paraconp) edge[bend left=25] (dp);
  \path[draw,->] (dp) -- (fptnpk);
  \path[draw,->] (fptnpk) -- (delta2);
  \path[draw,->] (delta2) edge[bend right] (sigma2);
  \path[draw,->] (delta2) edge[bend left] (pi2);
  \path[draw,->] (paranp) edge[out=5, in=-100] (ake);
  \path[draw,->] (paraconp) edge[out=175, in=-80] (eka);
  \path[draw,->] (eka) edge[bend left=20] (sigma2);
  \path[draw,->] (ake) edge[bend right=20] (pi2);
%  \path[draw,-] (w1) -- (wdots);
%  \path[draw,->] (wdots) -- (paranp);
%  \path[draw,-] (cow1) -- (cowdots);
%  \path[draw,->] (cowdots) -- (paraconp);
  \path[draw,->,dashed, dash pattern=on .5mm] (w1) -- (paranp);
  \path[draw,->,dashed, dash pattern=on .5mm] (cow1) -- (paraconp);
  \path[draw,->] (fpt) edge[bend left=15] (w1);
  \path[draw,->] (fpt) edge[bend right=15] (cow1);
  
%  \begin{scope}[xshift=1.0]
%  \fill[black!50,nearly transparent,rounded corners=.3cm]
%    %(paranp) -- (dp) -- (paraconp)  -- cycle;
%    (-0.6,-5.6) -- (-0.6,-3.8) -- (0.4,-1.6) -- (2.6,-1.6) -- (3.6,-3.8) -- (3.6,-5.6) -- (3.0,-6.6) -- (0,-6.6) -- cycle;
%  \end{scope}
\end{tikzpicture}
\end{center}
\vspace{-15pt}
\caption{Parameterized complexity classes
up to the second level of the polynomial
hierarchy.
Arrows indicate inclusion relations.
%The classes highlighted in gray allow
%efficient fpt-reductions to \SAT{}; the other classes
%are believed not to allow this.
(We omit the full definition of some of the
parameterized complexity classes depicted in the figure.
For a detailed definition of these, we refer to other
sources
\cite{Papadimitriou94,DowneyFellows99,FlumGrohe06,DowneyFellows13}.)}
\vspace{-5pt}
\label{fig:classes}
\end{figure}

An overview
of all relevant parameterized complexity classes
can be found in Figure~\ref{fig:classes}.
Locating problems in the complexity landscape as laid out in
this figure
can provide a guideline for practitioners, to indicate what algorithmic
approaches are possible and where complexity theoretic obstacles lie.

There are two fundamental aspects that are relevant for an algorithm
that makes queries to a SAT solver: (i)~the running time of the
algorithm (which does not take into account the time needed by the SAT solver to
answer the queries) and (ii)~the number of SAT calls. Results from
classical \emph{bounded query complexity} \cite{Hartmanis90,Wagner90}
suggests that if the running
time is polynomial, then increasing the numbers of SAT calls increases
the computing power. Several such separation results are known
\cite{ChangKadin96,Krentel88}.
From a practical point of view, the number of SAT calls
may seem to be relatively insignificant, assuming that the queries are
easy for the solver, and the solver can reuse information
from previous calls
\cite{BenedettiBernardini04,Hooker93,WhittemoreKimSakallah01}.
For a theoretical
worst-case model, however, one must assume that all queries involve hard SAT
instances, and that no information from previous calls can be
reused. Therefore, in a theoretical analysis,
it makes sense to study the number of SAT
calls made by fpt-time algorithms. In the parameterized setting, it is
natural to bound the number of SAT calls by a function of the
parameter. This yields the class $\mtext{FPT}^{\mtext{NP}[f(k)]}$,
which lies between \para{\DP} (two calls)
and \para{\DeltaP{2}} (unrestricted number of calls).

%\paragraph{Treewidth}
%\PreliminariesTreewidth{}
}

\squeezing{
\satversion{
\begin{table}[h!]
  \scriptsize
  \centering
  \vspace{-10pt}
  \begin{tabular}{@{}p{5.2cm}@{\qquad}p{2.8cm}@{\qquad}p{2.8cm}@{}}
 \toprule
    Problem & Complexity & Fpt-reducible to SAT? \\
 \midrule
    \SSIC{} & \EkA{}-c (Prop~\ref{prop:ssic-hardness}~and~\ref{prop:ssic-membership})
      & no \\
    \LSIC{} & \EkA{}-c (Prop~\ref{prop:lsic}) & no \\
    \LargeMinDNF{} & \EkA{}-c (Prop~\ref{prop:largemindnf-hardness}~and~\ref{prop:largemindnf-membership}) & no \\
    \SmallMinDNF{} & in~$\mtext{FPT}^{\mtext{FNP}[k]}$ (Thm~\ref{thm:smallmindnf-membership2}) & yes \\
    \EAtwESat{} & \para{\SigmaP{2}}-c (Prop~\ref{prop:qsat2-twE-fpt}) & no \\
    \EAtwASat{} & \para{\NP}-c (Thm~\ref{thm:qsat2-twA-fpt}) & yes \\
    \FewAQBFSat{} & \para{\NP}-c (Prop~\ref{prop:qbf-fewa}) & yes \\
 \bottomrule
  \end{tabular}
  \vspace{5pt}
  \caption{Overview of main complexity results (related to fpt-reducibility to SAT).}
  \label{table:results}
  \vspace{-30pt}
\end{table}
}
}

\satversion{
%%% MINIMIZATION PROBLEMS FOR DNF FORMULAS (SAT)
\section{Minimization Problems for DNF Formulas}

We consider several problems related to minimizing
implicants of DNF formulas and minimizing DNF formulas.
We consider several parameterizations,
and we show that some of these allow an fpt-reduction to SAT,
whereas others apparently do not.

The following decision problem,
that is concerned with reducing a given
implicant of a DNF formula in size,
is \SigmaP{2}-complete~\cite{Umans00}.
\probdef{
  \SIC{}
  
  \emph{Instance:} A DNF formula~$\varphi$,
  an implicant~$C$ of~$\varphi$ of size~$n$, and an integer~$m$.

  \emph{Question:} Does there exist an implicant~$C' \subseteq C$
  of~$\varphi$ of size~$m$?
}
We consider two parameterizations of this problem:
(1)~\SSIC{}, where the parameter~$k = m$ is
the size of the minimized implicant,
and (2)~\LSIC{}, where the parameter~$k = n-m$ is
the difference in size between the original implicant
and the minimized implicant.
We show that neither of these restrictions is enough
to admit an fpt-reduction to SAT.
All results can straightforwardly be extended to the
variant of the problem where implicants of size at most~$m$ are accepted.

%\boxedprob{
%  \SSIC{}
%  
%  \emph{Instance:} A DNF formula~$\varphi$,
%  an implicant~$C$ of~$\varphi$, and an integer~$k$.
%
%  \emph{Parameter:} $k$.
%
%  \emph{Question:} Does there exist an implicant~$C' \subseteq C$
%  of~$\varphi$ of size~$k$?
%}{\textbf{Complexity:} \EkA{}-complete
%(Proposition~\ref{prop:ssic-completeness}).}
%\vspace{-5pt}
%\boxedprob{
%  \LSIC{}
%  
%  \emph{Instance:} A DNF formula $\varphi$,
%  an implicant~$C$ of~$\varphi$ of size~$n$, and an integer~$k$.
%
%  \emph{Parameter:} $k$.
%
%  \emph{Question:} Does there exist an implicant~$C' \subseteq C$
%  of~$\varphi$ of size~$n-k$?
%}{\textbf{Complexity:} \EkA{}-complete
%(Proposition~\ref{prop:lsic}).}

Next, consider the following decision problem,
that is concerned with deciding whether a
given DNF formula~$\varphi$ is logically equivalent to
a DNF formula~$\varphi'$ of size~$m$,
and that is \SigmaP{2}-complete~\cite{Umans00}.
\probdef{
  \MinDNF{}
  
  \emph{Instance:} A \DNF{} formula~$\varphi$
  of size~$n$,
  and an integer~$m$.

  \emph{Question:} Does there exist a term-wise subformula~$\varphi'$ 
  of~$\varphi$ of size~$m$ such
  that~$\varphi \equiv \varphi'$?
}
We consider the following two parameterizations
of this problem:
(1)~\LargeMinDNF{}, where the parameter~$k = n-m$ is the difference
in size between the original formula~$\varphi$ and the
minimized formula~$\varphi'$,
and (2)~\SmallMinDNF{}, where the parameter~$k = m$ is the
size of the minimized formula~$\varphi'$.
We show that the former parameterization is not enough
to allow an fpt-reduction to SAT,
but that for the latter parameterization, the problem can be
solved with an fpt-algorithm that uses at
most~$\lceil \log_2 k \rceil+1$ many SAT calls.
Moreover, this algorithm works even for the case
where equivalent DNF formulas that are not term-wise subformulas
of~$\varphi$ are also accepted.

%\boxedprob{
%  \LargeMinDNF{}
%  
%  \emph{Instance:} A \DNF{} formula~$\varphi$
%  of size~$n$,
%  and an integer~$k$.
%
%  \emph{Parameter:} $k$.
%
%  \emph{Question:} Does there exist an \DNF{} formula~$\varphi'$
%  of size~$n-k$ such
%  that~$\varphi \equiv \varphi'$?
%}{\textbf{Complexity:} \EkA{}-complete
%(Proposition~\ref{prop:largemindnf-completeness}).}
%\vspace{-5pt}
%\boxedprob{
%  \SmallMinDNF{}
%  
%  \emph{Instance:} A \DNF{} formula~$\varphi$
%  of size~$n$,
%  and an integer~$k$.
%
%  \emph{Parameter:} $k$.
%
%  \emph{Question:} Does there exist an \DNF{} formula
%  $\varphi'$ of size~$k$,
%  such that $\varphi \equiv \varphi'$?
%}{\textbf{Complexity:} \para{\co\NP}-hard
%and in \EkA{}
%(Proposition~\ref{prop:smallmindnf-results1}), and
%solvable in fpt-time using~$\lceil \log_2 k \rceil +1$
%functional SAT calls
%(Theorem~\ref{thm:smallmindnf-membership2}).}

We will now set out to prove the complexity results
mentioned in the discussion above.
%
%%%
%%%
%%%
%\begin{proposition}
%\label{prop:sikc}
%\sloppypar \SIkC{} is \para{\co\NP}-complete.
%\end{proposition}
%\begin{proof}
%\sloppypar
%To show \para{\co\NP}-hardness,
%we give a polynomial-time reduction from UNSAT
%to the problem~$P = \SB (\psi,C) \in \SIkC \SM \Card{C} = 1 \SE$.
%Let~$\varphi$ be an instance of UNSAT consisting of a CNF formula.
%Let~$x \not\in \Var{\varphi}$ be a fresh variable.
%We construct an instance~$(\psi, \SBs x \SEs)$ of~\SIkC{}
%as follows.
%We let~$\psi = (\neg \varphi \vee x)$.
%We know that there is exactly one~$C' \subsetneq C$,
%namely~$C' = \emptyset$.
%It is straightforward to verify that~$\emptyset$ is an implicant of~$\psi$
%if and only if~$\neg \varphi$ is valid, or equivalently,
%if and only if~$\varphi$ is unsatisfiable.
%
%\sloppypar
%To show \para{\co\NP}-membership,
%we give an fpt-reduction from \SIkC{} to UNSAT.
%Let~$(\psi,C)$ be an instance of \SIkC{}.
%For each~$D \subsetneq C$, we let~$\psi^{D}$
%be a copy of~$\psi$ where each~$x \in \Var{\psi}$
%is replaced by a copy~$x^{D}$ of~$x$.
%Furthermore, for each~$D \subsetneq C$, we
%define the set~$\sigma(D) = \SB x^{D} \SM x \in C' \SE$
%containing a copy~$x^{D}$ for each~$x \in C'$.
%We construct an instance~$\varphi$ of UNSAT
%by letting
%$\varphi = \bigwedge\nolimits_{D \subsetneq C}
%( \bigwedge \sigma(D) \wedge \neg\psi^{C'} )$.
%Clearly,~$\varphi$ can be constructed in time~$2^k \cdot \size{\psi}$.
%It is straightforward to verify that~$\varphi$ is satisfiable
%if and only if no~$D \subsetneq C$ is an implicant
%of~$\psi$.
%\end{proof}
%
%%%
%%%
%%%
In order to prove \EkA{}-hardness of \SSIC{},
we need the following technical lemma
(we omit its straightforward proof).

\begin{lemma}
\label{lem:exactly-k}
Let $(\varphi,k)$ be an instance of \EkAWSat{}.
In polynomial time, we can construct an equivalent instance $(\varphi',k)$
of \EkAWSat{} with $\varphi' = \exists X. \forall Y. \psi$, such that
for every assignment $\alpha : X \rightarrow \SBs 0,1 \SEs$ that has weight
$m \neq k$, it holds that $\forall Y. \psi[\alpha]$ is true.
\end{lemma}

\begin{proposition}
\label{prop:ssic-completeness}
\SSIC{} is \EkA{}-complete.
\end{proposition}
\ProofSSICCompleteness{}

%\begin{proposition}
%\label{prop:ssic-hardness}
%\SSIC{} is \EkA{}-hard.
%\end{proposition}
%\ProofSSICHardness{}
%
%\begin{proposition}
%\label{prop:ssic-membership}
%\SSIC{} is in \EkA{}.
%\end{proposition}
%\ProofSSICMembership{}

%%%
%%%
%%%

\begin{proposition}
\label{prop:lsic}
\LSIC{} is \EkA{}\hy{}complete.
\end{proposition}
\begin{proof}[sketch]
As an auxiliary problem, we consider the parameterized problem
\EnminkAWSat{}, which is a variant of \EkAWSat{}.
Given an input consisting of a QDNF~$\varphi = \exists X. \forall Y.
\psi$ with~$\Card{X} = n$ and an integer~$k$,
the problem is to decide whether there exists an assignment~$\alpha$
to~$X$ with weight~$n-k$
such that~$\forall Y. \psi[\alpha]$ is true.
The parameter for this problem is~$k$.
We claim that this problem has the following properties.
We omit a proof of these claims.

\smallskip \noindent \textit{Claim 1.} \EnminkAWSat{} is \EkA{}-complete.

\smallskip \noindent \textit{Claim 2.} Let $(\varphi,k)$ be an instance of \EnminkAWSat{}.
In polynomial time, we can construct an equivalent instance $(\varphi',k)$
of \EnminkAWSat{} with $\varphi' = \exists X. \forall Y. \psi$, such that
for any assignment $\alpha : X \rightarrow \SBs 0,1 \SEs$ that has weight
$m \neq (\Card{X} - k)$, it holds that $\forall Y. \psi[\alpha]$ is true.

\smallskip\noindent Using these claims, both membership and hardness for \EkA{} follow 
straightforwardly using arguments similar to the \EkA{}-completeness proof
of \SSIC{}.
The fpt-reductions in the proof of Proposition~\ref{prop:ssic-completeness}
show that \LSIC{} fpt-reduces to and from \EnminkAWSat{}.
\end{proof}
%
%%%
%%%
%%%
\smallskip\noindent We can now turn to proving complexity results
for the problems of minimizing DNF formulas.

\begin{proposition}
\label{prop:largemindnf-completeness}
\LargeMinDNF{} is \EkA{}-complete.
\end{proposition}
\begin{proof}[sketch]%
To show \EkA{}-hardness, we use
the reduction from the literature
that is used to show \SigmaP{2}-hardness
of the unparameterized version of \LargeMinDNF{}.
The polynomial-time reduction
from the unparameterized version of \LSIC{}
to the unparameterized version of \LargeMinDNF{}
given by Umans~\cite[Theorem~2.2]{Umans00}
is an fpt-reduction from \LSIC{}
to \LargeMinDNF{}.

To show membership in \EkA{}, one can give an fpt-reduction
to \EkAWSat{}.
We describe the main idea behind this reduction,
and we omit a detailed proof.
Given an instance~$(\varphi,k)$ of \LargeMinDNF{}
we construct an instance~$(\varphi',k)$ of \EkAWSat{}
where the assignment to the existentially quantified variables of~$\varphi'$
represents the~$k$ many literal occurrences that are to be removed,
and where universally quantified part of~$\varphi'$ is used
to verify the equivalence of~$\varphi$ and the formula
obtained from~$\varphi$ by removing the~$k$ literals
chosen by the assignment to the existential variables.
\end{proof}

\CompleteSmallMinDNF{}
}

\satversion{
%%% QBF SATISFIABILITY
\section{QBF Satisfiability and Treewidth}

\newcommand{\itw}{\mtext{incid.tw}}
\newcommand{\Eitw}{\mtext{$\exists$-incid.tw}}
\newcommand{\Aitw}{\mtext{$\forall$-incid.tw}}

The graph parameter \emph{treewidth} measures in a certain sense the
tree-likeness of a graph (for a definition of treewidth, see, e.g.,
\cite{Bodlaender98,Bodlaender12}). Many hard problems are
fixed-parameter tractable when parameterized by the treewidth of a
graph associated with the
input~\cite{Bodlaender12,GottlobPichlerWei10}.  By associating the
following graph with a QDNF formulas one can apply the parameter
treewidth also to QDNF formulas (for QCNF formulas the graph can be
defined analogously, taking clauses instead of terms).

The \emph{incidence graph} of a QDNF formula $\varphi$ is the
bipartite graph where one side of the partition consists of the
variables and the other side consists of the terms; a variable and a term
are adjacent if the variable appears positively or negatively in the
term.  The \emph{incidence treewidth} of $\varphi$, in symbols
$\itw(\varphi)$, is the treewidth of the incidence graph of $\varphi$.
It is well known that checking the truth of a QDNF formula whose
number of quantifier alternations is bounded by a constant is
fixed-parameter tractable when parameterized by $\itw$ (this can be
easily shown using Courcelle's Theorem \cite{GottlobPichlerWei10}).

Bounding the treewidth of the entire incidence graph is very
restrictive.  In this section we investigate whether bounding the
treewidth of certain subgraphs of the incidence graph is sufficient to
reduce the complexity. To this aim we define the \emph{existential
  incidence treewidth} of a QDNF formula~$\varphi$, in symbols
$\Eitw(\varphi)$, as the treewidth of the incidence graph of~$\varphi$
after deletion of all universal variables.  The \emph{universal
  incidence treewidth}, in symbols $\Aitw(\varphi)$, is the treewidth
of the incidence graph of $\varphi$ after deletion of all existential
variables.

The existential and universal treewidth can be small for formulas
whose incidence treewidth is arbitrarily large.  Take for instance a
QDNF formula $\varphi$ whose incidence graph is an $n\times n$ square
grid, as in Figure~\ref{fig:grid}.
In this example, $\Eitw(\varphi)=\Aitw(\varphi)=2$
(since after the deletion of the universal or the existential variables the
incidence graph becomes a collection of trivial path-like graphs), but
\mbox{$\itw(\varphi)=n$}~\cite{Bodlaender98}. Hence, a tractability
result in terms of existential or universal incidence treewidth would
apply to a significantly larger class of instances than a tractability
result in terms of incidence treewidth.
\begin{figure}
\centering
  \vspace{-10pt}
  \begin{tikzpicture}[scale=.5]
    \newcommand{\edgesep}[0]{-0mm}
    \tikzstyle{mybox}=[draw, rectangle, text width=3pt, text height=3pt]
    \tikzstyle{mybullet}=[draw, circle, text width=3pt, text height=3pt, fill=black]
    \tikzstyle{mycirc}=[draw, circle, text width=3pt, text height=3pt]

    % NODES
    \node[inner sep=\edgesep, mybullet] (0-0) at (0,0) {};
    \node[inner sep=\edgesep, mybox] (1-0) at (1,0) {};
    \node[inner sep=\edgesep, mycirc] (2-0) at (2,0) {};
    \node[inner sep=\edgesep, mybox] (3-0) at (3,0) {};
    \node[inner sep=\edgesep, mybullet] (4-0) at (4,0) {};
    \node[inner sep=\edgesep, mybox] (5-0) at (5,0) {};
    \node[inner sep=\edgesep, mybox] (0-1) at (0,1) {};
    \node[inner sep=\edgesep, mybullet] (1-1) at (1,1) {};
    \node[inner sep=\edgesep, mybox] (2-1) at (2,1) {};
    \node[inner sep=\edgesep, mycirc] (3-1) at (3,1) {};
    \node[inner sep=\edgesep, mybox] (4-1) at (4,1) {};
    \node[inner sep=\edgesep, mybullet] (5-1) at (5,1) {};
    \node[inner sep=\edgesep, mycirc] (0-2) at (0,2) {};
    \node[inner sep=\edgesep, mybox] (1-2) at (1,2) {};
    \node[inner sep=\edgesep, mybullet] (2-2) at (2,2) {};
    \node[inner sep=\edgesep, mybox] (3-2) at (3,2) {};
    \node[inner sep=\edgesep, mycirc] (4-2) at (4,2) {};
    \node[inner sep=\edgesep, mybox] (5-2) at (5,2) {};
    \node[inner sep=\edgesep, mybox] (0-3) at (0,3) {};
    \node[inner sep=\edgesep, mycirc] (1-3) at (1,3) {};
    \node[inner sep=\edgesep, mybox] (2-3) at (2,3) {};
    \node[inner sep=\edgesep, mybullet] (3-3) at (3,3) {};
    \node[inner sep=\edgesep, mybox] (4-3) at (4,3) {};
    \node[inner sep=\edgesep, mycirc] (5-3) at (5,3) {};
    \node[inner sep=\edgesep, mybullet] (0-4) at (0,4) {};
    \node[inner sep=\edgesep, mybox] (1-4) at (1,4) {};
    \node[inner sep=\edgesep, mycirc] (2-4) at (2,4) {};
    \node[inner sep=\edgesep, mybox] (3-4) at (3,4) {};
    \node[inner sep=\edgesep, mybullet] (4-4) at (4,4) {};
    \node[inner sep=\edgesep, mybox] (5-4) at (5,4) {};
    \node[inner sep=\edgesep, mybox] (0-5) at (0,5) {};
    \node[inner sep=\edgesep, mybullet] (1-5) at (1,5) {};
    \node[inner sep=\edgesep, mybox] (2-5) at (2,5) {};
    \node[inner sep=\edgesep, mycirc] (3-5) at (3,5) {};
    \node[inner sep=\edgesep, mybox] (4-5) at (4,5) {};
    \node[inner sep=\edgesep, mybullet] (5-5) at (5,5) {};
    \node at (2.5,-1) {(a)};
    % EDGES
    \begin{scope}[on background layer]
    \path[draw] (0-0) -- (0-1);
    \path[draw] (0-1) -- (0-2);
    \path[draw] (0-2) -- (0-3);
    \path[draw] (0-3) -- (0-4);
    \path[draw] (0-4) -- (0-5);
    \path[draw] (1-0) -- (1-1);
    \path[draw] (1-1) -- (1-2);
    \path[draw] (1-2) -- (1-3);
    \path[draw] (1-3) -- (1-4);
    \path[draw] (1-4) -- (1-5);
    \path[draw] (2-0) -- (2-1);
    \path[draw] (2-1) -- (2-2);
    \path[draw] (2-2) -- (2-3);
    \path[draw] (2-3) -- (2-4);
    \path[draw] (2-4) -- (2-5);
    \path[draw] (3-0) -- (3-1);
    \path[draw] (3-1) -- (3-2);
    \path[draw] (3-2) -- (3-3);
    \path[draw] (3-3) -- (3-4);
    \path[draw] (3-4) -- (3-5);
    \path[draw] (4-0) -- (4-1);
    \path[draw] (4-1) -- (4-2);
    \path[draw] (4-2) -- (4-3);
    \path[draw] (4-3) -- (4-4);
    \path[draw] (4-4) -- (4-5);
    \path[draw] (5-0) -- (5-1);
    \path[draw] (5-1) -- (5-2);
    \path[draw] (5-2) -- (5-3);
    \path[draw] (5-3) -- (5-4);
    \path[draw] (5-4) -- (5-5);
    \path[draw] (0-0) -- (1-0);
    \path[draw] (1-0) -- (2-0);
    \path[draw] (2-0) -- (3-0);
    \path[draw] (3-0) -- (4-0);
    \path[draw] (4-0) -- (5-0);
    \path[draw] (0-1) -- (1-1);
    \path[draw] (1-1) -- (2-1);
    \path[draw] (2-1) -- (3-1);
    \path[draw] (3-1) -- (4-1);
    \path[draw] (4-1) -- (5-1);
    \path[draw] (0-2) -- (1-2);
    \path[draw] (1-2) -- (2-2);
    \path[draw] (2-2) -- (3-2);
    \path[draw] (3-2) -- (4-2);
    \path[draw] (4-2) -- (5-2);
    \path[draw] (0-3) -- (1-3);
    \path[draw] (1-3) -- (2-3);
    \path[draw] (2-3) -- (3-3);
    \path[draw] (3-3) -- (4-3);
    \path[draw] (4-3) -- (5-3);
    \path[draw] (0-4) -- (1-4);
    \path[draw] (1-4) -- (2-4);
    \path[draw] (2-4) -- (3-4);
    \path[draw] (3-4) -- (4-4);
    \path[draw] (4-4) -- (5-4);
    \path[draw] (0-5) -- (1-5);
    \path[draw] (1-5) -- (2-5);
    \path[draw] (2-5) -- (3-5);
    \path[draw] (3-5) -- (4-5);
    \path[draw] (4-5) -- (5-5);
    \end{scope}
    
    \begin{scope}[xshift=8cm]
    % NODES
    \node[inner sep=\edgesep, mybox] (1-0) at (1,0) {};
    \node[inner sep=\edgesep, mycirc] (2-0) at (2,0) {};
    \node[inner sep=\edgesep, mybox] (3-0) at (3,0) {};
    \node[inner sep=\edgesep, mybox] (5-0) at (5,0) {};
    \node[inner sep=\edgesep, mybox] (0-1) at (0,1) {};
    \node[inner sep=\edgesep, mybox] (2-1) at (2,1) {};
    \node[inner sep=\edgesep, mycirc] (3-1) at (3,1) {};
    \node[inner sep=\edgesep, mybox] (4-1) at (4,1) {};
    \node[inner sep=\edgesep, mycirc] (0-2) at (0,2) {};
    \node[inner sep=\edgesep, mybox] (1-2) at (1,2) {};
    \node[inner sep=\edgesep, mybox] (3-2) at (3,2) {};
    \node[inner sep=\edgesep, mycirc] (4-2) at (4,2) {};
    \node[inner sep=\edgesep, mybox] (5-2) at (5,2) {};
    \node[inner sep=\edgesep, mybox] (0-3) at (0,3) {};
    \node[inner sep=\edgesep, mycirc] (1-3) at (1,3) {};
    \node[inner sep=\edgesep, mybox] (2-3) at (2,3) {};
    \node[inner sep=\edgesep, mybox] (4-3) at (4,3) {};
    \node[inner sep=\edgesep, mycirc] (5-3) at (5,3) {};
    \node[inner sep=\edgesep, mybox] (1-4) at (1,4) {};
    \node[inner sep=\edgesep, mycirc] (2-4) at (2,4) {};
    \node[inner sep=\edgesep, mybox] (3-4) at (3,4) {};
    \node[inner sep=\edgesep, mybox] (5-4) at (5,4) {};
    \node[inner sep=\edgesep, mybox] (0-5) at (0,5) {};
    \node[inner sep=\edgesep, mybox] (2-5) at (2,5) {};
    \node[inner sep=\edgesep, mycirc] (3-5) at (3,5) {};
    \node[inner sep=\edgesep, mybox] (4-5) at (4,5) {};
    \node at (2.5,-1) {(b)};
    % EDGES
    \path[draw] (0-1) -- (0-2);
    \path[draw] (0-2) -- (0-3);
    \path[draw] (1-2) -- (1-3);
    \path[draw] (1-3) -- (1-4);
    \path[draw] (2-0) -- (2-1);
    \path[draw] (2-3) -- (2-4);
    \path[draw] (2-4) -- (2-5);
    \path[draw] (3-0) -- (3-1);
    \path[draw] (3-1) -- (3-2);
    \path[draw] (3-4) -- (3-5);
    \path[draw] (4-1) -- (4-2);
    \path[draw] (4-2) -- (4-3);
    \path[draw] (5-2) -- (5-3);
    \path[draw] (5-3) -- (5-4);
    \path[draw] (1-0) -- (2-0);
    \path[draw] (2-0) -- (3-0);
    \path[draw] (2-1) -- (3-1);
    \path[draw] (3-1) -- (4-1);
    \path[draw] (0-2) -- (1-2);
    \path[draw] (3-2) -- (4-2);
    \path[draw] (4-2) -- (5-2);
    \path[draw] (0-3) -- (1-3);
    \path[draw] (1-3) -- (2-3);
    \path[draw] (4-3) -- (5-3);
    \path[draw] (1-4) -- (2-4);
    \path[draw] (2-4) -- (3-4);
    \path[draw] (2-5) -- (3-5);
    \path[draw] (3-5) -- (4-5);
    \end{scope}
    
    \begin{scope}[xshift=16cm]
    % NODES
    \node[inner sep=\edgesep, mybullet] (0-0) at (0,0) {};
    \node[inner sep=\edgesep, mybox] (1-0) at (1,0) {};
    \node[inner sep=\edgesep, mybox] (3-0) at (3,0) {};
    \node[inner sep=\edgesep, mybullet] (4-0) at (4,0) {};
    \node[inner sep=\edgesep, mybox] (5-0) at (5,0) {};
    \node[inner sep=\edgesep, mybox] (0-1) at (0,1) {};
    \node[inner sep=\edgesep, mybullet] (1-1) at (1,1) {};
    \node[inner sep=\edgesep, mybox] (2-1) at (2,1) {};
    \node[inner sep=\edgesep, mybox] (4-1) at (4,1) {};
    \node[inner sep=\edgesep, mybullet] (5-1) at (5,1) {};
    \node[inner sep=\edgesep, mybox] (1-2) at (1,2) {};
    \node[inner sep=\edgesep, mybullet] (2-2) at (2,2) {};
    \node[inner sep=\edgesep, mybox] (3-2) at (3,2) {};
    \node[inner sep=\edgesep, mybox] (5-2) at (5,2) {};
    \node[inner sep=\edgesep, mybox] (0-3) at (0,3) {};
    \node[inner sep=\edgesep, mybox] (2-3) at (2,3) {};
    \node[inner sep=\edgesep, mybullet] (3-3) at (3,3) {};
    \node[inner sep=\edgesep, mybox] (4-3) at (4,3) {};
    \node[inner sep=\edgesep, mybullet] (0-4) at (0,4) {};
    \node[inner sep=\edgesep, mybox] (1-4) at (1,4) {};
    \node[inner sep=\edgesep, mybox] (3-4) at (3,4) {};
    \node[inner sep=\edgesep, mybullet] (4-4) at (4,4) {};
    \node[inner sep=\edgesep, mybox] (5-4) at (5,4) {};
    \node[inner sep=\edgesep, mybox] (0-5) at (0,5) {};
    \node[inner sep=\edgesep, mybullet] (1-5) at (1,5) {};
    \node[inner sep=\edgesep, mybox] (2-5) at (2,5) {};
    \node[inner sep=\edgesep, mybox] (4-5) at (4,5) {};
    \node[inner sep=\edgesep, mybullet] (5-5) at (5,5) {};
    \node at (2.5,-1) {(c)};
    % EDGES
    \path[draw] (0-0) -- (0-1);
    \path[draw] (0-3) -- (0-4);
    \path[draw] (0-4) -- (0-5);
    \path[draw] (1-0) -- (1-1);
    \path[draw] (1-1) -- (1-2);
    \path[draw] (1-4) -- (1-5);
    \path[draw] (2-1) -- (2-2);
    \path[draw] (2-2) -- (2-3);
    \path[draw] (3-2) -- (3-3);
    \path[draw] (3-3) -- (3-4);
    \path[draw] (4-0) -- (4-1);
    \path[draw] (4-3) -- (4-4);
    \path[draw] (4-4) -- (4-5);
    \path[draw] (5-0) -- (5-1);
    \path[draw] (5-1) -- (5-2);
    \path[draw] (5-4) -- (5-5);
    \path[draw] (0-0) -- (1-0);
    \path[draw] (3-0) -- (4-0);
    \path[draw] (4-0) -- (5-0);
    \path[draw] (0-1) -- (1-1);
    \path[draw] (1-1) -- (2-1);
    \path[draw] (4-1) -- (5-1);
    \path[draw] (1-2) -- (2-2);
    \path[draw] (2-2) -- (3-2);
    \path[draw] (2-3) -- (3-3);
    \path[draw] (3-3) -- (4-3);
    \path[draw] (0-4) -- (1-4);
    \path[draw] (3-4) -- (4-4);
    \path[draw] (4-4) -- (5-4);
    \path[draw] (0-5) -- (1-5);
    \path[draw] (1-5) -- (2-5);
    \path[draw] (4-5) -- (5-5);
    \end{scope}
    
  \end{tikzpicture}

\caption{Incidence graph of a QDNF formula (a). Universal variables are
  drawn with black round shapes, existential variables with white
  round shapes, and terms are drawn with square shapes.
  Both deleting the universal variables (b) and the existential variables (c)
  significantly decreases the treewidth of the incidence graph.}
 \label{fig:grid}
 \vspace{-15pt}
\end{figure}
In the following we pinpoint the exact complexity of checking the
satisfiability of $\exists\forall$-QDNF formulas parameterized by
$\Eitw$ and $\Aitw$.
We let \EAtwASat{} denote the problem \EASat{} parameterized by \Aitw{},
and similarly we let \EAtwESat{} denote the problem \EASat{}
parameterized by \Eitw{}.
We show that
parameterizing by $\Eitw$
does not decrease the complexity and leaves the problem on the second
level of the~PH,
but for \EAtwASat{} we get an fpt-reduction to SAT.

\begin{proposition}
\label{prop:qsat2-twE-fpt}
$\EAtwESat{}$ is \para{\SigmaP{2}}-complete.
\end{proposition}
\ProofEAtwECompleteness{}

\begin{theorem}
\label{thm:qsat2-twA-fpt}
$\EAtwASat{}$ is \para{\NP}-complete.
\end{theorem}
\ProofEAtwACompleteness{}

\smallskip\noindent
Instead of incidence graphs one can also use \emph{primal graphs} to
model the structure of QBF formulas (see,
e.g.,~\cite{AtseriasOliva13,PanVardi06}). One can define corresponding parameters
primal treewidth, universal primal treewidth, and existential primal
treewidth.  The proof of Proposition~\ref{prop:qsat2-twE-fpt} shows
that $\EASat$ is
\para{\SigmaP{2}}-hard when parameterized by existential primal
treewidth. The parameter incidence treewidth is more general than
primal treewidth in the sense that small primal treewidth implies
small incidence treewidth~\cite{KolaitisVardi00,SamerSzeider10a}, but
the converse does not hold in general. Hence,
Theorem~\ref{thm:qsat2-twA-fpt} also holds for the parameter universal
primal treewidth.

%
%\paragraph{Restricting the number of universal variables}
%Next, we consider another parameterization of the satisfiability
%problem of QBFs that allows an fpt-reduction to SAT.
%The result follows straightforwardly from known methods in QBF
%solving, namely universal quantifier expansion \cite{AyariBasin02,Biere04}.
%We show how quantifier expansion can be used
%to get an fpt-reduction to SAT
%in cases where there are only a small number of universally quantified
%variables.
%Consider the following parameterization of the
%quantified Boolean satisfiability problem.
%
%\boxedprob{
%  \FewAQBFSat{}
%  
%  \emph{Instance:} A QBF~$\varphi = \exists X_1. \forall X_2 \dotsc \exists X_{j}. \psi$,
%    where~$j$ is odd (and possibly~$X_1 = \emptyset$ or $X_j = \emptyset$).
%
%  \emph{Parameter:} $\Card{X_2} + \Card{X_4} \dotsc + \Card{X_{j-1}}$.
%
%  \emph{Question:} Is~$\varphi$ true?
%}{\textbf{Complexity:} \para{\NP}-complete
%(Proposition~\ref{prop:qbf-fewa}).}
%
%\begin{proposition}
%\label{prop:qbf-fewa}
%\FewAQBFSat{} is \para{\NP}-complete.
%\end{proposition}
%\begin{proof}[sketch]
%Removing a single universally quantified variable by means of quantifier
%expansion doubles the QBF in size, in the worst case.
%Therefore, repeatedly applying this operation
%to eliminate the~$k$ many universally quantified variables
%leads to an existentially quantified formula
%that is at most a factor of~$2^k$ larger than the original formula.
%%
%A detailed proof of this proposition can be found in the appendix.
%\end{proof}
}

\satversion{
%%% CONCLUSION (SAT)
\section{Conclusion}

%We studied the fpt-reducibility to SAT
%for some problems related to Boolean satisfiability.
%%
%In particular, we studied the problems of minimizing
%an implicant of a DNF formula, minimizing a DNF formula,
%and the satisfiability problem of QBFs with one quantifier
%alternation, under various parameterizations.
%%
%We hope that the foundational work in this paper
%leads to a structured theory
%that captures the possibilities offered by fpt-algorithms
%that can use SAT solvers as a subroutine.
%
%A natural direction for further research is to perform a similar
%analysis for other problems related to Boolean satisfiability
%that lie at the second level of the Polynomial Hierarchy,
%or even higher.
%%
%In particular, one possibility for future work is to investigate how
%the results obtained by bounding the treewidth of the incidence graphs
%can be extended to QBFs with more quantifier alternations.
%%
%It would also be interesting to find out the actual parameter values
%for instances of the discussed problems that occur in
%practical settings.

We studied the fpt-reducibility to SAT for several problems beyond NP
under natural parameters.  Our positive results show that in some
cases it is possible to utilize structure in terms of parameters to
break through the barriers between classical complexity classes.
Parameters that admit an fpt-reduction to SAT can be less restrictive
than parameters that provide fixed-parameter tractability of the
problem itself, hence our approach extends the scope of
fixed-parameter tractability.  Additionally, we show that fpt-time
algorithms that can query a SAT solver (i.e., fpt-time Turing
reductions to SAT) exist for some problems that cannot be solved in
polynomial-time with the help of queries to a SAT solver (unless the
PH collapses), hence our approach also extends the scope of algorithms
using SAT queries.  Our negative results point out the limits of the
approach, showing that some problems do, most likely, not admit
fpt-reductions to SAT under certain natural parameters.
}

%%%
%%% BIBLIOGRAPHY (SAT VERSION)
%%%
\satversion{
  \vfill
  \pagebreak\clearpage
  \bibliographystyle{splncs03}
  \DeclareRobustCommand{\DE}[3]{#3}
  %\bibliography{literature}

}

\krlongversion{
\longversion{
%%% APPENDIX: PROBLEMS COMPLETE FOR \EkA
\pagebreak
\appendix
\section{Compendium}
\label{sec:appendix-a}

In this compendium, we give an overview of a number of parameterized complexity results
based on classical decision problems that lie at higher levels of the polynomial hierarchy.
This illustrates the applicability of the hardness theory developed in this paper.
We group the results by the type of problems.
To provide a systematic and concise overview, the compendium also
contains the results established in the main text of the paper,
and proofs for new results can be found in Appendix B below.

%%% LOGIC PROBLEMS
\subsection{Logic Problems}

We start with the quantified circuit satisfiability
problems on which the \kstar{} and \stark{} hierarchies are based.
We present only a two canonical forms of the problems in the \kstar{} hierarchy.
For problems in the \stark{} hierarchy, we let~$\CCC$ range over classes of Boolean circuits.\\

\boxedprob{
  $\EkAWSat{}$
  
  \emph{Instance:} A quantified Boolean
  formula~$\phi = \exists X. \forall Y. \psi$,
  and an integer~$k$

  \emph{Parameter:} $k$.

  \emph{Question:} Does there exist an assignment~$\alpha$
  to~$X$ with weight~$k$,
  such that~$\forall Y. \psi[\alpha]$ evaluates to true?
}{\textbf{Complexity:} \EkA-complete (Theorem~\ref{thm:kstar-collapse}).}

\boxedprob{
  $\EkAWSat(3\DNF)$
  
  \emph{Instance:} A quantified Boolean
  formula~$\phi = \exists X. \forall Y. \psi$ with~$\psi \in 3\DNF$,
  and an integer~$k$

  \emph{Parameter:} $k$.

  \emph{Question:} Does there exist an assignment~$\alpha$
  to~$X$ with weight~$k$,
  such that~$\forall Y. \psi[\alpha]$ evaluates to true?
}{\textbf{Complexity:} \EkA-complete (Theorem~\ref{thm:kstar-collapse}).}

\boxedprob{
  $\EAkWSat(\CCC)$
  
  \emph{Instance:} A Boolean circuit~$C \in \CCC$ over two disjoint
  sets~$X$ and~$Y$ of variables,
  and an integer~$k$

  \emph{Parameter:} $k$.

  \emph{Question:} Does there exist an assignment~$\alpha$ to~$X$,
  such that for all assignments~$\beta$ to~$Y$ of weight~$k$
  the assignment~$\alpha \cup \beta$ satisfies~$C$?
}{\textbf{Complexity:}\\
\EAkW{t}-complete when restricted to circuits of weft~$t$, for any~$t \geq 1$ (by definition);\\
\EAkW{SAT}-complete if~$\CCC = \Phi$ (by definition);\\
\EAkW{P}-complete if~$\CCC = \Gamma$ (by definition).}

%% BOOLEAN SATISFIABILITY 1
\subsubsection{Weighted Quantified Boolean Satisfiability for \kstar{}}

\hspace{0pt}

\boxedprob{
  $\EleqkAWSat{}$
  
  \emph{Instance:} A quantified Boolean
  formula~$\phi = \exists X. \forall Y. \psi$,
  and an integer~$k$

  \emph{Parameter:} $k$.

  \emph{Question:} Does there exist an assignment~$\alpha$
  to~$X$ with weight at most~$k$,
  such that~$\forall Y. \psi[\alpha]$ evaluates to true?
}{\textbf{Complexity:} \EkA-complete
(Propositions~\ref{prop:kstar-leqk-hardness}~and~\ref{prop:kstar-leqk-membership}).}

\boxedprob{
  \EnminkAWSat{}
  
  \emph{Instance:} A quantified Boolean
  formula~$\phi = \exists X. \forall Y. \psi$,
  and an integer~$k$

  \emph{Parameter:} $k$.

  \emph{Question:} Does there exist an assignment~$\alpha$
  to~$X$ with weight~$\Card{X}-k$,
  such that~$\forall Y. \psi[\alpha]$ evaluates to true?
}{\textbf{Complexity:} \EkA-complete (Proposition~\ref{prop:kstar-nmink}).}

\boxedprob{
  $\EgeqkAWSat{}$
  
  \emph{Instance:} A quantified Boolean
  formula~$\phi = \exists X. \forall Y. \psi$,
  and an integer~$k$

  \emph{Parameter:} $k$.

  \emph{Question:} Does there exist an assignment~$\alpha$
  to~$X$ with weight at least~$k$,
  such that~$\forall Y. \psi[\alpha]$ evaluates to true?
}{\textbf{Complexity:} \para{}\SigmaP{2}-complete
(Propositions~\ref{prop:kstar-geqk-hardness}~and~\ref{prop:kstar-geqk-membership}).}

%% BOOLEAN SATISFIABILITY 2
\subsubsection{Weighted Quantified Boolean Satisfiability in the \stark{} Hierarchy}

\hspace{0pt}

\boxedprob{
  $\EAkWSat(2\DNF)$
  
  \emph{Instance:} A quantified Boolean
  formula~$\varphi = \exists X. \forall Y. \psi$ with~$\psi \in 2\DNF$,
  and an integer~$k$

  \emph{Parameter:} $k$.

  \emph{Question:} Does there exist an assignment~$\alpha$ to~$X$,
  such that for all assignments~$\beta$ to~$Y$ of weight~$k$
  the assignment~$\alpha \cup \beta$ satisfies~$\psi$?
}{\textbf{Complexity:} \EAkW{1}-complete
(Theorem~\ref{thm:stark-2cnf-hardness} and Corollary~\ref{cor:stark-2cnf-hardness}).}

\boxedprob{
  $\EAkWSat(\forall\mtext{-monotone})$
  
  \emph{Instance:} A Boolean circuit~$C \in \CCC$ over two disjoint
  sets~$X$ and~$Y$ of variables that is in negation normal form
  and that is monotone in~$Y$,
  and an integer~$k$

  \emph{Parameter:} $k$.

  \emph{Question:} Does there exist an assignment~$\alpha$ to~$X$,
  such that for all assignments~$\beta$ to~$Y$ of weight~$k$
  the assignment~$\alpha \cup \beta$ satisfies~$C$?
}{\textbf{Complexity:} \EAkW{P}-complete
(Proposition~\ref{prop:eakwp-monotone}).}

\noindent Let~$C$ be a quantified Boolean circuit over
two disjoint sets~$X$ and~$Y$ of variables
that is in negation normal form.
We say that~$C$ is \emph{anti-monotone in the variables~$Y$} 
if the only nodes having variables in~$Y$ as input
are negation nodes.
i.e., the variables in~$Y$ can appear only negatively in the circuit.\\

\boxedprob{
  $\EAkWSat(\forall\mtext{-anti-monotone})$
  
  \emph{Instance:} A Boolean circuit~$C \in \CCC$ over two disjoint
  sets~$X$ and~$Y$ of variables that is in negation normal
  form and that is anti-monotone in~$Y$,
  and an integer~$k$

  \emph{Parameter:} $k$.

  \emph{Question:} Does there exist an assignment~$\alpha$ to~$X$,
  such that for all assignments~$\beta$ to~$Y$ of weight~$k$
  the assignment~$\alpha \cup \beta$ satisfies~$C$?
}{\textbf{Complexity:} \EAkW{P}-complete
(Proposition~\ref{prop:eakwp-antimonotone}).}

%% BOOLEAN SATISFIABILITY 3
\subsubsection{Quantified Boolean Satisfiability with Bounded Treewidth}

%If~$\psi$ is a DNF formula,~$Z \subseteq \Var{\psi}$ is a subset of variables,
%and~$(\mathcal{T},(B_t)_{t \in T})$ is a tree decomposition
%of~$\mtext{IG}(Z,\psi)$, we let~$\Var{t}$ denote~$B_t \cap Z$,
%for any~$t \in T$.

\PreliminariesTreewidth{}

Let~$\psi = \delta_1 \vee \dotsm \vee \delta_u$ be a \DNF{} formula.
For any subset~$Z \subseteq \Var{\psi}$ of variables, we define
\emph{the incidence graph~$\mtext{IG}(Z,\psi)$
of~$\psi$ with respect to~$Z$}
to be the graph~$\mtext{IG}(Z,\psi) = (V,E)$,
where~$V = Z \cup \SBs \delta_1,\dotsc,\delta_u \SEs$
and~$E = \SB \SBs \delta_j, z \SEs \SM 1 \leq j \leq u, z \in Z,
\mtext{$z$ occurs in the clause~$\delta_j$} \SE$.

One way of capturing structure in an instance~$\varphi =
\exists X. \forall Y. \psi$ of \QSat{2}
is to take the treewidth of the incidence graph (the \emph{incidence
treewidth}) of~$\psi$ (with respect to all variables)
as a parameter.
However, as this is very restrictive, one could alternatively consider
taking as parameter the incidence treewidth with respect to
certain subsets of the variables.
In order to do so, we define the \emph{existential incidence treewidth}
of~$\varphi$
as the treewidth of the incidence graph of~$\psi$
with respect to the set~$X$ of existential variables. 
The \emph{universal incidence treewidth}
is the treewidth of the incidence graph of~$\psi$
with respect to the set~$Y$ of universal variables.

The existential and universal treewidth can be small for formulas
whose incidence treewidth is arbitrarily large.  Take for instance an
instance of~\QSat{2}
whose incidence graph is an $n\times n$ square
grid, as in Figure~\ref{fig:grid}.
In this example, both the existential and the universal incidence
treewidth are~$2$
(since after the deletion of the universal or the existential variables the
incidence graph becomes a collection of trivial path-like graphs), but
the incidence treewidth is~$n$~\cite{Bodlaender98}.
Hence, a tractability
result in terms of existential or universal incidence treewidth would
apply to a significantly larger class of instances than a tractability
result in terms of the incidence treewidth.
\begin{figure}[h!]
\centering
  \begin{tikzpicture}[scale=.5]
    \newcommand{\edgesep}[0]{-0mm}
    \tikzstyle{mybox}=[draw, rectangle, text width=3pt, text height=3pt]
    \tikzstyle{mybullet}=[draw, circle, text width=3pt, text height=3pt, fill=black]
    \tikzstyle{mycirc}=[draw, circle, text width=3pt, text height=3pt, fill=black!20]

    % NODES
    \node[inner sep=\edgesep, mybullet] (0-0) at (0,0) {};
    \node[inner sep=\edgesep, mybox] (1-0) at (1,0) {};
    \node[inner sep=\edgesep, mycirc] (2-0) at (2,0) {};
    \node[inner sep=\edgesep, mybox] (3-0) at (3,0) {};
    \node[inner sep=\edgesep, mybullet] (4-0) at (4,0) {};
    \node[inner sep=\edgesep, mybox] (5-0) at (5,0) {};
    \node[inner sep=\edgesep, mybox] (0-1) at (0,1) {};
    \node[inner sep=\edgesep, mybullet] (1-1) at (1,1) {};
    \node[inner sep=\edgesep, mybox] (2-1) at (2,1) {};
    \node[inner sep=\edgesep, mycirc] (3-1) at (3,1) {};
    \node[inner sep=\edgesep, mybox] (4-1) at (4,1) {};
    \node[inner sep=\edgesep, mybullet] (5-1) at (5,1) {};
    \node[inner sep=\edgesep, mycirc] (0-2) at (0,2) {};
    \node[inner sep=\edgesep, mybox] (1-2) at (1,2) {};
    \node[inner sep=\edgesep, mybullet] (2-2) at (2,2) {};
    \node[inner sep=\edgesep, mybox] (3-2) at (3,2) {};
    \node[inner sep=\edgesep, mycirc] (4-2) at (4,2) {};
    \node[inner sep=\edgesep, mybox] (5-2) at (5,2) {};
    \node[inner sep=\edgesep, mybox] (0-3) at (0,3) {};
    \node[inner sep=\edgesep, mycirc] (1-3) at (1,3) {};
    \node[inner sep=\edgesep, mybox] (2-3) at (2,3) {};
    \node[inner sep=\edgesep, mybullet] (3-3) at (3,3) {};
    \node[inner sep=\edgesep, mybox] (4-3) at (4,3) {};
    \node[inner sep=\edgesep, mycirc] (5-3) at (5,3) {};
    \node[inner sep=\edgesep, mybullet] (0-4) at (0,4) {};
    \node[inner sep=\edgesep, mybox] (1-4) at (1,4) {};
    \node[inner sep=\edgesep, mycirc] (2-4) at (2,4) {};
    \node[inner sep=\edgesep, mybox] (3-4) at (3,4) {};
    \node[inner sep=\edgesep, mybullet] (4-4) at (4,4) {};
    \node[inner sep=\edgesep, mybox] (5-4) at (5,4) {};
    \node[inner sep=\edgesep, mybox] (0-5) at (0,5) {};
    \node[inner sep=\edgesep, mybullet] (1-5) at (1,5) {};
    \node[inner sep=\edgesep, mybox] (2-5) at (2,5) {};
    \node[inner sep=\edgesep, mycirc] (3-5) at (3,5) {};
    \node[inner sep=\edgesep, mybox] (4-5) at (4,5) {};
    \node[inner sep=\edgesep, mybullet] (5-5) at (5,5) {};
    \node at (2.5,-1) {(a)};
    % EDGES
    \begin{scope}[on background layer]
    \path[draw] (0-0) -- (0-1);
    \path[draw] (0-1) -- (0-2);
    \path[draw] (0-2) -- (0-3);
    \path[draw] (0-3) -- (0-4);
    \path[draw] (0-4) -- (0-5);
    \path[draw] (1-0) -- (1-1);
    \path[draw] (1-1) -- (1-2);
    \path[draw] (1-2) -- (1-3);
    \path[draw] (1-3) -- (1-4);
    \path[draw] (1-4) -- (1-5);
    \path[draw] (2-0) -- (2-1);
    \path[draw] (2-1) -- (2-2);
    \path[draw] (2-2) -- (2-3);
    \path[draw] (2-3) -- (2-4);
    \path[draw] (2-4) -- (2-5);
    \path[draw] (3-0) -- (3-1);
    \path[draw] (3-1) -- (3-2);
    \path[draw] (3-2) -- (3-3);
    \path[draw] (3-3) -- (3-4);
    \path[draw] (3-4) -- (3-5);
    \path[draw] (4-0) -- (4-1);
    \path[draw] (4-1) -- (4-2);
    \path[draw] (4-2) -- (4-3);
    \path[draw] (4-3) -- (4-4);
    \path[draw] (4-4) -- (4-5);
    \path[draw] (5-0) -- (5-1);
    \path[draw] (5-1) -- (5-2);
    \path[draw] (5-2) -- (5-3);
    \path[draw] (5-3) -- (5-4);
    \path[draw] (5-4) -- (5-5);
    \path[draw] (0-0) -- (1-0);
    \path[draw] (1-0) -- (2-0);
    \path[draw] (2-0) -- (3-0);
    \path[draw] (3-0) -- (4-0);
    \path[draw] (4-0) -- (5-0);
    \path[draw] (0-1) -- (1-1);
    \path[draw] (1-1) -- (2-1);
    \path[draw] (2-1) -- (3-1);
    \path[draw] (3-1) -- (4-1);
    \path[draw] (4-1) -- (5-1);
    \path[draw] (0-2) -- (1-2);
    \path[draw] (1-2) -- (2-2);
    \path[draw] (2-2) -- (3-2);
    \path[draw] (3-2) -- (4-2);
    \path[draw] (4-2) -- (5-2);
    \path[draw] (0-3) -- (1-3);
    \path[draw] (1-3) -- (2-3);
    \path[draw] (2-3) -- (3-3);
    \path[draw] (3-3) -- (4-3);
    \path[draw] (4-3) -- (5-3);
    \path[draw] (0-4) -- (1-4);
    \path[draw] (1-4) -- (2-4);
    \path[draw] (2-4) -- (3-4);
    \path[draw] (3-4) -- (4-4);
    \path[draw] (4-4) -- (5-4);
    \path[draw] (0-5) -- (1-5);
    \path[draw] (1-5) -- (2-5);
    \path[draw] (2-5) -- (3-5);
    \path[draw] (3-5) -- (4-5);
    \path[draw] (4-5) -- (5-5);
    \end{scope}
    
    \begin{scope}[xshift=8cm]
    % NODES
    \node[inner sep=\edgesep, mybox] (1-0) at (1,0) {};
    \node[inner sep=\edgesep, mycirc] (2-0) at (2,0) {};
    \node[inner sep=\edgesep, mybox] (3-0) at (3,0) {};
    \node[inner sep=\edgesep, mybox] (5-0) at (5,0) {};
    \node[inner sep=\edgesep, mybox] (0-1) at (0,1) {};
    \node[inner sep=\edgesep, mybox] (2-1) at (2,1) {};
    \node[inner sep=\edgesep, mycirc] (3-1) at (3,1) {};
    \node[inner sep=\edgesep, mybox] (4-1) at (4,1) {};
    \node[inner sep=\edgesep, mycirc] (0-2) at (0,2) {};
    \node[inner sep=\edgesep, mybox] (1-2) at (1,2) {};
    \node[inner sep=\edgesep, mybox] (3-2) at (3,2) {};
    \node[inner sep=\edgesep, mycirc] (4-2) at (4,2) {};
    \node[inner sep=\edgesep, mybox] (5-2) at (5,2) {};
    \node[inner sep=\edgesep, mybox] (0-3) at (0,3) {};
    \node[inner sep=\edgesep, mycirc] (1-3) at (1,3) {};
    \node[inner sep=\edgesep, mybox] (2-3) at (2,3) {};
    \node[inner sep=\edgesep, mybox] (4-3) at (4,3) {};
    \node[inner sep=\edgesep, mycirc] (5-3) at (5,3) {};
    \node[inner sep=\edgesep, mybox] (1-4) at (1,4) {};
    \node[inner sep=\edgesep, mycirc] (2-4) at (2,4) {};
    \node[inner sep=\edgesep, mybox] (3-4) at (3,4) {};
    \node[inner sep=\edgesep, mybox] (5-4) at (5,4) {};
    \node[inner sep=\edgesep, mybox] (0-5) at (0,5) {};
    \node[inner sep=\edgesep, mybox] (2-5) at (2,5) {};
    \node[inner sep=\edgesep, mycirc] (3-5) at (3,5) {};
    \node[inner sep=\edgesep, mybox] (4-5) at (4,5) {};
    \node at (2.5,-1) {(b)};
    % EDGES
    \path[draw] (0-1) -- (0-2);
    \path[draw] (0-2) -- (0-3);
    \path[draw] (1-2) -- (1-3);
    \path[draw] (1-3) -- (1-4);
    \path[draw] (2-0) -- (2-1);
    \path[draw] (2-3) -- (2-4);
    \path[draw] (2-4) -- (2-5);
    \path[draw] (3-0) -- (3-1);
    \path[draw] (3-1) -- (3-2);
    \path[draw] (3-4) -- (3-5);
    \path[draw] (4-1) -- (4-2);
    \path[draw] (4-2) -- (4-3);
    \path[draw] (5-2) -- (5-3);
    \path[draw] (5-3) -- (5-4);
    \path[draw] (1-0) -- (2-0);
    \path[draw] (2-0) -- (3-0);
    \path[draw] (2-1) -- (3-1);
    \path[draw] (3-1) -- (4-1);
    \path[draw] (0-2) -- (1-2);
    \path[draw] (3-2) -- (4-2);
    \path[draw] (4-2) -- (5-2);
    \path[draw] (0-3) -- (1-3);
    \path[draw] (1-3) -- (2-3);
    \path[draw] (4-3) -- (5-3);
    \path[draw] (1-4) -- (2-4);
    \path[draw] (2-4) -- (3-4);
    \path[draw] (2-5) -- (3-5);
    \path[draw] (3-5) -- (4-5);
    \end{scope}
    
    \begin{scope}[xshift=16cm]
    % NODES
    \node[inner sep=\edgesep, mybullet] (0-0) at (0,0) {};
    \node[inner sep=\edgesep, mybox] (1-0) at (1,0) {};
    \node[inner sep=\edgesep, mybox] (3-0) at (3,0) {};
    \node[inner sep=\edgesep, mybullet] (4-0) at (4,0) {};
    \node[inner sep=\edgesep, mybox] (5-0) at (5,0) {};
    \node[inner sep=\edgesep, mybox] (0-1) at (0,1) {};
    \node[inner sep=\edgesep, mybullet] (1-1) at (1,1) {};
    \node[inner sep=\edgesep, mybox] (2-1) at (2,1) {};
    \node[inner sep=\edgesep, mybox] (4-1) at (4,1) {};
    \node[inner sep=\edgesep, mybullet] (5-1) at (5,1) {};
    \node[inner sep=\edgesep, mybox] (1-2) at (1,2) {};
    \node[inner sep=\edgesep, mybullet] (2-2) at (2,2) {};
    \node[inner sep=\edgesep, mybox] (3-2) at (3,2) {};
    \node[inner sep=\edgesep, mybox] (5-2) at (5,2) {};
    \node[inner sep=\edgesep, mybox] (0-3) at (0,3) {};
    \node[inner sep=\edgesep, mybox] (2-3) at (2,3) {};
    \node[inner sep=\edgesep, mybullet] (3-3) at (3,3) {};
    \node[inner sep=\edgesep, mybox] (4-3) at (4,3) {};
    \node[inner sep=\edgesep, mybullet] (0-4) at (0,4) {};
    \node[inner sep=\edgesep, mybox] (1-4) at (1,4) {};
    \node[inner sep=\edgesep, mybox] (3-4) at (3,4) {};
    \node[inner sep=\edgesep, mybullet] (4-4) at (4,4) {};
    \node[inner sep=\edgesep, mybox] (5-4) at (5,4) {};
    \node[inner sep=\edgesep, mybox] (0-5) at (0,5) {};
    \node[inner sep=\edgesep, mybullet] (1-5) at (1,5) {};
    \node[inner sep=\edgesep, mybox] (2-5) at (2,5) {};
    \node[inner sep=\edgesep, mybox] (4-5) at (4,5) {};
    \node[inner sep=\edgesep, mybullet] (5-5) at (5,5) {};
    \node at (2.5,-1) {(c)};
    % EDGES
    \path[draw] (0-0) -- (0-1);
    \path[draw] (0-3) -- (0-4);
    \path[draw] (0-4) -- (0-5);
    \path[draw] (1-0) -- (1-1);
    \path[draw] (1-1) -- (1-2);
    \path[draw] (1-4) -- (1-5);
    \path[draw] (2-1) -- (2-2);
    \path[draw] (2-2) -- (2-3);
    \path[draw] (3-2) -- (3-3);
    \path[draw] (3-3) -- (3-4);
    \path[draw] (4-0) -- (4-1);
    \path[draw] (4-3) -- (4-4);
    \path[draw] (4-4) -- (4-5);
    \path[draw] (5-0) -- (5-1);
    \path[draw] (5-1) -- (5-2);
    \path[draw] (5-4) -- (5-5);
    \path[draw] (0-0) -- (1-0);
    \path[draw] (3-0) -- (4-0);
    \path[draw] (4-0) -- (5-0);
    \path[draw] (0-1) -- (1-1);
    \path[draw] (1-1) -- (2-1);
    \path[draw] (4-1) -- (5-1);
    \path[draw] (1-2) -- (2-2);
    \path[draw] (2-2) -- (3-2);
    \path[draw] (2-3) -- (3-3);
    \path[draw] (3-3) -- (4-3);
    \path[draw] (0-4) -- (1-4);
    \path[draw] (3-4) -- (4-4);
    \path[draw] (4-4) -- (5-4);
    \path[draw] (0-5) -- (1-5);
    \path[draw] (1-5) -- (2-5);
    \path[draw] (4-5) -- (5-5);
    \end{scope}
    
  \end{tikzpicture}

\caption{Incidence graph of an instance of \QSat{2} (a).
  Universal variables are
  drawn with black round shapes, existential variables with gray
  round shapes, and terms are drawn with square shapes.
  Both restricting to the the existential variables (b)
  and to the universal variables (c)
  significantly decreases the treewidth of the incidence graph.}
 \label{fig:grid}
\end{figure}
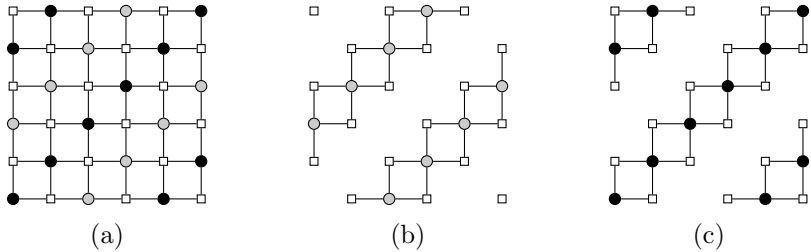

The following parameterized decision problems
are variants of \QSat{2},
where the treewidth of the incidence graph
with respect to certain subsets of variables is bounded.\\

\boxedprob{
  $\EAtwSat{}$
  
  \emph{Instance:} A quantified Boolean
  formula~$\varphi = \exists X. \forall Y. \psi$,
  with~$\psi$ in \DNF{},
  and an tree decomposition~$(\mathcal{T},(B_t)_{t \in T})$
  of~$\mtext{IG}(X \cup Y, \psi)$
  of width~$k$.

  \emph{Parameter:} $k$.

  \emph{Question:} Is~$\varphi$ satisfiable?
}{\textbf{Complexity:} fixed-parameter tractable
\cite{Chen04,FederKolaitis06}.}

\boxedprob{
  $\EAtwESat{}$
  
  \emph{Instance:} A quantified Boolean
  formula~$\varphi = \exists X. \forall Y. \psi$,
  with~$\psi$ in \DNF{},
  and an tree decomposition~$(\mathcal{T},(B_t)_{t \in T})$
  of~$\mtext{IG}(X, \psi)$
  of width~$k$.

  \emph{Parameter:} $k$.

  \emph{Question:} Is~$\varphi$ satisfiable?
}{\textbf{Complexity:} \para{\SigmaP{2}}-complete
(Proposition~\ref{prop:qsat2-twE-fpt}).}

\boxedprob{
  $\EAtwASat{}$
  
  \emph{Instance:} A quantified Boolean
  formula~$\varphi = \exists X. \forall Y. \psi$,
  with~$\psi$ in \DNF{},
  and an tree decomposition~$(\mathcal{T},(B_t)_{t \in T})$
  of~$\mtext{IG}(Y, \psi)$
  of width~$k$.

  \emph{Parameter:} $k$.

  \emph{Question:} Is~$\varphi$ satisfiable?
}{\textbf{Complexity:} \para{\NP}-complete
(Proposition~\ref{prop:qsat2-twA-fpt}).}

%% DNF MINIMIZATION
\subsubsection{DNF Minimization}

Let~$\varphi$ be a propositional formula in \DNF{}.
We say that a set~$C$ of literals is an \emph{implicant of~$\varphi$}
if all assignments that satisfy~$\bigwedge_{l \in C} l$
also satisfy~$\varphi$.
Moreover, we say that a DNF
formula~$\varphi'$ is a \emph{term-wise subformula}
of~$\varphi'$ if for all terms~$t' \in \varphi'$
there exists a term~$t \in \varphi$ such that~$t' \subseteq t$.
The following parameterized problems
are natural parameterizations of problems
shown to be \SigmaP{2}-complete by Umans~\cite{Umans00}.\\

\boxedprob{
  \SSIC{}
  
  \emph{Instance:} A DNF formula~$\varphi$,
  an implicant~$C$ of~$\varphi$, and an integer~$k$.

  \emph{Parameter:} $k$.

  \emph{Question:} Does there exists an implicant~$C' \subseteq C$
  of~$\varphi$ of size~$k$?
}{\textbf{Complexity:} \EkA{}-complete
(Propositions~\ref{prop:smallshortestimplicant-hardness}~and~\ref{prop:smallshortestimplicant-membership}).}

\boxedprob{
  \LSIC{}
  
  \emph{Instance:} A DNF formula~$\varphi$,
  an implicant~$C$ of~$\varphi$ of size~$n$, and an integer~$k$.

  \emph{Parameter:} $k$.

  \emph{Question:} Does there exists an implicant~$C' \subseteq C$
  of~$\varphi$ of size~$n-k$?
}{\textbf{Complexity:} \EkA{}-complete
(Propositions~\ref{prop:largeshortestimplicant-hardness}~and~\ref{prop:largeshortestimplicant-membership}).}

\boxedprob{
  \LargeMinDNF{}
  
  \emph{Instance:} A \DNF{} formula~$\varphi$
  of size~$n$,
  and an integer~$k$.

  \emph{Parameter:} $k$.

  \emph{Question:} Does there exist a term-wise subformula~$\varphi'$
  of~$\varphi$ of size~$n-k$ such
  that~$\varphi \equiv \varphi'$?
}{\textbf{Complexity:} \EkA{}-complete
(Propositions~\ref{prop:largemindnf-hardness}~and~\ref{prop:largemindnf-membership}).}

\boxedprob{
  \SmallMinDNF{}
  
  \emph{Instance:} A \DNF{} formula~$\varphi$
  of size~$n$,
  and an integer~$k$.

  \emph{Parameter:} $k$.

  \emph{Question:} Does there exist an \DNF{}
  formula~$\varphi'$ of size~$k$,
  such that~$\varphi \equiv \varphi'$?
}{\textbf{Complexity:} \para{\co\NP}-hard
(Proposition~\ref{prop:smallmindnf-hardness}),
solvable in fpt-time using~$k+1$ many SAT calls
(Proposition~\ref{prop:smallmindnf-membership2}),
and in \EkA{} (Proposition~\ref{prop:smallmindnf-membership1}).}

%%% KR PROBLEMS
\subsection{Knowledge Representation and Reasoning Problems}

%% DISJUNCTIVE ANSWER SET PROGRAMMING
\subsubsection{Disjunctive Answer Set Programming}

The following problems from the setting of disjunctive answer set programming
are based on the notions of disjunctive logic programs
and answer sets for such programs.
For definitions of these and other notions
that are used in the problem statements below,
we refer to Section~\ref{sec:introduction}.\\

\boxedprob{
  \ASPcons\ContAtoms{}
  
  \emph{Instance:} A disjunctive logic program~$P$.

  \emph{Parameter:} The number of contingent atoms of~$P$.

  \emph{Question:} Does~$P$ have an answer set?
}{\textbf{Complexity:} \para{\co\NP}-complete
(Proposition~\ref{prop:asp-contatoms-completeness}).}

\boxedprob{
  \ASPcons\ContRules{}
  
  \emph{Instance:} A disjunctive logic program~$P$.

  \emph{Parameter:} The number of contingent rules of~$P$.

  \emph{Question:} Does~$P$ have an answer set?
}{\textbf{Complexity:} \EkA{}-complete
(Theorem~\ref{thm:asp-contrules-completeness}).}

\boxedprob{
  \ASPcons\DisjRules{}
  
  \emph{Instance:} A disjunctive logic program~$P$.

  \emph{Parameter:} The number of disjunctive rules of~$P$.

  \emph{Question:} Does~$P$ have an answer set?
}{\textbf{Complexity:} \EAkW{P}-complete (Theorem~\ref{thm:asp-disjrules}).}

\boxedprob{
  \ASPcons\NonDualNormalRules{}
  
  \emph{Instance:} A disjunctive logic program~$P$.

  \emph{Parameter:} The number of non-dual-normal rules of~$P$.

  \emph{Question:} Does~$P$ have an answer set?
}{\textbf{Complexity:} \EAkW{P}-complete (Theorem~\ref{thm:asp-nondualnormalrules}).}

\boxedprob{
  \ASPcons\SNbd{}
  
  \emph{Instance:} A disjunctive logic program~$P$.

  \emph{Parameter:} The size of the smallest normality-backdoor for~$P$.

  \emph{Question:} Does~$P$ have an answer set?
}{\textbf{Complexity:} \para{\NP}-complete
\cite{FichteSzeider13}.}

\boxedprob{
  \ASPcons\AtomOcc{}
  
  \emph{Instance:} A disjunctive logic program~$P$.

  \emph{Parameter:} The maximum number of times that any atom occurs in~$P$.

  \emph{Question:} Does~$P$ have an answer set?
}{\textbf{Complexity:} \para{\SigmaP{2}}-complete
(Corollary~\ref{cor:asp-varocc-completeness}).}

%% ROBUST CONSTRAINT SATISFACTION
\subsubsection{Robust Constraint Satisfaction}

The following reasoning problem originates in the domain
of knowledge representation.
We consider the class of robust constraint satisfaction problems,
introduced recently by
Gottlob~\cite{Gottlob12} and
Abramsky, Gottlob and Kolaitis~\cite{AbramskyGottlobKolaitis13}.
These problems are concerned with the question of whether every partial assignment
of a particular size can be extended to a full solution,
in the setting of constraint satisfaction problems.
%As we will see, a natural parameterized variant of this class of problems
%is complete for the class \AkE{}.

A \emph{CSP instance}~$N$
is a triple~$(X,D,C)$, where~$X$ is a finite set of \emph{variables},
the \emph{domain}~$D$ is a finite set of \emph{values}, and~$C$ is a finite
set of \emph{constraints}. Each constraint~$c \in C$
is a pair~$(S,R)$, where~$S = \Var{c}$, the \emph{constraint scope},
is a finite sequence of distinct variables from~$X$,
and~$R$, the \emph{constraint relation}, is a relation over~$D$
whose arity matches the length of~$S$,
i.e.,~$R \subseteq D^r$ where~$r$ is the length of~$S$.

Let~$N = (X,D,C)$ be a CSP instance. A \emph{partial instantiation}
of~$N$ is a mapping~$\alpha : X' \rightarrow D$
defined on some subset~$X' \subseteq X$.
We say that~$\alpha$ \emph{satisfies} a
constraint~$c = ((x_1,\dotsc,x_r),R) \in C$
if~$\Var{c} \subseteq X'$ and~$(\alpha(x_1),\dotsc,\alpha(x_r)) \in R$.
If~$\alpha$ satisfies all constraints of~$N$ then it is a \emph{solution} of~$N$.
We say that~$\alpha$ \emph{violates} a
constraint~$c = ((x_1,\dotsc,x_r),R) \in C$
if there is no extension~$\beta$ of~$\alpha$ defined
on~$X' \cup \Var{c}$ such that~$(\beta(x_1),\dotsc,\beta(x_r)) \in R$.

Let~$k$ be a nonnegative integer.
We say that a CSP instance~$N = (X,D,C)$ is \emph{$k$-robustly satisfiable}
if for each instantiation~$\alpha : X' \rightarrow D$
defined on some subset~$X' \subseteq X$ of~$k$ many variables
(i.e.,~$\Card{X'} = k$) that does not violate any constraint in~$C$,
it holds that~$\alpha$ can be extended to
a solution for the CSP instance~$(X,D,C)$.
The following parameterized problem is \AkE{}-complete.

\boxedprob{
  \RobustCSPSat{}
  
  \emph{Instance:} A CSP instance~$(X,D,C)$,
  and an integer~$k$.

  \emph{Parameter:} $k$.

  \emph{Question:} Is~$(X,D,C)$~$k$-robustly satisfiable?
}{\textbf{Complexity:} \AkE-complete
(Theorems~\ref{thm:robustcsp-membership}~and~\ref{thm:robustcsp-hardness}).}

%%% GRAPH PROBLEMS
\subsection{Graph Problems}

%% CLIQUE EXTENSIONS
\subsubsection{Clique Extensions}

Let~$G = (V,E)$ be a graph.
A clique~$C \subseteq V$ of~$G$ is a subset of vertices
that induces a complete subgraph of~$G$,
i.e.~$\SBs v,v' \SEs \in E$ for all~$v,v' \in C$ such that~$v \neq v'$.
The \W{1}-complete problem of determining whether a graph has
a clique of size~$k$
is an important problem in the W-hierarchy,
and is used in many \W{1}-hardness proofs.\\

\boxedprob{
  \SmallCliqueExt{}
  
  \emph{Instance:} A graph~$G = (V,E)$,
  a subset~$V' \subseteq V$,
  and an integer~$k$.

  \emph{Parameter:} $k$.

  \emph{Question:} Is it the case that for each
  clique~$C \subseteq V'$, there is some~$k$-clique~$D$ of~$G$
  such that~$C \cup D$ is a~$(\Card{C}+k)$-clique?
}{\textbf{Complexity:} \AEkW{1}-complete
(Propositions~\ref{prop:smallcliqueext-hardness}~and~\ref{prop:smallcliqueext-membership}).}

%%% GRAPH COLORING EXTENSIONS
\subsubsection{Graph Coloring Extensions}

The following problem
related to extending colorings to the leaves of a graph
to a coloring on the entire graph,
is \PiP{2}-complete in the most general setting
\cite{AjtaiFaginStockmeyer00}.

Let~$G = (V,E)$ be a graph.
We will denote those vertices~$v$ that have
degree~$1$ by \emph{leaves}.
We call a (partial) function~$c : V \rightarrow \SBs 1,2,3 \SEs$
a \emph{$3$-coloring (of~$G$)}.
Moreover, we say that a $3$-coloring~$c$ is \emph{proper}
if~$c$ assigns a color to every vertex~$v \in V$,
and if for each edge~$e = \SBs v_1,v_2 \SEs \in E$ holds
that~$c(v_1) \neq c(v_2)$.
The problem of deciding, given a graph~$G = (V,E)$ with~$n$
many leaves and an integer~$m$,
whether any $3$-coloring that assigns a color to exactly~$m$
leaves of~$G$ (and to no other vertices)
can be extended to a proper $3$-coloring of~$G$,
is \PiP{2}-complete~\cite{AjtaiFaginStockmeyer00}.\\

\boxedprob{
  \ThreeColExt{}\Degree{}
  
  \emph{Instance:} a graph~$G = (V,E)$
    with~$n$ many leaves, and an integer~$m$.

  \emph{Parameter:} the degree of~$G$.

  \emph{Question:} can any $3$-coloring that assigns
    a color to exactly~$m$ leaves of~$G$ (and to no other vertices)
    be extended to a proper $3$-coloring of~$G$?
}{\textbf{Complexity:} \para{}\PiP{2}-complete
(Proposition~\ref{prop:3ce-degree}).}

\boxedprob{
  \ThreeColExt{}\NumLeaves{}
  
  \emph{Instance:} a graph~$G = (V,E)$
    with~$n$ many leaves, and an integer~$m$.

  \emph{Parameter:} $n$.

  \emph{Question:} can any $3$-coloring that assigns
    a color to exactly~$m$ leaves of~$G$ (and to no other vertices)
    be extended to a proper $3$-coloring of~$G$?
}{\textbf{Complexity:} \para{}\NP{}-complete
(Proposition~\ref{prop:3ce-numleaves}).}

\boxedprob{
  \ThreeColExt{}\NumColLeaves{}
  
  \emph{Instance:} a graph~$G = (V,E)$
    with~$n$ many leaves, and an integer~$m$.

  \emph{Parameter:} $m$.

  \emph{Question:} can any $3$-coloring that assigns
    a color to exactly~$m$ leaves of~$G$ (and to no other vertices)
    be extended to a proper $3$-coloring of~$G$?
}{\textbf{Complexity:} \AkE{}-complete
(Proposition~\ref{prop:3ce-numcolleaves}).}

\boxedprob{
  \ThreeColExt{}\NumUncolLeaves{}
  
  \emph{Instance:} a graph~$G = (V,E)$
    with~$n$ many leaves, and an integer~$m$.

  \emph{Parameter:} $n-m$.

  \emph{Question:} can any $3$-coloring that assigns
    a color to exactly~$m$ leaves of~$G$ (and to no other vertices)
    be extended to a proper $3$-coloring of~$G$?
}{\textbf{Complexity:} \para{}\PiP{2}-complete
(Proposition~\ref{prop:3ce-numuncolleaves}).}

%% OTHER PROBLEMS
\subsection{Other Problems}

%% FIRST-ORDER MODEL CHECKING
\subsubsection{First-order Model Checking}
\label{sec:app-fomc}

First-order model checking is at the basis of a well-known hardness
theory in parameterized complexity theory~\cite{FlumGrohe06}.
The following problem, also based on first-order model checking,
offers another characterization of the parameterized complexity class \EkA{}.
For a formal definition of first-order model checking we refer to
Section~\ref{sec:mc}.

\boxedprob{
  \EkAMC{}
  
  \emph{Instance:} A first-order logic
  sentence~$\varphi = \exists x_1,\dotsc,x_k. \forall y_1,\dotsc,y_n. \psi$
  over a vocabulary~$\tau$,
  where~$\psi$ is quantifier-free,
  and a finite $\tau$-structure~$\AAA$.

  \emph{Parameter:} $k$.

  \emph{Question:} Is it the case that~$\AAA \models \varphi$?
}{\textbf{Complexity:} \EkA-complete
(Theorem~\ref{thm:modelchecking-completeness}).}

%% BOUNDED MODEL CHECKING
\subsubsection{Bounded Model Checking}
\label{sec:app-bmc}

The following problem is concerned with the problem
of verifying whether a linear temporal logic formula
is satisfied on all paths in a Kripke structure.
This problem is of importance in the area of software and hardware
verification~\cite{Biere09}.
\emph{Linear temporal logic} (LTL) is a modal temporal logic
where one can encode properties related to the future of paths.
\emph{LTL formulas} are defined recursively as follows:
propositional variables and their negations are in LTL;
then, if~$\varphi_1,\varphi_2 \in \mtext{LTL}$,
then so are~$\varphi_1 \vee \varphi_2$,~$\mtext{F}\varphi_1$
(Future),~$\mtext{X}\varphi_1$ (neXt),~$\varphi_1\mtext{U}\varphi_2$
($\varphi_1$~Until~$\varphi_2$).
(Further temporal operators that are considered in the literature
can be defined in terms of the operators~$\mtext{X}$ and~$\mtext{U}$.)

The semantics of LTL is defined along paths of Kripke structures.
A \emph{Kripke structure} is a tuple~$K = (S,I,T,L)$ such that
(i)~$S$ is a set of states, where states are defined by valuations
to a set~$V$ of propositional variables,
(ii)~$I \subseteq S$ is a nonempty set of initial states,
(iii)~$T \subseteq S \times S$ is the transition relation
and (iv)~$L : S \rightarrow 2^V$ is the labeling function.
The initial states~$I$ and the transition relation~$T$ are given as
functions in terms of~$S$.
A path~$\pi$ of~$K$ is an infinite sequence~$(s_0,s_1,s_2,\dotsc)$ of states,
where~$s_i \in S$ and~$T(s_i,s_{i+1})$ for all~$i \in \mathbb{N}$.
A path is initialized if~$s_0 \in I$.
We let~$\pi(i) = s_i$ denote the $i$-th state of~$\pi$.
A suffix of a path is defined as~$\pi^i = (s_i,s_{i+1},\dotsc)$.
We give the standard semantics of LTL formulas,
defined recursively over the formula structure.
We closely follow the definitions as given by Biere~\cite{Biere09}.
In what cases an LTL formula~$\varphi$ holds along a path~$\pi^{i}$,
written~$\pi^{i} \models \varphi$, is specified by the following conditions:
\[ \begin{array}{l c l l c l}
  \pi^{i} \models v \in V &\mtext{iff}& v \in L(\pi(i)), &
  \pi^{i} \models \neg v &\mtext{iff}& v \not\in L(\pi(i)), \\
  \pi^{i} \models \varphi_1 \vee \varphi_2 &\mtext{iff}&
    \pi^{i} \models \varphi_1 \mtext{ or } \pi^{i} \models \varphi_2, &
  \pi^{i} \models \mtext{X} \varphi &\mtext{iff}& \pi^{i+1} \models \varphi, \\
  \pi^{i} \models \mtext{F} \varphi &\mtext{iff}&
    \mtext{for some $j \in \mathbb{N}$, } \pi^{i+j} \models \varphi,\qquad\qquad &
  \pi^{i} \models \varphi_1 \mtext{U} \varphi_2 &\mtext{iff}& 
    \mtext{for some $j \in \mathbb{N}$, } \pi^{i+j} \models \varphi_2 \mtext{ and } \\
  &&&&&\pi^{\ell} \models \varphi_1 \mtext{ for all $i \leq \ell < i+j$.} \\
\end{array} \]
Then, an LTL formula~$\varphi$ holds in a Kripke structure~$K$
if and only if~$\pi \models \varphi$ for all initialized paths~$\pi$ of~$K$.
Related to the model checking problem is the question whether
a witness exists: a formula~$\varphi$ has a witness in~$K$
if there is an initialized path~$\pi$ of~$K$ with~$\pi \models \varphi$.

The idea of \emph{bounded model checking} is to consider only
those paths that can be represented by a prefix of length at most~$k$,
and prefixes of length~$k$.
Observe that some infinite paths can be represented by a finite prefix
with a ``loop'': an infinite path is a~$(k,l)$-lasso
if~$\pi(k+1+j) = \pi(l+j)$, for all~$j \in \mathbb{N}$.
In fact, the search for witnesses can be restricted to lassos
if~$K$ is finite.
This leads to the following bounded semantics.
In what cases an LTL formula~$\varphi$ holds
along a suffix~$\pi^i$ of a $(k,l)$-lasso~$\pi$ in the bounded semantics,
written~$\pi^i \models_k \varphi$,
is specified by the following conditions:
\[ \begin{array}{l c l}
  \pi^{i} \models_k \mtext{X} \varphi &\mtext{iff}& 
    \begin{dcases*}
      \pi^{i+1} \models_k \varphi & if $i < k$, \\
      \pi^{l} \models_k \varphi & if $i = k$, \\
    \end{dcases*} \\
  \pi^{i} \models_k \mtext{F} \varphi &\mtext{iff}&
    \mtext{for some $j \in \SBs \min(i,l), \dotsc, k \SEs$, } \pi^{j} \models_k \varphi, \\
  \pi^{i} \models_k \varphi_1 \mtext{U} \varphi_2 &\mtext{iff}& 
    \mtext{for some $j \in \SBs \min(i,l), \dotsc, k \SEs$, }
    \pi^{j} \models_k \varphi_2, \\
    &&\mtext{and } \begin{dcases*}
      \pi^{\ell} \models_k \varphi_1 \mtext{ for all $i \leq \ell < k$ and all $l \leq \ell < j$}
        & if $j < i$, \\
      \pi^{\ell} \models_k \varphi_1 \mtext{ for all $l \leq \ell < j$}
        & if $j \geq i$. \\
    \end{dcases*} \\
\end{array} \]
In the case where~$\pi$ is not a $(k,l)$-lasso for any~$l$,
the bounded semantics only gives an approximation.
In what cases an LTL formula~$\varphi$ holds
along a suffix~$\pi^i$ of a path~$\pi$ that is not a $(k,l)$-lasso
for any~$l$, written~$\pi^i \models_k \varphi$,
is specified by the following conditions:
\[ \begin{array}{l c l}
  \pi^{i} \models_k \mtext{X} \varphi &\mtext{iff}& 
    \pi^{i+1} \models_k \varphi \mtext{ and } i < k, \\
  \pi^{i} \models_k \mtext{F} \varphi &\mtext{iff}&
    \mtext{for some $j \in \SBs i, \dotsc, k \SEs$, } \pi^{j} \models_k \varphi, \\
  \pi^{i} \models_k \varphi_1 \mtext{U} \varphi_2 &\mtext{iff}& 
    \mtext{for some $j \in \SBs i, \dotsc, k \SEs$, }
    \pi^{j} \models_k \varphi_2, \\
    &&\mtext{and } \pi^{\ell} \models_k \varphi_1 \mtext{ for all $i \leq \ell < j$.} \\
\end{array} \]
Note that~$\pi \models_k \varphi$ implies~$\pi \models \varphi$
for all paths~$\pi$.
However, it might be the case that~$\pi \models \varphi$
but not~$\pi \models_k \varphi$.

For a detailed definition and discussion of Kripke structures
and the syntax and semantics of LTL we refer to other sources
\cite{BaierKatoen08,ClarkeGrumbergPeled99}.
For a detailed definition of the bounded semantics for LTL formulas,
we refer to the bounded model checking literature
\cite{Biere09,BiereCimattiClarkeZhu99}.
The following problem is central to bounded model checking.
The unparameterized problem based on this problem
is \PSPACE{}-complete
(cf.~\cite{BaierKatoen08,ClarkeKroeningOuaknineStrichman04}).\\

\boxedprob{
  \BMCwitness{}
  
  \emph{Instance:} An LTL formula~$\varphi$,
  a Kripke structure~$K$,
  and an integer~$k \geq 1$.

  \emph{Parameter:} $k$.

  \emph{Question:} Is there some path~$\pi$ of~$K$ such that~$\pi \models_{k} \varphi$?
}{\textbf{Complexity:} in \para{\NP}~\cite{BiereCimattiClarkeZhu99}.}

%% QUANTIFIED FAGIN DEFINABILITY
\subsubsection{Quantified Fagin Definability}

The W-hierarchy can also be defined by means of
Fagin-definable parameterized problems~\cite{FlumGrohe06},
which are based on Fagin's characterization of \NP{}.
We provide an additional characterization of the class \AkE{}
by means of some parameterized problems
that are quantified analogues of Fagin-defined problems.

Let~$\tau$ be an arbitrary vocabulary,
and let~$\tau' \subseteq \tau$ be a subvocabulary of~$\tau$.
We say that a~$\tau$-structure~$\mathcal{A}$ \emph{extends}
a~$\tau'$-structure~$\mathcal{B}$ if
(i)~$\mathcal{A}$ and~$\mathcal{B}$ have the same domain, and
(ii)~$\mathcal{A}$ and~$\mathcal{B}$ coincide on the interpretation
of all relational symbols in~$\tau'$,
i.e.~$R^{\mathcal{A}} = R^{\mathcal{B}}$ for all~$R \in \tau'$.
We say that~\emph{$\mathcal{A}$ extends~$\mathcal{B}$
with weight~$k$} if~$\sum_{R \in \tau \backslash \tau'}
\Card{R^{\mathcal{A}}} = k$.
Let~$\varphi$ be a first-order formula over~$\tau$
with a free relation variable~$X$ of arity~$s$.

We let~$\Pi_2$ denote the class of all first-order formulas
of the form~$\forall y_1,\dotsc,y_n. \exists x_1,\dotsc,x_m. \psi$,
where~$\psi$ is quantifier-free.
Let~$\varphi(X)$ be a first-order formula over~$\tau$,
with a free relation variable~$X$ with arity~$s$.\\

\boxedprob{
  $\AkEFDform{(\tau,\tau')}{\varphi}$
  
  \emph{Instance:} A $\tau'$-structure~$\mathcal{B}$, and an integer~$k$.

  \emph{Parameter:} $k$.

  \emph{Question:} Is it the case that for
  each~$\tau$-structure~$\mathcal{A}$ extending~$\mathcal{B}$
  with weight~$k$,
  there exists some relation~$S \subseteq A^s$
  such that~$\mathcal{A} \models \varphi(S)$?
}{\textbf{Complexity:}\\
in \AkE{} for each~$\varphi(X)$,~$\tau'$ and~$\tau$
(Proposition~\ref{prop:ake-fd-membership});\\
\AkE{}-hard for some~$\varphi(X) \in \Pi_2$,~$\tau'$ and~$\tau$
(Proposition~\ref{prop:ake-fd-hardness}).}

\boxedprob{
  $\AEkFDform{(\tau,\tau')}{\varphi}$
  
  \emph{Instance:} A $\tau'$-structure~$\mathcal{B}$, and an integer~$k$.

  \emph{Parameter:} $k$.

  \emph{Question:} Is it the case that for
  each~$\tau$-structure~$\mathcal{A}$ extending~$\mathcal{B}$,
  there exists some relation~$S \subseteq A^s$
  with~$\Card{S} = k$
  such that~$\mathcal{A} \models \varphi(S)$?
}{\textbf{Complexity:}
\AEkW{1}-hard for some~$\varphi(X) \in \Pi_2$
(Proposition~\ref{prop:aek-fd-hardness})}

\paragraph{Characterization of \AkE{} by means of quantified Fagin definability}
\noindent We let~$T$ denote the set of all relational vocabularies,
and for any~$\tau \in T$ we let~$\mtext{FO}^{X}_{\tau}$
denote the set of all first-order formulas over the vocabulary~$\tau$
with a free relation variable~$X$.
We then get the following characterization of \AkE{}:
\[ \AkE{} = 
\fptclosure{ \SB \AkEFDform{(\tau',\tau)}{\varphi} \SM \tau \in T,
\tau' \subseteq \tau, \varphi \in \mtext{FO}^{X}_{\tau} \SE }. \]

%%%
%%% ADDITIONAL PROOFS
%%%
\pagebreak{}
\section{Additional Proofs}
\label{sec:appendix-b}

%%% Reset Theorem counter, and append Appendix name to numbering.
\setcounter{theorem}{0}
\renewcommand{\thetheorem}{\Alph{section}.\arabic{theorem}}
% Also modify hyperref-anchors accordingly
\renewcommand{\theHtheorem}{\Alph{section}.\arabic{theorem}}
%%%

%%%
%%% PROOFS OF TM CHARACTERIZATION
%%%
\subsection{Proofs for Section~\ref{sec:atm-char}}
\label{sec:app-atm-char}
We give detailed proofs of
the Propositions~\ref{prop:tm-char-1}--\ref{prop:tm-char-4}
that were used in the proof of
Theorems~\ref{thm:atm-char1}~and~\ref{thm:atm-char2}.
We first prove the following technical lemma.

\begin{lemma}
\label{lem:eka-tmhalt-exact}
Let~$\mathbb{M}$ be an \EA{}-machine with~$m$ tapes
and let~$k,t \in \mathbb{N}$.
We can construct an \EA{}-machine~$\mathbb{M}'$ with~$m$ tapes
(in time polynomial in~$\Card{(\mathbb{M},k,t)}$) such that
the following are equivalent:
\begin{itemize}
  \item there is an accepting run~$\rho$ of~$\mathbb{M}'$ with input~$\epsilon$
    and each computation path in~$\rho$ contains \emph{exactly}~$k$ existential configurations
    and \emph{exactly}~$t$ universal configurations
  \item $\mathbb{M}$ halts on~$\epsilon$ with existential cost~$k$
    and universal cost~$t$.
\end{itemize}
\end{lemma}
\begin{proof}
Let~$\mathbb{M} = (S_{\exists},S_{\forall},\Sigma,\Delta,s_0,F)$ be an \EA{}-machine
with~$m$ tapes.
Now construct~$\mathbb{M} = (S'_{\exists},S'_{\forall},\Sigma,\Delta',s_0,F')$
as follows:
\[ \begin{array}{r l}
  S'_{\exists} = &
    \SB s_i \SM s \in S_{\exists}, 1 \leq i \leq k+t \SE, \\
  S'_{\forall} = &
    \SB s_i \SM s \in S_{\forall}, 1 \leq i \leq k+t \SE, \\  
  \Delta' = &
    \SB (s_i,\overline{a},s'_{i+1},\overline{a}',\overline{d})
    \SM (s,\overline{a},s',\overline{a}',\overline{d}) \in \Delta, 1 \leq i \leq k+t-1 \SE\ \cup \\
    & \SB (s_i,\overline{a},s_{i+1},\overline{a},\mathbf{S}^{m})
    \SM s \in S_{\exists} \cup S_{\forall}, \overline{a} \in \Sigma^{m} \SE, \mtext{ and} \\
  F' = &
    \SB f_{k+t} \SM f \in F \SE.
\end{array} \]
To see that~$\mathbb{M}'$ satisfies the required properties,
it suffices to see that for each (accepting) computation
path~$C_1 \rightarrow \dotsc \rightarrow C_{k'+t'}$
of~$\mathbb{M}$ with input~$\epsilon$
that contains existential configurations~$C_1,\dotsc,C_{k'}$
and universal configurations~$C_{k'+1},\dotsc,C_{k'+t'}$
for~$1 \leq k' \leq k$ and~$1 \leq t' \leq t$,
it holds that
\[ C^1_1 \rightarrow \dotsc \rightarrow C^{k'}_{k'} \rightarrow
  C^{k'+1}_{k'} \rightarrow \dotsc \rightarrow C^{k}_{k'} \rightarrow
  C^{k+1}_{k'+1} \rightarrow \dotsc \rightarrow C^{k+t'}_{k'+t'} \rightarrow
  C^{k+t'+1}_{k'+t'} \rightarrow \dotsc \rightarrow C^{k+t}_{k'+t'} \]
is an (accepting) computation path of~$\mathbb{M}'$ with input~$\epsilon$,
where for each~$1 \leq i \leq k+t$ and each~$1 \leq j \leq k'+t'$
we let~$C^i_j$ be the configuration~$(s_i,x_1,p_1,\dotsc,x_m,p_m)$,
where~$C_j = (s,x_1,p_1,\dotsc,x_m,p_m)$.
\end{proof}

\begin{proposition}
\label{prop:tm-char-1}
\EkATMhalt{*} \fptred{} \EkAMC{}
\end{proposition}
\begin{proof}
Let~$(\mathbb{M},k,t)$ be an instance of \EkATMhalt{*},
where~$\mathbb{M} = (S_{\exists},S_{\forall},\Sigma,\Delta,s_0,F)$
is an \EA{}-machine with~$m$ tapes,
and~$k$ and~$t$ are positive integers.
We constuct in fpt-time an instance~$(\AAA,\varphi)$ of \EkAMC{},
such that~$(\mathbb{M},k,t) \in \EkATMhalt{*}$
if and only if~$(\AAA,\varphi) \in \EkAMC{}$.
By Lemma~\ref{lem:eka-tmhalt-exact}, 
it suffices to construct~$(\AAA,\varphi)$ in such a way
that~$(\AAA,\varphi) \in \EkAWSat{}$ if and only if there exists an accepting
run~$\rho$
of~$\mathbb{M}$ with input~$\epsilon$ such that each computation path
of~$\rho$
contains exactly~$k$ existential configurations and exactly~$t$ universal configurations.

We construct~$\AAA$ to be a~$\tau$-structure
with domain~$A$.
We will define the vocabulary~$\tau$ below.
The domain~$A$ of~$\AAA$ is defined as follows:
\[ A = S \cup \Sigma \cup \SBs \$,\Box \SEs \cup
  \SBs \mathbf{L},\mathbf{R},\mathbf{S} \SEs \cup
  \SBs 0,\dotsc,\max \SBs m,k+t-1 \SEs \SEs \cup T, \]
where~$T$ is the set of
tuples~$(a_1,\dotsc,a_m) \in (\Sigma \cup \SBs \$,\Box \SEs)^m$
and of
tuples~$(d_1,\dotsc,d_m) \in \SBs \mathbf{L},\mathbf{R},\mathbf{S} \SEs^{m}$
occurring in transitions of~$\Delta$.
Observe that~$\Card{A} = O(k + t + \Card{\mathbb{M}})$.

We now describe the relation symbols in~$\tau$
and their interpretation in~$\AAA$.
The vocabulary~$\tau$ contains the 5-ary relation symbol~$D$
(intended as ``transition relation''),
and the ternary relation symbol~$P$
(intended as ``projection relation''),
with the following interpretations:
\[ \begin{array}{r l}
  D^{\AAA} = &
    \Delta, \mtext{ and} \\
  P^{\AAA} = &
    \SB (j,\overline{b},b_j) \SM 1 \leq j \leq m, \overline{b} \in T,
      \overline{b} = (b_1,\dotsc,b_m) \SE.
\end{array} \]
Moreover,~$\tau$ contains the unary relation symbols
$R_{\mtext{tape}}$,
$R_{\mtext{cell}}$,
$R_{\mtext{blank}}$,
$R_{\mtext{end}}$,
$R_{\mtext{symbol}}$,
$R_{\mtext{init}}$,
$R_{\mtext{acc}}$,
$R_{\mtext{left}}$,
$R_{\mtext{right}}$,
$R_{\mtext{stay}}$,
$R_{\exists}$,
$R_{\forall}$,
$R_{i}$ for each~$1 \leq i \leq k+t-1$, and
$R_{a}$ for each~$a \in \Sigma$,
which are interpreted in~$\AAA$ as follows:
\[ \begin{array}{c}
  R^{\AAA}_{\mtext{tape}} = 
    \SBs 1,\dotsc,m \SEs, 
  R^{\AAA}_{\mtext{cell}} = 
    \SBs 1,\dotsc,k+t \SEs, 
  R^{\AAA}_{\mtext{blank}} = 
    \SBs \Box \SEs, 
  R^{\AAA}_{\mtext{end}} = 
    \SBs \$ \SEs, 
  R^{\AAA}_{\mtext{symbol}} =
    \Sigma, 
  R^{\AAA}_{\mtext{init}} = 
    \SBs s_0 \SEs, \\[5pt]
  R^{\AAA}_{\mtext{acc}} = 
    F,
  R^{\AAA}_{\mtext{left}} = 
    \SBs \textbf{L} \SEs, 
  R^{\AAA}_{\mtext{right}} = 
    \SBs \textbf{R} \SEs, 
  R^{\AAA}_{\mtext{stay}} = 
    \SBs \textbf{S} \SEs,
  R^{\AAA}_{\exists} = 
    S_{\exists},
  R^{\AAA}_{\forall} = 
    S_{\forall}, \\[5pt]
  R^{\AAA}_{i} =
    \SBs i \SEs \mtext{ for each } 1 \leq i \leq k+t-1,
  \mtext{and } R^{\AAA}_{a} = \SBs a \SEs \mtext{ for each } a \in \Sigma.
\end{array} \]
Next, we define a formula that is intended to provide
a fixed interpretation of some variables
that we can use to refer to the elements of the singleton relations of~$\AAA$:
\[ \begin{array}{r l}
  \psi_{\mtext{constants}} = &
    R_{\mtext{blank}}(z_{\Box}) \wedge
    R_{\mtext{end}}(z_{\$}) \wedge
    R_{\mtext{blank}}(z_{\Box}) \wedge
    R_{\mtext{left}}(z_{\mtext{left}}) \wedge
    R_{\mtext{right}}(z_{\mtext{right}})\ \wedge \\
    & R_{\mtext{stay}}(z_{\mtext{stay}}) \wedge
    \underset{0 \leq i \leq k+t-1}{\bigwedge} R_{i}(z_{i}) \wedge
    \underset{a \in \Sigma}{\bigwedge} R_{a}(z_{a}).
\end{array} \]

The formula~$\varphi$ that we will construct aims to express that there
exist~$k$ transitions (from existential states),
such that for any sequence of~$t-1$ many transitions (from universal states),
the entire sequence of transitions results in an accepting state.
It will contain variables~$s_i,t_i,s'_i,t'_i,d_i$, for~$1 \leq i \leq k+t-1$.

The formula~$\varphi$ will also contain variables~$p_{i,j}$ and~$q_{i,j,\ell}$,
for each~$k+1 \leq i \leq k+t$, each~$1 \leq j \leq m$
and each~$1 \leq \ell \leq k+t$.
The variables~$p_{i,j}$ will encode the position of the tape head
for tape~$j$ at the $i$-th configuration in the computation path,
and the variables~$q_{i,j,\ell}$ will encode the symbol that is at cell~$\ell$
of tape~$j$ at the $i$-th configuration in the computation path.

The position of the tape heads and the contents of the tapes for
configurations~$1$ to~$k$ in the computation path,
will not be encoded by means of variables, but by means
of the formulas~$\psi_{\mtext{symbol},i}$ and~$\psi_{\mtext{position},i}$,
which we define below.

In addition the formula~$\varphi$ will contain
variables~$z_{\Box},z_{\$},z_{\mtext{init}},z_{\mtext{left}},z_{\mtext{right}},z_{\mtext{stay}},
z_{1},\dotsc,z_{k+t},z_{a_1},\dotsc,z_{a_{\Card{\Sigma}}}$,
where~$\Sigma = \SBs a_1,\dotsc,a_{\Card{\Sigma}} \SEs$,
that we will use to refer to
elements of the singleton relations of~$\AAA$.
We define~$\varphi$ as follows:
\[ \begin{array}{r l}
  \varphi = &
    \exists s_1,t_1,s'_1,t'_1,d_1,\dotsc,s_k,t_k,s'_k,t'_k,d_k. \\
    & \forall z_{\Box},z_{\$},z_{\mtext{init}},z_{\mtext{left}},z_{\mtext{right}},z_{\mtext{stay}},
z_{1},\dotsc,z_{k+t}, z_{a_1},\dotsc,z_{a_{\Card{\Sigma}}}. \\
    & \forall s_{k+1},t_{k+1},s'_{k+1},t'_{k+1},d_{k+1},\dotsc,
    s_{k+t-1},t_{k+t-1},s'_{k+t-1},t'_{k+t-1},d_{k+t-1}. \\
    & \forall p_{k+1,1},\dotsc,p_{k+t,m}.q_{k+1,1,1},\dotsc,q_{k+t,m,k+t}. \psi, \\[5pt]
  \psi = &
    \psi_{\mtext{constants}} \rightarrow
    \left ( \psi_{\exists\mtext{-states}} \wedge
    \psi_{\exists\mtext{-tapes}} \wedge
    \left ( (\psi_{\forall\mtext{-states}} \wedge
    \psi_{\forall\mtext{-tapes}}) \rightarrow
    \psi_{\mtext{accept}} \right ) \right ), \\[5pt]
  \psi_{\exists\mtext{-states}} = &
    (s_1 = z_{\mtext{init}}) \wedge
    \underset{1 \leq i \leq k}{\bigwedge} D(s_i,t_i,s'_i,t'_i,d_i) \wedge
    \underset{1 \leq i \leq k-1}{\bigwedge} \left ( (s_{i+1} = s'_{i}) \wedge R_{\exists}(s_{i+1}) \right ), \\[15pt]
  \psi_{\forall\mtext{-states}} = &
    \underset{k+1 \leq i \leq k+t-1}{\bigwedge} D(s_i,t_i,s'_i,t'_i,d_i) \wedge
    \underset{k \leq i \leq k+t-2}{\bigwedge} \left ( (s_{i+1} = s'_{i}) \wedge R_{\forall}(s_{i+1}) \right ),
    \mtext{ and} \\[15pt]
  \psi_{\mtext{accept}} = &
    R_{\mtext{acc}}(s'_{k+t-1}),
\end{array} \]
where we define the formulas~$\psi_{\exists\mtext{-tapes}}$
and~$\psi_{\forall\mtext{-tapes}}$ below.
In order to do so, for each~$1 \leq i \leq k+1$
we define the quantifier-free formulas
\[
  \psi_{\mtext{symbol},i}(w,p,a,\overline{v}_i)
    \qquad\mtext{ and }\qquad
  \psi_{\mtext{position},i}(w,p,\overline{v}_i),
\]
with~$\overline{v}_i = s_1,t_1,s'_1,t'_1,d_1,\dotsc,
s_{i-1},t_{i-1},s'_{i-1},t'_{i-1},d_{i-1}$.
Intuitively:
\begin{itemize}
  \item $\psi_{\mtext{symbol},i}(w,p,a,\overline{v}_i)$ represents whether,
    starting with empty tapes, whenever the sequence of transitions in~$\overline{v}_i$
    has been carried out, then the $p$-th cell of the $w$-th tape contains the
    symbol~$a$; and
  \item $\psi_{\mtext{position},i}(w,p,\overline{v}_i)$ represents whether,
    starting with empty tapes, whenever the sequence of transitions in~$\overline{v}_i$
    has been carried out, then the head of the $w$-th tape is at position~$p$.
\end{itemize}
We define~$\psi_{\mtext{symbol},i}(w,p,a,\overline{v}_i)$ and
$\psi_{\mtext{position},i}(w,p,\overline{v}_i)$ simultaneously
by induction on~$i$ as follows:
\[ \begin{array}{r l}
  \psi_{\mtext{symbol},1}(w,p,a) = &
    R_{\mtext{tape}}(w) \wedge
    R_{\mtext{cell}}(p) \wedge
    (p = z_0 \rightarrow a = z_{\$}) \wedge
    (p \neq z_0 \rightarrow a = z_{\Box}), \\[5pt]
  \psi_{\mtext{position},1}(w,p) = &
    R_{\mtext{tape}}(w) \wedge
    (p = z_1), \\[5pt]
  \psi_{\mtext{symbol},i+1}(w,p,a,\overline{v}_{i+1}) = &
    R_{\mtext{tape}}(w) \wedge
    R_{\mtext{cell}}(p)\ \wedge \\
    & \left ( \left (
      \psi_{\mtext{position}}(w,p,\overline{v}_i) \wedge P(w,t'_i,x)
    \right ) \vee
    \left (
      \neg \psi_{\mtext{position}}(w,p,\overline{v}_i) \wedge
      \psi_{\mtext{symbol},i}(w,p,a,\overline{v}_{i})
    \right ) \right ), \\[5pt]
  \psi_{\mtext{position},i+1}(w,p,\overline{v}_{i+1}) = &
    R_{\mtext{tape}}(w) \wedge
    \psi_{\mtext{left},i+1}(w,p,\overline{v}_{i+1}) \wedge
    \psi_{\mtext{right},i+1}(w,p,\overline{v}_{i+1}) \wedge
    \psi_{\mtext{stay},i+1}(w,p,\overline{v}_{i+1}), \\[5pt]
  \psi_{\mtext{left},i+1}(w,p,\overline{v}_{i+1}) = &
    P(w,d_i,z_{\mtext{left}}) \wedge
    \underset{1 \leq j \leq i+1}{\bigvee}
    (\psi_{\mtext{position},i}(w,z_j,\overline{v}_i) \wedge (p = z_{j-1})), \\[5pt]
  \psi_{\mtext{right},i+1}(w,p,\overline{v}_{i+1}) = &
    P(w,d_i,z_{\mtext{right}}) \wedge
    \underset{1 \leq j \leq i+1}{\bigvee}
    (\psi_{\mtext{position},i}(w,z_j,\overline{v}_i) \wedge (p = z_{j+1})), \mtext{ and} \\[5pt]
  \psi_{\mtext{stay},i+1}(w,p,\overline{v}_{i+1}) = &
    P(w,d_i,z_{\mtext{stay}}) \wedge
    \underset{1 \leq j \leq i+1}{\bigvee}
    (\psi_{\mtext{position},i}(w,z_j,\overline{v}_i) \wedge (p = z_{j})). \\
\end{array} \]
Note that for each~$1 \leq i \leq k$,
the size of the formulas~$\psi_{\mtext{symbol},i}(w,p,a,\overline{v}_i)$
and~$\psi_{\mtext{position},i}(w,p,\overline{v}_i)$ only depends on~$k$.
We can now define~$\psi_{\exists\mtext{-tapes}}$:
\[
  \psi_{\exists\mtext{-tapes}} = \forall w. \forall p. \forall a.
    \underset{1 \leq i \leq k}{\bigwedge}
    \left (
      (\psi_{\mtext{position},i}(w,p,\overline{v}_i) \wedge
      \psi_{\mtext{symbol},i}(w,p,a,\overline{v}_i)) \rightarrow
      P(w,t_i,a)
    \right ).
\]
Intuitively, the formulas~$\psi_{\exists\mtext{-states}}$
and~$\psi_{\exists\mtext{-tapes}}$
together represent
whether the transitions specified
by~$s_i,t_i,s'_i,\allowbreak{}t'_i,d_i$, for~$1 \leq i \leq k$,
together constitute a valid (partial) computation path.

Next, we define the formula~$\psi_{\forall\mtext{-tapes}}$:
\[ \begin{array}{r l}
  \psi_{\forall\mtext{-tapes}} = &
    \psi_{\forall\mtext{-tapes-}1} \wedge
    \psi_{\forall\mtext{-tapes-}2} \wedge
    \psi_{\forall\mtext{-tapes-}3} \wedge
    \psi_{\forall\mtext{-tapes-}4} \wedge
    \psi_{\forall\mtext{-tapes-}5}, \\[5pt]
  \psi_{\forall\mtext{-tapes-}1} = &
    \underset{k+1 \leq i \leq k+t}{\bigwedge}\ 
    \underset{1 \leq j \leq m}{\bigwedge}\ 
    \left (
      R_{\mtext{cell}}(p_{i,j}) \wedge
      \bigwedge\limits_{1 \leq \ell \leq k+t}
      R_{\mtext{symbol}}(q_{i,j,\ell})
    \right ), \\[15pt]
  \psi_{\forall\mtext{-tapes-}2} = &
    \underset{1 \leq j \leq m}{\bigwedge}
    \left (
    \begin{array}{l}
      \mathop{\bigwedge\limits_{1 \leq \ell \leq k+1}}\limits_{a \in \Sigma}
      ((q_{k+1,j,\ell} = z_{a}) \leftrightarrow
      \psi_{\mtext{symbol},k+1}(z_j,z_{k+1},z_a,\overline{v}_{k+1}))\ \wedge \\
      \hfill \underset{k+2 \leq \ell \leq k+t}{\bigwedge}
      (q_{k+1,j,\ell} = z_{\Box}) \\
    \end{array}
    \right ), \\[25pt]
  \psi_{\forall\mtext{-tapes-}3} = &
    \underset{1 \leq j \leq m}{\bigwedge}\ 
    \underset{1 \leq i \leq k+1}{\bigwedge}\
    \left (
      (p_{k+1,j} = z_{i}) \leftrightarrow
      \psi_{\mtext{position},k+1}(z_{j},z_{i},\overline{v}_{k+1})
    \right ), \\[10pt]
  \psi_{\forall\mtext{-tapes-}4} = &
%    \underset{1 \leq j \leq m}{\bigwedge}\ 
    \bigwedge\limits_{1 \leq j \leq m}\ 
    \underset{k+1 \leq i \leq k+t-1}{\bigwedge}\ 
    \underset{1 \leq \ell \leq k+t}{\bigwedge}
    \left (
    \begin{array}{l}
      (P(z_j,d_i,z_{\mtext{left}}) \wedge (p_{i,j} = z_{\ell})) \rightarrow (p_{i+1,j} = z_{\ell-1})\ \wedge \\
      (P(z_j,d_i,z_{\mtext{right}}) \wedge (p_{i,j} = z_{\ell})) \rightarrow (p_{i+1,j} = z_{\ell+1})\ \wedge \\
      (P(z_j,d_i,z_{\mtext{stay}}) \wedge (p_{i,j} = z_{\ell})) \rightarrow (p_{i+1,j} = z_{\ell})\ \wedge \\
    \end{array}
    \right ), \mtext{ and} \\[20pt]
  \psi_{\forall\mtext{-tapes-}5} = &
    \underset{1 \leq j \leq m}{\bigwedge}\ 
    \underset{k+1 \leq i \leq k+t-1}{\bigwedge}\ 
    \underset{1 \leq \ell \leq k+t}{\bigwedge}\ 
    \underset{a \in \Sigma}{\bigwedge}
    \left (
    \begin{array}{l}
      ((p_{i,j} \neq z_{\ell}) \rightarrow (q_{i+1,j,\ell} = q_{i,j,\ell}))\ \wedge \\
      ((p_{i,j} = z_{\ell}) \wedge P(z_{j},t_i,z_{a}) \rightarrow (q_{i,j,\ell} = z_{a}))\ \wedge \\
      ((p_{i,j} = z_{\ell}) \wedge P(z_{j},t'_i,z_{a}) \rightarrow (q_{i+1,j,\ell} = z_{a})) \\
    \end{array}
    \right ). \\
\end{array} \]
Intuitively, the formulas~$\psi_{\forall\mtext{-states}}$
and~$\psi_{\forall\mtext{-tapes}}$
together represent
whether the transitions specified by~$s_i,t_i,s'_i,\allowbreak{}t'_i,d_i$,
for~$k+1 \leq i \leq k+t-1$,
together constitute a valid (partial) computation path,
extending the computation path represented by the
transitions~$s_i,t_i,s'_i,t'_i,d_i$, for~$1 \leq i \leq k$.

It is straightforward to verify that~$\varphi$ is (logically equivalent to a formula)
of the right form, containing~$k' = 5k$ existentially quantified variables.
Also, it is now straightforward to verify that~$(\mathbb{M},k,t) \in \EkATMhalt{*}$
if and only if~$(\AAA,\varphi) \in \EkAMC{}$.
\end{proof}

\begin{proposition}
\label{prop:tm-char-2}
For any parameterized problem~$P$ that is decided by some \EkA{}-machine
with~$m$ tapes,
it holds that~$P \fptred{} \EkATMhalt{m+1}$.
\end{proposition}
\begin{proof}
Let~$P$ be a parameterized problem,
and let~$\mathbb{M} = (S_{\exists},S_{\forall},\Sigma,\Delta,s_0,F)$
be an \EkA{}-machine with~$m$ tapes that decides it,
i.e., there exists some computable function~$f$
and some polynomial~$p$
such that for any instance~$(x,k)$ of~$P$
we have that any computation path of~$\mathbb{M}$
with input~$(x,k)$
has length at most~$f(k) \cdot p(\Card{x})$ and
contains at most~$f(k) \cdot \log\Card{x}$ nondeterministic existential configurations.
We show how to construct in fpt-time for each instance~$(x,k)$ of~$P$
an \EA{}-machine~$\mathbb{M}^{(x,k)}$ with~$m+1$ tapes,
and positive integers~$k',t \in \mathbb{N}$ such
that~$\mathbb{M}^{(x,k)}$ accepts the empty string with existential cost~$k'$
and universal cost~$t$ if and only if~$\mathbb{M}$ accepts~$(x,k)$.

The idea of this construction is the following.
We add to~$\Sigma$ a fresh symbol~$\sigma_{(C_1,\dotsc,C_u)}$
for each sequence of possible ``transitions''~$T_1,\dotsc,T_u$ of~$\mathbb{M}$,
where~$u \leq \lceil \log\Card{x} \rceil$.
The machine~$\mathbb{M}^{(x,k)}$ starts with nondeterministically
writing down~$f(k)$ symbols~$\sigma_{(T_1,\dotsc,T_{\lceil \log\Card{x} \rceil})}$
to tape~$m+1$ (stage 1).
This can be done using~$k'$ nondeterministic existential steps.
Then, using universal steps, it writes down the input~$(x,k)$ to its first tape (stage 2).
It continues with simulating the existential steps in the execution of~$\mathbb{M}$
with input~$(x,k)$ (stage 3):
each deterministic existential step can simply be performed
by a deterministic universal step, and each nondeterministic existential
step can be simulated by ``reading off'' the next configuration from
the symbols on tape~$m+1$, and transitioning into this configuration
(if this step is allowed by~$\Delta$).
Finally, the machine~$\mathbb{M}^{(x,k)}$ simply performs the universal
steps in the execution of~$\mathbb{M}$ with input~$(x,k)$ (stage 4).

Let~$(x,k)$ be an arbitrary instance of~$P$.
We construct~$\mathbb{M}^{(x,k)} = (S'_{\exists},S'_{\forall},\Sigma',\Delta',s'_0,F')$.
We split the construction of~$\mathbb{M}^{(x,k)}$ into several steps
that correspond to the various stages in the execution of~$\mathbb{M}^{(x,k)}$
described above.
We begin with defining~$\Sigma'$:
\[ \Sigma' = \Sigma \cup
  \SB \sigma_{(T_1,\dotsc,T_u)}
  \SM 0 \leq u \leq \lceil \log\Card{x} \rceil,\ 1 \leq i \leq u,
  T_u = (s,\overline{a},\overline{d}), s \in S,
  \overline{a} \in \Sigma^{m},
  \overline{d} \in \SBs \mathbf{L},\mathbf{R},\mathbf{S} \SEs^{m} \SE. \]
Observe that for each~$s \in S$ and each~$\overline{a} \in \Sigma^{m}$,
each~$T_u = (s',\overline{a}',\overline{d})$ specifies a
tuple~$(s,\overline{a},s',\overline{a}',\overline{d})$
that may or may not be contained in~$\Delta$, i.e.,
a ``possible transition.''
Note that also~$\sigma_{()} \in \Sigma'$, where~$()$ denotes the
empty sequence.
Moreover, it is straightforward to verify
that~$\Card{\Sigma'} = \Card{\Sigma} + O(\Card{x})$.

We now construct the formal machinery that executes the first stage
of the execution of~$\mathbb{M}^{(x,k)}$.
We let:
\[ \begin{array}{r l}
  S_{1,\exists} = &
    \SBs s_{1,\mtext{guess}}, s_{1,\mtext{done}} \SEs, \mtext{ and} \\[5pt]
  \Delta'_1 = &
    \SB (s_{1,\mtext{guess}},\overline{a},s_{1,\mtext{guess}},\overline{a}',\overline{d})
    \SM \overline{a} = \Box^{m+1},
    \overline{a}' = \Box^{m}\sigma_{(T_1,\dotsc,T_{\lceil \log\Card{x} \rceil})},\\[5pt]
    & \phantom{\SB} 1 \leq i \leq \lceil \log\Card{x} \rceil, T_i \in S \times \Sigma^{m} \times
    \SBs \mathbf{L},\mathbf{R},\mathbf{S} \SEs^{m},
    \overline{d} = \textbf{S}^{m}\textbf{R} \SE\ \cup \\[5pt]
  & \SB (s_{1,\mtext{guess}},\overline{a},s_{1,\mtext{done}},\overline{a},\overline{d})
  \SM \overline{a} = \Box^{m+1}, \overline{d} = \textbf{S}^{m}\textbf{L} \SE\ \cup \\[5pt]
  & \SB (s_{1,\mtext{done}},\overline{a},s_{1,\mtext{done}},\overline{a},\overline{d})
  \SM \overline{a} \in \SBs\Box\SEs^{m} \times \Sigma',
  \overline{d} = \textbf{S}^{m}\textbf{L} \SE\ \cup \\[5pt]
  & \SB (s_{1,\mtext{done}},\overline{a},s_{2,0},\overline{a},\overline{d})
  \SM \overline{a} = \Box^{m}\$, \overline{d} = \textbf{S}^{m}\textbf{R} \SE,
\end{array} \]
where we will define~$s_{2,0} \in S'_{\forall}$ below
($s_{2,0}$ will be the first state of the second stage of~$\mathbb{M}^{(x,k)}$).
Furthermore, we let:
\[ s'_0 = s_{1,\mtext{guess}}. \]
The intuition behind the above construction is that state~$s_{1,\mtext{guess}}$
can be used as many times as necessary to write the
symbol~$\sigma_{(T_1,\dotsc,T_{\lceil \log\Card{x} \rceil})}$ to
the~$(m+1)$-th tape,
for some sequence~$T_1,\dotsc,T_{\lceil \log\Card{x} \rceil}$ of ``possible transitions.''
Then, the state~$s_{1,\mtext{done}}$ moves the tape head
of tape~$m+1$ back to the first position,
in order to continue with the second stage
of the execution of~$\mathbb{M}^{(x,k)}$.

We continue with the definition of those parts of~$\mathbb{M}^{(x,k)}$ that
perform the second stage of the execution of~$\mathbb{M}^{(x,k)}$,
i.e., writing down the input~$(x,k)$ to the first tape.
Let the sequence~$(\sigma_1,\dotsc,\sigma_n) \in \Sigma^n$
denote the representation of~$(x,k)$ using the alphabet~$\Sigma$.
We define:
\[ \begin{array}{r l}
  S_{2,\forall} = &
    \SB s_{2,i} \SM 1 \leq i \leq n \SE \cup
    \SBs s_{2,n+1} = s_{2,\mtext{done}} \SEs, \mtext{ and} \\[5pt]
  \Delta'_{2} = &
    \SB (s_{2,i},\overline{a},s_{2,i+1},\overline{a}',\overline{d})
    \SM 1 \leq i \leq n, \overline{a} \in \Box^{m}\sigma, \sigma \in \Sigma',
    \overline{a} = \sigma_i\Box^{m-1}\sigma, \overline{d} = \textbf{R}(\textbf{S})^{m} \SE\ \cup \\[5pt]
  & \SB (s_{2,\mtext{done}},\overline{a},s_{2,\mtext{done}},\overline{a},\overline{d})
  \SM \overline{a} \in \Sigma \times \SBs\Box\SEs^{m-1} \times \Sigma',
  \overline{d} = \textbf{L}\textbf{S}^{m} \SE\ \cup \\[5pt]
  & \SB (s_{2,\mtext{done}},\overline{a},s_{3,0},\overline{a},\overline{d})
  \SM \overline{a} = \SBs \$ \SEs \times \SBs \Box \SEs^{m-1} \times \Sigma',
  \overline{d} = \textbf{R}\textbf{S}^{m} \SE,
\end{array} \]
where we will define~$s_{3,0} \in S'_{\forall}$ below
($s_{3,0}$ will be the first state of the second stage of~$\mathbb{M}^{(x,k)}$).
Intuitively, each state~$s_{2,i}$ writes the~$i$-th symbol of the representation
of~$(x,k)$
(that is, symbol~$\sigma_i$) to the first tape,
and state~$s_{2,n+1} = s_{2,\mtext{done}}$ moves the tape head of the first tape
back to the first position.
Note that the states in~$S_{2,\forall}$ are deterministic.

Next, we continue with the definition of those parts of~$\mathbb{M}^{(x,k)}$ that
perform the third stage of the execution of~$\mathbb{M}^{(x,k)}$,
i.e., simulating the existential steps in the execution of~$\mathbb{M}$ with
input~$(x,k)$.
We define:
\[ \begin{array}{r l}
  S_{3,\forall} = &
    S_{\exists}, \mtext{ and} \\
  \Delta'_3 = &
    \SB \Delta'_{3,s} \SM s \in S_{\exists} \SE,
\end{array} \]
where for each~$s \in S_{\exists}$ we define the
set~$\Delta'_{3,s} \subseteq \Delta'_3$
as follows:
\[
  \Delta'_{3,s} = \SB \Delta'_{3,s,\overline{a}} \SM \overline{a} \in \Sigma^{m} \SE,
\]
and where for each~$s \in S_{\exists}$ and each~$\overline{a} \in \Sigma^{m}$ we define:
\[ \begin{array}{r l}
  \Delta_{(s,\overline{a})} = &
    \SB (s',\overline{a}',\overline{d}) \SM
      (s,\overline{a},s',\overline{a}',\overline{d}) \in \Delta \SE, \\[5pt]
  \Delta'_{3,s,\overline{a}} = &
  \begin{dcases*}
    \SB (s,\overline{a}\sigma',s',\overline{a}'\sigma',\overline{d}\textbf{S})
    \SM \sigma' \in \Sigma' \SE
      & if $\Delta_{(s,\overline{a})} = \SBs (s',\overline{a}',\overline{d}) \SEs$, \\
    \!\!\!\begin{array}{l}
      \SB (s,\overline{a}\sigma_{(T_1,\dotsc,T_u)},s',
        \overline{a}'\sigma_{(T_2,\dotsc,T_u)},\overline{d}\mathbf{S}) \SM \\
      \phantom{\SB} 1 \leq u \leq \lceil \log\Card{x} \rceil, T_1 = (s',\overline{a}',\overline{d}),
      (s,\overline{a},s',\overline{a}',\overline{d}) \in \Delta \SE\ \cup \\
      \SBs (s,\overline{a}\sigma_{()},s,\overline{a}\Box,\mathbf{S}^m\mathbf{R}) \SEs \\
    \end{array}
      & if $\Card{\Delta_{(s,\overline{a})}} > 1$, \\
    \emptyset
      & otherwise. \\
  \end{dcases*}
\end{array} \]
Observe that there exist transitions from states in~$S_{3,\forall}$ to
states in~$S_{\forall}$; this will be unproblematic, since we will have
that~$S_{\forall} \subseteq S'_{\forall}$ (see below).
Intuitively, each state in~$S_{\exists}$ that is deterministic in~$\mathbb{M}$
simply performs its behavior from~$\mathbb{M}$ on the first~$m$ tapes,
and ignores tape~$m+1$.
Each state in~$S_{\exists}$ that would lead to nondeterministic behavior
in~$\mathbb{M}$,
performs the transition~$T_1$ that is written as first ``possible transition'' in the currently
read symbol~$\sigma_{(T_1,\dotsc,T_u)}$ on tape~$m+1$ (if this transition is allowed
by~$\Delta$), and removes~$T_1$ from tape~$m+1$
(by replacing~$\sigma_{(T_1,\dotsc,T_u)}$ by~$\sigma_{(T_2,\dotsc,T_u)}$).
Note that the states in~$S_{3,\forall}$ are deterministic.

We continue with formally defining the part of~$\mathbb{M}^{(x,k)}$ that
performs stage 4, i.e., performing the (possibly nondeterministic) universal
steps in the execution of~$\mathbb{M}$ with input~$(x,k)$.
We define:
\[ \begin{array}{r l}
  S_{4,\forall} = &
    S_{\forall}, \mtext{ and} \\
  \Delta'_{4} = &
    \SB (s,\overline{a}\delta',s',\overline{a}'\delta',\overline{d}\mathbf{S})
    \SM s \in S_{\forall}, \overline{a} \in \Sigma^{m},
    (s,\overline{a},s',\overline{a}',\overline{d}) \in \Delta \SE.
\end{array} \]
Intuitively, each state in~$S_{\forall}$ simply performs its behavior
from~$\mathbb{M}$
on the first~$m$ tapes, and ignores tape~$m+1$.
Note that the states in~$S_{4,\forall}$ may be nondeterministic.

We conclude our definition
of~$\mathbb{M}^{(x,k)} = (S'_{\exists},S'_{\forall},\Sigma',\Delta',s'_0,F')$:
\[ \begin{array}{r l}
  S'_{\exists} = &
    S_{1,\exists}, \\[5pt]
  S'_{\forall} = &
    S_{2,\forall} \cup S_{3,\forall} \cup S_{4,\forall}, \\[5pt]
  \Delta' = &
    \Delta'_1 \cup \Delta'_2 \cup \Delta'_3 \cup \Delta'_4, \\[5pt]
  s'_0 = &
    s_{1,\mtext{guess}} \mtext{ (as mentioned above), and} \\[5pt]
  F' = &
    F.
\end{array} \]
Finally, we define~$k'$ and~$t$:
\[
  k' = 2f(k)+2
    \qquad\mtext{ and }\qquad
  t = 2\Card{(x,k)} + f(k) \cdot (p(\Card{x}) + 1) + 2.
\]
Intuitively,~$\mathbb{M}'$ needs~$k' = 2f(k)+2$ existential steps
to write down~$f(k)$ symbols~$\sigma_{(T_1,\dotsc,T_{\lceil \log\Card{x} \rceil})}$
and return the tape head of tape~$m+1$ to the first position.
It needs~$2\Card{(x,k)}+2$ steps to write the input~$(x,k)$ to the first tape
and return the tape head of tape~$1$ to the first position.
It needs~$\leq f(k) \cdot p(\Card{x}) + f(k)$ steps to simulate
the existential steps in the execution of~$\mathbb{M}$ with input~$(x,k)$,
and to perform the universal steps in the execution of~$\mathbb{M}$
with input~$(x,k)$.

This concludes our construction of the instance~$(\mathbb{M}^{(x,k)},k',t)$
of \EkATMhalt{m+1}.
It is straightforward to verify that~$(x,k) \in P$ if and only
if~$(\mathbb{M}^{(x,k)},k',t) \in \EkATMhalt{m+1}$,
by showing that~$\mathbb{M}$ accepts $(x,k)$ if and only
if~$(\mathbb{M}^{(x,k)},k',t) \in \EkATMhalt{m+1}$.
\end{proof}

\begin{proposition}
\label{prop:tm-char-3}
There is an \EkA{}-machine with a single tape that decides \EleqkAWSat{}.
\end{proposition}
\begin{proof}
We describe an \EkA{}-machine~$\mathbb{M}$ with~$1$ tape for \EleqkAWSat{},
that accepts the language \EleqkAWSat{}.
We will not spell out the
machine~$\mathbb{M} = (S_{\exists},S_{\forall},\Sigma,\Delta,s_0,F)$ in full detail,
but describe~$\mathbb{M}$ in such detail that the working of~$\mathbb{M}$ is clear
and writing down the complete formal description of~$\mathbb{M}$
can be done straightforwardly.

We assume that instances~$(\varphi,k)$ are encoded as
strings~$\sigma_1\sigma_2\dotsc\sigma_n$
over an alphabet~$\Sigma' \subseteq \Sigma$.
We denote the representation of an instance~$(\varphi,k)$ using the
alphabet~$\Sigma'$
by~$\Repr{\varphi,k}$.
Also, for any Boolean formula~$\psi(Z)$ over variables~$Z$
and any (partial) assignment~$\gamma : Z \rightarrow \SBs 0,1 \SEs$,
we let~$\Repr{\psi,\gamma}$ denote the representation
(using alphabet~$\Sigma$)
of the formula~$\psi$, where each variable~$z \in \Dom{\gamma}$
is replaced by the constant value~$\gamma(z)$.

Let~$(\varphi,k)$ be an instance of \EleqkAWSat{},
where~$\varphi = \exists X. \forall Y. \psi$,~$X = \SBs x_1,\dotsc,x_n \SEs$,
and~$Y = \SBs y_1,\dotsc,y_m \SEs$.
In the initial configuration of~$\mathbb{M}$, the tape contains the
word~$\Repr{\varphi,k}$.
We construct~$\mathbb{M}$ in such a way that it proceeds in seven stages.
Intuitively, in stage 1,~$\mathbb{M}$ adds~$\Box\Repr{\psi,\emptyset}$ to the right of the
tape contents.
We will refer to this word~$\Repr{\psi,\emptyset}$ as the representation of~$\psi$.
In stage 2, it appends the word~$(\Box1\dotsc 1)$, containing~$\lceil \log n \rceil = u$ many~$1$s,~$k$ times to the right of the tape contents.
Next, in stage 3,~$\mathbb{M}$ (nondeterministically) overwrites each such
word~$(\Box1\dotsc 1)$ by~$(\Box b_1,\dotsc,b_u)$,
for some bits~$b_1,\dotsc,b_u \in \SBs 0,1 \SEs$.
Then, in stage 4, it repeatedly reads some word~$(\Box b_1,\dotsc,b_u)$
written at the rightmost part of the tape,
and in the representation of~$\psi$, written as ``second word'' on the tape,
instantiates variable~$x_i$ to the value~$1$,
where~$b_1\dotsc b_u$ is the binary representation of~$i$.
After stage 4, at most~$k$ variables~$x_i$ are instantiated to~$1$.
Then, in stage 5,~$\mathbb{M}$ instantiates the remaining variables~$x_i$
in the representation of~$\psi$ to the value~$0$.
These first five stages are all implemented using states in~$S_{\exists}$.
The remaining two stages are implemented using states in~$S_{\forall}$.
In stage 6,~$\mathbb{M}$ nondeterministically instantiates each variable~$y_j$
in the representation of~$\psi$ to some truth value~$0$ or~$1$.
Finally, in stage 7, the machine verifies whether the fully instantiated
formula~$\psi$ evaluates to true or not, and accepts if and only if
the formula~$\psi$ evaluates to true.

We now give a more detailed description of the seven stages of~$\mathbb{M}$,
by describing what each stage does to the tape contents,
and by giving bounds on the number of steps that each stage needs.
In the initial configuration, the tape contents~$w_0$ are as follows
(we omit trailing blank symbols):
\[ w_0 = \$\Repr{\varphi,k}. \]
In stage 1,~$\mathbb{M}$ transforms the tape contents~$w_0$
to the following contents~$w_1$:
\[ w_1 = \$\Repr{\varphi,k}\Box\Repr{\psi,\emptyset}, \]
where~$\emptyset$ denotes the empty assignment
to the variables~$X \cup Y$.
This addition to the tape contents can be done by means
of~$O(\Card{\Repr{\varphi,k}})$ many deterministic existential steps.

Next, in stage 2,~$\mathbb{M}$ adds to the tape contents~$k$ words
of the form~$(\Box1\dotsc 1)$, each containing~$\lceil \log n \rceil$ many~$1$s,
resulting in the tape contents~$w_2$ after stage 2:
\[ w_2 = \$\Repr{\varphi,k}\Box\Repr{\psi,\emptyset}
  \underbrace{\Box\overbrace{1\dotsc 1}^{\lceil \log n \rceil}}_{\mtext{word } 1}
  \underbrace{\Box\overbrace{1\dotsc 1}^{\lceil \log n \rceil}}_{\mtext{word } 2}
  \Box\dotsc
  \underbrace{\Box\overbrace{1\dotsc 1}^{\lceil \log n \rceil}}_{\mtext{word } k}.
\]
This addition to the tape contents can be done by means
of~$O(k \cdot \Card{\Repr{\varphi,k}}^2)$ many deterministic existential steps.

Then, in stage 3,~$\mathbb{M}$ proceeds nondeterministically.
It replaces each word of the form~$(\Box1\dotsc 1)$ that were written to the tape
in stage 2 by a word of the form~$(\Box b_1,\dotsc,b_u)$,
for some bits~$b_1,\dotsc,b_u \in \SBs 0,1 \SEs \subseteq \Sigma$.
Here we let~$u = \lceil \log n \rceil$.
Resultingly, the tape contents~$w_3$ after stage 3 are:
\[ w_3 = \$\Repr{\varphi,k}\Box\Repr{\psi,\emptyset}
  \Box b^1_1 \dotsc b^1_u
  \Box b^2_1 \dotsc b^2_u
  \Box\dotsc
  \Box b^k_1 \dotsc b^k_u,
\]
where for each~$1 \leq i \leq k$ and each~$1 \leq j \leq u$,~$b^i_j \in
\SBs 0,1 \SEs$.
This transformation of the tape contents can be done
by means of~$O(k \cdot \lceil \log n \rceil)$ many
nondeterministic existential steps.

In stage 4,~$\mathbb{M}$ repeatedly performs the following transformation
of the tape contents, until all words~$\Box b^i_1 \dotsc b^i_u$ are removed.
The tape contents~$w'_3$ before each such transformation are as follows:
\[ w'_3 = \$\Repr{\varphi,k}\Box\Repr{\psi,\alpha}
  \Box b^1_1 \dotsc b^1_u
  \Box b^2_1 \dotsc b^2_u
  \Box\dotsc
  \Box b^{\ell}_1 \dotsc b^{\ell}_u,
\]
for some partial assignment~$\alpha : X \rightarrow \SBs 0,1 \SEs$,
and some~$1 \leq \ell \leq k$.
Each such transformation functions in such a way that the tape
contents~$w''_3$ afterwards are:
\[ w''_3 = \$\Repr{\varphi,k}\Box\Repr{\psi,\alpha'}
  \Box b^1_1 \dotsc b^1_u
  \Box b^2_1 \dotsc b^2_u
  \Box\dotsc
  \Box b^{\ell-1}_1 \dotsc b^{\ell-1}_u,
\]
where the bit string~$b^{\ell}_1\dotsc b^{\ell}_u$ is the binary representation
of the integer~$i \leq 2^u$, and where the assignment~$\alpha'$ is defined
for all~$1 \leq j \leq n$, by:
\[
  \alpha'(x_j) = \begin{dcases*}
    \alpha(x_j) & if $x_j \in \Dom{\alpha}$, \\
    1 & if $x_j \not\in \Dom{\alpha}$ and $j = i$, \\
    \mtext{undefined} & otherwise. \\
  \end{dcases*}
\]
Each such transformation can be implemented by means
of~$O(\Card{\Repr{\psi,\alpha}}^2 \cdot k \lceil \log n \rceil)$
many deterministic existential steps.
After all the~$k$ many transformation of stage 4 are performed,
the tape contents~$w_4$ are thus as follows:
\[ w_4 = \$\Repr{\varphi,k}\Box\Repr{\psi,\alpha_{\mtext{pos}}}, \]
where~$\alpha_{\mtext{pos}} : X \rightarrow \SBs 0,1 \SEs$ is the partial assignment
such that~$\Dom{\alpha_{\mtext{pos}}} = \SBs i_1,\dotsc,i_k \SEs$,
and~$\alpha_{\mtext{pos}}(x_{i_j}) = 1$ for each~$1 \leq j \leq k$,
where for each~$1 \leq j \leq k$, the integer~$i_j$
is such that~$b^j_1\dotsc b^j_u$ is the binary representation of~$i_j$.
The operations in stage 4 can be implemented by means
of~$O(\Card{\Repr{\psi,\emptyset}}^2 \cdot k^2 \lceil \log n \rceil)$
many deterministic existential steps.

Next, in stage 5, the machine~$\mathbb{M}$ transforms the tape contents
by modifying the word~$\Repr{\psi,\alpha_{\mtext{pos}}}$, resulting in~$w_5$:
\[ w_5 = \$\Repr{\varphi,k}\Box\Repr{\psi,\alpha}, \]
where the complete assignment~$\alpha' : X \rightarrow \SBs 0,1 \SEs$
is defined as follows:
\[ \alpha(x) = \begin{dcases*}
  \alpha_{\mtext{pos}}(x) & if~$x \in \Dom{\alpha_{\mtext{pos}}}$, \\
  0 & otherwise. \\
\end{dcases*} \]
This can be done using~$O(\Card{\Repr{\psi,\alpha_{\mtext{pos}}}})$ nondeterministic
existential steps.
Note that the assignment~$\alpha$ has weight at most~$k$.

Now, in stage 6, the machine~$\mathbb{M}$ alternates to universal steps.
It nondeterministically transforms the tape contents
using~$O(\Card{\Repr{\psi,\alpha}})$ nondeterministic universal steps,
resulting in the tape contents~$w_6$:
\[ w_6 = \$\Repr{\varphi,k}\Box\Repr{\psi,\alpha \cup \beta}, \]
for some complete assignment~$\beta : Y \rightarrow \SBs 0,1 \SEs$.

Finally, in stage 7,~$\mathbb{M}$ checks whether the
assignment~$\alpha \cup \beta$
satisfies the formula~$\psi$.
This check can be done by means of~$O(\Card{\Repr{\psi,\alpha \cup \beta}})$ many
deterministic universal steps.
The machine~$\mathbb{M}$ accepts if and only if~$\alpha \cup \beta$
satisfies~$\psi$.

It is straightforward to verify that there exists a computable function~$f$ and
a polynomial~$p$ such that each computation path of~$\mathbb{M}$
with input~$(\varphi,k)$ has length at most~$f(k) \cdot p(\Card{\varphi})$
and contains at most~$f(k) \cdot \log \Card{\varphi}$ many
nondeterministic existential configurations.
Also, it is straightforward to verify that~$\mathbb{M}$ accepts an
input~$(\varphi,k)$
if and only if~$(\varphi,k) \in \EleqkAWSat{}$.
This concludes our proof that the \EkA{}-machine~$\mathbb{M}$
decides \EleqkAWSat{}.
\end{proof}

\begin{proposition}
\label{prop:tm-char-4}
Let~$A$ and~$B$ be parameterized problem,
and let~$m \in \mathbb{N}$ be a positive integer.
If~$B$ is decided by some \EkA{}-machine with~$m$ tapes,
and if~$A \fptred{} B$,
then~$A$ is decided by some \EkA{}-machine with~$m$ tapes.
\end{proposition}
\begin{proof}
Let~$R$ be the fpt-reduction from~$A$ to~$B$,
and let~$M$ be an algorithm that decides~$B$
and that can be implemented by an \EkA{}-machine with~$m$ tapes.
Clearly, the composition of~$R$ and~$M$ is an algorithm
that decides~$A$.
It is straightforward to verify that the composition
of~$R$ and~$M$ can be implemented by an \EkA{}-machine with~$m$ tapes.
\end{proof}

\noindent The following statement is known from the literature.
\begin{proposition}[cf.~{\cite[Thm~8.9~and~Thm~8.10]{HopcroftMotwaniUllman01}}]
\label{prop:multitape-TM-simulation}
Let~$m \geq 1$ be a (fixed) positive integer.
For each ATM~$\mathbb{M}$ with~$m$ tapes,
there exists an ATM~$\mathbb{M}'$ with~$1$ tape such that:
\begin{itemize}
  \item $\mathbb{M}$ and~$\mathbb{M}'$ are equivalent, i.e.,
    they accept the same language;
  \item $\mathbb{M}'$ simulates~$n$ many steps of~$\mathbb{M}$
    using~$O(n^2)$ many steps; and
  \item $\mathbb{M}'$ simulates existential steps of~$\mathbb{M}$
    using existential steps, and
    simulates universal steps of~$\mathbb{M}$ using universal steps.
\end{itemize}
\end{proposition}
\begin{corollary}
\label{cor:multitape-TM-simulation}
\EkATMhalt{2} \fptred{} \EkATMhalt{1}.
\end{corollary}

%%%
%%% PROOFS FOR THE COMPENDIUM
%%%
\subsection{Proofs for the Compendium}

\begin{proposition}
\label{prop:kstar-leqk-hardness}
\EkAWSat{} \fptred{} \EleqkAWSat{}.
\end{proposition}
\begin{proof}
Let~$(\varphi,k)$ be an instance of \EkAWSat{},
where~$\varphi = \exists X. \forall Y. \psi$,
and~$X = \SBs x_1,\dotsc,x_n \SEs$.
We construct an instance~$(\varphi',k')$ of \EleqkAWSat{}.
Let~$C = \SB c^i_j \SM 1 \leq i \leq k, 1 \leq j \leq n \SE$
be a set of fresh propositional variables.
Intuitively, we can think of the variables~$c^i_j$ as being placed
in a matrix with~$k$ columns and~$n$ rows: variable~$c^i_j$
is positioned in the $i$-th column and in the $j$-th row.
We ensure that in each column, exactly one variable is set to true
(see~$\psi_{\mtext{col}}$ below),
and that in each row, at most one variable is set to true
(see~$\psi_{\mtext{row}}$ below).
This way, any satisfying assignment must set
exactly~$k$ variables in the matrix to true,
in different rows.
Next, we ensure that if any variable in the $j$-th row is set to true,
that~$x_j$ is set to true
(see~$\psi_{\mtext{corr}}$ below).
This way, we know that exactly~$k$ many variables~$x_j$ must be set
to true in any satisfying assignment.

Formally, we define:
\[ \begin{array}{r l}
  \varphi' = & \exists X \cup C. \forall Y. \psi'; \\
  k' = & 2k; \\
  \psi' = &
    \psi_{\mtext{col}} \wedge \psi_{\mtext{row}} \wedge \psi_{\mtext{corr}} \wedge \psi; \\
  \psi_{\mtext{col}} = &
    \bigwedge\limits_{1 \leq j \leq n}\ 
    \left (
      \bigvee\limits_{1 \leq i \leq k} c^i_j \wedge
      \bigwedge\limits_{1 \leq i < i' \leq k} (\neg c^i_j \vee \neg c^{i'}_j)
    \right ) ; \\
  \psi_{\mtext{row}} = &
    \bigwedge\limits_{1 \leq i \leq k}\ 
    \bigwedge\limits_{1 \leq j < j' \leq n}
    (\neg c^i_j \vee \neg c^i_{j'}); \mtext{ and} \\
  \psi_{\mtext{corr}} = &
    \bigwedge\limits_{1 \leq i \leq k}\ 
    \bigwedge\limits_{1 \leq j \leq n}
    c^i_j \rightarrow x_j. \\
\end{array} \]

Any assignment
$\alpha : X \cup C \rightarrow \SBs 0,1 \SEs$ that satisfies
$\psi_{\mtext{col}} \wedge \psi_{\mtext{row}}$ must set the variables
$c^1_{j_1},\dotsc,c^k_{j_k}$ to true, for some
$1 \leq j_1 < \dotsm < j_k \leq n$.
Furthermore, if~$\alpha$ satisfies~$\psi_{\mtext{corr}}$,
it must also set~$x_{j_1},\dotsc,x_{j_k}$ to true.

It is now easy to show that~$(\varphi,k) \in \EkAWSat{}$ if and only
if~$(\varphi',k') \in \EleqkAWSat{}$.
Let~$\alpha : X \rightarrow \SBs 0,1 \SEs$ be an assignment of weight~$k$
such that~$\forall Y. \psi[\alpha]$ is true,
where~$\SB x_i \SM 1 \leq i \leq n, \alpha(x_i) = 1 \SE
= \SBs x_{j_1},\dotsc,x_{j_k} \SEs$.
Then consider the assignment~$\gamma : C \rightarrow \SBs 0,1 \SEs$
where~$\gamma(c^i_j) = 1$ if and only if~$j = j_i$.
Then the assignment~$\alpha \cup \gamma$ has weight~$k'$,
and has the property that~$\forall Y. \psi'[\alpha \cup \gamma]$ is true.

Conversely, let~$\gamma : X \cup C \rightarrow \SBs 0,1 \SEs$ be an assignment
of weigth~$k'$ such that~$\forall Y. \psi'[\gamma]$ is true.
Then the restriction~$\alpha$ of~$\gamma$ to the variables~$X$
has weight~$k$, and has the property that~$\forall Y. \psi[\alpha]$ is true.
\end{proof}

\begin{proposition}
\label{prop:kstar-leqk-membership}
\EleqkAWSat{} \fptred{} \EkAWSat{}.
\end{proposition}
\begin{proof}
Let~$(\varphi,k)$ be an instance of \EleqkAWSat{},
with~$\varphi = \exists X. \forall Y. \psi$.
We construct an instance~$(\varphi',k)$ of \EkAWSat{}.
Let~$X' = \SBs x'_1,\dotsc,x'_k \SEs$.
Now define~$\varphi' = \exists X \cup X'. \forall Y. \psi$.
We show that~$(\varphi,k) \in \EleqkAWSat{}$ if and only
if~$(\varphi',k) \in \EkAWSat{}$.

$(\Rightarrow)$
Assume that~$(\varphi,k) \in \EleqkAWSat{}$.
This means that there exists an
assignment~$\alpha : X \rightarrow \SBs 0,1 \SEs$
of weight~$\ell \leq k$ such that~$\forall Y. \psi[\alpha]$ evaluates to true.
Define the assignment~$\alpha' : X' \rightarrow \SBs 0,1 \SEs$ as follows.
We let~$\alpha'(x'_i) = 1$ if and only if~$1 \leq i \leq k-\ell$.
Then the assignment~$\alpha \cup \alpha'$ has weight~$k$,
and~$\forall Y. \psi[\alpha \cup \alpha']$ evaluates to true.
Therefore,~$(\varphi',k) \in \EkAWSat{}$.

$(\Leftarrow)$
Assume that~$(\varphi',k) \in \EkAWSat{}$.
This means that there exists an
assignment~$\alpha : X \cup X' \rightarrow \SBs 0,1 \SEs$
of weight~$k$ such that~$\forall Y. \psi[\alpha]$ evaluates to true.
Now let~$\alpha'$ be the restriction of~$\alpha$ to the set~$X$ of variables.
Clearly,~$\alpha'$ has weight at most~$k$.
Also, since~$\psi$ contains no variables in~$X'$, we know
that~$\forall Y. \psi[\alpha']$
evaluates to true.
Therefore,~$(\varphi,k) \in \EleqkAWSat{}$.
\end{proof}

%%%
%%%
%%%
\begin{proposition}
\label{prop:kstar-geqk-hardness}
\EgeqkAWSat{} is \para{}\SigmaP{2}-hard.
\end{proposition}
\begin{proof}
To show $\para{\SigmaP{2}}$-hardness, it suffices to show that the
problem is already \SigmaP{2}-hard when the parameter value
is restricted to~$1$~\cite{FlumGrohe03}.
We give a polynomial-time reduction from \QSat{2}
to the slice of \EgeqkAWSat{} where~$k = 1$.
Let~$\varphi = \exists X. \forall Y. \psi$ be an instance of \QSat{2}.
We let~$\varphi' = \exists X. \forall Y'. \psi$,
where~$Y' = Y \cup \SBs y_0 \SEs$ for a fresh variable~$y_0$.
Moreover, we let~$k=1$.
We claim that~$\varphi \in \QSat{2}$ if and only
if~$(\varphi',k) \in \EgeqkAWSat{}$.

$(\Rightarrow)$
Assume that~$\varphi \in \QSat{2}$, i.e., that there exists a truth
assignment~$\alpha : X \rightarrow \SBs 0,1 \SEs$ such
that~$\forall Y. \psi[\alpha]$ is true.
We show that for all truth assignments~$\beta : Y' \rightarrow \SBs 0,1 \SEs$
of weight at least~$1$ it holds that~$\psi[\alpha \cup \beta]$ is true.
Let~$\beta : Y' \rightarrow \SBs 0,1 \SEs$ be an arbitrary such truth assignment.
Clearly, since~$\psi$ contains only variables in~$Y$
and since~$\forall Y. \psi[\alpha]$ is true, we know that~$\psi[\alpha \cup \beta]$
is true.
Thus, we can conclude that~$(\varphi,k) \in \EgeqkAWSat{}$.

$(\Leftarrow)$
Conversely, assume that~$(\varphi,k) \in \EgeqkAWSat{}$,
i.e., that there exists a truth assignment~$\alpha : X \rightarrow \SBs 0,1 \SEs$
such that for all truth assignments~$\beta : Y' \rightarrow \SBs 0,1 \SEs$
of weight at least~$1$ it holds that~$\psi[\alpha \cup \beta]$ is true.
We show that~$\forall Y. \psi[\alpha]$ is true.
Let~$\beta : Y \rightarrow \SBs 0,1 \SEs$ be an arbitrary truth assignment.
Construct the truth assignment~$\beta' : Y' \rightarrow \SBs 0,1 \SEs$
as follows.
On the variables in~$Y$,~$\beta$ and~$\beta'$ coincide,
and~$\beta'(y') = 1$.
Clearly,~$\beta'$ has weight at least~$1$, and thus~$\psi[\alpha \cup \beta']$
is true.
Since~$\psi$ does not contain the variable~$y_0$, we then know
that~$\psi[\alpha \cup \beta]$ is true as well.
Then, since~$\beta$ was arbitrary, we can conclude that~$\varphi \in \QSat{2}$.
\end{proof}

\begin{proposition}
\label{prop:kstar-geqk-membership}
\EgeqkAWSat{} is in \para{}\SigmaP{2}.
\end{proposition}
\begin{proof}
To show $\para{\SigmaP{2}}$-membership, we show that the
problem is in \SigmaP{2}-hard
after a precomputation on the parameter~\cite{FlumGrohe03},
i.e., we give a reduction from \EgeqkAWSat{} to \QSat{2}
that is allowed to run in fpt-time.
Let~$(\varphi,k)$ be an instance of \EgeqkAWSat{},
where~$\varphi = \exists X. \forall Y. \psi$.
We construct an instance~$\varphi'$ of \QSat{2} as follows.

We let~$Z$ be a set of fresh variables.
Then, let~$\chi$ be a propositional formula on the variables~$Y \cup Z$
that is unsatisfiable if and only if less than~$k$ variables in~$Y$
are set to true.
This formula~$\chi$ is straightforward to construct, and we omit
the details of the construction here.
Then, we let~$\varphi' = \exists X. \forall Y \cup Z. \psi'$,
where~$\psi' = \chi \rightarrow \psi$.
We claim that~$(\varphi,k) \in \EgeqkAWSat{}$ if and only
if~$\varphi' \in \QSat{2}$.

$(\Rightarrow)$
Assume that~$(\varphi,k) \in \EgeqkAWSat{}$, i.e.,
that there exists a truth assignment~$\alpha : X \rightarrow \SBs 0,1 \SEs$
such that for all truth assignments~$\beta : Y \rightarrow \SBs 0,1 \SEs$
of weight at least~$k$ it holds that~$\psi[\alpha \cup \beta]$ is true.
We show that~$\varphi' \in \QSat{2}$.
We show that~$\forall Y \cup Z. \psi'[\alpha]$ is true.
Let~$\gamma : Y \cup Z \rightarrow \SBs 0,1 \SEs$ be an arbitrary truth
assignment.
We distinguish two cases: either~(i)~$\gamma$ satisfies~$\chi$
or~(ii)~this is not the case.
In case~(i), we know that~$\chi$ is satisfied, and thus that~$\gamma$
sets at least~$k$ variables in~$Y$ to true.
Therefore, we know that~$\psi[\alpha \cup \gamma]$ is true,
and thus that~$\alpha \cup \gamma$ satisfies~$\psi'$.
In case~(ii),~$\alpha \cup \gamma$ satisfies~$\psi'$ because
it does not satisfy~$\chi$.
Then, since~$\gamma$ was arbitrary, we know that~$\varphi' \in \QSat{2}$.

$(\Leftarrow)$
Conversely, assume that~$\varphi' \in \QSat{2}$, i.e., that there
exists an truth assignment~$\alpha : X \rightarrow \SBs 0,1 \SEs$
such that for all truth assignments~$\gamma : Y \cup Z \rightarrow \SBs 0,1 \SEs$
it holds that~$\psi'[\alpha \cup \gamma]$ is true.
We show that~$(\varphi,k) \in \EgeqkAWSat{}$.
In particular, we show that for all truth
assignments~$\beta : Y \rightarrow \SBs 0,1 \SEs$ of weight at least~$k$
it holds that~$\psi[\alpha \cup \beta]$ is true.
Let~$\beta : Y \rightarrow \SBs 0,1 \SEs$ be a truth assignment
of weight at least~$k$.
Since, by construction,~$\chi$ is satisfiable if and only if at least~$k$
variables in~$Y$ are set to true, we know that we can extend the
assignment~$\beta$ to a truth assignment~$\gamma : Y \cup Z \rightarrow
\SBs 0,1 \SEs$ that satisfies~$\chi$.
Then, since~$\psi'[\alpha \cup \gamma]$ is true,
and since~$\gamma$ satisfies~$\chi$, we know that~$\alpha \cup \gamma$
satisfies~$\psi$.
Moreover, since~$\psi$ contains only variables in~$X \cup Y$,
and since~$\gamma$ coincides with~$\beta$ on the variables in~$Y$,
we know that~$\psi[\alpha \cup \beta]$ is true.
Since~$\beta$ was arbitrary, we can conclude
that~$(\varphi,k) \in \EgeqkAWSat{}$.
\end{proof}

%%%
%%%
%%%
\begin{proposition}
\label{prop:kstar-nmink}
\EnminkAWSat{} is \EkA{}-complete.
\end{proposition}
\begin{proof}
\ProofKstarNmink{}
\end{proof}

%%%
%%%
%%%
\begin{proposition}
The problem~\EAkWSat{} is \EAkW{P}-hard,
even when restricted to quantified circuits that
are in negation normal form and that
are anti-monotone in the universal variables.
\label{prop:eakwp-antimonotone}
\end{proposition}
\begin{proof}[Proof (sketch)]
We show the equivalent result that \AEkWSat{}
is \AEkW{P}-hard, even when restricted to quantified circuits
that are in negation normal form and that are
monotone in the existential variables.
The proof of this statement is very similar
to the proof of Proposition~\ref{prop:eakwp-monotone}.
We describe the modifications that need to be made to the construction.
In a similar fashion, given a quantified circuit~$C$,
we construct a quantified circuit~$C'$ that is monotone
in the existential variables,
by introducing~$k$ sets~$Y_1,\dotsc,Y_k$ of copies of the input nodes,
and simulating the original input nodes~$y_j$ and their negations~$\neg y_j$
by internal nodes~$y_j$ and~$y'_j$ in the constructed circuit.
In the modified construction, we can dispose of the nodes~$z^{j,j'}_{i}$
and the nodes~$z_i$.
Moreover, we modify the subcircuit~$B$ to check whether at
least~$k$ of its inputs are set to true,
and we let its input nodes be the nodes~$y_1,\dotsc,y_m$
in the constructed circuit.
The resulting circuit~$C'$ is in negation normal form,
is monotone in the existential variables,
and constitutes an equivalent instance of \AEkWSat{}.
\end{proof}

\begin{proposition}
\label{prop:qsat2-twE-fpt}
$\EAtwESat{}$ is \para{\SigmaP{2}}-complete.
Furthermore, \para{\SigmaP{2}}-hardness holds even for the case
where the input formula is in \DNF{}.
\end{proposition}
\ProofEAtwECompleteness{}

\begin{proposition}
\label{prop:qsat2-twA-fpt}
$\EAtwASat{}$ is \para{\NP}-complete.
Furthermore, \para{\NP}-hardness holds even for the case
where the input formula is in \DNF{}.
\end{proposition}
\ProofEAtwACompleteness{}

%%%
%%%
%%%
\begin{lemma}
\label{lem:exactly-k}
Let~$(\varphi,k)$ be an instance of \EkAWSat{}.
In polynomial time, we can construct an instance~$(\varphi',k)$
of \EkAWSat{} with~$\varphi' = \exists X. \forall Y. \psi$, such that:
\begin{itemize}
  \item $(\varphi',k) \in \EkAWSat{}$ if and only if~$(\varphi,k) \in \EkAWSat{}$; and
  \item for any assignment~$\alpha : X \rightarrow \SBs 0,1 \SEs$ that has
    weight~$m \neq k$, it holds that~$\forall Y. \psi[\alpha]$ is true.
\end{itemize}
\end{lemma}
\begin{proof}
Let~$(\varphi,k)$ be an instance of \EkAWSat{},
where~$\varphi = \exists X. \forall Y. \psi$,
and where~$X = \SBs x_1,\dotsc,x_n \SEs$.
We construct an instance~$(\varphi',k)$,
with~$\varphi' = \exists X. \forall Y \cup Z. \psi'$.
We define:
\[ Z = \SB z^i_j \SM 1 \leq i \leq k, 1 \leq j \leq m \SE. \]
Intuitively, one can think of the variables~$z^i_j$ as being positioned
in a matrix with~$n$ rows and~$k$ columns:
the variable~$z^i_j$ is placed in the $j$-th row and the $i$-th column.
We will use this matrix to verify whether exactly~$k$ variables in~$X$
are set to true.

We define~$\psi'$ as follows:
\[ \begin{array}{r l}
  \psi' = &
    (\psi^{X,Z}_{\mtext{corr}} \wedge
     \psi^{Z}_{\mtext{row}} \wedge
     \psi^{Z}_{\mtext{col}}) \rightarrow \psi; \\
  \psi^{X,Z}_{\mtext{corr}} = &
    \bigwedge\limits_{1 \leq j \leq n}\ 
    \left (
      \bigwedge\limits_{1 \leq i \leq k}
      (z^i_j \rightarrow x_j) \wedge
      (x_j \rightarrow
      \bigvee\limits_{1 \leq i \leq k} z^i_j )
    \right ); \\
  \psi^{Z}_{\mtext{row}} = &
    \bigwedge\limits_{1 \leq j \leq n}\ 
    \bigwedge\limits_{1 \leq i < i' \leq k}
    (\neg z^i_j \vee \neg z^{i'}_j); \mtext{ and} \\
  \psi^{Z}_{\mtext{col}} = &
    \bigwedge\limits_{1 \leq i \leq k}
    \left (
      \bigwedge\limits_{1 \leq j < j' \leq n}
      (\neg z^i_j \vee \neg z^i_{j'}) \wedge
      \bigvee\limits_{1 \leq j \leq n} z^i_j
    \right ). \\
\end{array} \]
Intuitively, the formula~$\psi^{X,Z}_{\mtext{corr}}$
ensures that exactly those~$x_j$ are set to true
for which there exists some~$z^i_j$ that is set to true.
The formula~$\psi^{Z}_{\mtext{row}}$ ensures that
in each row of the matrix filled with variables~$z^i_j$,
there is at most one variable set to true.
The formula~$\psi^{Z}_{\mtext{col}}$ ensures
that in each column there is exactly one variable set to true.

We show that~$(\varphi,k) \in \EkAWSat{}$ if and only
if~$(\varphi',k) \in \EkAWSat{}$.
Assume that there is some assignment~$\alpha : X \rightarrow \SBs 0,1 \SEs$
of weight~$k$ such that~$\forall Y. \psi[\alpha]$ is true.
Then~$\forall Y \cup Z. \psi'[\alpha]$ is true too,
since for any assignment~$\beta : Y \rightarrow \SBs 0,1
\SEs$,~$\alpha \cup \beta$ already satisfies the consequent of the
implication in~$\psi'$.
Conversely, assume that there is some
assignment~$\alpha : X \rightarrow \SBs 0,1 \SEs$
of weight~$k$ such that~$\forall Y \cup Z. \psi'[\alpha]$ is true.
We show that~$\forall Y. \psi[\alpha]$ is true as well.
Let~$\beta : Y \rightarrow \SBs 0,1 \SEs$ be any assignment.
Now, let~$\SB x \in X \SM \alpha(x) = 1 \SE = \SBs x_{j_1},\dotsc,x_{j_k} \SEs$.
Define~$\gamma : Z \rightarrow \SBs 0,1 \SEs$ by
letting~$\gamma(z^i_j) = 1$ if and only if~$j = j_i$.
We know that~$\psi'[\alpha \cup \beta \cup \gamma]$ is true.
Then, it is easy to verify
that~$(\psi^{X,Z}_{\mtext{corr}} \wedge \psi^{Z}_{\mtext{row}} \wedge
\psi^{Z}_{\mtext{col}})[\beta \cup \gamma]$ is true.
Therefore,~$\psi[\alpha \cup \beta \cup \gamma]$ is true.
Since~$\Var{\psi} \cap \Dom{\gamma} = \emptyset$, we know
that~$\psi[\alpha \cup \beta]$ is true.

Also, we show that for any assignment~$\alpha : X \rightarrow \SBs 0,1 \SEs$
of weight~$m \neq k$ it holds that~$\forall Y. \psi[\alpha]$ is true.
Let~$\alpha$ be any such assignment.
For any assignment~$\gamma : Z \rightarrow \SBs 0,1 \SEs$ it holds
that~$(\psi^{X,Z}_{\mtext{corr}} \wedge \psi^{Z}_{\mtext{row}} \wedge
\psi^{Z}_{\mtext{col}})[\alpha \cup \gamma]$ is false.
Therefore,~$\forall Y \cup Z. \psi[\alpha]$ is true.
\end{proof}

\begin{lemma}
\label{lem:exactly-nmink}
Let~$(\varphi,k)$ be an instance of \EnminkAWSat{}.
In polynomial time, we can construct an instance~$(\varphi',k)$
of \EnminkAWSat{} with~$\varphi' = \exists X. \forall Y. \psi$, such that:
\begin{itemize}
  \item $(\varphi',k) \in \EnminkAWSat{}$ if and only
    if~$(\varphi,k) \in \EnminkAWSat{}$; and
  \item for any assignment~$\alpha : X \rightarrow \SBs 0,1 \SEs$ that has
    weight~$m \neq (\Card{X} - k)$, it holds that~$\forall Y. \psi[\alpha]$ is true.
\end{itemize}
\end{lemma}
\begin{proof}[Proof (sketch).]
The result follows directly from the proofs of Proposition~\ref{prop:kstar-nmink}
and Lemma~\ref{lem:exactly-k}.
\end{proof}

\begin{proposition}
\label{prop:smallshortestimplicant-hardness}
\SSIC{} is \EkA{}-hard.
\end{proposition}
\ProofSSICHardness{}

\begin{proposition}
\label{prop:smallshortestimplicant-membership}
\SSIC{} is in \EkA{}.
\end{proposition}
\ProofSSICMembership{}

%%%
%%%
%%%
\begin{proposition}
\label{prop:largeshortestimplicant-hardness}
\LSIC{} is \EkA{}-hard.
\end{proposition}
\begin{proof}[Proof (sketch).]
The fpt-reduction given in the proof of
Proposition~\ref{prop:smallshortestimplicant-hardness}
is also an fpt-reduction from from \EnminkAWSat{} to
\LSIC{}.
In order to argue that this
is a correct reduction from \EnminkAWSat{} to \LSIC{},
Lemma~\ref{lem:exactly-nmink} is needed instead of Lemma~\ref{lem:exactly-k}.
Verifying that this is indeed a correct fpt-reduction is straightforward.
\end{proof}

\begin{proposition}
\label{prop:largeshortestimplicant-membership}
\LSIC{} is in \EkA{}.
\end{proposition}
\begin{proof}[Proof (sketch).]
The fpt-reduction given in the proof of
Proposition~\ref{prop:smallshortestimplicant-membership}
is also an fpt-reduction from \LSIC{}
to \EnminkAWSat{}.
Verifying that this is indeed a correct fpt-reduction is straightforward.
\end{proof}

%%%
%%%
%%%
\begin{proposition}
\label{prop:largemindnf-hardness}
\LargeMinDNF{} is \EkA{}-hard.
\end{proposition}
\begin{proof}
\sloppypar
The polynomial-time reduction
from the unparameterized version of \LSIC{}
to the unparameterized version of \LargeMinDNF{}
given by Umans~\cite[Theorem~2.2]{Umans00}
is an fpt-reduction from \LSIC{}
to \LargeMinDNF{}.
This is straightforward to verify.
\end{proof}

\begin{proposition}
\label{prop:largemindnf-membership}
\LargeMinDNF{} is in \EkA{}.
\end{proposition}
\begin{proof}
\ProofLargeMinDNFMembership{}
\end{proof}

\begin{proposition}
\label{prop:largemindnf-hardness}
\LargeMinDNF{} is \EkA{}-hard.
\end{proposition}
\begin{proof}
\sloppypar
The polynomial-time reduction
from the unparameterized version of \LSIC{}
to the unparameterized version of \LargeMinDNF{}
given by Umans~\cite[Theorem~2.2]{Umans00}
is an fpt-reduction from \LSIC{}
to \LargeMinDNF{}.
This is straightforward to verify.
\end{proof}

\begin{proposition}
\label{prop:largemindnf-membership}
\LargeMinDNF{} is in \EkA{}.
\end{proposition}
\begin{proof}
\ProofLargeMinDNFMembership{}
\end{proof}
\CompleteSmallMinDNF{}

%%%
%%%
%%%
\begin{theorem}
\label{thm:robustcsp-membership}
\RobustCSPSat{} is in \AkE{}.
\end{theorem}
\begin{proof}
We give an fpt-reduction from \RobustCSPSat{} to \AkEWSat{}.
Let~$(X,D,C,k)$ be an instance of \RobustCSPSat{},
where~$(X,D,C)$ is a CSP
instance,~$X = \SBs x_1,\dotsc,x_n \SEs$,~$D =
\SBs d_1,\dotsc,d_m \SEs$,
and~$k$ is an integer.
We construct an instance~$(\varphi,k)$ of \AkEWSat{}.
For the formula~$\varphi$, we use propositional
variables~$Z = \SB z^i_j \SM 1 \leq i \leq n, 1 \leq j \leq m \SE$
and~$Y = \SB y^i_j \SM 1 \leq i \leq n, 1 \leq j \leq m \SE$.
Intuitively, the variables~$z^i_j$ will represent an arbitrary
assignment~$\alpha$ that assigns values to~$k$ variables in~$X$.
Any variable~$z^i_j$ represents that
variable~$x_i$ gets assigned value~$d_j$.
The variables~$y^i_j$ will represent the solution~$\beta$ that extends
the arbitrary assignment~$\alpha$.
Similarly, any variable~$y^i_j$ represents that
variable~$x_i$ gets assigned value~$d_j$. %
\longversion{

}%
We then let~$\varphi = \forall Z. \exists Y. \psi$
with~$\psi = (\psi^{Z}_{\mtext{proper}} \wedge \neg \psi^{Z}_{\mtext{violate}})
\rightarrow (\psi^{Y,Z}_{\mtext{corr}} \wedge \psi^{Y}_{\mtext{proper}}
\wedge \bigwedge_{c \in C} \psi^{Y}_{c})$.
We will describe the subformulas of~$\varphi$ below,
as well as the intuition behind them. %
\longversion{

}%
We start with the formula~$\psi^{Z}_{\mtext{proper}}$.
This formula represents whether for each variable~$x_i$
at most one value is chosen for the assignment~$\alpha$.
\krversion{We let~$\psi^{Z}_{\mtext{proper}} = \bigwedge_{1 \leq i \leq n}
\bigwedge_{1 \leq j < j' \leq m} (\neg z^i_j \vee \neg z^i_{j'})$.}%
\longversion{We let: \[ \psi^{Z}_{\mtext{proper}} = \bigwedge_{1 \leq i \leq n}
\bigwedge_{1 \leq j < j' \leq m} (\neg z^i_j \vee \neg z^i_{j'}). \]} %
\longversion{

}%
\longversion{\noindent}%
Next, we consider the formula~$\psi^{Z}_{\mtext{violate}}$.
This subformula encodes whether the assignment~$\alpha$
violates some constraint~$c \in C$.
\krversion{We let $\psi^{Z}_{\mtext{violate}} =
\bigvee_{c = (S,R) \in C} \bigwedge_{\overline{d} \in R}
\bigvee_{z \in \Psi^{\overline{d},c}}$,}%
\longversion{We let: \[\psi^{Z}_{\mtext{violate}} =
\bigvee_{c = (S,R) \in C}\ \bigwedge_{\overline{d} \in R}\ 
\bigvee_{z \in \Psi^{\overline{d},c}} z, \]} %
where we define the set~$\Psi^{\overline{d},c} \subseteq Z$
as follows.
Let~$c = ((x_{i_1},\dotsc,x_{i_r}),R) \in C$
and~$\overline{d} = (d_{j_1},\dotsc,d_{j_r}) \in R$.
\krversion{Then we let $\Psi^{\overline{d},c} =
\SB z^{i_\ell}_j \SM 1 \leq \ell \leq r, j \neq j_{\ell} \SE$.}%
\longversion{Then we let: \[\Psi^{\overline{d},c} =
\SB z^{i_\ell}_j \SM 1 \leq \ell \leq r, j \neq j_{\ell} \SE. \]} %
Intuitively, the set~$\Psi^{\overline{d},c}$
contains the variables~$z^i_j$ that represent
those variable assignments in~$\alpha$ that prevent
that~$\beta$ satisfies~$c$ by assigning~$\Var{c}$ to~$\overline{d}$. %
\longversion{

}%
Then, the formula~$\psi^{Y}_{\mtext{proper}}$ ensures
that for each variable~$x_i$ exactly one value~$d_j$
is chosen in~$\beta$.
We define:
\krversion{$\psi^{Y}_{\mtext{proper}} = \bigwedge_{1 \leq i \leq n}
[ \bigvee_{1 \leq j \leq m} y^i_j \wedge
\bigwedge_{1 \leq j < j' \leq m} (\neg y^i_j \vee \neg y^i_{j'}) ]$.}%
\longversion{\[ \psi^{Y}_{\mtext{proper}} = \bigwedge_{1 \leq i \leq n}
[ \bigvee_{1 \leq j \leq m} y^i_j \wedge
\bigwedge_{1 \leq j < j' \leq m} (\neg y^i_j \vee \neg y^i_{j'}) ]. \]} %
\longversion{

}%
\longversion{\noindent}%
Next, the formula~$\psi^{Y,Z}_{\mtext{corr}}$ ensures
that~$\beta$ is indeed an extension of~$\alpha$.
\krversion{We define: $\psi^{Y,Z}_{\mtext{corr}} = \bigwedge_{1 \leq i \leq n}
\bigwedge_{1 \leq j \leq m} (z^i_j \rightarrow y^i_j)$.}%
\longversion{We define: \[ \psi^{Y,Z}_{\mtext{corr}} = \bigwedge_{1 \leq i \leq n}\ 
\bigwedge_{1 \leq j \leq m} (z^i_j \rightarrow y^i_j). \]} %
\longversion{

}%
\longversion{\noindent}%
Finally, for each~$c \in C$, the formula~$\psi^{Y}_{c}$ represents
whether~$\beta$ satisfies~$c$.
Let~$c = ((x_{i_1},\dotsc,x_{i_r}),R) \in C$.
\krversion{We define $\psi^{Y}_{c} = \bigvee_{(d_{j_1},\dotsc,d_{j_r}) \in R}
\bigwedge_{1 \leq \ell \leq r} y^{i_\ell}_{j_\ell}$.}%
\longversion{We define: \[ \psi^{Y}_{c} = \bigvee_{(d_{j_1},\dotsc,d_{j_r}) \in R}\ 
\bigwedge_{1 \leq \ell \leq r} y^{i_\ell}_{j_\ell}. \]} %
\longversion{

}%
\krversion{A full proof that~$(X,D,C,k) \in \RobustCSPSat{}$ if and only
if~$(\varphi,k) \in \AkEWSat{}$ can be found in the technical report.}
\longversion{
\noindent We now argue that~$(X,D,C,k) \in \RobustCSPSat{}$ if and only if
$(\varphi,k) \in \AkEWSat{}$.

$(\Rightarrow)$
Assume that~$(X,D,C)$ is~$k$-robustly satisfiable.
We show that~$(\varphi,k) \in \AkEWSat{}$.
Let~$\alpha : Z \rightarrow \SBs 0,1 \SEs$ be an arbitrary assignment of
weight~$k$.
If~$\alpha(z^i_j) = \alpha(z^i_{j'}) = 1$ for some~$1 \leq i \leq n$ and
some~$1 \leq j < j' \leq m$, then~$\alpha$ (and any extension of it)
satisfies~$\neg \psi^{Z}_{\mtext{proper}}$.
Therefore,~$\exists Y. \psi[\alpha]$ is true.
We can thus restrict our attention to the case where for
each~$1 \leq i \leq n$, there is at most one~$1 \leq j_i \leq m$
such that~$\alpha(z^i_{j_i}) = 1$.
Define the subset~$X' \subseteq X$ and the
instantiation~$\mu : X' \rightarrow D$
by letting~$x_i \in X'$ and~$\mu(x_i) = d_j$ if and only
if~$\alpha(z^i_j) = 1$.
Clearly,~$\mu$ assigns values to~$k$ different variables.

We distinguish two cases: either (i)~$\mu$ violates some constraint~$c \in C$,
or (ii)~$\mu$ violates no constraint in~$C$.
In case (i), it is straightforward to verify that~$\alpha$
satisfies~$\psi^{Z}_{\mtext{violate}}$.
Therefore,~$\exists Y. \psi[\alpha]$ is true.
Next, consider case~(ii).
Since~$(X,D,C)$ is~$k$-robustly satisfiable,
we know that there exists some complete instantiation~$\nu : X \rightarrow D$
that extends~$\mu$, and that satisfies all constraints~$c \in C$.
Now, define the assignment~$\beta : Y \rightarrow \SBs 0,1 \SEs$
by letting~$\beta(y^i_j) = 1$ if and only if~$\nu(x_i) = d_j$.
It is straightforward to verify
that~$\alpha \cup \beta$ satisfies~$\psi^{Y,Z}_{\mtext{corr}}$,
and that~$\beta$ satisfies~$\psi^{Y}_{\mtext{proper}}$
and~$\psi^{Y}_{c}$ for all~$c \in C$.
Therefore,~$\alpha \cup \beta$ satisfies~$\psi$.
This concludes our proof that~$(\varphi,k) \in \AkEWSat{}$.

$(\Leftarrow)$
Assume that~$(\varphi,k) \in \AkEWSat{}$.
We show that~$(X,D,C)$ is~$k$-robustly satisfiable.
Let~$X' \subseteq X$ be an arbitrary subset of size~$k$,
and let~$\mu : X' \rightarrow D$ be an arbitrary instantiation
that does not violate any constraint~$c \in C$.
Define the assignment~$\alpha : Z \rightarrow \SBs 0,1 \SEs$
by letting~$\alpha(z^i_j) = 1$ if and only if~$x_i \in X'$
and~$\mu(x_i) = d_j$.
Clearly, the assignment~$\alpha$ has weight~$k$.
Therefore, there must exist an assignment~$\beta : Y \rightarrow \SBs 0,1 \SEs$
such that~$\alpha \cup \beta$ satisfies~$\psi$.
It is straightforward to verify that~$\alpha$
satisfies~$\psi^{Z}_{\mtext{proper}} \wedge \neg \psi^{Z}_{\mtext{violate}}$.
Therefore,~$\alpha \cup \beta$ must
satisfy~$\psi^{Y,Z}_{\mtext{corr}} \wedge \psi^{Y}_{\mtext{proper}}$
and~$\psi^{Y}_{c}$ for all~$c \in C$.
Since~$\alpha \cup \beta$ satisfies~$\psi^{Y}_{\mtext{proper}}$,
we know that for each~$1 \leq i \leq n$, there is a unique~$1 \leq j_i \leq m$
such that~$\beta(y^i_{j_i}) = 1$.
Define the complete instantiation~$\nu : X \rightarrow D$
by letting~$\nu(x_i) = d_{j_i}$.
It is straightforward to verify that~$\nu$ extends~$\mu$,
since~$\alpha \cup \beta$ satisfies~$\psi^{Y,Z}_{\mtext{corr}}$.
Also, since~$\alpha \cup \beta$ satisfies~$\psi^{Y}_{c}$ for all~$c \in C$,
it follows that~$\nu$ satisfies each~$c \in C$.
Therefore,~$\nu$ is a solution of the CSP instance~$(X,D,C)$.
This concludes our proof that~$(X,D,C)$ is $k$-robustly satisfiable.
}
\end{proof}

\noindent In order to prove \AkE{}-hardness,
we need the following technical lemma.

\begin{lemma}
\label{lem:exactly-k-false}
\krversion{Let $(\varphi,k)$ be an instance of \EkAWSat{}
with $\varphi = \exists X. \forall Y. \psi$.}%
\longversion{Let $(\varphi,k)$ be an instance of \EkAWSat{}.}
\krversion{In polynomial time, we can construct an equivalent
instance $(\varphi',k)$ of \EkAWSat{} with $\varphi' = \exists X. \forall Y'. \psi'$,
such that}\longversion{In polynomial time,
we can construct an instance $(\varphi',k)$
of \EkAWSat{} with $\varphi' = \exists X. \forall Y. \psi$, such that:}
\krversion{for any assignment $\alpha : X \rightarrow \SBs 0,1 \SEs$
that has weight $m \neq k$, it holds that $\forall Y'. \psi'[\alpha]$ is true.
}\longversion{\begin{itemize}
  \item $(\varphi',k) \in \EkAWSat{}$ if and only if
    $(\varphi,k) \in \EkAWSat{}$; and
  \item for any assignment $\alpha : X \rightarrow \SBs 0,1 \SEs$ that has weight
    $m \neq k$, it holds that $\forall Y. \psi[\alpha]$ is true.
\end{itemize}}
\end{lemma}
\krversion{\begin{proof}[Proof (sketch).]
We introduce a set~$Z$ of additional universally quantified variables,
and use these to verify whether $k$ many existentially quantified variables
are set to true.
We do so by constructing a propositional formula~$\chi$
containing variables in~$X$ and~$Z$
that is satisfiable if and only if exactly~$k$ many variables in~$X$
are set to true.
We let $\chi' = (\exists Z. \chi) \rightarrow \psi$.
We then know that $\forall Y. \chi'[\alpha]$ is true for all
truth assignments~$\alpha : X \rightarrow \SBs 0,1 \SEs$ of weight
$m \neq k$.
Moreover, the formula~$\forall Y. \chi'$
is equivalent to the formula~$\forall Y \cup Z. (\neg \chi) \vee \psi$.
\end{proof}}%
\longversion{\begin{proof}
Let $(\varphi,k)$ be an instance of \EkAWSat{},
where $\varphi = \exists X. \forall Y. \psi$,
and where $X = \SBs x_1,\dotsc,x_n \SEs$.
We construct an instance $(\varphi',k)$,
with $\varphi' = \exists X. \forall Y \cup Z. \psi'$.
\krversion{We define
$Z = \SB z^i_j \SM 1 \leq i \leq k, 1 \leq j \leq n \SE$.}\longversion{We define:
\[ Z = \SB z^i_j \SM 1 \leq i \leq k, 1 \leq j \leq n \SE. \]}
Intuitively, one can think of the variables $z^i_j$ as being positioned
in a matrix with $n$ rows and $k$ columns:
the variable $z^i_j$ is placed in the $j$-th row and the $i$-th column.
We will use this matrix to verify whether exactly $k$ variables in $X$
are set to true.

\krversion{
We let $\psi' = (\psi^{X,Z}_{\mtext{corr}} \wedge
\psi^{Z}_{\mtext{row}} \wedge
\psi^{Z}_{\mtext{col}}) \rightarrow \psi$.
Here we let
$\psi^{X,Z}_{\mtext{corr}} = 
    \bigwedge_{1 \leq j \leq n}\ 
    [
      \bigwedge_{1 \leq i \leq k}
      (z^i_j \rightarrow x_j) \wedge
      (x_j \rightarrow
      \bigvee_{1 \leq i \leq k} z^i_j )
    ]$,
we let $\psi^{Z}_{\mtext{row}} = 
    \bigwedge_{1 \leq j \leq n}\ 
    \bigwedge_{1 \leq i < i' \leq k}
    (\neg z^i_j \vee \neg z^{i'}_j)$,
and we let $\psi^{Z}_{\mtext{col}} = 
    \bigwedge_{1 \leq i \leq k}
    [
      \bigwedge_{1 \leq j < j' \leq n}
      (\neg z^i_j \vee \neg z^i_{j'}) \wedge
      \bigvee_{1 \leq j \leq n} z^i_j
    ]$.
}\longversion{We define $\psi'$ as follows:
\[ \begin{array}{r l}
  \psi' = &
    (\psi^{X,Z}_{\mtext{corr}} \wedge
    \psi^{Z}_{\mtext{row}} \wedge
    \psi^{Z}_{\mtext{col}}) \rightarrow \psi; \\
  \psi^{X,Z}_{\mtext{corr}} = &
    \bigwedge\limits_{1 \leq j \leq n}\ 
    \left (
      \bigwedge\limits_{1 \leq i \leq k}
      (z^i_j \rightarrow x_j) \wedge
      (x_j \rightarrow
      \bigvee\limits_{1 \leq i \leq k} z^i_j )
    \right ); \\
  \psi^{Z}_{\mtext{row}} = &
    \bigwedge\limits_{1 \leq j \leq n}\ 
    \bigwedge\limits_{1 \leq i < i' \leq k}
    (\neg z^i_j \vee \neg z^{i'}_j); \mtext{ and} \\
  \psi^{Z}_{\mtext{col}} = &
    \bigwedge\limits_{1 \leq i \leq k}
    \left (
      \bigwedge\limits_{1 \leq j < j' \leq n}
      (\neg z^i_j \vee \neg z^i_{j'}) \wedge
      \bigvee\limits_{1 \leq j \leq n} z^i_j
    \right ). \\
\end{array} \]}

Intuitively, the formula $\psi^{X,Z}_{\mtext{corr}}$
ensures that exactly those $x_j$ are set to true
for which there exists some $z^i_j$ that is set to true.
The formula $\psi^{Z}_{\mtext{row}}$ ensures that
in each row of the matrix filled with variables $z^i_j$,
there is at most one variable set to true.
The formula $\psi^{Z}_{\mtext{col}}$ ensures
that in each column there is exactly one variable set to true.

\krversion{A full proof that $(\varphi',k')$ satisfies the required
properties can be found in the technical report.}
\longversion{
We show that $(\varphi,k) \in \EkAWSat{}$ if and only if
$(\varphi',k') \in \EkAWSat{}$.

$(\Rightarrow)$
Let~$\alpha : X \rightarrow \SBs 0,1 \SEs$ be an assignment of weight~$k$
such that~$\forall Y. \psi[\alpha]$ is true.
We show that~$\forall Y \cup Z. \psi'[\alpha]$ is true.
Let~$\beta : Y \cup Z \rightarrow \SBs 0,1 \SEs$ be an arbitrary truth assignment.
We distinguish two cases: either (i) for each~$1 \leq i \leq k$
there is a unique~$1 \leq j_i \leq n$ such that~$\beta(z^i_{j_i}) = 1$
and~$\alpha(x_{j_i}) = 1$, or (ii) this is not the case.
In case (i), it is straightforward to verify that~$\alpha \cup \beta$
satisfies~$(\psi^{X,Z}_{\mtext{corr}} \wedge
\psi^{Z}_{\mtext{row}} \wedge
\psi^{Z}_{\mtext{col}})$.
Then, since~$\forall Y. \psi[\alpha]$ is true,
we know that~$\alpha \cup \beta$ satisfies~$\psi$,
and thus that~$\alpha \cup \beta$ satisfies~$\psi'$.
In case (ii), we know that~$\alpha \cup \beta$ does not
satisfy~$(\psi^{X,Z}_{\mtext{corr}} \wedge
\psi^{Z}_{\mtext{row}} \wedge
\psi^{Z}_{\mtext{col}})$,
and thus that~$\alpha \cup \beta$ satisfies~$\psi'$.
Therefore,~$\forall Y \cup Z. \psi'[\alpha]$ is true.

$(\Leftarrow)$
Assume that there exists an assignment~$\alpha : X \rightarrow \SBs 0,1 \SEs$
of weight~$k$ such that~$\forall Y \cup Z. \psi'[\alpha]$ is true.
We show that~$\forall Y. \psi[\alpha]$ is true.
Let~$\beta : Y \rightarrow \SBs 0,1 \SEs$ be an arbitrary truth assignment.
Let~$\SB x_j \SM 1 \leq j \leq n, \alpha(x_j) = 1 \SE = \SBs x_{j_1},\dotsc,x_{j_k} \SEs$.
We construct the assignment~$\gamma : Z \rightarrow \SBs 0,1 \SEs$
by letting~$\gamma(z^i_j) = 1$ if and only if~$j = j_i$.
We know that~$\alpha \cup \beta \cup \gamma$ satisfies~$\psi'$.
Clearly,~$\alpha \cup \gamma$ satisfies~$(\psi^{X,Z}_{\mtext{corr}} \wedge
\psi^{Z}_{\mtext{row}} \wedge
\psi^{Z}_{\mtext{col}})$.
Hence,~$\alpha \cup \beta$ must satisfy~$\psi$.
Since~$\beta$ was arbitrary, we know that~$\forall Y. \psi[\alpha]$ is true.

Next, it is straightforward to verify that for any
assignment~$\alpha : X \rightarrow \SBs 0,1 \SEs$
that has weight~$m \neq k'$, it holds that
for no assignment~$\gamma : Z \rightarrow \SBs 0,1 \SEs$ 
it is the case that~$\alpha \cup \gamma$
satisfies~$(\psi^{X,Z}_{\mtext{corr}} \wedge
\psi^{Z}_{\mtext{row}} \wedge
\psi^{Z}_{\mtext{col}})$, and thus
that~$\forall Y. \psi[\alpha]$ is true.
}
\end{proof}}

\begin{theorem}
\label{thm:robustcsp-hardness}
\RobustCSPSat{} is \AkE{}-hard,
even when the domain size~$\Card{D}$ is restricted to~$2$.
\end{theorem}
\begin{proof}
We give an fpt-reduction from~$\AkEWSat(3\CNF)$ to \RobustCSPSat{}.
Let~$(\varphi,k)$ be an instance of~$\AkEWSat(3\CNF)$,
with~$\varphi = \forall X. \exists Y. \psi$,
and~$\psi = c_1 \wedge \dotsm \wedge c_u$.
By Lemma~\ref{lem:exactly-k-false}, we may assume without loss of generality
that for any assignment~$\alpha : X \rightarrow \SBs 0,1 \SEs$
of weight~$m \neq k$, we have that~$\exists Y. \psi[\alpha]$ is false.
We construct an instance~$(Z,D,C,k)$ of \RobustCSPSat{} as follows.
We define the set~$Z$ of variables
by~$Z = X \cup Y'$, where~$Y' = \SB y^i \SM y \in Y, 1 \leq i \leq 2k+1 \SE$,
and we let~$D = \SBs 0,1 \SEs$.
We will define the set~$C$ of constraints below,
by representing them as a set of clauses whose length
is bounded by~$f(k)$, for some fixed function~$f$.

The intuition behind the construction of~$C$ is the following.
We replace each variable~$y \in Y$, by~$2k+1$ copies~$y^i$ of it.
Assigning a variable~$y \in Y$ to a value~$b \in \SBs 0,1 \SEs$ will
then correspond to assigning a majority of variables~$y^i$ to~$b$,
i.e., assigning at least~$k+1$ variables~$y^i$ to~$b$.
In order to encode this transformation in the constraints of~$C$,
intuitively, we will replace each occurrence of a variable~$y$
\krversion{by the conjuction
$\psi_{y} = \bigwedge_{1 \leq i_1 < \dotsm < i_{k+1} \leq 2k+1}
(y^{i_1} \vee \dotsm \vee y^{i_{k+1}})$,
}\longversion{by the conjunction:
\[ \psi_{y} = \bigwedge\limits_{1 \leq i_1 < \dotsm < i_{k+1} \leq 2k+1}
  (y^{i_1} \vee \dotsm \vee y^{i_{k+1}}), \]}
and replace each occurrence of a literal~$\neg y$
by a similar conjunction.
We will then multiply the resulting formula out into CNF.
Note that whenever a majority of variables~$y^i$ is set to~$b \in \SBs 0,1 \SEs$,
then the formula~$\psi_{y}$ will also evaluate to~$b$.

In the construction of~$C$,
we will directly encode the CNF formula that is a result
of the transformation described above.
For each literal~$l = y \in Y$, let~$l^i$ denote~$y^i$,
and for each literal~$l = \neg y$ with~$y \in Y$, let~$l^i$ denote~$\neg y^i$.
For each literal~$l$ over the variables~$X \cup Y$,
we define a set~$\sigma(l)$ of clauses:
\longversion{\[ \sigma(l) = \begin{dcases*}
  \SB (l^{i_1} \vee \dotsm \vee l^{i_{k+1}}) \SM 1 \leq i_1 < \dotsm < i_{k+1} \leq 2k+1 \SE
  & if~$l$ is a literal over~$Y$; \\
  \SBs l \SEs
  & if~$l$ is a literal over~$X$. \\
\end{dcases*} \]
}\krversion{
$\sigma(l) = (l^{i_1} \vee \dotsm \vee l^{i_{k+1}})
\SM 1 \leq i_1 < \dotsm < i_{k+1} \leq 2k+1 \SE$
if $l$ is a literal over $Y$,
and $\sigma(l) = l$ if $l$ is a literal over $X$.
}

Note that for each literal~$l$, it holds that~$\Card{\sigma(l)} \leq g(k) = \binom{2k+1}{k+1}$.
Next, for each clause~$c_i = l^i_1 \vee l^i_2 \vee l^i_3$ of~$\psi$,
we introduce to~$C$ a set $\sigma(c_i)$ of clauses:
\longversion{\[ \sigma(c_i) = \SB d_1 \vee d_2 \vee d_3 \SM
  d_1 \in \sigma(l^i_1), d_2 \in \sigma(l^i_2), d_3 \in \sigma(l^i_3) \SE. \]}
\krversion{$\sigma(c_i) = \SB d_1 \vee d_2 \vee d_3 \SM
  d_1 \in \sigma(l^i_1), d_2 \in \sigma(l^i_2), d_3 \in \sigma(l^i_3) \SE$.}
Note that~$\Card{\sigma(c_i)} \leq g(k)^3$.
Formally, we let~$C$ be the set of constraints
corresponding to the set~$\bigcup_{1 \leq i \leq u} \sigma(c_i)$ of clauses.
Since each such clause is of length at most~$3(k+1)$,
representing a clause by means of a constraint can be done
by specifying~$\leq 2^{3(k+1)}-1$ tuples,
i.e., all tuples satisfying the clause.
Therefore, the instance~$(Z,D,C,k)$ can be constructed
in fpt-time.
\krversion{A full proof that~$(\varphi,k) \in \AkEWSat(3\CNF)$ if and only
if~$(Z,D,C,k) \in \RobustCSPSat{}$ can be found in the
technical report.}\longversion{
We now argue that~$(\varphi,k) \in \AkEWSat(3\CNF)$ if and only
if~$(Z,D,C,k) \in \RobustCSPSat{}$.

$(\Rightarrow)$
Assume that~$(\varphi,k) \in \AkEWSat(3\CNF)$.
We show that~$(Z,D,C)$ is $k$-robustly satisfiable.
Let~$\mu : Z \rightarrow D$ be an arbitrary partial assignment
with~$\Card{\Dom{\mu}} = k$
that does not violate any constraint in~$C$.
We know that~$\Card{\Dom{\mu} \cap X} \leq k$,
and in particular that~$\Card{\SB x \in \Dom{\mu} \cap X \SM \mu(x) = 1 \SE} = m \leq k$.
Now define the assignment~$\alpha : X \rightarrow \SBs 0,1 \SEs$ as follows.
For any~$x \in X$, if~$x \in \Dom{\mu}$, then
let~$\alpha(x) = \mu(x)$.
Also, for~$k-m$ many variables~$x' \in X \backslash \Dom{\mu}$,
we let~$\alpha(x') = 1$.
For all other variables~$x' \in X \backslash \Dom{\mu}$,
we let~$\alpha(x') = 0$.
Then~$\alpha$ has weight~$k$.
Therefore, there must exist an assignment~$\beta : Y \rightarrow \SBs 0,1 \SEs$
such that~$\psi[\alpha \cup \beta]$ is true.
Now, define the assignment~$\nu : Z \rightarrow D$ extending~$\mu$
as follows.
For each~$z \in \Dom{\mu}$, we let~$\nu(z) = \mu(z)$.
For each~$x \in X \backslash \Dom{\mu}$, we let~$\nu(x) = \alpha(x)$.
For each~$y^i \in Y' \backslash \Dom{\mu}$, we let~$\nu(y^i) = \beta(y)$.
It is straightforward to verify that for
each~$y \in Y$,~$\beta(y) = b \in \SBs 0,1 \SEs$ if and only
if~$\nu$ sets at least~$k+1$ variables~$y^i$ to~$b$.
Using this fact, and the fact that~$\alpha \cup \beta$ satisfies
each clause~$c_i$ of~$\psi$,
it is straightforward to verify that~$\nu$ satisfies each constraint of~$C$,
and therefore that~$\nu$ is a solution of the CSP instance~$(Z,D,C)$.
This concludes our proof that~$(Z,D,C,k) \in \RobustCSPSat{}$.

$(\Leftarrow)$
Now, assume that~$(Z,D,C)$ is~$k$-robustly satisfiable.
We show that~$(\varphi,k) \in \AkEWSat(3\CNF)$.
Let~$\alpha : X \rightarrow \SBs 0,1 \SEs$ be an arbitrary assignment
of weight~$k$.
Now define the partial assignment~$\mu : Z \rightarrow D$
by letting~$x \in \Dom{\mu}$ and~$\mu(x) = \alpha(x)$ if and only
if~$\alpha(x) = 1$, for all~$x \in X$.
Clearly,~$\Card{\Dom{\mu}} = k$, and~$\mu$ does not violate
any constraints of~$C$.
Therefore, we know that there exists an extension~$\nu : Z \rightarrow D$
of~$\mu$ that is a solution for the CSP instance~$(Z,D,C)$.
For all~$x \in X \backslash \Dom{\mu}$ it holds that~$\nu(x) = 0$.
If this weren't the case, this would violate our assumption that for all
assignments~$\alpha' : X \rightarrow \SBs 0,1 \SEs$ with weight~$m \neq k$
we have that~$\exists Y. \psi[\alpha']$ is false.
Therefore, we know that~$\nu$ coincides with~$\alpha$ on the variables~$X$.
Now, we define the assignment~$\beta : Y \rightarrow \SBs 0,1 \SEs$
by letting~$\beta(y) = 1$ if and only if for at least~$k+1$ different~$y^i$
it holds that~$\nu(y^i) = 1$.

We then verify that~$\psi[\alpha \cup \beta]$ is true.
Let~$c_i$ be a clause of~$\psi$.
We know that~$\nu$ satisfies all clauses in~$\sigma(c_i)$.
Assume that~$\alpha \cup \beta$ does not
satisfy~$c_i = l^i_1 \vee l^i_2 \vee l^i_3$.
Then,~$\alpha \cup \beta$ sets all literals~$l^i_1$,~$l^i_2$ and~$l^i_3$ to false.
Therefore, for any literal~$l^i_j$ over~$Y$,~$\nu$ sets at least~$k+1$ copies
of~$l^i_j$ to false.
From this, we can conclude that there is some clause in~$\sigma(c_i)$
that is set to false by~$\nu$, which is a contradiction.
Therefore, we know that~$\psi[\alpha \cup \beta]$ is true.
This concludes our proof that~$(\varphi,k) \in \AkEWSat(3\CNF)$.}
\end{proof}

%%%
%%%
%%%
\begin{proposition}
\label{prop:smallcliqueext-hardness}
\SmallCliqueExt{} is \AEkW{1}-hard.
\end{proposition}
\begin{proof}
We give an fpt-reduction from $\AEkWSat(2\CNF)$
to \SmallCliqueExt{}.
Let~$(\varphi,k)$ be an instance of $\AEkWSat(2\CNF)$,
where~$\varphi = \forall Y. \exists X. \psi$.
By step 2 in the proof of Theorem~\ref{thm:stark-2cnf-hardness},
we may assume without loss of generality that~$\psi$ is
antimonotone in~$X$,
i.e., all literals of~$\psi$
that contain variables in~$X$ are negative.

We construct an instance~$(G,V',k)$ of $\SmallCliqueExt{}$ as follows.
We define:
\begin{align*}
  G =&\ (V,E); \\
  V' =&\ \SB v_{y}, v_{\neg y} \SM y \in Y \SE; \\
  V =&\ V' \cup \SB v_{x} \SM x \in X \SE; \\
  E =&\ E_{Y} \cup E_{XY}; \\
  E_{Y} =&\ 
    \SB \SBs v_{y},v' \SEs \SM y \in Y, v' \in V', v' \neq v_{\neg y} \SE\ \cup \\
    &\ \SB \SBs v_{\neg y},v' \SEs \SM y \in Y, v' \in V', v' \neq v_{y} \SE; \mtext{ and} \\
  E_{XY} =&\ 
    \SB \SBs v_{x},v_{x'} \SEs \SM x \in X, x' \in X, \SBs \neg x, \neg x' \SEs \not\in \psi \SE\ \cup \\
    &\ \SB \SBs v_{x},v_{y} \SEs \SM x \in X, y \in Y, \SBs \neg x, \neg y \SEs \not\in \psi \SE\ \cup \\
    &\ \SB \SBs v_{x},v_{\neg y} \SEs \SM x \in X, y \in Y, \SBs \neg x, y \SEs \not\in \psi \SE.
\end{align*}
We claim that~$(\varphi,k) \in \AEkWSat(2\CNF)$ if and only
if~$(G,V',k) \in \SmallCliqueExt{}$.

$(\Rightarrow)$
Assume that~$(\varphi,k) \in \AEkWSat(2\CNF)$.
Let~$C \subseteq V'$ be an arbitrary clique of~$G$.
It suffices to consider maximal cliques~$C$,
i.e., assume there is no clique~$C'$
such that~$C \subsetneq C' \subseteq V'$.
If a maximal clique~$C$ can be extended with~$k$ elements in~$V$
to another clique,
then clearly this holds for all its subsets as well.

We show that for all~$y \in Y$, either~$v_{y} \in C$
or~$v_{\neg y} \in C$.
Assume the contrary, i.e., assume that for some~$y \in Y$
it holds that~$v_{y} \not\in C$ and~$v_{\neg y} \not\in C$.
Then~$C \cup \SBs v_{y} \SEs \supsetneq C$ is a clique.
This contradicts our assumption that~$C$ is maximal.
Now define the assignment~$\alpha_{C} : Y \rightarrow \SBs 0,1 \SEs$
by letting~$\alpha_{C}(y) = 1$ if and only if~$v_{y} \in C$.
We then know that there exists an assignment~$\beta$ to the variables~$X$
of weight~$k$ such that~$\alpha_{C} \cup \beta$ satisfies~$\psi$.
Consider the set~$D_{\beta} = \SB v_{x} \in V \SM x \in X, \beta(x) =
1 \SE \subseteq V$.
Since~$\beta$ has weight~$k$, we know that~$\Card{D_{\beta}} = k$.
Also,~$D_{\beta} \cap C = \emptyset$,
therefore~$\Card{C \cup D_{\beta}} = \Card{C} + k$.
By the construction of~$E$,
and by the fact that~$\alpha_C \cup \beta$ satisfies~$\psi$,
it follows that~$C \cup D_{\beta}$ is a clique.
To see this, assume the contrary, i.e.,
assume that there exist~$v,v' \in C \cup D_{\beta}$
such that~$\SBs v,v' \SEs \not\in E$.
We assume that~$v = v_{x}$ and~$v' = v_{x'}$ for~$x,x' \in X$.
The other cases are analogous.
Then~$\SBs \neg x, \neg x' \SEs \in \psi$.
Since~$v_{x},v_{x'} \in D_{\beta}$, we know that~$\beta(x) = \beta(x') = 1$.
Then~$\alpha_C \cup \beta$ does not satisfy~$\psi$,
which is a contradiction.
From this we can conclude that~$C \cup D_{\beta}$ is a~$(\Card{C}+k)$-clique
of~$G$, and thus that~$(G,V',k) \in \SmallCliqueExt{}$.

$(\Leftarrow)$
Assume~$(G,V',k) \in \SmallCliqueExt{}$.
We show that for all assignments~$\alpha : Y \rightarrow \SBs 0,1 \SEs$
there exists an assignment~$\beta : X \rightarrow \SBs 0,1 \SEs$ of weight~$k$
such that~$\psi[\alpha \cup \beta]$ evaluates to true.
Let~$\alpha : Y \rightarrow \SBs 0,1 \SEs$ be an arbitrary assignment.
Consider~$C_{\alpha} = \SB v_{y} \in V \SM y \in Y, \alpha(y) = 1 \SE \cup
\SB v_{\neg y} \SM y \in Y, \alpha(y) = 0 \SE \subseteq V'$.
By construction of~$E_{Y} \subseteq E$, it follows
that~$C_{\alpha}$ is a clique of~$G$.
We then know that there exists some set~$D \subseteq V$
of size~$k$ such that~$C_{\alpha} \cup D$ is a clique.
We show that~$D \subseteq V \backslash V'$.
To show the contrary, assume that this is not the case,
i.e., assume that there exists some~$v \in D \cap V'$.
By construction of~$C_{\alpha}$, we know that for each~$y \in Y$,
either~$v_y \in C$ or~$v_{\neg y} \in C_{\alpha}$.
We also know that~$v \in \SBs v_{y'}, v_{\neg y'} \SEs$
for some~$y' \in Y$.
Assume that~$v = v_{y'}$; the other case is analogous.
If~$v_{y'} \in C_{\alpha} \cap D$, then~$C_{\alpha} \cup D$ cannot be a clique
of size~$\Card{C_{\alpha}} + k = \Card{C_{\alpha}} + \Card{D}$.
Therefore,~$v_{y'} \not\in C_{\alpha}$, and thus~$v_{\neg y'} \in C_{\alpha}$.
We then know that~$\SBs v_{y'}, v_{\neg y'} \SEs \subseteq C_{\alpha} \cup D$.
However,~$\SBs v_{y'},v_{\neg y'} \SEs \not\in E$,
and therefore~$C_{\alpha} \cup D$ is not a clique.
This is a contradiction, and thus~$D \subseteq V \backslash V'$.

We define~$\beta_{D} : X \rightarrow \SBs 0,1 \SEs$ as follows.
We let~$\beta_{D}(x) = 1$ if and only if~$v_{x} \in D$.
Clearly,~$\beta_D$ is of weight~$k$.
We show that~$\psi[\alpha \cup \beta_{D}]$ evaluates to true.
Consider an arbitrary clause~$c$ of~$\psi$.
Assume that~$c = \SBs \neg x, y \SEs$ for some~$x \in X$
and some~$y \in Y$; the other cases are analogous.
To show the contrary, assume that~$\alpha \cup \beta_{D}$
does not satisfy~$c$, i.e.,~$(\alpha \cup \beta_{D})(x) = 1$
and~$(\alpha \cup \beta_{D})(y) = 0$.
Then~$v_x \in D$ and~$v_{\neg y} \in C$.
However~$\SBs v_x, v_{\neg y} \SEs \not\in E$, and
thus~$C_{\alpha} \cup D$ is not a clique.
This is a contradiction, and therefore we can conclude
that~$\alpha \cup \beta_{D}$
satisfies~$c$.
Since~$c$ was arbitrary, we know that~$\psi[\alpha \cup \beta_{D}]$
evaluates to true.
Thus,~$(\varphi,k) \in \AEkWSat(2\CNF)$.
\end{proof}

\begin{proposition}
\label{prop:smallcliqueext-membership}
\SmallCliqueExt{} is in \AEkW{1}.
\end{proposition}
\begin{proof}
We give an fpt-reduction from \SmallCliqueExt{}
to $\AEkWSat(\Gamma_{1,3})$.
Let~$(G,V',k)$ be an instance \SmallCliqueExt{},
with~$G = (V,E)$.
We construct an equivalent instance~$(C,k)$ of $\AEkWSat(\Gamma_{1,3})$.
We will define a circuit~$C$,
that has universally quantified variables~$Y$ and
existentially quantified variables~$X$.
We present the circuit~$C$ as a propositional formula.
We define:
\begin{align*}
  Y =&\ \SB y_{v'} \SM v' \in V' \SE; \\
  X =&\ \SB x_{v} \SM v \in V \SE; \\
  C =&\ C_1 \vee (C_2 \wedge C_3); \\
  C_1 =&\ \underset{\SBs v_1,v_2 \SEs \in (V' \times V') \backslash E}{\bigvee}
      \left ( y_{v_1} \wedge y_{v_2} \right ); \\
  C_2 =&\ 
      \underset{v' \in V'}{\bigwedge} \left ( \neg y_{v'} \vee \neg x_{v'} \right ); \\
  C_3 =&\ 
      \underset{e \in (V \times V) \backslash E}{\bigwedge} \chi_{e}; \\
  \chi_e =&\ 
    \underset{z_1 \in \Theta(v_1)}{\underset{z_2 \in \Theta(v_2)}{\bigwedge}}
    \left ( \neg z_1 \vee \neg z_2 \right ),
    & \mtext{for all $\SBs v_1,v_2 \SEs = e \in (V \times V) \backslash E$; and} \\
  \Theta(v) =&\ 
    \begin{dcases*}
      \SBs x_{v}, y_{v} \SEs & if $v \in V'$, \\
      \SBs x_{v} \SEs & otherwise, \\
    \end{dcases*} & \mtext{for all $v \in V$.}
\end{align*}
We show that~$(G,V',k) \in \SmallCliqueExt{}$ if and only
if~$(C,k) \in \AEkWSat(\Gamma_{1,3})$.

$(\Rightarrow)$
Assume that~$(G,V',k) \in \SmallCliqueExt{}$.
This means that for each clique~$F \subseteq (V' \times V')$
there exists some set~$D \subseteq V$ of vertices
such that~$F \cup D$ is a~$(\Card{F} + k)$-clique.
We show that for each assignment~$\alpha : Y \rightarrow \SBs 0,1 \SEs$
there exists an assignment~$\beta : X \rightarrow \SBs 0,1 \SEs$ of weight~$k$
such that~$C[\alpha \cup \beta]$ evaluates to true.
Let~$\alpha : Y \rightarrow \SBs 0,1 \SEs$ be an arbitrary assignment.
We define~$F_{\alpha} = \SB v \in V' \SM \alpha(y_{v}) = 1 \SE$.
We distinguish two cases: either (i)~$F_{\alpha}$ is not a clique in~$G$,
or (ii)~$F_{\alpha}$ is a clique in~$G$.
In case (i), we know that there exist~$v_1,v_2 \in F_{\alpha}$ such
that~$\SBs v_1,v_2 \SEs \not\in E$.
Then~$\alpha$ satisfies~$(y_{v_1} \wedge y_{v_2})$,
and therefore~$\alpha$ satisfies~$C_1$ and~$C$ as well.
Thus for every assignment~$\beta : X \rightarrow \SBs 0,1 \SEs$ of
weight~$k$,~$\alpha \cup \beta$ satisfies~$C$.

Consider case (ii).
Let~$m = \Card{F_{\alpha}}$.
We know that there exists a subset~$D \subseteq V$ of vertices
such that~$F_{\alpha} \cup D$ is an~$(m+k)$-clique.
Define the assignment~$\beta : X \rightarrow \SBs 0,1 \SEs$
by letting~$\beta(x_v) = 1$ if and only if~$v \in D$.
Clearly,~$\beta$ has weight~$k$.
We show that~$\alpha \cup \beta$ satisfies~$C$.
We know there is no~$v \in V$, such that~$v \in F_{\alpha} \cap D$,
since otherwise~$F_{\alpha} \cup D$ could not be a clique of size~$m+k$.
Therefore,~$\alpha \cup \beta$ satisfies~$C_2$.
Since~$F_{\alpha} \cup D$ is a clique, we know that~$\alpha \cup \beta$
satisfies~$C_3$ as well.
Thus~$\alpha \cup \beta$ satisfies~$C$.
This concludes our proof that~$(C,k) \in \AEkWSat(\Gamma_{1,3})$.

$(\Leftarrow)$
Assume that~$(C,k) \in \AEkWSat(\Gamma_{1,3})$.
This means that for each assignment~$\alpha : Y \rightarrow \SBs 0,1 \SEs$
there exists an assignment~$\beta : X \rightarrow \SBs 0,1 \SEs$ of weight~$k$
such that~$C[\alpha \cup \beta]$ evaluates to true.
Let~$F \subseteq (V' \times V')$ be an arbitrary clique,
and let~$m = \Card{F}$.
We show that there exists a set~$D \subseteq V$ of vertices
such that~$F \cup D$ is an $(m+k)$-clique.
Let~$\alpha_{F} : Y \rightarrow \SBs 0,1 \SEs$ be the assignment defined
by letting~$\alpha_{F}(y_v) = 1$ if and only if~$v \in F$.
We know that there must exist an
assignment~$\beta : X \rightarrow \SBs 0,1 \SEs$
of weight~$k$ such that~$C[\alpha_{F} \cup \beta]$ evaluates to true.
Since~$F$ is a clique, it is straightforward to verify that~$\alpha_{F} \cup \beta$
does not satisfy~$C_1$.
Therefore,~$\alpha_{F} \cup \beta$ must satisfy~$C_2$ and~$C_3$.
Consider~$D_{\beta} = \SB v \in V \SM \beta(x_v) = 1 \SE$.
Since~$\beta$ has weight~$k$,~$\Card{D} = k$.
Because~$\alpha_{F} \cup \beta$ satisfies~$C_2$,
we know that~$F \cap D_{\beta} = \emptyset$,
and thus that~$\Card{F \cup D_{\beta}} = m+k$.
Because~$\alpha_{F} \cup \beta$ satisfies~$C_3$,
we know that~$F \cup D_{\beta}$ is a clique in~$G$.
Since~$F$ was arbitrary, we can conclude
that~$(G,V',k) \in \SmallCliqueExt{}$.
\end{proof}

%%%
%%%
%%%
\begin{proposition}
\label{prop:3ce-degree}
\ThreeColExt{}\Degree{} is \para{}\PiP{2}-complete.
\end{proposition}
\begin{proof}
It is known that
\ThreeColExt{} is already \PiP{2}-hard when restricted
to graphs of degree~$4$~\cite{AjtaiFaginStockmeyer00}.
From this, it follows immediately that \ThreeColExt{} is
\para{}\PiP{2}-hard~\cite{FlumGrohe03}.
Membership in \para{}\PiP{2} follows directly from
the \PiP{2}-membership of \ThreeColExt{}.
\end{proof}

%%%
\begin{proposition}
\label{prop:3ce-numleaves}
\ThreeColExt{}\NumLeaves{} is \para{}\NP{}-complete.
\end{proposition}
\begin{proof}
To show membership in \para{}\NP{},
we give an fpt-reduction to \thSAT{}.
Let~$(G,m)$ be an instance of \ThreeColExt{}\NumLeaves{},
where~$k$ denotes the number of leaves of~$G$.
We construct a propositional formula that is satisfiable
if and only if~$(G,m) \in \ThreeColExt{}\NumLeaves{}$.
Let~$V'$ denote the set of leaves of~$G$,
and let~$C$ be the set of all $3$-colorings
that assigns~$m$ many vertices in~$V'$ to any color.
We know that~$\Card{C} \leq 3^k$.
For each~$c \in C$, we know that the problem of deciding
whether~$c$ can be extended to a proper $3$-coloring of~$G$
is in \NP{}.
Therefore, for each~$c$, we can construct a propositional
formula~$\varphi_c$ that is satisfiable if and only if~$c$ can be extended
to a proper $3$-coloring of~$G$.
We may assume without loss of generality that for any distinct~$c,c' \in C$
it holds that~$\varphi_c$ and~$\varphi_{c'}$ are variable-disjoint.
We then let~$\varphi = \bigwedge\nolimits_{c \in C} \varphi_c$.
Clearly,~$\varphi$ is satisfiable
if and only~$(G,m) \in \ThreeColExt{}$.
Moreover,~$\varphi$ is of size~$O^*(3^k)$.

\sloppypar
Hardness for \para{}\NP{} follows directly from the \NP{}-hardness
of deciding whether a given graph has a proper $3$-coloring,
which corresponds to the restriction of \ThreeColExt{}
to instances with~$k = 0$, i.e., to graphs that have no leaves.
\end{proof}

%%%
\begin{proposition}
\label{prop:3ce-numcolleaves}
\ThreeColExt{}\NumColLeaves{} is \AkE{}-complete.
\end{proposition}
\begin{proof}
To show membership,
we give an fpt-reduction from \ThreeColExt{}\NumColLeaves{}
to \AkEMC{}.
Let~$(G,m)$ be an instance of \ThreeColExt{}\NumColLeaves{},
where~$V'$ denotes the set of leaves of~$G$,
and where~$k = m$ is the number of edges that can be pre-colored.
Moreover, let~$V' = \SBs v_1,\dotsc,v_n \SEs$
and let~$V = V' \cup \SBs v_{n+1},\dotsc,v_u \SEs$.
We construct an instance~$(\mathcal{A},\varphi)$
of \AkEMC{}.
We define the domain~$A = \SB a_{v,i} \SM v \in V', 1 \leq i \leq 3 \SE \cup
\SBs 1,2,3 \SEs$.
Next, we define~$C^{\mathcal{A}} = \SBs 1,2,3 \SEs$,
$S^{\mathcal{A}} =
\SB (a_{v,i},a_{v,i'}) \SM v \in V', 1 \leq i,i' \leq 3 \SE$,
and~$F^{\mathcal{A}} = \SB (j,j') \SM 1 \leq j,j' \leq 3, j \neq j' \SE$.
Then, we can define the formula~$\varphi$,
by letting~$\varphi =
  \forall x_1,\dotsc,x_k.
  \exists y_1,\dotsc,y_u.
(\psi_{1} \rightarrow (\psi_{2} \wedge
\psi_{3} \wedge \psi_{4}))$,
where~$\psi_{1} =
\bigwedge\nolimits_{1 \leq j < j' \leq k} \neg S(x_i,x_{i'})$,
and~$\psi_{2} =
\bigwedge\nolimits_{1 \leq j \leq u} C(y_j)$,
and~$\psi_{3} =
  \bigwedge\nolimits_{v_j \in V', 1 \leq i \leq 3}
  ((\bigvee\nolimits_{1 \leq \ell \leq k} (x_{\ell} = a_{v_j,i}))
  \rightarrow (y_j = i))$,
and~$\psi_{4} =
  \bigwedge\nolimits_{\SBs v_j,v_{j'} \SEs \in E}
  F(y_j,y_{j'})$.
It is straightforward to verify that~$(G,m) \in \ThreeColExt{}$
if and only if~$\mathcal{A} \models \varphi$.

Intuitively, the assignments to the variables~$x_i$ correspond
to the pre-colorings of the vertices in~$V'$.
This is done by means of elements~$a_{v,i}$, which represent
the coloring of vertex~$v$ with color~$i$.
The subformula~$\psi_{1}$ is used to disregard any assignments
where variables~$x_i$ are not assigned to the intended elements.
Moreover, the assignments to the variables~$y_i$ correspond
to a proper $3$-coloring extending the pre-coloring.
The subformula~$\psi_{2}$ ensures that the variables~$y_i$
are assigned to a color in~$\SBs 1,2,3 \SEs$,
the subformula~$\psi_{3}$ ensures that
the coloring encoded by the assignment to the variables~$y_i$
extends the pre-coloring encoded by the assignment to the
variables~$x_i$, and the subformula~$\psi_{4}$ ensures
that the coloring is proper.

To show hardness,
it suffices to observe that the polynomial-time reduction
from \co{}\QSat{2} to \ThreeColExt{} to show \PiP{2}-hardness
\cite[Appendix]{AjtaiFaginStockmeyer00}
can be seen as an fpt-reduction from \AkEWSat{}
to \ThreeColExt{}\NumColLeaves{}
(which leaves the parameter unchanged).
\end{proof}

%%%
\begin{proposition}
\label{prop:3ce-numuncolleaves}
\ThreeColExt{}\NumUncolLeaves{} is \para{}\PiP{2}-complete.
\end{proposition}
\begin{proof}
We know that \ThreeColExt{} is already \PiP{2}-hard when restricted
to instances where~$n = m$ (and thus
where~$k = 0$)~\cite{AjtaiFaginStockmeyer00}.
From this, it follows immediately that \ThreeColExt{} is
\para{}\PiP{2}-hard~\cite{FlumGrohe03}.
Membership in \para{}\PiP{2} follows directly from
the \PiP{2}-membership of \ThreeColExt{}.
\end{proof}

%%%
%%%
%%%
\begin{proposition}
\label{prop:ake-fd-membership}
Let~$\tau$ be an arbitrary relational vocabulary,
and let~$\tau' \subseteq \tau$ be a subvocabulary of~$\tau$.
Let~$\varphi(X)$ be a first-order formula over~$\tau$
with a free relation variable~$X$ of arity~$s$.
The problem $\AkEFDform{(\tau,\tau')}{\varphi}$ is in \AkE{}.
\end{proposition}

\begin{proof}
We show membership in \AkE{} by giving an fpt-reduction to \AkEWSat{}.
Let~$(\mathcal{B},k)$ be an instance of~$\AkEFDform{(\tau,\tau')}{\varphi}$,
where~$\mathcal{B}$ is a $\tau'$-structure with domain~$B$.
We construct an instance~$(\psi,k)$ of~$\AkEWSat{}$ as follows.
We consider the set~$W = \SB w_{R,\overline{b}} \SM R \in \tau \backslash \tau',
\overline{b} \in B^{\arity{R}} \SE$ of variables that will represent
the interpretation of the relation symbols~$R \in \tau \backslash \tau'$.
In addition, we consider the set~$Y = \SB y_{\overline{b}} \SM \overline{b} \in B^s \SE$
of variables that will represent the choice of the relation~$S \subseteq B^{s}$.
We will let:
\[ \psi = \forall W. \exists Y. \psi^{\varphi}, \]
where~$\psi^{\varphi}$ is defined below by induction
on the structure of~$\varphi$.
Concretely, for each formula~$\chi$ such that~$\chi = \chi'[\beta]$
for some subformula~$\chi'$ of~$\varphi$ and some
assignment~$\beta : \Free{\chi'} \rightarrow B$,
we will define the formula~$\psi^{\chi}$.
We assume without loss of generality that~$\varphi$
contains only existential quantifiers and only the connectives~$\neg$
and~$\wedge$.
We distinguish several cases.
Consider the case where~$\chi = R(b_1,\dotsc,b_n)$
for some~$b_1,\dotsc,b_n \in B$,
and where~$R \in \tau'$.
We define:
\[ \psi^{\chi} = \begin{dcases*}
  \top & if~$(b_1,\dotsc,b_n) \in R^{\mathcal{B}}$, \\
  \bot & otherwise. \\
\end{dcases*} \]
Next, consider the case where~$\chi = R(b_1,\dotsc,b_n)$
for some~$b_1,\dotsc,b_n \in B$,
and where~$R \in \tau \backslash \tau'$.
We let~$\overline{b} = (b_1,\dotsc,b_n)$.
We then define:
\[ \psi^{\chi} = w_{R,\overline{b}} . \]
Next, consider the case where~$\chi = X(b_1,\dotsc,b_s)$
for some~$b_1,\dotsc,b_s \in B$.
We let~$\overline{b} = (b_1,\dotsc,b_s)$, and
we define:
\[ \psi^{\chi} = y_{\overline{b}}. \]
Next, consider the case where~$\chi = \exists z. \chi'$.
We define:
\[ \psi^{\chi} = \bigvee\limits_{b \in B} \psi^{\chi'[z \mapsto b]}. \]
Finally,
in the case where~$\chi = \chi_1 \wedge \chi_2$, we
define~$\psi^{\chi} = \psi^{\chi_1} \wedge \psi^{\chi_2}$,
and in the case where~$\chi = \neg \chi_1$,
we let~$\psi^{\chi} = \neg \psi^{\chi_1}$.
We claim that~$(\psi,k) \in \AkEWSat{}$ if and only
if~$(\mathcal{B},k) \in \AkEFDform{(\tau,\tau')}{\varphi}$.

$(\Rightarrow)$
Assume that for all assignments~$\alpha : W \rightarrow \SBs 0,1 \SEs$
of weight~$k$ the
formula~$\exists Y. \psi^{\varphi}[\alpha]$
is true.
Now let~$\mathcal{A}$ be an arbitrary $\tau$-structure extending~$\mathcal{B}$
with weight~$k$. We show that there exists some~$S \subseteq B^s$
such that~$\mathcal{A} \models \varphi(S)$.
We define the assignment~$\alpha_{\mathcal{A}} : W \rightarrow \SBs 0,1 \SEs$
as follows.
For each~$R \in \tau \backslash \tau'$ and each~$\overline{b} \in B^{\arity{R}}$,
we let~$\alpha_{\mathcal{A}}(w_{R,\overline{b}}) = 1$ if and only
if~$\overline{b} \in R^{\mathcal{A}}$.
Since~$\mathcal{A}$ extends~$\mathcal{B}$ with weight~$k$,
we know that~$\alpha_{\mathcal{A}}$ has weight~$k$.
Therefore, we know that there must exist some
assignment~$\beta : Y \rightarrow \SBs 0,1 \SEs$ such
that~$\psi^{\varphi}[\alpha_{\mathcal{A}} \cup \beta]$ is true.
We define the relation~$S_{\beta} \subseteq B^s$ as follows.
For each~$\overline{b} \in B^s$, we let~$\overline{b} \in S_{\beta}$
if and only if~$\beta(y_{\overline{b}}) = 1$.
It is now straightforward to verify by induction
on the structure of~$\varphi$ that~$\mathcal{A} \models \varphi(S_{\beta})$
if and only if~$\psi^{\varphi}[\alpha_{\mathcal{A}} \cup \beta]$ is true.
This concludes our proof
that~$(\mathcal{B},k) \in \AkEFDform{(\tau,\tau')}{\varphi}$.

$(\Leftarrow)$
Conversely, assume that for every $\tau$-structure~$\mathcal{A}$
that extends~$\mathcal{B}$ with weight~$k$
there exists some relation~$S \subseteq B^s$ such
that~$\mathcal{A} \models \varphi(S)$.
Now let~$\alpha : W \rightarrow \SBs 0,1 \SEs$ be an arbitrary
assignment of weight~$k$.
We define the $\tau$-structure~$\mathcal{A}_{\alpha}$ as follows.
For each~$R \in \tau \backslash \tau'$ we let~$R^{\mathcal{A}_{\alpha}}
= \SB \overline{b} \in B^{\arity{R}} \SM \alpha(w_{R,\overline{b}}) = 1 \SE$.
Since~$\alpha$ has weight~$k$, we know that~$\mathcal{A}_{\alpha}$
extends~$\mathcal{B}$ with weight~$k$.
Therefore, we know that there exists some relation~$S \subseteq B^s$
such that~$\mathcal{A}_{\alpha} \models \varphi(S)$.
We now define the assignment~$\beta_{S} : Y \rightarrow \SBs 0,1 \SEs$
as follows.
For each~$\overline{b} \in B^s$ we let~$\beta_{S}(y_{\overline{b}}) = 1$
if and only if~$\overline{b} \in S$.
It is now straightforward to verify by induction
on the structure of~$\varphi$ that~$\mathcal{A}_{\alpha} \models \varphi(S)$
if and only if~$\psi^{\varphi}[\alpha \cup \beta_{S}]$ is true.
This concludes our proof
that~$(\psi,k) \in \AkEWSat{}$.
\end{proof}

\begin{proposition}
\label{prop:ake-fd-hardness}
The problem $\AkEFDform{(\tau,\tau')}{\varphi}$ is \AkE{}-hard,
for some relational vocabulary~$\tau$,
some subvocabulary~$\tau' \subseteq \tau$ of~$\tau$,
and some first-order formula~$\varphi(Z) \in \Pi_2$ over~$\tau$
with a free relation variable~$Z$ of arity~$s$.
\end{proposition}
\begin{proof}
We let~$\tau'$ contain the unary relation symbols~$G$ and~$C$,
and the binary relation symbols~$C^{F}_{-}$,~$C^{F}_{+}$,~$C^{G}_{-}$
and~$C^{G}_{+}$.
We let~$\tau \supseteq \tau'$ in addition contain the unary relation symbol~$F$.
Furthermore, we let~$\varphi(Z) \in \Pi_2$ be the first-order formula
defined as follows:
\begin{eqnarray*}
  \varphi = & \forall v. (Zv \rightarrow Gv) \wedge
    \forall v. ( Cv \rightarrow (\exists w. \psi(v,w,Z))  ); \mtext{ and} \\
  \psi(v,w,Z) = &
    \left ( \neg Fw \wedge C^{F}_{-}(v,w) \right ) \vee
    \left ( Fw \wedge C^{F}_{+}(v,w) \right ) \vee
    \left ( \neg Zw \wedge C^{G}_{-}(v,w) \right ) \vee
    \left ( Zw \wedge C^{G}_{+}(v,w) \right ).
\end{eqnarray*}
Note that~$\varphi(Z)$ is not a formula in~$\Pi_2$, but can easily be
transformed to
an equivalent formula in~$\Pi_2$.

We show that the problem $\AkEFDform{(\tau,\tau')}{\varphi}$ is \AkE{}-hard,
by giving an fpt-reduction from $\AleqkEWSat(\CNF)$.
Let~$(\varphi',k)$ be an instance of $\AleqkEWSat(\CNF)$,
where~$\varphi' = \forall Y. \exists X. \psi'$,~$Y = \SBs y_1,\dotsc,y_n
\SEs$,~$X = \SBs x_1,\dotsc,x_m \SEs$,
and~$\psi' = c_1 \wedge \dotsm \wedge c_u$.
We construct an instance~$(\mathcal{A},k)$
of $\AkEFDform{(\tau,\tau')}{\varphi}$.
We define the structure~$\mathcal{A}$ with domain~$A$ as follows:
\[ \begin{array}{r l}
  A = & X \cup Y \cup \SBs c_1,\dotsc,c_u \SEs; \\
  G^{\mathcal{A}} = & X; \\
  C^{\mathcal{A}} = & \SBs c_1,\dotsc,c_u \SEs; \\
  (C^{F}_{-})^{\mathcal{A}} = & \SB (c_i,y) \SM 1 \leq i \leq u, y \in Y, \neg y \in c_i \SE; \\
  (C^{F}_{+})^{\mathcal{A}} = & \SB (c_i,y) \SM 1 \leq i \leq u, y \in Y, y \in c_i \SE; \\
  (C^{G}_{-})^{\mathcal{A}} = & \SB (c_i,x) \SM 1 \leq i \leq u, x \in X, \neg x \in c_i \SE; \mtext{ and} \\
  (C^{G}_{+})^{\mathcal{A}} = & \SB (c_i,x) \SM 1 \leq i \leq u, x \in X, x \in c_i \SE.
\end{array} \]
We show that~$(\varphi',k) \in \AleqkEWSat(\CNF)$ if and only
if~$(\mathcal{A},k) \in \AkEFDform{(\tau,\tau')}{\varphi}$.

$(\Rightarrow)$
Assume that~$(\varphi',k) \in \AleqkEWSat(\CNF)$.
This means that for any assignment~$\alpha : Y \rightarrow \SBs 0,1 \SEs$
of weight at most~$k$, the formula~$\exists X. \psi'[\alpha]$ evaluates to true.
We show that~$(\mathcal{A},k) \in \AkEFDform{(\tau,\tau')}{\varphi}$.
Let~$\mathcal{A}'$ be an arbitrary $\tau$-structure
extending~$\mathcal{A}$ with weight~$k$.
Then we know that~$\Card{F^{\mathcal{A}'} \cap Y} \leq k$.
Consider the assignment~$\alpha_{\mathcal{A}'} : Y \rightarrow \SBs 0,1 \SEs$,
where~$\alpha_{\mathcal{A}'}(y) = 1$ if and only if~$y \in F^{\mathcal{A}'} \cap Y$.
Clearly,~$\alpha_{\mathcal{A}'}$ has weight at most~$k$.
Therefore,~$\exists X. \psi'[\alpha_{\mathcal{A}'}]$ evaluates to true.
This means that there exists an assignment~$\beta : X \rightarrow \SBs 0,1 \SEs$
such that~$\psi'[\alpha_{\mathcal{A}'} \cup \beta]$ evaluates to true.
Consider the relation~$W = \SB x \in X \SM \beta(x) = 1 \SE \subseteq A$.

We show that~$\mathcal{A}' \models \varphi(W)$.
Clearly,~$\mathcal{A}' \models \forall v. (Wv \rightarrow Gv)$.
To see that~$\mathcal{A}' \models \forall v.
(Cv \rightarrow (\exists w. \psi(v,w,W)))$,
consider the assignment~$\SBs v \mapsto c_i \SEs$ for an
arbitrary~$1 \leq i \leq u$.
We show that~$\mathcal{A}' \models \exists w. \psi(w,c_i,W)$.
We know that~$\psi'[\alpha_{\mathcal{A}'} \cup \beta]$ evaluates to true,
so therefore there must exist a literal~$l \in c_i$ such
that~$\alpha_{\mathcal{A}'} \cup \beta$ satisfies~$l$.
Assume~$l = x \in X$.
Then~$x \in W$, and therefore~$\mathcal{A}' \models (Wx \wedge C^{G}_{+}(c_i,x))$.
Assume~$l = y \in Y$.
Then~$\alpha_{\mathcal{A}'}(y) = 1$, and thus~$y \in F^{\mathcal{A}'}$.
Therefore,~$\mathcal{A}' \models (Fy \wedge C^{F}_{+}(c_i,y))$
The cases where~$l = \neg z$ for some~$z \in X \cup Y$ are analogous.
This concludes our proof that~$\mathcal{A}' \models \varphi(W)$,
and so we can conclude
that~$(\mathcal{A},k) \in \AkEFDform{(\tau,\tau')}{\varphi}$.

$(\Leftarrow)$
Assume that~$(\mathcal{A},k) \in \AkEFDform{(\tau,\tau')}{\varphi}$.
This means that for each $\tau$-structure~$\mathcal{A}'$ that
extends~$\mathcal{A}$
with weight~$k$, there exists a relation~$W \subseteq A$
such that~$\mathcal{A}' \models \varphi(W)$.
Let~$\alpha : Y \rightarrow \SBs 0,1 \SEs$ be an arbitrary assignment of
weight~$\ell \leq k$.
We show that~$\exists X. \psi'[\alpha]$ evaluates to true.
Define the $\tau$-structure~$\mathcal{A}_{\alpha}$
by letting~$F^{\mathcal{A}_{\alpha}} = \SB y \in Y \SM \alpha(y) = 1 \SE
\cup \SBs y_1,\dotsc,y_{k-\ell} \SEs$.
Assuming without loss of generality that~$\Card{Y} \geq \Card{X}$,
it follows that~$\Card{F^{\mathcal{A}_{\alpha}}} = k$.
Thus,~$\mathcal{A}_{\alpha}$ extends~$\mathcal{A}$ with weight~$k$.
We therefore know that there exists a relation~$W \subseteq A$
such that~$\mathcal{A}_{\alpha} \models \varphi(W)$.
Now define~$\beta_{W} : X \rightarrow \SBs 0,1 \SEs$
by letting~$\beta_{W}(x) = 1$ if and only if~$x \in W$.

We show that~$\psi'[\alpha \cup \beta_{W}]$ evaluates to true.
Let~$c_i$ be an arbitrary clause of~$\psi'$.
Since~$\mathcal{A}_{\alpha} \models \forall v.
(Cv \rightarrow (\exists w. \psi(v,w,W)))$,
we know that there must be some~$a \in A$ such
that~$\mathcal{A}_{\alpha} \models \psi(c_i,a,W)$.
It is straightforward to verify that~$a \in X \cup Y$.

We distinguish four cases.
Consider the case where~$a \in X$
and~$a \in (C^{G}_{+})^{\mathcal{A}_{\alpha}}$.
Then~$\mathcal{A}_{\alpha} \models Wa$, and thus~$a \in W$.
Also, since~$a \in (C^{G}_{+})^{\mathcal{A}_{\alpha}}$,
$a$ is a positive literal in the clause~$c_i$.
Therefore,~$\beta_{W}(a) = 1$, and thus~$\alpha \cup \beta_{W}$ satisfies~$c_i$.
Consider the case where~$a \in Y$
and~$a \in (C^{F}_{+})^{\mathcal{A}_{\alpha}}$.
Then~$\mathcal{A}_{\alpha} \models Fa$, and thus~$\alpha(y) = 1$.
Also, since~$a \in (C^{F}_{+})^{\mathcal{A}_{\alpha}}$,~$a$
is a positive literal in the clause~$c_i$.
Therefore,~$\alpha(a) = 1$, and thus~$\alpha \cup \beta_{W}$ satisfies~$c_i$.
The cases where~$a \in (C^{G}_{-})^{\mathcal{A}_{\alpha}}$
and where~$a \in (C^{F}_{-})^{\mathcal{A}_{\alpha}}$ are analogous.

Since~$c_i$ was an arbitrary clause of~$\psi'$,
we know that~$\alpha \cup \beta_{W}$ satisfies~$\psi'$.
This concludes our proof that~$(\varphi',k) \in \AleqkEWSat(\CNF)$.
\end{proof}

\begin{proposition}
\label{prop:aek-fd-hardness}
The problem $\AEkFDform{(\tau,\tau')}{\varphi}$ is \AEkW{1}-hard,
for some relational vocabulary~$\tau$,
some subvocabulary~$\tau' \subseteq \tau$ of~$\tau$,
and some first-order formula~$\varphi(Z) \in \Pi_2$ over~$\tau$
with a free relation variable~$Z$ of arity~$s$.
\end{proposition}
\begin{proof}%[Proof (sketch).]
We show \AEkW{1}-hardness by giving an fpt-reduction
from $\AEkWSat(2\CNF)$.
Let~$(\varphi',k)$ be an instance of $\AEkWSat(2\CNF)$,
where~$\varphi' = \forall Y. \exists X. \psi'$,~$Y =
\SBs y_1,\dotsc,y_n \SEs$,~$X = \SBs x_1,\dotsc,x_m \SEs$,
and~$\psi' = c_1 \wedge \dotsm \wedge c_u$.
We use the same construction as the one used in the
fpt-reduction from $\AleqkEWSat(\CNF)$ to $\AEkFDform{(\tau,\tau')}{\varphi}$,
in the proof of Proposition~\ref{prop:ake-fd-hardness}
to construct an instance~$(\mathcal{A},k)$ of $\AEkFDform{(\tau,\tau')}{\varphi}$. 
This is possible since instance of $\AEkWSat(2\CNF)$ are
also instances of $\AleqkEWSat(\CNF)$.
We show that this reduction is also a correct fpt-reduction
from $\AEkWSat(2\CNF)$ to $\AEkFDform{(\tau,\tau')}{\varphi}$.

$(\Rightarrow)$
Assume that~$(\varphi',k) \in \AEkWSat(2\CNF)$.
This means that for any assignment~$\alpha : X \rightarrow \SBs 0,1 \SEs$,
there exists an assignment~$\beta : Y \rightarrow \SBs 0,1 \SEs$ of weight~$k$
such that~$\psi'[\alpha \cup \beta]$ evaluates to true.
We show that~$(\mathcal{A}',k) \in \AEkFDform{(\tau,\tau')}{\varphi}$.
Let~$\mathcal{A}'$ be an arbitrary $\tau$-structure extending~$\mathcal{A}$.
Consider the assignment~$\alpha_{\mathcal{A}'} : Y \rightarrow \SBs 0,1 \SEs$,
where~$\alpha_{\mathcal{A}'}(y) = 1$ if and only if~$y \in F^{\mathcal{A}'} \cap Y$.
There must be an assignment~$\beta : X \rightarrow \SBs 0,1 \SEs$
such that~$\psi'[\alpha_{\mathcal{A}'} \cup \beta]$ evaluates to true.
Consider the relation~$W = \SB x \in X \SM \beta(x) = 1 \SE \subseteq A$.
Clearly,~$\Card{W} = k$.

We show that~$\mathcal{A}' \models \varphi(W)$.
Clearly~$\mathcal{A}' \models \forall v.(Wv \rightarrow Gv)$.
To see that~$\mathcal{A}' \models \forall v. (Cv \rightarrow (\exists w. \psi(v,w,W)))$,
consider the assignment~$\SBs v \mapsto c_i \SEs$ for an arbitrary~$1 \leq i \leq u$.
We show that~$\mathcal{A}' \models \exists w. \psi(w,c_i,W)$.
%x
We know that~$\psi'[\alpha_{\mathcal{A}'} \cup \beta]$ evaluates to true,
so therefore there must exist a literal~$l \in c_i$ such
that~$\alpha_{\mathcal{A}'} \cup \beta$ satisfies~$l$.
Let~$z$ be the variable of the literal~$l$.
Verifying that~$\mathcal{A}' \models \psi(z,c_i,W)$
is analogous to the proof of Proposition~\ref{prop:ake-fd-hardness}.
This completes our proof that~$\mathcal{A}' \models \varphi(W)$,
and we can thus conclude that~$(\mathcal{A}',k) \in \AEkFDform{(\tau,\tau')}{\varphi}$.

$(\Leftarrow)$
Assume that~$(\mathcal{A},k) \in \AEkFDform{(\tau,\tau')}{\varphi}$.
This means that for each $\tau$-structure~$\mathcal{A}'$ that
extends~$\mathcal{A}$,
there exists a relation~$W \subseteq A$ with~$\Card{W} = k$
such that~$\mathcal{A}' \models \varphi(W)$.
Let~$\alpha : Y \rightarrow \SBs 0,1 \SEs$ be an arbitrary assignment.
We show that there exists an assignment~$\beta : X \rightarrow \SBs 0,1 \SEs$
of weight~$k$ such that~$\psi'[\alpha \cup \beta]$ evaluates to true.
Define the $\tau$-structure~$\mathcal{A}_{\alpha}$
by letting~$F^{\mathcal{A}_{\alpha}} = \SB y \in Y \SM \alpha(y) = 1 \SE$.
Clearly,~$\mathcal{A}_{\alpha}$ extends~$\mathcal{A}$.
We therefore know that there exists a relation~$W \subseteq A$
with~$\Card{W} = k$
such that~$\mathcal{A}_{\alpha} \models \varphi(W)$.
Since~$\mathcal{A}_{\alpha} \models \forall v. (Wv \rightarrow Gv)$,
we know that~$W \subseteq X$.
Now define~$\beta_{W} : X \rightarrow \SBs 0,1 \SEs$
by letting~$\beta_{W}(x) = 1$ if and only if~$x \in W$.
We know that~$\beta_{W}$ has weight~$k$.

We show that~$\psi'[\alpha \cup \beta_{W}]$ evaluates to true.
Let~$c_i$ be an arbitrary clause of~$\psi'$.
Since~$\mathcal{A}_{\alpha} \models \forall v. (Cv \rightarrow (\exists w. \psi(v,w,W)))$,
we know that there must be some~$a \in A$ such
that~$\mathcal{A}_{\alpha} \models \psi(c_i,a,W)$.
It is straightforward to verify that~$a \in X \cup Y$.
Verifying that~$\mathcal{A}_{\alpha} \models \psi(c_i,a,W)$
implies that~$\alpha \cup \beta_{W}$ satisfies~$c_i$
is analogous to the proof of Proposition~\ref{prop:ake-fd-hardness}.
This completes our proof that~$\psi'[\alpha \cup \beta_{W}]$ evaluates to true,
and we can thus conclude that~$(\varphi',k) \in \AEkWSat(2\CNF)$.
\end{proof}

%%% /APPENDIX
}
}

% References section
\krversion{
  %\newpage{}
  \bibliographystyle{aaai}
}
\longversion{
  \pagebreak{}
  \bibliographystyle{plain}
}
\krlongversion{
  \DeclareRobustCommand{\DE}[3]{#3}
  \krversion{
    \begin{small}
      \bibliography{literature}
    \end{small}
  }\longversion{
  
%    \bibliography{literature}

  }
}

\end{document}